\documentclass[aps,prx,twocolumn,superscriptaddress,nofootinbib]{revtex4}
\pdfoutput=1

\usepackage{footmisc}
\usepackage{hyperref}
\usepackage{mathtools}
\usepackage{amsmath,amssymb,graphicx,xcolor,braket,hyphenat,makeidx,amsthm}
\usepackage{bbm}
\usepackage{relsize}
\usepackage{comment}
\usepackage{appendix}
\usepackage{thmtools,thm-restate}

\theoremstyle{theorem}

\newtheorem{definition}{Definition}
\newtheorem{theorem}{Theorem}
\newtheorem{theorem*}{Theorem}
\newtheorem{corollary}{Corollary}
\newtheorem{corollary*}{Corollary}
\newtheorem{lemma}{Lemma}
\newtheorem{remark}{Remark}

\usepackage{lipsum}

\newcommand\blfootnote[1]{%
	\begingroup
	\renewcommand\thefootnote{}\footnote{#1}%
	\addtocounter{footnote}{-1}%
	\endgroup
}

\usepackage{dsfont,bbm}

\newcommand{\ketbra}[2]{\ket{#1}\!\!\bra{#2}}

\newcommand{\tr}{\textup{tr}}
\newcommand{\be}{\begin{equation}}
\newcommand{\ee}{\end{equation}}
\newcommand{\nn}{{\mathbbm{N}}}
\newcommand{\nnp}{{\mathbbm{N}_{> 0}}}
\newcommand{\nno}{{\mathbbm{N}_{\geq 0}}}
\newcommand{\rrp}{{\mathbbm{R}_{> 0}}}

\newcommand{\hh}{{\mathbbm{H}}}
\newcommand{\rr}{{\mathbbm{R}}}
\newcommand{\cc}{{\mathbbm{C}}}
\newcommand{\zz}{{\mathbbm{Z}}}
\newcommand{\me}{\mathrm{e}}
\newcommand{\mi}{\mathrm{i}}
\newcommand{\id}{{\mathbbm{1}}}
\newcommand{\bo}{\mathcal{O}}
\newcommand{\lo}{o}

\newcommand{\Mspace}{\vspace{0.2cm}}
\newcommand{\MspaceFig}{\vspace{0.15cm}}

\def\ba#1\ea{\begin{align}#1\end{align}}

\newcommand{\doc}{\text{manuscript}}  

\newcommand{\app}{\text{supplemental}}
\newcommand{\wso}{\text{Quasi-Ideal}}
\newcommand{\Wso}{\text{Quasi-Ideal}}

\newcommand{\cl}{\text{C}}

\newcommand\mpwST[1]{MW:{\let\helpcmd\sout\parhelp#1\par\relax\relax} }
\long\def\parhelp#1\par#2\relax{%
	\helpcmd{#1}\ifx\relax#2\else\par\parhelp#2\relax\fi%
}

\usepackage[normalem]{ulem}

\newcommand{\proj}[1]{\ketbra{#1}{#1}}

\newcommand*{\cC}{\mathcal{C}}
\newcommand*{\cI}{\mathcal{I}}
\newcommand*{\cM}{\mathcal{M}}
\newcommand*{\cN}{\mathcal{N}}
\newcommand*{\cP}{\mathcal{P}}
\newcommand*{\cT}{\mathcal{T}}

\newcommand{\tickmatrix}{\mathcal{T}}
\newcommand{\notickmatrix}{\mathcal{N}}
\newcommand{\conv}{\; * \;}
\newcommand{\onenorm}[1]{\left| \! \left| #1 \right| \! \right|_\Sigma}
\newcommand{\mat}[1]{\left[ #1 \right]}

\newcommand{\green}{\color{green}}
\newcommand*{\RR}[1]{{\green [RR: #1]}}

\begin{document}
\title{Quantum clocks are more precise than classical ones 
}


\author{Mischa P. Woods$^*$}
\affiliation{Institute for Theoretical Physics, ETH Zurich, Switzerland}
\affiliation{Department of Computer Science, University College London, UK}
\author{Ralph Silva$^*$}
\affiliation{Institute for Theoretical Physics, ETH Zurich, Switzerland}
\affiliation{D\'epartement de Physique Appliqu\'ee, Universit\'e de Gen\`eve, Switzerland}
\author{Gilles P\"utz}
\affiliation{Institute for Theoretical Physics, ETH Zurich, Switzerland}
\author{Sandra Stupar}
\affiliation{Institute for Theoretical Physics, ETH Zurich, Switzerland}
\author{Renato Renner}
\affiliation{Institute for Theoretical Physics, ETH Zurich, Switzerland}

\begin{abstract}
  A \emph{clock} is, from an information-theoretic perspective, a system that emits information about time. One may therefore ask whether the theory of information imposes any constraints on the maximum precision of clocks. Here we show a quantum-over-classical advantage for clocks or, more precisely, the task of generating information about what time it is. The argument is based on information-theoretic considerations: we analyse how the precision of a clock scales with its size, measured in terms of the number of bits that could be stored in it. We find that a quantum clock can achieve a quadratically improved precision compared to a purely classical one of the same size.
\end{abstract}
\maketitle

\section{Introduction}\label{Introduction}

\blfootnote{$^*$M.W. and R.S. contributed equally to the results.} Timekeeping is one of the oldest ways in which humanity has organised its activities, dating back to ancient civilisations that observed the solar cycles. Eventually, we invented our own devices to mark the passage of time, and the advancements in these \emph{clocks} allowed for revolutionary capabilities such as maritime navigation, and enabled the industrial revolution. The best clocks today are very sophisticated and need a quantum description to understand how they work \cite{RevModPhys.87.637}. The next generation of quantum clocks will enable new applications, such as faster telecommunications, non-satellite based GPS systems, and also foster advances in fundamental physics, e.g., in the context of gravitational wave detection \cite{qrevolution}. 


However, quantum theory suggests that there is a limit to the maximum precision of clocks. In contrast to position, momentum and energy, time cannot be made into an ``ideal observable", that is to say, one whose outcomes deterministically determine time without error \cite{pauli1, pauli2,PauliGeneralPrinciples,holevo2011probabilistic}. Furthermore, a clock must not only evolve with time, but also \textit{emit} information about its state to the outside world \cite{RaLiRe15}, like in the case of a ticking watch, or bell tower. It is thus vulnerable to the disturbance inherent to any quantum measurement \cite{FuchsPeres96}; as can be seen in the settings of autonomous quantum control~\cite{WSO16} and thermodynamics~\cite{Pauletal2017}. 

So we currently find ourselves at an interesting juncture:  on the one hand, clocks are increasingly more precise | and just as pendulum clocks enabled the industrial revolution, the next generation of atomic clocks will do the same for a new technological age. However, on the other hand, quantum mechanics suggests that there must be a limit to their increasing precision. As an analogy, consider the birth of thermodynamics in the late 18th century: even as heat engines were developed and improved upon, Clausius, Carnot and others found fundamental limits to their efficiency by relating it to temperature and heat. In the case of clocks a natural question is thus: \emph{Can we relate their precision to physical variables such as entropy}
\emph{, energy, size, or information contents, and by doing so, quantify the fundamental limits to their precision?}

To clarify what is meant by ``a clock" in this work, we distinguish between two types of devices for measuring time: timepieces that output time information on request, like a stopwatch, and clocks that output time information autonomously, like a chiming clock. They serve different purposes. Stop watches are used to measure a time interval between events triggered by  external processes (e.g., between the event that a train leaves the station at~$A$ and the event that it arrives at~$B$) \cite{OptimalStopwatch}. Conversely, chiming clocks ``generate'' events themselves, which may then be used to trigger external events (e.g., that the train leaves the station at~$A$),  see Fig.~\ref{fig:clockcartoon1}.
\Mspace 

This work is concerned with the second type of time-keeping. Hence, from now on (and with the exception of the review of earlier work at the beginning of Section~\ref{sec:Modeling clocks}) we use the term \emph{clock} for devices that output information about time autonomously.\footnote{The word ``clock'' derives from the Medieval Latin ``clocca'', which means ``bell''. The hourly ringing of the bells may be regarded as an autonomous process.}  Specifically, we take a clock to be a device that generates a sequence of individual events, which we call \emph{ticks}. For the purpose of this discussion, we assume that the ticks are the only information output by the clock. \Mspace

We investigate the effect of the size of a clock, motivated by the general observation that the disturbance suffered by large mechanical clocks by the act of reading-off time appears insignificant, while tiny clocks are more prone to be disturbed. There are a number of ways to quantify the size of a clock, e.g. by its mass \cite{SaleckerWigner58}. We take an information-theoretic approach, and consider the size of the state space of the clock, which is the number of perfectly distinguishable states that it can be in, or alternatively, the dimension $d$ of its associated Hilbert space. Indeed, a clock of size~$d$ is a clock that could in principle store at most $\log_2 d$ bits of information in its internal state, and thus $d$ is a measure of its information contents. In the context of stopwatches, bounds on the precision given a bound on the size were derived in \cite{OptimalStopwatch} (also see reviews \cite{TiQMVol1,TiQMVol2} for related references).

Moreover, it is interesting to ask whether quantum features in clocks could provide an advantage. In order to make a comparison, one can introduce the notion of a classical clock as a quantum clock which has lost its quantum properties through decoherence.

This \doc~proves a fundamental connection between the size of a quantum or classical clock and its attainable precision. Namely we find that there exist quantum mechanical clocks based on \cite{WSO16}, whose precision represents a quadratic improvement over the best classical clocks of the same size.

\begin{figure}[!htb]
	\includegraphics
	[scale=0.28]{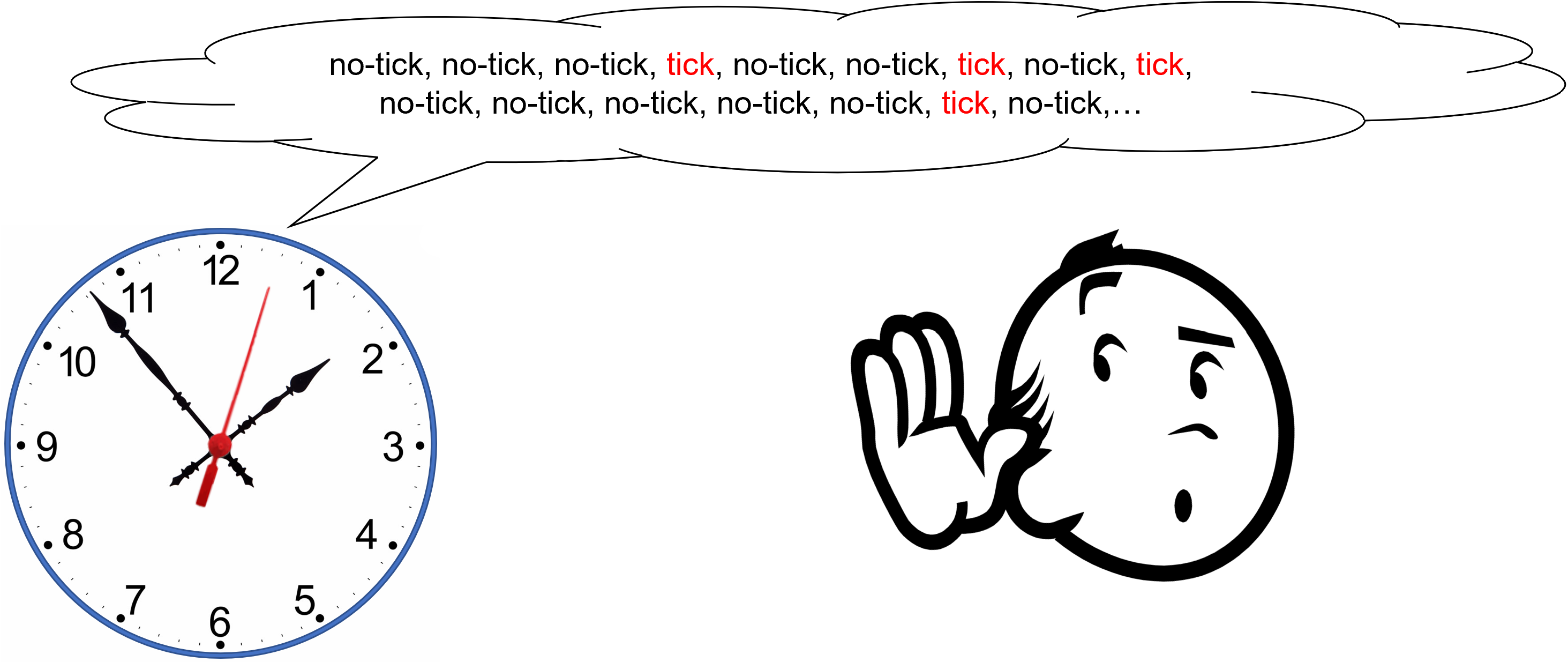}
	\caption{\footnotesize{Illustration of a (chiming) clock. A (chiming) clock produces a continuous stream of ``no-tick", ``tick" information. The ``no-ticks" represent the silence between ticks. Note that the silence between ticks is just as important as the ticks themselves for the functioning of a clock. To illustrate how (chiming) clocks and stopwatches serve different purposes, consider two examples: case a) Two people agree to meet at a given time in the future, say 38 hours from now. If one has a (chiming) clock, one can simply count the emitted ticks until the time interval of 38 hours has passed. However, in contrast, if one has a stopwatch there is no external trigger one can use and measuring the stopwatch at a random (unknown) time, would result in a measure zero probability of getting (even approximately) the outcome of 38 hours later. Case b) A race takes place and one wants to know the time of the winner. Here there is an external trigger (the winner crossing the finish line), and a stopwatch will suffice. In case b) a (chiming) clock will also suffice of course. So a (chiming) clock may be used to replace a stopwatch but not vice versa.
			\label{fig:clockcartoon1}}}
\end{figure}

The precision of a clock can then be defined via the regularity of its ticks. We ask the simple operational question: \textit{How many ticks can a clock output before the uncertainty in its next tick has grown to be as large as the interval between ticks?} This precision measure, introduced in~\cite{Pauletal2017}, is referred to as $R$ (Section \ref{Accuracy of Quantum Clocks}).

We use the term \emph{quantum clock} for a clock whose dynamics is not subject to any constraints other than those imposed by quantum theory. Their internal state can therefore be represented by a density operator in a $d$-dimensional Hilbert space, and the transition from the clock's state at a time $t_1$ to its state at time $t_2$  corresponds to a trace-preserving completely positive map. We also consider the special case of \emph{classical clocks}, where decoherence is assumed to occur on a timescale that is much shorter than the processes responsible for the generation of ticks. Their state space is therefore restricted to a fixed set consisting of $d$ perfectly distinguishable states and their probabilistic mixtures | the ``classical'' states. In this case, a state transition from time $t_1$ to $t_2$ is most generally represented by a stochastic map. \Mspace

Our main results are bounds on the precision which depend on the clock's size~$d$.  On the one hand, we prove that, for any fixed $\eta > 0$, there exist quantum clocks 
 whose precision scales as
\begin{align} \label{eq_quantumlowerbound}
  R_{\mathrm{quantum}} \gtrsim d^{2-\eta}. 
\end{align}
That is, quantum clocks can have a precision that grows essentially quadratically in the clock's size for large $d$. We prove this statement by construction, showing that the so-called \wso~clocks proposed in~\cite{WSO16} can achieve this scaling under the appropriate circumstances. On the other hand, we prove that the precision of any classical clock is upper bounded by
\begin{align} \label{eq_classicalupperbound}
  R_{\mathrm{classical}} \leq d,
\end{align}
and show that a simple stochastic clock, studied in \cite{ATGRandomWalk} in the context of the Alternate Ticks Game, saturates this bound. Combining Eqs.~\eqref{eq_quantumlowerbound} and~\eqref{eq_classicalupperbound}, we conclude that for large size~$d$, quantum clocks outperform classical ones quadratically in terms of their precision~$R$.\Mspace

\section{Modelling clocks}\label{sec:Modeling clocks}

To motivate our framework for describing clocks, we first have a look at existing models that have been considered in the literature and discuss their features and limitations. (An extensive review on prior literature regarding clocks and the general issue of time in quantum mechanics can be found in \cite{TiQMVol1,TiQMVol2}.) \Mspace

Pauli regarded an ``ideal clock'' as a device that has an observable $T$ whose value is in one-to-one correspondence to the time parameter~$t$ in the quantum-mechanical equation of motion. The observable $T$ would need to satisfy $\frac{d}{dt}T = \id$. Furthermore, since neither $T$ nor the Hamiltonian of the system, $H$, should depend on time explicitly, they would  need to satisfy the commutation relation $i [H, T] = \id$.\footnote{We set $\hbar =1$, so that $\frac{i}{\hbar} [H, T] = i [H,T] = \id$.} Pauli then argued that this implies that $H$ has as its spectrum the full real line \cite{Pashby14}. 
Since such Hamiltonians are unphysical, he concluded that an observable $T$ with the desired properties, and hence an ideal clock, cannot exist~\cite{pauli1,PauliGeneralPrinciples}.\footnote{We note that this conclusion has been challenged and it has been argued that the relation $i [H, T] = \id$ can be satisfied for Hamiltonians $H$ with semi-bounded spectrum if one considers operators with restricted domains of definition (see~\cite{Pashby14} for a discussion). Such restrictions however still correspond to unphysical assumptions, such as infinite potentials to keep a particle in a confined region.}  As such, these objects are referred to as \emph{Idealised clocks}. 
\Mspace

This raises the question whether one can at least approximate an Idealised clock. Salecker and Wigner~\cite{SaleckerWigner58} and Peres~\cite{Peres80} considered finite-dimensional constructions. Specifically, they showed that for any dimension~$d$ and for any fixed time interval~$\Delta$ there exists a clock, which we will refer to as the \emph{SWP clock}, whose Hamiltonian satisfies
\begin{align*}
  \forall\, k \in \{0, \ldots, d-1\} : \quad e^{i H \Delta} \ket{\theta_k} = \ket{\theta_{k+1 \, (\mathrm{mod} \, d)}}
\end{align*} 
where $\smash{\{\ket{\theta_i}\}_{i=0}^{d-1}}$ is the \emph{SWP basis} | an orthonormal basis of the clock's Hilbert space. Hence, if the clock was initialised to state $\ket{\theta_0}$ and if one did read the clock at a time $t \in \{0, \Delta, 2 \Delta, \ldots \}$ by applying a projective measurement with respect to the SWP basis, the outcome would be precise information about time $n \, (\mathrm{mod} \, d)$. However, in between these particular points in time, the amplitudes of the basis states are in general all non-zero~\cite{Gross2012}. Hence, if the clock was measured, say, at  $t = \frac{5}{2} \Delta$, the outcome would be uncertain\footnote{At intermediate time intervals, the variance of the state w.r.t. the basis states $\ket{\theta_0}$ is as much as $\sqrt{d}$.}. In addition, such a measurement would disturb the clock's state, effectively resetting it to a random time. This problem was resolved in recent work by some of us, with the introduction of the so-called \emph{\wso~clock} \cite{WSO16}, which is able to approximate the dynamical behaviour of Pauli's Idealised clock while maintaining a finite dimension. Another approach to time operators for clocks, is to consider covariant time observables~(see e.g.,~\cite{BUSCH1994357}) that are unsharp. We will not discuss these here, since they do not bear upon the question of precision.
\Mspace

The constructions from~\cite{SaleckerWigner58,Peres80} do however not include a mechanism to output time information autonomously. Hence, to use the terminology introduced earlier, they are stopwatches rather than chiming clocks. To extract time information from them, one would have to apply measurements from the outside. But then the outcome depends on when and how these measurements are performed. Thus, in order to reasonably talk about their precision | in terms of operationally motivated quantities | we need a more complete description. In \cite{WSO16}, a potential term was added to the Hamiltonian. In the case that this potential is pure imaginary, it will allow us to model information about time being extracted autonomously. This feature, together with the definition of quantum clocks as outlined in the following section, will allow for the precision of quantum clocks to be bounded. 

\subsection{Quantum Clocks}\label{Quantum Clocks}

The modelling of clocks that we use here follows the operational approach introduced in~\cite{RaLiRe15} with some adjustments. 
We now explain this setup in detail.

 A $d$-dimensional quantum clock consists of a (generally open) quantum system~$C$ which we call the \emph{clockwork}. 
 The transition of a clockwork's state $\rho_{C,t}$ at some time $t$ to its state $\rho_{C, t+\Delta}$ at a later time $t + \Delta$ can hence most generally be described by a trace-preserving completely positive map
\begin{align*}
  \cM^{\Delta}_{C \to C} : \quad \rho_{C, t} \mapsto \rho_{C, t+\Delta}  ,
\end{align*}
which depends on~$\Delta \in \rr_{\geq 0}$ but not on $t \in \rr$.  Note that these maps form a family parameterised by $\Delta \in \rr_{\geq 0}$. For the particular choice $\Delta = 0$ it is the identity map,
\begin{align} \label{eq_identity}
  \cM^{(0)}_{C \to C} = \cI_C  .
\end{align}
Furthermore, the maps are mutually commutative under composition, i.e., 
\begin{align} \label{eq_composition}
  \cM^{\Delta + \Delta'}_{C \to C} = \cM^{\Delta'}_{C \to C} \circ  \cM^{\Delta}_{C \to C} =  \cM^{\Delta}_{C \to C} \circ  \cM^{\Delta'}_{C \to C}  ,
\end{align}
for any $\Delta, \Delta' \in \rr_{\geq 0}$. In other words, the evolution of $C$  is determined by a one-parameter family of maps, $\{\cM^{\Delta}_{C \to C}\}_{\Delta \in \rr_{\geq 0}}$, and which are Markovian. The Markovianity assumption is necessary, otherwise the generators of the dynamics could change at regular intervals, providing an unaccounted timing resource for the clock. \Mspace

Assuming that the energy that drives the clockwork's evolution is finite, we may additionally assume that the clockwork's state changes at a finite speed. This means that the function $\Delta \mapsto \smash{\cM_{C \to C}^{\Delta}}$ is continuous. But, using Eqs.~\eqref{eq_identity} and~\eqref{eq_composition}, this is in turn equivalent to the requirement that
\begin{align} \label{eq_continuity}
  \lim_{\Delta \to 0} \cM^{\Delta}_{C \to C} = \cI_{C}  ,
\end{align}
which may be regarded as a strengthening of Eq.~\eqref{eq_identity}. \Mspace

Since we assumed that the clockwork's evolution is time-independent, its description in terms of the entire family $\smash{\cM^{\Delta}_{C \to C}}$, for $\Delta \in \rr_{\geq 0}$, is highly redundant. Indeed, using Eq.~\eqref{eq_composition} we may write
\begin{align} \label{eq_mapDeltadelta}
  \cM^{\Delta}_{C \to C} = \lim_{\delta \to 0} \bigl(\cM^{\delta}_{C \to C}\bigr)^{\lfloor \frac{\Delta}{\delta} \rfloor},
\end{align}
where we have used the notation
\begin{align}
(\cM^{\delta}_{C \to C})^k = \underbrace{\cM^{\delta}_{C \to C} \circ \cdots \circ \cM^{\delta}_{C \to C}}_{\text{$k$ times}}.
\end{align}
It thus suffices to specify the evolution map for arbitrarily small time parameters, which we will in the following denote by~$\delta$. (The evolution is thus governed by the Lindblad equation, a fact that we will exploit in Section~\ref{sec_generators}).\Mspace

\begin{figure}[t]
	\mbox{}\hspace{-1.9mm}\includegraphics[trim= 6.5cm 4.6cm  8.8cm 4.4cm, clip=true, scale=0.35]
	{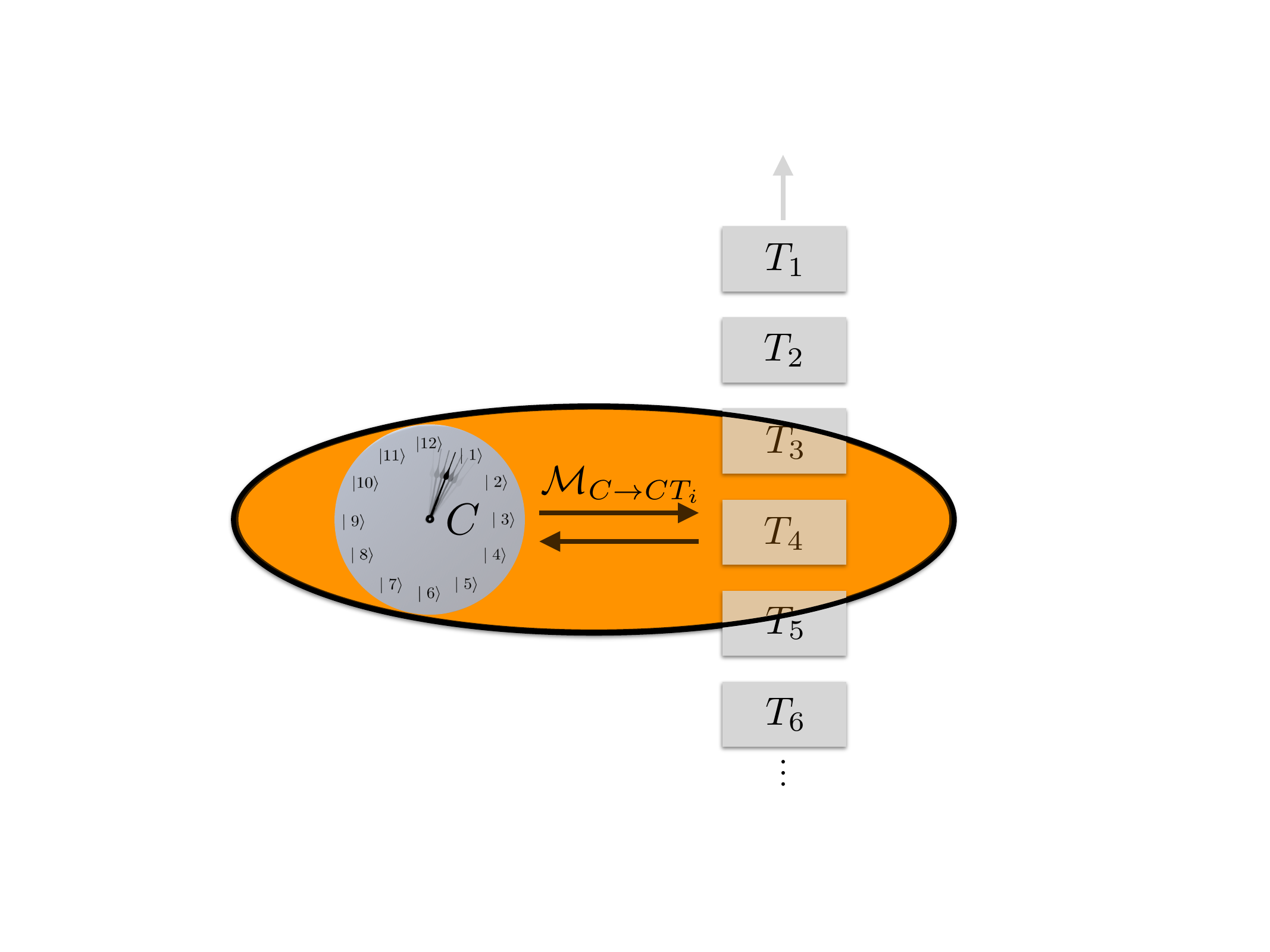}
	\caption{
		\footnotesize{\textbf{Information-theoretic clock model.} A clock is modelled as a device that emits information about time. This flow of information is described by small discrete time steps $\delta$, specified by maps $\cM_{C \to C T_i}^\delta$ \cite{RaLiRe15}. In each time step $\delta$ the clockwork $C$ interacts with a \emph{tick register} $T_i$, $i\in\nn$. The latter models the part of the information that the clock emits to the outside and hence is unavailable to the clock at any later step. Since we are considering the case in which time is continuous, in our analysis we take the continuous limit in which each time step is infinitesimal. Conversely, $C$ must include all information that remains accessible to the clock over more than one time step. The emission of information into the registers $\{T_i\}_i$ induces quantum back-action on~$C$, and thus a degradation of the information it contains.}}
	\label{fig:clockcartoon}
\end{figure}

The maps $\cM^{\delta}_{C \to C}$ describe the evolution of the state on~$C$. But, as argued above, we are generally interested in the information that the clock transmits to the outside. This can be included in our description by virtue of extensions of the maps $\smash{\cM^{\delta}_{C \to C}}$. That is, we consider maps $\smash{\cM^{\delta}_{C \to C T_i}}$ whose target space is a composite system, consisting of $C$ and an additional system $T_i$, such that
\begin{align} \label{eq_extension}
  \cM^{\delta}_{C \to C} = \tr_{T_i} \circ \cM^{\delta}_{C \to C T_i}  .
\end{align}
We call~$T_i$ \emph{tick registers}, alluding to the idea that the basic elements of information emitted by a clock are its ``ticks'', see Fig~\ref{fig:clockcartoon}. Note that while the model of a clock considered here involves an unbounded sequence of finite-dimensional registers that carry the time information it generates, one can show, see~\cite{mischa}, that one is able to achieve the precision as reported here for quantum and classical clocks with a single finite register attached to the clock. The model in~\cite{mischa} can achieve this by only utilising a new register state when the clock ticks, in contrast to requiring a new qubit register, $T_i$, for every infinitesimal time step. Therefore, its register only needs to be as large as the number of ticks one wishes to record with it. Furthermore, this alternative model has a master equation description for the entire register and clockwork; and hence, contrary to the model considered here, does not require additional degrees of freedom to account for the alignment of the clockwork with a new register $T_i$ at every infinitesimal time step. 

After these general remarks, we are now ready to state the technical definition. In the following we let all the tick registers $T_i$ be isomorphic to a single tick register denoted $T$. 
\begin{definition}[From \cite{RaLiRe15}] \label{def_quantumclock}
  A \emph{(quantum) clock} is a pair $(\rho_C^0, \{\smash{\cM^{\delta}_{C \to C T}}\}_{\delta})$, consisting of a density operator $\rho^0_C$ on a $d$-dimensional Hilbert space~$C$ together with a family of trace preserving and completely positive maps $\smash{\cM^{\delta}_{C \to C T}}$ from $C$ to $C \otimes T$, where $T$ is an arbitrary system, such that the following limits exist and take on the value
  \begin{align} \label{eq_quantumclockcondition}
     \lim_{\Delta \to 0}  \lim_{\delta \to 0} \bigl(\tr_T \circ \cM^{\delta}_{C \to C T}\bigr)^{\lfloor \frac{\Delta}{\delta} \rfloor} = \cI_{C}  .
  \end{align}
  \end{definition}

Using Eq.~\eqref{eq_mapDeltadelta}, it is easy to see that any family of maps whose reduction to~$C$ satisfies Eqs.~\eqref{eq_composition} and~\eqref{eq_continuity} also satisfies Eq.~\eqref{eq_quantumclockcondition}. The converse is however not necessarily true. Nevertheless, given a family of maps $\{\smash{\cM_{C \to C T}^{\delta}}\}_\delta$ as in Def.~\ref{def_quantumclock}, one may always define a family of maps
\begin{align} \label{eq_limitclock}
 \left\{ \bar{\cM}^{\Delta}_{C \to C} \right\}_\Delta= \left\{ \lim_{\delta \to 0} \bigl(\tr_T \circ \cM^{\delta}_{C \to C T}\bigr)^{\lfloor \frac{\Delta}{\delta} \rfloor}  \right\}_\Delta,
\end{align}
which meet both Eqs.~\eqref{eq_composition} and~\eqref{eq_continuity}. In this sense, specifying a map that satisfies Eq.~\eqref{eq_quantumclockcondition} is indeed sufficient to define the continuous and time-independent evolution of a clock. What is more, one may be concerned about technical issues which can arise when dealing with infinite tensor product spaces. Since in any finite time interval, the clock can only tick finitely many times, at any given moment, the state of the register is a classical ensemble of states, each of which contains infinitely many registers $T_i$ in the ``no-tick" state, and only finitely many in the ``tick'' state. Thus the resulting infinite dimensional tensor product space is well-defined in our case; see Sec. 2.5. of~\cite{weaver2001mathematical}.\Mspace 

The definition does not yet impose any constraints on the tick register, $T$. Since we want to compare different clocks, it will however be convenient to  assume that $T$ contains two designated orthogonal states, $\ket{1}$ and $\ket{0}$, which we interpret as ``tick'' and ``no tick'', respectively.  The idea is that ticks are the most basic units of time information that a clock can emit. Roughly speaking, a tick indicates that a certain time interval has passed since the last tick. \Mspace

To know if the clock has ticked after the application of the map $\cM^\delta_{C \to C T}$, one has to measure the tick register in the ``tick" basis $\{\ket{0},\ket{1}\}$. In general, this represents an additional map on the clockwork and register, as Def. \ref{def_quantumclock} allows for the tick register to be coherent in the tick basis, and even entangled with the clockwork system $C$. However, in this work we are only concerned with the probability distribution of ticks (as we characterise the performance of the clock from this alone), and so we incorporate the additional measurement into the map $\cM^\delta_{C \to C T}$ itself. This is equivalent to requiring the map to restrict the state of the clockwork and tick register to be block-diagonal states in the basis $\{\ket{0},\ket{1}\}$.\Mspace

Furthermore we consider the behaviour of the tick register in the limit $\delta\to 0$. In principle, the probability of a tick in this limit need not be zero. However, such a clock would correspond to one that has a probability of ticking on every application of the map $\cM^\delta_{C \to C T}$ \emph{independently of the state of the clockwork system}, and thus does not provide any information about time.\footnote{More precisely, we could express such a clock via the convex combination of two maps, one that does have a zero tick probability for $\delta \to 0$, and one that does not. The second one would provide no time information, and thus only worsen the performance of the clock.} \Mspace

Following the above considerations, we continue with clocks whose maps $\cM_{C \to C T}^\delta$ provide states on $CT$ that are diagonal in the tick basis $\{\ket{0},\ket{1}\}$, and also satisfy the limit
\begin{align}
	\tr_{C} \left[ \lim_{\delta \to 0} \cM_{C \to C T}^\delta (\rho_C) \right] = \proj{0}_T.
\end{align}

One may feel inclined to think of a clock whose ``ticks" convey additional information, such as the number of previous ticks produced by the clock. For example, often a church bell will produce different chimes to specify the passing of different hours. To treat this within our model, one may think of a (classical) counter, which merely counts the number of tick registers in the state $\ket{1}$. This way, if a tick occurs, one can read the counter and discern the time. Clearly the counter does not need any additional timing devices to function. Importantly, since such a counter only interacts with the tick registers and not the clockwork, it does not directly affect the evolution of the clockwork system $C$.\Mspace

This concludes our discussion of the generic model of clocks. Real life clocks may also be subject to additional  constraints, such as unavoidable de-coherence or power constraints \cite{Erker},\cite{Pauletal2017}. Since we are considering finite dimensional maps from the clockwork to itself which are continuous, this naturally leads to a finite power consumption, and de-coherence is addressed later with our classical clock case. We furthermore comment on aspects of the clock model in the conclusions, Section \ref{sec:conclusions}.

\subsection{Representation in Terms of Generators} \label{sec_generators}\label{sec:continuouslimit}


As explained above, the specification of the individual maps of the family $\{\cM^{\delta}_{C \to C T}\}_\delta$ is  redundant. The following lemma, which is basically a variant of the Lindblad representation theorem \cite{Lindblad},  asserts that the family can equivalently be specified in terms of generators. 

\begin{restatable}{lemma}{lemclockgenerators} \label{lem_clockgenerators}
  Let $(\rho^0_C, \{\cM_{C \to C T}^\delta\}_\delta)$ be a clock with a classical tick register, having as a basis the states $\{\ket{0}, \ket{1}\}$. Then there exists a Hermitian operator $H$ as well as two families of orthogonal operators $\{L_j\}_{j=1}^m$ and $\{J_j\}_{j=1}^m$ on $C$ such that
  \begin{align}\label{eq:MC to CT gens}
  \begin{split}
    &\cM^{\delta}_{C \to C T}(\rho_C)
    = \rho_C \otimes \proj{0}_T\\
     &- \delta\Bigl(  i [H, \rho]  + \sum_{j=1}^m \frac{1}{2} \{L^{\dagger}_j L_j + J^{\dagger}_j J_j, \rho\} - L_j \rho L_j^{\dagger}\Bigr) \otimes \proj{0}_T\\
      &+ \delta \sum_{j=1}^m J_j \rho J_j^{\dagger} \otimes \proj{1}_T + F^\delta_{C \to C T}(\rho_C),
   \end{split}
  \end{align} 
  for $\delta > 0$, and where $F^\delta_{C \to C T} = O(\delta^2)$. Conversely, given any Hermitian operator $H$ and orthogonal families of operators $\{L_j\}_{j=1}^m$ and $\{J_j\}_{j=1}^m$ on $C$, Eq. \eqref{eq:MC to CT gens} defines a clock $(\rho^0_C, \{\cM^\delta_{C \to C T}\})$ with a classical tick register.
  
  In the case of classical clocks with basis $\{\ket{c_n}\}$, $H$ is the zero operator  and the operators  $L_j$ and $J_j$ can all be chosen to be proportional to operators of the form $\ket{c_m} \! \! \bra{c_n}$.
\end{restatable}

The proof of this Lemma, which is provided in Appendix~\ref{sec_GeneratorLemmaProof}, follows the description in Section 3.5.2 of~\cite{PreskillLectureNotes}. We call $\rho^0_C$ the \emph{initial state} of the clockwork. Furthermore, the operators $J_j$ are called \emph{tick generators}.\footnote{While Eq. \ref{eq:MC to CT gens} does not define a dynamical semigroup $CT \rightarrow CT$, it is possible to do so, see \app, Sec. \ref{appsec:Lindbladian}.}

In addition to determining when the clock ticks, the tick generators $J_j$ also define the clockwork's state after a tick. Clocks for which this state coincides with the initial state $\rho^0_C$ are of special interest, for they have a particularly appealing mathematical structure and 
 are optimal in terms of their precision in the case of classical clocks.



\begin{definition}\label{def:reset clock}
A \emph{reset clock} is a quantum clock $(\rho_C^0, \{ \cM^{\delta}_{C \to C T}\}_{\delta})$ whose tick generators induce a mapping to the clock's initial state\footnote{More generally, the tick generators induce a mapping to some fixed state, but there is very little loss of generality setting the initial state to be the same, since only the first tick of the clock is affected, every subsequent tick behaves as if the initial state is the fixed state.}, i.e., 
\begin{equation} \label{eq:reset clock}
\sum_{j=1}^m J_j \rho J_j^\dag \propto \rho_C^0\quad \forall \rho\in\mathcal{S}(\mathcal{H}_C).
\end{equation} 
\end{definition}


One may also use the Lindbladian generators to describe the evolution of the clockwork system $C$ as continuous, parametrised by a real variable $t$. From Lemma \ref{lem_clockgenerators}, the following differential equation governs the evolution of the clockwork,
\begin{align}
\begin{split}
	\frac{d}{dt}\rho_{C}(t) =& \lim_{\delta\rightarrow 0} \frac{ \tr_T\Big[ \cM^{\delta}_{C \to C T} (\rho_C(t))\Big] - \rho_C(t)}{\delta} \\
	=& -i [H, \rho_C(t)]  + \sum_{j=1}^m L_j \rho_C(t) L_j^{\dagger} + J_j \rho_C(t) J_j^{\dagger}\\ 
	&- \frac{1}{2} \{L^{\dagger}_j L_j + J^{\dagger}_j J_j, \rho_C(t)\}.\label{eq:continuousquantumclock}\end{split}
\end{align}
For the tick register, one may take the same limit to find the probability density of a tick being recorded, via the probability that the register is in the state $\ket{1}$,
\begin{align}
	P(t) &= \lim_{\delta\rightarrow 0} \frac{ \tr_{CT} \Big[ \proj{1}_T \cM_{C \to C T}^\delta (\rho_C(t)) \Big] }{\delta}\nonumber\\
	& = \tr_C \left[ \sum_{j=1}^m J_j \rho_C(t) J_j^\dag \right].\label{eq:continuousquantumtickdensity}
\end{align}
This limit and the sequence of ticks is discussed in more detail in Section \ref{Accuracy of Quantum Clocks}.\Mspace

Furthermore, consider the case of a clock in which one focuses on a single tick, and tracks the state of the clock only up to the first tick. In this case one can remove the ``tick" channel $\sum_{j=1}^m J_j \rho_C J_j^\dagger$ from the Lindbladian of the clock in Eq. (\ref{eq:continuousquantumclock}), as it represents the state of the clockwork \emph{after} a tick (see Lemma \ref{lem_clockgenerators}). Thus the description of the entire family of tick generators $\{J_j\}_{j=1}^m$ is redundant. Labelling the state of the clockwork for just a single tick as $\rho_C^{(1)}$, its dynamics are given by (taking Eq. (\ref{eq:continuousquantumclock}) with the tick channel removed)
\begin{align}
	\frac{d}{dt}\rho_{C}^{(1)}(t) =& -i [H, \rho_C^{(1)}(t)] - \{V,\rho_C^{(1)}(t)\} \nonumber\\
	&
	+ \sum_{j=1}^m L_j \rho_C^{(1)}(t) L_j^{\dagger} - \frac{1}{2} \{L^{\dagger}_j L_j, \rho_C^{(1)}(t)\},\label{eq:resetclockdynamics}
\end{align}
where 
\be 
V = \frac{1}{2} \sum_{j=1}^m J_j^\dagger J_j
\ee
is an arbitrary positive operator representing the ticking dynamics of the clockwork. In this case, the probability density of \textit{the first tick} being recorded is, from Eq. \ref{eq:continuousquantumtickdensity},
\begin{align}
	P^{(1)}(t) &= \tr_C \left[ 2 V \rho_C^{(1)}(t) \right].
\end{align}


This proves useful in the case of reset clocks. As we discuss later, the ticks of a reset clock are a sequence of independent and identically distributed random variables, and thus the first tick suffices to characterise such a clock.

\subsection{Example}\label{sec:example}

When describing a clock, one may want to distinguish between the intrinsic evolution of the state of the clockwork and the mechanism that transfers information about this state to the outside. A rather generic way to do this is to describe the evolution of the clockwork by a Hamiltonian $\hat{H}_C$ on the system $C$, and the transfer of information to the outside by a continuous measurement of the system's state with respect to a fixed basis $\{\ket{t_i}\}_{i =0}^{d-1}$, which we will refer to as the \emph{time basis}. In order to ensure that the measurement does not disturb the clock's state too much, the coupling between clockwork and measurement mechanism must be weak. We quantify it in the following by assigning a coupling parameter $V_i \in \rr_{\geq 0}$ to each of the elements $\ket{t_i}$ of the time basis and consider a reset clock (Def. \ref{def:reset clock}). We could then define a quantum clock $(\proj{\psi_0}, \{ \cM^{\delta}_{C \to C T}\}_{\delta})$ with initial state $\proj{\psi_0}$ and maps
\begin{align}
\begin{split}
  \cM^{\delta}_{C \to C T} : \quad \rho_C \,  \mapsto  \, e^{-i \delta \hat{H}_C} \hat{M}^\delta_0  \rho_C \hat{M}^\delta_0{}^{\dagger} e^{i \delta \hat{H}_C} & \otimes \proj{0} \\ + \sum_{j=0}^{d-1}  e^{-i \delta \hat{H}_C}\hat{M}^\delta_{1,j}  \rho_C \hat{M}^\delta_{1,j}{}^{\dagger}  e^{i \delta \hat{H}_C}&\otimes \proj{1} \label{eq_resetclockdefinition}
\end{split}
\end{align}
where $\smash{\hat{M}^\delta_{1,j} = \sqrt{2 \delta V_j} \ket{\psi_0} \!\! \bra{t_j} }$ and $\smash{  \hat{M}^\delta_0 = \sqrt{\id_C - 2 \delta \hat{V}_C} }$ with $\hat{V}_C =  \sum_{i=0}^{d-1} V_i \proj{t_i}.$  \Mspace

For sufficiently small $\delta$, the quantities $\smash{\{ \hat{M}^\delta_{1,i}{}^{\dagger} \hat{M}^\delta_{1,i}\}_{i=0}^{d-1}}$ together with $\hat{M}^\delta_{0}{}^\dag \hat{M}^\delta_{0}$ form a positive-operator valued measure (POVM), since
\begin{align*}
\hat{M}^\delta_0{}^{\dagger} \hat{M}^{\delta}_0 + \sum_{i=0}^{d-1}   \hat{M}^\delta_{i,1}{}^{\dagger} \hat{M}^{\delta}_{i,1} 
= \Bigl| \id_C - 2 \delta \hat{V}_C \Bigr| + 2 \delta \hat{V}_C  .
\end{align*}
As such, one can interpret Eq. \eqref{eq_resetclockdefinition}, in the following light. The initial state of the clockwork $\rho_C$ is measured via the POVMs, followed by allowing the clockwork to freely evolve according to its internal Hamiltonian $\smash{\hat H_C}$ for an infinitesimal time step $\delta$ and repeating the process. In accordance with Eq. \eqref{eq_resetclockdefinition} one would then associate the outcome ``no-tick" with the POVM element $\hat{M}^\delta_{0}{}^\dag \hat{M}^\delta_{0}$ and the ``tick" outcome with the elements $\smash{\{ \hat{M}^\delta_{1,i}{}^{\dagger} \hat{M}^\delta_{1,i}\}_{i=0}^{d-1}}$. Since the POVM defines a measurement with classical outcome, one may regard the tick as a classical value, i.e., the tick register could be assumed to be classical in this case. \Mspace

Furthermore, by expanding in $\delta$, Eq. \eqref{eq_resetclockdefinition} can be written in the form
\begin{align} 
\cM^{\delta}_{C \to C T} : \quad &\rho_C \,  \mapsto \rho_C \otimes \proj{0} \nonumber\\
&- \delta\Bigl(  i [\hat H_C, \rho]  + \sum_{j=1}^m \frac{1}{2} \{ J^{\dagger}_j J_j, \rho\}\Bigr) \otimes \proj{0}\nonumber\\
& + \delta \sum_{j=1}^m J_j \rho J_j^{\dagger} \otimes \proj{1}+ O(\delta^2),\label{eq_resetclockdefinition 2}
\end{align}
where we have defined $J_j=\sqrt{2V_j}\ketbra{\psi_0}{t_j}$. With the further identifications $H=\hat H_C$, and $\{L_j\}_{j=1}^m$ with the set of zero operators, we see that Eq. \eqref{eq_resetclockdefinition 2} is in the form prescribed by Lemma \eqref{lem_clockgenerators}. This ensures that the map $\cM^{\delta}_{C \to C T}$ is indeed a clock, according to our definition \ref{eq_quantumclockcondition}. Consequently, it is easily verified that the $J_j$ operators satisfy Eq. \eqref{eq:reset clock} and the clock is thus a reset clock. It also follows from Section \ref{sec:continuouslimit} that in the continuous limit of clocks, the probability of not getting a ``tick" in the time interval $[0,t]$ followed by a tick in the interval time $t,t+\delta t$ is
\be 
\delta t P^{(1)}(t)= \delta t\,  \tr_C \left[ 2 V \rho_C^{(1)}(t) \right] =  \delta t\,  \tr_C \left[ 2\hat V_C \rho_C^{(1)}(t) \right],
\ee 
where $\rho^{(1)}_C(t)=\ketbra{\bar{\psi}_t}{\bar{\psi}_t}$, with 
\be \label{eq:psi bar t ex}
\ket{\bar{\psi}_t}= \me^{-\mi t \hat H_C- t \hat V_C} \ket{\psi_0}.
\ee

\subsection{Classical Clocks as a Special Case}\label{sec:classicalspecialcase}

Classical physics is widely believed to arise from quantum mechanics due to a mechanism called decoherence. It is a naturally occurring process caused by phenomena in which the quantum state becomes incoherent in some preferred basis \cite{2007decoherence,UnversalDeco,QuantumDarwinism}. Roughly speaking, a classical clock may be regarded as a clock that satisfies Def.~\ref{def_quantumclock}, but whose state space is restricted to classical states due to decoherence effects which happen on a time-scale much shorter than the times between ticks.

We allow for any preferred basis. Let us denote it by an arbitrary fixed orthonormal basis $\{\ket{c_k}\}_{k=0}^{d-1}$, of the Hilbert space $C$ of the clockwork:
\begin{align*}
\cC_{\{\ket{c_k}\}} = \bigg\{\rho_C = \sum_{i=0}^{d-1} p_i \proj{c_i} : \, p_i \geq 0, \sum_{i=0}^{d-1} p_i = 1\bigg\}.
\end{align*}

\begin{definition} \label{eq_classicalclock}
  A clock $(\rho^0_C, \{\cM^{\delta}_{C \to C T}\}_{\delta})$ is called \emph{classical} if there exists a basis $\{\ket{c_k}\}_k$ (called the \emph{classical basis}) such that
  \begin{align*}
     \rho^0_C \in \cC_{\{\ket{c_k}\}} \quad \text{and} \quad \tr_T\circ\cM^\delta_{C \to C T}(\cC_{\{\ket{c_k}\}}) \subseteq \cC_{\{\ket{c_k}\}}
  \end{align*}
  for all $\delta\geq 0$. 
\end{definition}

Since the dynamics is restricted to a single basis, we only require the vector of diagonal elements in that basis to describe the clock, and we label this by 
\begin{align}
	v_C &= \sum_m v_{C,m} \mathbf{e}_m,
\end{align}
where $\mathbf{e}_m$ represents the basis vector $\ket{c_m}_C\!\bra{c_m}$ and $v_{C,m} = \braket{c_m|\rho_C|c_m}$ are the diagonal elements of the clock in the preferred basis.

With these definitions in hand, we find that the clock generators take on the simple form of stochastic generators, namely:
\begin{restatable}{corollary}{collclassicalclocks}\label{coll:classical clocks} 
	Let $(v^0_C, \{\cM_{C \to C T}^\delta\}_\delta)$ be a classical clock and suppose that the tick register has basis $\{\ket{0}, \ket{1}\}$. Then there exist $d\times d$-matrices $\cN$ and $\cT$ such that
	\begin{align*}
		\cM^{\delta}_{C \to C T}&(v_C) = v_C \otimes \proj{0} \\
			&+ \delta \left( \cN v_C \otimes \proj{0} + \cT v_C \otimes \proj{1} \right) + O (\delta^2).
	\end{align*}
	with $\cT$ being an entry-wise positive matrix, and the sum $\cN + \cT$ being an infinitesimal generator (also known as a transition rate matrix). More precisely,
	\begin{align} 
	\cN_{mn} & \begin{cases} \leq 0 & \text{for $m=n$} \\ \geq 0 & \text{for $m \neq n$} \end{cases} \\
	\cT_{mn} & \geq 0
	\end{align}
	for any $m,n$ and
	\begin{align} \label{eq_sumticknotick}
	\sum_{m=1}^d \cN_{mn} + \sum_{m=1}^d \cT_{mn} = 0
	\end{align}
	for any $n$.\footnote{Eq. \eqref{eq_sumticknotick} will be relaxed in the \app~by replacing the ``$=$" sign with ``$\leq $". By doing so, we prove that our results for classical clocks hold under more general circumstances. The example of the maximally precise classical clock in Section \ref{sec:example classical} satisfies Eq. \eqref{eq_sumticknotick}.}
\end{restatable}

See Appendix~\ref{sec_ClassicalClocksCorollaryProof} for a proof of this corollary. Analogous to the quantum case, we see that $\cT$ is the classical version of the tick generator.

In the case of quantum clocks, we used the Lindbladian generators rather than the maps to describe the evolution of the clockwork as continuous and parametrised by $t$ (Section \ref{sec:continuouslimit}). We can do the same for classical clocks, by taking the limit $\delta \rightarrow 0$, as in Eq. \ref{eq:continuousquantumclock}. However, in the classical case, since the state is always diagonal w.r.t. a fixed orthonormal basis, we only require the dynamics of the vector of diagonal elements, which is seen to be
\begin{align}
\frac{d}{dt} v_C(t) &= \left( \cN + \cT \right) v_C(t).
\end{align}

As in the quantum case, in the continuous limit we replace the register by a probability density of a tick being recorded, Eq. \ref{eq:continuousquantumtickdensity}, found to be
\begin{align}
P(t) &= \left|\left| \cT v_C(t) \right|\right|_1.
\end{align}

Furthermore, if one is focused on a single tick, as in Eq. \ref{eq:resetclockdynamics}, the reduced dynamics of the state of the clock for a single tick, $v_C^{(1)}$ is simply
\begin{align}
	\frac{d}{dt} v_C^{(1)}(t) &= \cN v_C(t).
\end{align}

\subsection{Precision of Clocks}\label{Accuracy of Quantum Clocks}

As mentioned in the introduction, we use the regularity of the tick output of a clock as a measure for its precision. We will now introduce definitions that allow us to express this quantity formally in terms of the clock maps. \Mspace

A clock (Def. \ref{def_quantumclock}) after $N$ applications of the map $\mathcal{M}^\delta_{C\to CT}$ gives rise to a probability distribution $P_{T_j} (t_N)$ corresponding to the probability that $j-1$ ticks have occurred during the $N-1$ applications of the map, and the $j^\textup{th}$ has occurred at the $N^\textup{th}$ application of the map. In the limiting case of continuous clocks discussed in Section \ref{sec:continuouslimit}, the probability $P_{T_j} (t_N)$, converges to a probability density, such that $P_{T_j} (t)\delta t$ is the probability that $j-1$ ticks have occurred in the interval $[0,t)$ and the $j^\textup{th}$ has occurred during the interval $[t,t+\delta t]$ for infinitesimal $\delta t>0$. Such probability densities are also known as delay functions or waiting times. In particular, we call $P_{T_{j}}(t)$ the delay function associated with the $j^\text{th}$ tick.
\Mspace

The expected time $\mu_j$ of the $j^\text{th}$ tick and its variance $\sigma_j^2$ are then given by 
\begin{align*}
\mu_j &  = \int_{0}^{\infty} d t P_{T_{j}}(t) t,  \\
\sigma_j^2 &  = \int_{0}^{\infty}  d t P_{T_{j}}(t) (t-\mu_j)^2   ,
\end{align*} 
for any $j \in \nn_{> 0}$. Based on these quantities, we can now define the \emph{clock precisions} $R_j$. Note that this is different from the single \emph{clock precision} $R$, which will be defined below for the particular case of reset clocks.

\begin{definition}\label{def:clock accuracies}
	The \emph{clock precisions} $(R_j)_{j\in\nnp}$ of a clock $(\rho_0, \{\cM^{\delta}_{C \to CT}\}_{\delta})$ is a sequence of real numbers, where the $j^\text{th}$ element $R_j$ is the precision of the $j^\text{th}$ tick,
	\begin{align}
	R_j =  \frac{\mu_j^2}{\sigma_j^2}  . 
	\end{align}
\end{definition}

We will use this definition to define a partial ordering of clocks. For any two clocks $A$, and $B$, with clock accuracies $(R_j^A)_{j\in\nnp}$ and $(R_j^B)_{j\in\nnp}$ respectively, we will say that $A$ is \emph{more precise} than $B$ iff every tick of  $A$ is more precise than the corresponding tick of $B$, i.e., iff $R^A_j > R^B_j$  $\forall j\in\nnp$. It is this definition that we refer to when we later prove that quantum clocks are more precise than classical ones.  \Mspace

The characterisation of clocks provided by definition \ref{def:clock accuracies} has a number of nice properties. Firstly, it is scale invariant, meaning that the values $R_j$ are invariant under the mapping $t$ to $t/a$, for constants $a>0$. In other words, it is a measure of the closeness of the tick intervals to each other rather than to an external timescale, and is not affected by whether these ticks took place over a short or long time scale. Physically, this means that for every clock with precisions $(R_j)_{j \in \nn_{> 0}}$, and mean tick times $\mu_1, \mu_2,\mu_3,\dots$, there is another clock with the same precision, but with the ticks occurring on average at times $a\mu_1, a\mu_2,a\mu_3,\dots$. The new clock is constructed from the old clock by mapping $t$ to $t/a$, which is equivalent to rescaling the generators $\{L_j\}_{j=1}^m$, $\{J_j\}_{j=1}^m$ and the Hamiltonian $H$, introduced in Lemma \ref{lem_clockgenerators}, by constant factors. \Mspace


Furthermore, we can now appreciate the simplicity of reset clocks (Def. \ref{def:reset clock}). Since every time such a clock produces a tick, it is reset to its initial state, the ticks represent a sequence of independent events, which are identically distributed. It thus follows that the delay function of the $j^\text{th}$ tick, $P_{T_j}(t)$, is the convolution of the delay function associated with the 1$^\text{st}$ tick $P_{T_1}(t)$, with itself $j$ times. This in turn, gives rise to a simple relationship between the precisions $(R_j)_{j \in \nn_{> 0}}$ of reset clocks, (see \app~\ref{app:delaysequence}, and Remark \ref{remark:delaysequence} for a detailed argument)
\be \label{eq:R j reset}
R_j= j R_1,
\ee  
and takes on a particularly satisfactory interpretation. Namely, the precision of the 1$^\text{st}$ tick $R_1$, is the number of ticks the clock generates (on average), before the next tick has a standard deviation equal to the mean time between ticks, $\mu_1$. As such, roughly speaking, the clock's useful lifetime is $\mu_1 R_1$, beyond which one can no-longer distinguish between subsequent ticks. To compare two reset clocks, it follows that one only needs to compare their $R_1$ values. Given the special significance of $R_1$, we will sometimes simply refer to it as $R$.\Mspace 

A similar interpretation is also possible for the $R_j$ value of later ticks. For the purpose of illustration, suppose that the mean time between ticks, $\mu_1$ is one second. Then $R_{60}$ corresponds to the number of minutes (on average) that the clock can generate until the tick corresponding to the next minute has a standard deviation which is equal to one minute. As such, while according to Eq. \eqref{eq:R j reset}, $R_{60}$ is 60 times larger than $R_1$, this is not to say that ``the 60$^\text{th}$ tick is more precise than the 1$^\text{st}$ tick."\Mspace

 \section{Fundamental limitations for classical and quantum clocks}
 
 In this section, we will state our findings and explain their relevance and connection to related fields. There are two main theorems. The following one, which is about limitations on classical clocks, and Theorem \ref{thm:quantum lower bound}, which shows how quantum clocks can outperform classical clocks.
 
 \begin{restatable}{theorem}{upperboundclassical}[Upper bound for classical clocks] \label{thm_classicalbound}
 	For every $d$-dimensional classical clock, the clock precisions $(R_j)_{j\in\nnp}$ satisfy
 	\be \label{eq:R1 R2 ... classical}
 	R_j\leq j\, d
 	\ee 
 	for all $d\in\nnp$. Furthermore, for every dimension $d\in\nnp$, there exists a reset clock whose precisions $(R_j)_{j\in\nnp}$ saturate the bound Eq. \eqref{eq:R1 R2 ... classical},
 	\be \label{eq:saturates accuracy classical}
 	R_j= j\, d.
 	\ee 
 \end{restatable}
 \begin{proof}
 	See Section \ref{Sec:classical clocks theorem proof} for a proof of the above inequalities and Section \ref{sec:example classical} for an explicit construction of an optimal $d$-dimensional reset clock which saturates bound Eq. \eqref{eq:R1 R2 ... classical}. This clock is further discussed in Fig. \ref{fig:3clocks}  a).
 \end{proof}
 
 While the proof is a bit involved, there is an intuitive explanation to why reset clocks are optimal. If the clock is reset to its initial state after the 1st tick, then one can simply choose the initial state and dynamical map which has the highest possible precision for the 1st tick. Intuitively, the only way a non-reset clock could have a superior precision for later ticks than this one, would be for one to adjust the mean time of the following tick in the sequence to be longer or shorted than the previous one to make up for any lost or gained time due to the previous tick being ``early" or ``late". However, determining whether the clock ticked too early or late would require an additional clock, which is not available within the model.\Mspace

\begin{figure*}
	
	
	\centering
	\begin{minipage}[b]{0.49\textwidth}
		\includegraphics[width=\textwidth]{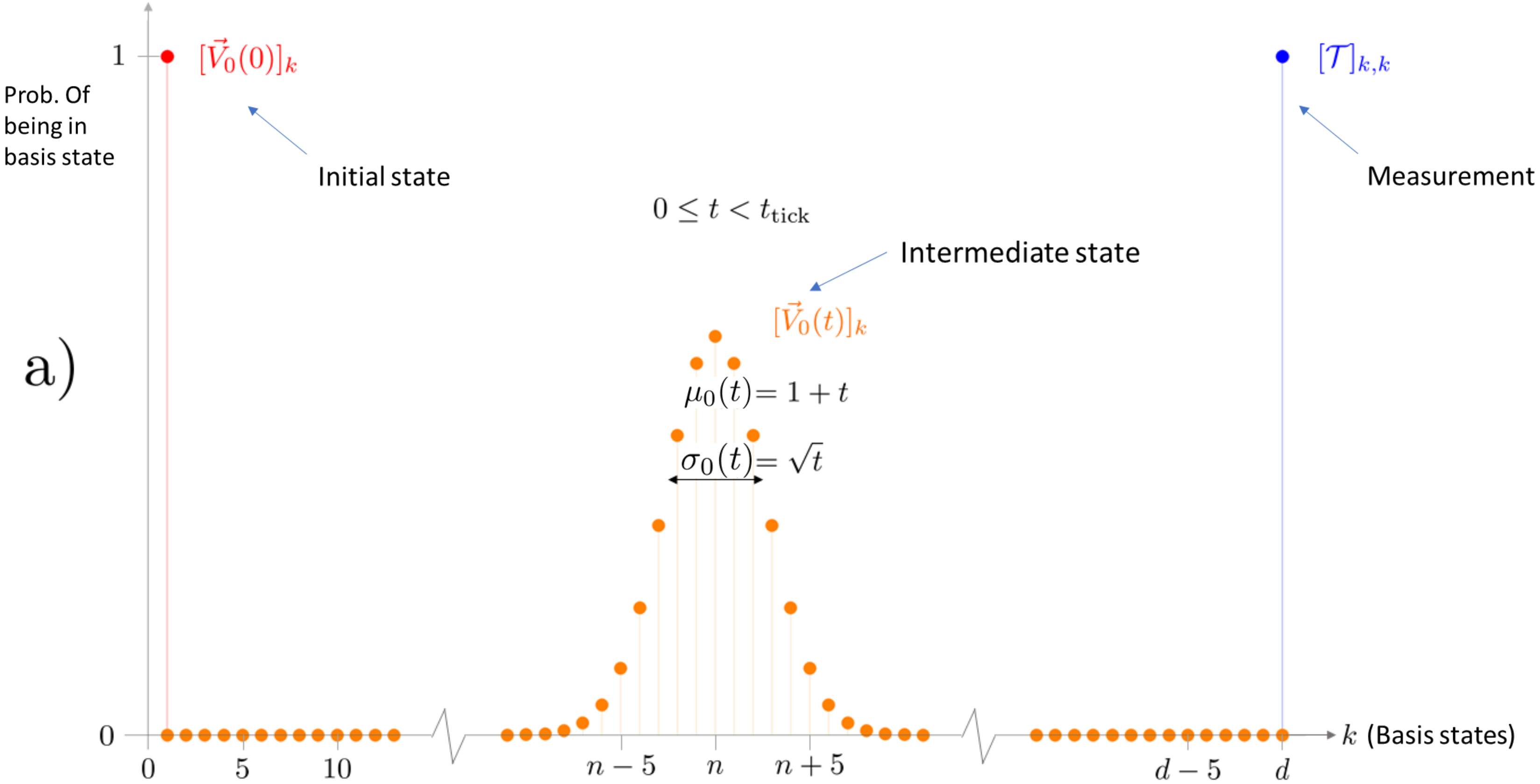}
	\end{minipage}
	\hfill
	\begin{minipage}[b]{0.47\textwidth}
		\includegraphics[width=\textwidth]{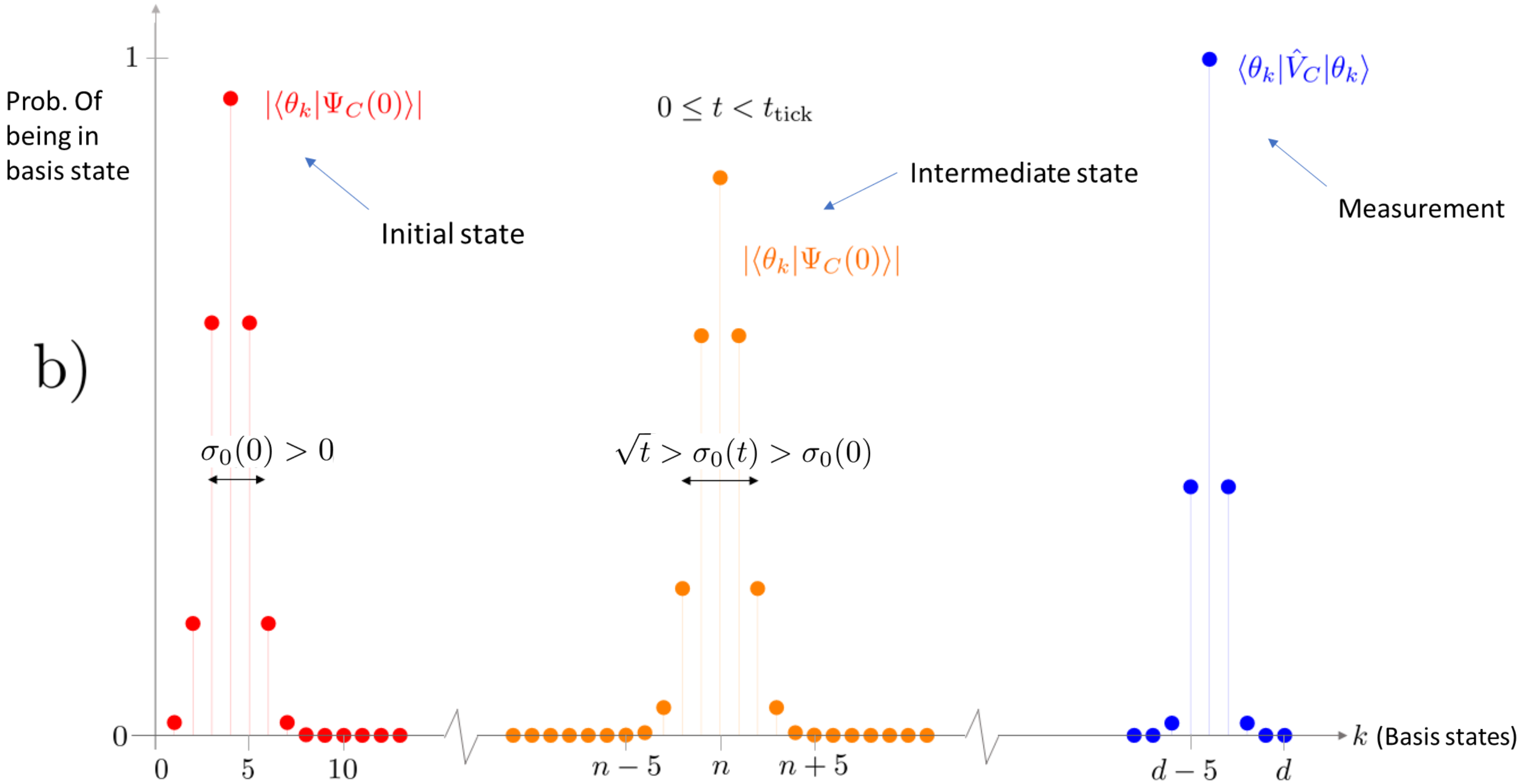}
	\end{minipage}
	\hfill
	\begin{minipage}[b]{0.47\textwidth}
		\includegraphics[width=\textwidth]{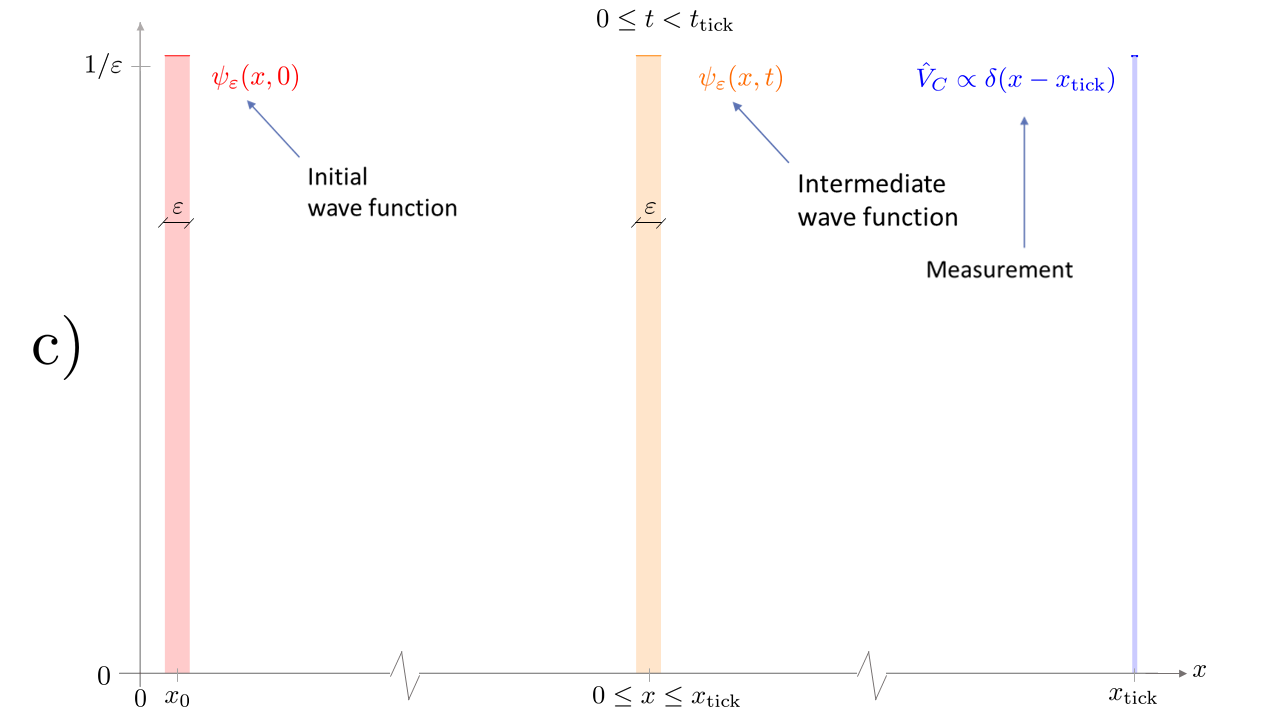}
	\end{minipage}
	
	\caption{\footnotesize{\textbf{Comparison of distributions of two clocks states before the 1st tick has occurred: a) An optimal $d$-dimensional classical clock. b) The \wso~(quantum) clock~\cite{WSO16}. c) Pauli's Idealised quantum clock.} In a nutshell, the purpose of these figures is to highlight how classical clocks disperse more than quantum clocks. In contrast, a hypothetical perfect clock would not disperse at all. Intuitively, this is why quantum clocks can be more precise than classical ones | but not perfect. $k$ labels orthogonal states; corresponding to real vectors in a) and pure quantum states in b).
			\MspaceFig\newline
			\textbf{a)} The $d$-dimensional probability vector $V_0(t)\in\cP_d$ associated with the clock having not ticked, starts off at $t=0$ with certainty at the 1st site, $V_0(0)=\textbf{e}_1$ (red plot). Its mean $\mu_0(t)$ then moves with uniform velocity towards the right with a standard deviation $\sigma_0(t)$ increasing with $\sqrt{t}$ (orange plot). The tick generator $\mathcal{T}$, is chosen so that the clock can only ``tick" from the last site, $\textbf{e}_d$ and the clock is reset (blue plot). This clock, whose full details are reserved for the \app~\ref{app:optimalclassical}, will ``tick" once it has reached the site $n=d$, at which point it will have dispersed considerably, and as such, its precision is limited to $R_1=d$. Furthermore, since it is a reset clock, later ticks are optimally precise, $R_j=j\, d$.
			\MspaceFig\newline
			\textbf{b)} The \wso~clock starts in a distribution in the time basis which is highly peaked around $\ket{\theta_0}$, resulting in a small standard deviation $\sigma_0$ in the time basis (red plot). The amplitudes of its distribution shift/move in time towards a large concentration around $\ket{\theta_{d-1}}$, where a ``tick" is measured with high probability and the clock is reset to its initial state (blue plot). During the time intervals between ticks this clock will disperse, but less than the classical clock in a) due to quantum interference. This results in a smaller standard deviation $\sigma_0(t)$ than in a) (orange plot). Furthermore, it will also be disturbed by time measurements, causing further unwanted dispersion. However, there is a trade-off | the smaller the standard deviation of the initial state, the more precise initial time measurements will be, but the larger the dispersion due to dynamics and measurements will be too, making later time measurements less reliable. Even so, due to quantum constructive/destructive interference, quantum mechanics allows for a $d$-dimensional state to disperse less as it travels to the region where a tick has the highest probability of occurring, than an optimal $d$ dimensional classical clock, such as in a). As such, this quantum clock can surpass the classical bound; see Theorem \ref{thm:quantum lower bound}.\MspaceFig\newline
 			\textbf{c)} The Idealised clock of Pauli starts with an arbitrarily highly peaked wave-function at position $x=0$ (red plot). It then moves according to $\psi(t,x)=\psi(0,x-t)$; towards $x=x_\textup{tick}$ (orange plot). At all times, its standard deviation is a constant $\epsilon>0$, which can be chosen to be arbitrarily small. It is not disturbed at all by time measurements, and ``ticks" exactly at time $t_\textup{tick}=x_\textup{tick}$ (blue plot), resulting in perfect precision $R_1=\infty$. Furthermore, one can add additional Dirac-delta distributions to the potential centred around $2x_\textup{tick}, 3 x_\textup{tick},\ldots$ without effecting the standard deviation of the Idealised clock. This results in perfect precision for all later ticks; $R_2=R_3=\ldots=\infty$.
	}}\label{fig:3clocks}
\end{figure*}
 
 The optimal reset clocks which saturate the bound in Eq. \eqref{eq:R1 R2 ... classical}, provide insight into our results. For these classical clock examples, the clock starts at one end of a length $d$ nearest neighbour chain, with the tick generator's support region located at the other end. The clock dynamics produce a classical continuous biased random walk along this chain, see Fig a) \ref{fig:3clocks} for details. The error in telling the time is a consequence of the state dispersing as it travels along the chain. Indeed, for these clocks, the standard result from random walk theory which predicts that the standard deviation in a state is proportional to the square root of the distance travelled, approximately holds. This dispersive behaviour achieves $R_1=d$ in the optimal case.\Mspace
 
 Before making a comparison with the quantum clocks described in this \doc, it is illustrative to compare this result with a recent clock in the literature, \cite{Pauletal2017}. Here a quantum clock is powered by two thermal baths, at different temperatures. This temperature difference drives a random walk of an atomic particle up a $d$-dimensional ladder, which spontaneously decays back to the initial state when reaching the top of the ladder, emitting a ``tick" in the decay process. As such, it is a reset clock whose precision $R_1$ depends on the entropy generated by the clock. Interestingly, in the limit of weak coupling and vanishing frequency of ticks, it is found \cite{Pauletal2017} that the clock dynamics becomes classical, represented by a biased random walk up the ladder. The precision of the clock is then
 \be 
 R\approx d \left(\frac{p_\uparrow-p_\downarrow}{p_\uparrow+p_\downarrow} \right) \approx d \left( \frac{e^{\Delta S/d} - 1}{e^{\Delta S/d} + 1} \right)
 \ee 
 where $p_\uparrow$, $p_\downarrow$ are the probabilities of moving up/down the ladder, induced by the thermal baths, and $\Delta S$ is the entropy generated by the clock for each instance that it ticks. As such, as far as the criterion of dimensionality is concerned, this classical thermodynamic clock is always less precise than the optimal classical clock in Theorem \ref{thm_classicalbound} by a constant factor, and only approaches optimal precision in the limit of infinite entropy generation $\Delta S \to \infty$. 
 \Mspace
 
 On the other hand, we can also compare the classical clock in Fig.  \ref{fig:3clocks} a) with the behaviour of the Idealised clock of Pauli, introduced in Section \ref{sec:Modeling clocks}. Here the clock Hamiltonian is the generator of translations, and it is not disturbed by continuous measurements, thus leading to no dispersion, and a clock precision of $R_1=\infty$, see Fig. \ref{fig:3clocks} c). Of course, as previously discussed, this high precision is unfortunately an artefact of requiring both infinite energy and dimension.\Mspace
 
 The important question is whether one can do better than the classical clock, and achieve higher precision with a quantum clock. For this task, we consider the \wso~clock introduced in Section \ref{sec:Modeling clocks} which is formed by taking a complex Gaussian amplitude superposition of the SWP clocks, namely
 \be \label{eq:WSO clock}
 \ket{\Psi_\textup{\wso~}} := \sum_{\mathclap{\substack{k\in \mathcal{S}_d(k_0)}}} A e^{-\frac{\pi}{{\sigma_0}^2}(k-k_0)^2} e^{i 2\pi n_0(k-k_0)/d} \ket{\theta_k},
 \ee
 where $\mathcal{S}_d(k_0)$ is a set of $d$ consecutive integers centred about $k_0\in\rr$, $n_0\in(0,d)$ determines the mean energy of the clock state and ${\sigma_0}$ its width in the $\ket{\theta_k}$ basis. Its clock Hamiltonian $\hat H_C$ is the 1st $d$ levels of a quantum harmonic oscillator with level spacing $\omega$; $\hat H_C=\sum_{n=0}^{d-1} n\omega \ketbra{E_n}{E_n}$. The dynamics of the clock is generated via the (possibly non-Hermitian) operator $\hat H=\hat H_C+\hat V_C$, where $\hat V_C=\sum_{i=0}^{d-1} V_i\ketbra{\theta_i}{\theta_i}$ with, 
 $\{\ket{\theta_k}\}$ the complementary basis to $\{\ket{E_k}\}$, formed by taking the discrete Fourier Transform. As such, $\hat H_C$ and $\hat V_C$ are diagonal in complementary bases to each other. This setup was introduced in \cite{WSO16} with the aim of studying unitary control of other quantum systems. For this objective, 
 the clock underwent unitary dynamics without being measured. Here we will use its construct to see how well quantum clocks can measure time. Indeed, the potential $\hat V_C$ can be  chosen to correspond to continuous measurements rather than unitary evolution by making it anti-Hermitian instead of Hermitian.
 It then follows that one can use the \wso~clock setup to perform continuous measurements on a reset clock as described in Section \ref{sec:example}. In particular, the dynamics under $\hat H=\hat H_C+\hat V_C$ takes the same form as Eq. \eqref{eq:psi bar t ex} on making the basis identification $\{\ket{\theta_k}=\ket{t_k} \}_k$ | which we now make.\Mspace
 
 If one chooses the width of the Gaussian in Eq. \eqref{eq:WSO clock} to be  $\sigma_0=\sqrt{d}$, then the width in the complementary basis $\{\ket{E_k}\}$ is also $\sqrt{d}$. In such cases, a precision $R_1$ proportional to $d$ can be achieved. However, if we choose a width that is narrower but not too narrow, namely $\sigma_0 = d^{\eta/2}$ for small $\eta>0$, then 
the \wso~clock is able to mimic approximately the dynamical behaviour of the Idealised clock, while maintaining finite energy and dimension \cite{WSO16}, see Fig. \ref{fig:3clocks}  b).   The following theorem formalises this. In addition to $\sigma_0 = d^{\eta/2}$, other parameters such as the particular potential $\hat V_C$ and coefficients $n_0, k_0$ used in the theorem are specified in the proof.

 \begin{theorem}[Achievable precision for quantum clocks]\label{thm:quantum lower bound}
 Consider the quantum clock construction in Section \ref{sec:example} and use a Quasi-Ideal Clock. For all fixed constants $0<\eta\leq1$, and appropriately chosen parameters, the \wso~Clock's precision satisfies
 	\be \label{eq:lower quantum}
 	R_1\geq d^{2-\eta} + \lo (d^{2-\eta}) 
 	\ee
 	in the large $d$ limit, where we have used little-o notation. 
 	Furthermore, since it is a reset clock; the precisions $(R_j)_{j\in\nnp}$ of later ticks satisfy
 	\be 
 	R_j= j\, R_1,
 	\ee
 	for all $j,d\in\nnp$.
 \end{theorem}
 \begin{proof}
 	See Section \ref{sec:Quantum Clocks appendix} for a proof of a slightly more general version of the theorem. 
 	The main difficulty of the proof is to come up with a potential $\hat V_C$ which satisfies all the necessary properties | if its derivatives are too large, the clock dynamics are too disturbed by the continuous  measurements, yet if they are not large enough, the measurements will not capture enough time information from the clock.
 \end{proof}
 
 Recently, an upper bound, which scales as $d^2$, has been derived for the precision of all quantum clocks \cite{2004.07857}. This proves that the lower bound derived in Theorem \ref{thm:quantum lower bound} for quantum clocks is tight | at least for the 1st tick.

 \section{Discussion of the Quantum Bound: relationship to related fields and open problems}\label{sec:relation to other bounds} 
 
 In this section, we discuss the relationship of the quantum bound to other concepts and fields which have been associated with time and clocks in quantum mechanics in the past: time-energy uncertainty relations, quantum metrology, quantum speed limits.
 
 \begin{figure}[!htb]
 	
 	
 	
 	
 	\minipage{0.45\textwidth}
 	\includegraphics[width=\linewidth
 	]{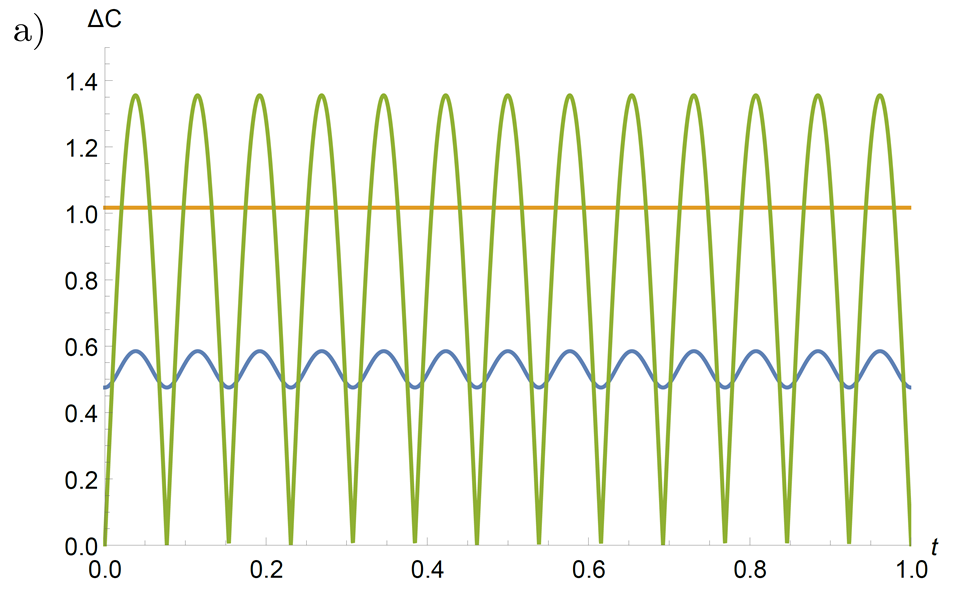}
 	\endminipage\hfill
 	\minipage{0.46\textwidth}
 	\includegraphics[width=\linewidth]{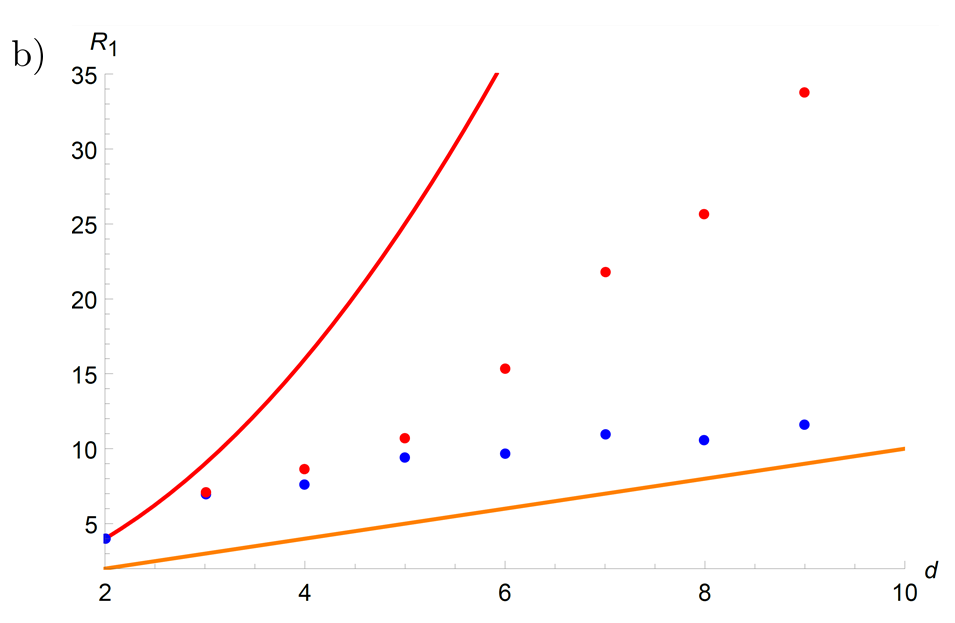}
 	\endminipage\hfill
 	

 	\caption{\footnotesize{\textbf{Comparison of the \wso~and SWP clocks. a) Standard deviations in the time basis as a function of time. b) $R_1$ as a function of clock dimension $d$.}\MspaceFig \newline
 			\textbf{a)} Standard deviation of clock states in the time basis $\Delta C(t)$ for different clocks as a function of time $t$, when time evolved according to their clock Hamiltonian $\hat H_C= \sum_{n=0}^{d-1}\omega \,n\ketbra{E_n}{E_n}$. Time runs from zero to one clock period $T_0=1$ with clock dimension $d=13$. Initial states are: \Wso~clock states for $\sigma_0=\sqrt{d}\approx 3.61$ (orange), $\sigma_0=1.8$ (blue), and a SWP state (green).  \MspaceFig\newline
 			\textbf{b)} Numerical optimization of $R_1$ for the \wso~quantum clock (red data points) and SWP quantum clock (blue data points) for a set of potentials. Both \wso~and SWP achieve $R_1=4$ for $d=2$, however, for large dimensions, the \wso~clock achieves higher precision. Red and orange solid lines ($R_1=d^2$ and $R_1=d$ respectively) are guides to the eye which represent the lower asymptotic bound for the \wso~clock and the upper bounds for the optimal classical clock respectively.  C.F. Theorems \ref{thm:quantum lower bound}, \ref{thm_classicalbound}.
 	}}\label{fig:Peres_Gauss_clock}
 \end{figure}

 \subsection{Time-energy uncertainty relation}\label{sec:t-E uncertainty rel}
 The time--energy uncertainty relation,
 \be \label{eq:E t uncert rel}
 \Delta t \Delta E\geq \frac{1}{2},
 \ee
 has been a controversial concept ever since its conception during the early days of quantum theory, with Bohr, Heisenberg, Pauli and Schr\"odinger giving it different interpretations and meanings. It is no longer so controversial, thanks to works such as~\cite{Busch2008,PhysRevA.66.052107}, which provide clarifying interpretations to the previous literature, and alternative quantifications, such as the Holevo variance; see, e.g.~\cite{OptimalStopwatch}. Often, at the heart of the controversy, was that in quantum mechanics, as explained in the introduction Section \ref{Introduction}, time was usually associated with a parameter, rather than an operator.\Mspace
 
 Since in the present context, we do have operators for time, a lot of this controversy can be circumvented. Indeed, Peres introduced a time operator, $\hat t_c:=\sum_{k=0}^{d-1} (T_0/d) k\ketbra{\theta_k}{\theta_k}$, where $T_0=2\pi/\omega$ is the period of the clock Hamiltonian $\hat H_C$ \cite{Peres80}. In \cite{WSO16} it was shown that the standard deviation of the initial clock state Eq. \eqref{eq:WSO clock} saturates a time-energy uncertainty relation; $\Delta E\Delta t=1/2$,\footnote{The saturation of the bound $\Delta E\Delta t=1/2$ by the \wso~clock, is up to an additive correction term which decays faster than any polynomial in $d$.} where the standard deviations $\Delta t,$ $\Delta E$ are calculated using the operators $\hat t_c$ and $\hat H_C$ respectively. One may be inclined to believe that one can increase the precision of the clock, by decreasing $\Delta t$ as a consequence of a larger $\Delta E$. While indeed decreasing $\sigma_0$ does have this effect, it would be naive to believe this paints the full picture. \Mspace
 
 To study this effect as the clock moves around the clock-face, let $\Delta E(t)$ be the standard deviation of the clock coefficients de-phased in the energy basis, $\{\braket{E_n|\rho_C(t)|E_n}\}_{n=0}^{d-1}$. We will define the standard deviation in the time basis similarly, but here one has to be careful since the time basis $\{\ket{\theta_k}\}$ has circular boundary conditions, meaning $\ket{\theta_k} = \ket{\theta_{k+1 \, (\mathrm{mod} \, d)}}$. Consequently, the state will ``jump" from the state $\ket{\theta_{d-1}}$ to $\ket{\theta_{0}}$ as it completes one period of its motion. We are not interested in the jumps due to the boundary effects, and therefore will denote $\Delta C(t)$ the standard deviation of the coefficients $\braket{\theta_k|\rho_C(t)|\theta_k}$ with the circular boundary conditions replaced with open boundaries\footnote{When the population $\braket{\theta_k|\rho_C(t)|\theta_k}$ reaches the position $d-1$, rather than subsequently appearing at position $0$, it will be assigned to position $d$, and subsequently to position $d+1$, etc.}. This way, we mimic the hands of a real clock, which visually do not suddenly ``jump" at 12 O'clock. The quantity $\Delta C(t)$ is plotted for the \wso~clock and SWP clocks in Fig. a) \ref{fig:Peres_Gauss_clock}. 
 
 Firstly, 
 if $\sigma_0$ is too small, the clock will disperse too much due to its dynamics and time measurements. Secondly, another way to decrease $\Delta t$ and increase $\Delta E$, is via reducing $T_0$, yet this has no effect on the precision of the clock. This latter observation is also related to how our measure of precision $R_1$, is invariant under a re-scaling of time, as discussed in Section \ref{Introduction}. In other words, for the lower bound in Eq. \ref{eq:lower quantum} to exceed the classically permitted value of $d$, the clock uncertainty in the time basis $\Delta t$, must be smaller than the uncertainty in the energy basis $\Delta E$ during the time in which the clock is running. In Fig. \ref{fig:Peres_Gauss_clock} a), the orange plot has $\Delta C(t)\approx\Delta E(t)\approx \sqrt{1/2}$. This is suboptimal, since the POVMs used to measure ticks are diagonal in the time basis $\{\ket{\theta_k}\}$. Furthermore, the blue curve has a smaller $\Delta C(t)$, and according to Theorem \ref{thm:quantum lower bound} can achieve $R_1>d$. However, if we continue to choose initial clock states with smaller $\Delta C(0)$, we run into a problem, namely in the limit $\sigma_0\rightarrow 0^+$, we recover the SWP clock state (green), which has a very large $\Delta C (t)$ at later times.\Mspace
 
 
 Another interesting aspect regarding time and energy, is what is the internal energy of the clock. We can easily answer this question for the \wso~clock. Note that for $n_0=(d-1)/2$, the mean energy of the \wso~clock is $ \omega\, d/2+\lo\left(d\right)$, thus Theorem \ref{thm:quantum lower bound} implies that high precision can be achieved with a modest linear increase in internal energy.
 
 \subsection{
 	Quantum metrology}\label{Quantum metrology}
 Metrology is the science of measuring unknown parameters. Here, like in other areas of quantum science, it has been shown that allowing for quantum effects can vastly improve the precision of measurements in comparison with the optimal classical protocols. Quantum metrology is now a mature discipline with a vast literature\cite{GiLlMa11}; and as such, it is appropriate to compare our results about clocks with them.\Mspace
 
 
 The basic setup in metrology is as follows. When attempting to measure an unknown parameter in a physical system, we prepare a probe state, let it evolve
 , and finally measure the evolved system. The evolution stage is according to the unitary operator $U_x:=\me^{-\mi x \hat H}$, where $\hat H$ is a known Hamiltonian. 
 The model of a clock described in this \doc~differs substantially from the metrology setup above. The easiest way to convey the main difference is via a simple example: for the purpose of illustration, imagine you were told to perform a particular task after 5 seconds. If you had a clock you could wait for 5 ticks to pass according to it, and perform your designated task. However, in the case of the metrology set-up one can only estimate the time when one \textit{happened} to measure the probe. As discussed in Section \ref{Introduction}, this marks an important division between time keeping devices, and those based on the standard metrology set-ups described here, fall into the category of stopwatches and not clocks, according to our definitions. One could of course use a very precise stopwatch in combination with a not so precise clock, by measuring the stopwatch when the clock ticks, and resetting the stopwatch immediately after measuring it. While one would know to high precision when in time these ticks occurred, their distribution in time would still be the same as if we did not have access to the stopwatch. As such, one can observe that even this combination cannot be substituted for a more precise clock. 

 \subsection{Quantum speed limits}
 The quantum speed limit, $\tau_\textup{QSL}$ is the minimum time $t_\textup{QSL}$ required for a state $\ket{\psi(0)}$ to become orthogonal to itself, $\braket{\psi(0)|\psi(t_\textup{QSL})}=0$ under unitary time evolution, $\ket{\psi(t)}=\me^{-\mi t \hat H}$. The celebrated Margolus-Levitin and Mandelstam-Tamm bounds, impose a tight lower limit in terms of the mean and standard deviation of $\hat H$ w.r.t. the initial state $\ket{\psi(0)}$ \cite{QSL_review}. It has found many applications in the field of thermodynamics, metrology, and the study of the rate at which information can be transmitted from a quantum system an external observer \cite{PhysRevD.9.3292,Bekenstein90}.\Mspace
 
 One may also be inclined to think that the fundamental limitations on the precision of clocks is related to how quickly the initial clock state becomes perfectly distinguishable to itself, when measured by an external observer. Indeed, the quicker states become distinguishable, the faster one can extract timing information from the clock state. Unfortunately, a simple example will reveal how the situation of a clock is too subtle to be captured by such simple arguments. The quantum speed limit $t_\textup{QSL}$ of the SWP clock is precisely $T_0/d=2\pi/(\omega d)$, since the states satisfy $\me^{-\mi \hat H_C T_0/d}\ket{\theta_k}=\ket{\theta_{k+1}}$. On the other hand, the \wso~clock has a quantum speed limit of much larger than $T_0$. 
 Yet, as described in Fig. \ref{fig:Peres_Gauss_clock} b)
 , numerics predict that while the SWP clock has a higher precision than the optimal classical clock, it is less precise than the \wso~clock. The explanation of this is that, contrary to the \wso~clock, the SWP clock | at the expense of a shorter quantum speed limit time | becomes highly spread-out in the time basis for times in-between becoming orthogonal to itself; thus incurring large $\Delta C(t)$ during these times [see green plot in Fig. \ref{fig:Peres_Gauss_clock} a)]. In summery, the quantum speed limit only tells us about how long it takes for a clock state to become distinguishable to itself, but fails to quantify the behaviour in-between. Since a clock has to produce a continuous stream of ``tick", ``no-tick" information, the nature of dynamics of the clock at \textit{all times} is of high importance.

 \section{Conclusion and Outlook}\label{sec:conclusions}
 
The workings of a clock requires a subtle interplay between two themes, measurement and dynamics | measure the time marked by the clock too strongly, and its dynamics will be very disturbed, adversely affecting later measurements of time. Yet measure too weakly, and you will not gain much information about time at all. Furthermore, the optimal state for minimising measurement disturbance, may possess a suboptimal dynamical evolution for distinguishing between different times | this poses another trade-off. Finding the optimal clock under measurements and dynamics is a fundamental and challenging problem.\Mspace 
 
Here we have motivated and used a general framework  where any clock is regarded as a $d$- dimensional system that autonomously emits information about time (the ticks) to the outside~\cite{RaLiRe15}, and we have shown a quantum-over-classical advantage for the precision of this time information: to achieve the same precision as a classical clock of size~$d$ a quantum clock only requires, roughly, dimension~$\sqrt{d}$. Moreover, due to recent developments,~\cite{2004.07857}, we know that this quantum-over-classical advantage is tight.

 
 A quantum-over-classical advantage characterised quantitatively by a square root is known for other tasks, in particular database search (where Grover's algorithm provides an advantage over any classical algorithm) or quantum metrology (where joint measurements provide an advantage over individual measurements). We stress however that these results are all of a different kind (see Section \ref{sec:relation to other bounds}). In the case of database search, the relevant quantity is the number of blackbox accesses to the database. In the case of metrology, the time keeping devices that one typically encounters, can only predict the time when they \textit{happened} to be measured; and thus are more akin to stopwatches than clocks, which autonomously emit a periodic time reference.\Mspace
 
 Traditionally, quantum metrology and concepts such as the time-energy uncertainty relation, and quantum speed limits; have been associated with characterising how precisely different physical processes involving time can be carried-out. Here we see that while these concepts have a role to play, a more discerning feature between the precision of classical and quantum clocks is how continuous-time quantum walks, under the right circumstances, allow for a smaller spread in the mean distribution, compared to classical standard walks which are limited to a standard deviation which is proportional to the square root of the mean distance travelled.\Mspace
 
Finally, we turn to discussing potential practical applications of our results. Given the description of a quantum clock with its corresponding Hamiltonian, one may understand the error in the clock's signal as arising from two sources. Firstly, the fundamental limitation of the clock due to its dimension, that we expose here, and secondly, the error due to its Hamiltonian not being perfectly stable; this can be thought of as a type of noise. It is the second type that is the dominant challenge for current atomic clocks that work by frequency stabilisation, a form of error correction. In order for our result to become practically relevant, the control over energy levels must increase to the point that the fundamental limitation exposed here (i.e. its effective dimension $d$) becomes dominant. In the case of atomic clocks, which work by coherently interacting with a qubit via a laser tuned to the energy gap of the qubit, the exact effective dimension $d$ is hard to determine, since the laser itself forms part of the clockwork.\Mspace

At the point where this is attained, the precision will be higher than that of current clocks, whose fractional error (the inverse of $R$) is of the order of $10^{-19}$~\cite{Oelker2019}. Thus the effective dimension\footnote{As noted earlier, we require a state highly coherent over all  $d$ degrees distinguishable states of the clock. In this sense, $d$ should be understood as an effective dimension, characterising the space over which one has full control (corresponding to the logical qubits of a quantum computer).} of a quantum clock that attains this level of precision must be at least of the order of $10^{10}$ (since the scaling of our quantum bound of $R\sim d^2$ is tight).\Mspace

These two conditions: stable Hamiltonians and high-dimensional coherent control, will allow us to build even more precise clocks by making full use of quantum properties, as demonstrated in this paper. A dimension of $10^{10}$ is attainable by $33$ qubits, which is not entirely outside the realm of possibility, given the current state of affair in quantum computing technology\footnote{Note that current quantum computers are able to manipulate a larger number $(50 - 100)$ of \emph{physical} qubits, while the number of \emph{logical} qubits is necessarily much smaller.}. Moreover, future work will investigate by how much the quantum-over-classical advantage is diminished by different types and intensities of noise.\Mspace

Another potential reason for pursuing a quantum clock is to explore the role quantum mechanics plays in gravity. Some theoretical models for quantum gravity predict that the general relativistic effect of time dilation will be slightly altered by quantum theory. For example, one semi-classical approach has predicted that while classical stopwatches are governed by standard time dilation when in the presence of a gravitational field, quantum stopwatches convey a, small yet important, modified time due to quantum fluctuations \cite{1904.02178,Anastopoulos_2018}. Similar effects have been reported in interferometry too \cite{zych2011quantum}. Therefore, while two clocks of the same precision (one classical, the other quantum) may be just as precise for telling the time, the quantum one has added value as a probe of quantum gravity. Recently, other foundational questions regarding clocks have been made, for example, how they can be derived from axiomatic principles \cite{mischa}, and how the precision of a quantum clock is related to a violation of Leggett-Garg-type inequalities, which provide further insight into their non classical nature~\cite{PhysRevResearch.3.033051}.\Mspace
 
  \vspace*{0.4cm} 
\acknowledgements
We thank Carlton Caves, Nicolas Gisin, Dominik Janzing, Christian Klumpp, Yeong-Cherng Liang and Yuxiang Yang for stimulating discussions. We all acknowledge the Swiss National Science Foundation (SNSF) via grant  No.~200020\_165843 and via the NCCR QSIT, Foundations Questions Institute via grant No.~FQXi-RFP-1610. M.W. acknowledges funding from his personal FQXi grant \emph{Finite dimensional Quantum Observers} (No.~FQXi-RFP-1623) for the programme \emph{Physics of the Observer}.
R.S. acknowledges funding from the SNSF via grant No.~200021\_169002 which funded him while at the University of Geneva.
 \vspace*{0.1cm}
 \subsection*{Author Contributions}
 All authors contributed to the ideas and results within this work. M.W. proved the lower bound for quantum clocks (Theorem \ref{thm:quantum lower bound}). R.S. proved the upper bound on classical clocks (Theorem \ref{thm_classicalbound}). S.S. provided numerics on quantum clocks (Fig. \ref{fig:Peres_Gauss_clock}). M.W. and R.R. wrote the main text.
 


\bibliography{RRBiblioV2,MWBiblioV2,RSBiblioV2}

\begin{thebibliography}{10}

\bibitem{RevModPhys.87.637}
Andrew~D. Ludlow, Martin~M. Boyd, Jun Ye, E.~Peik, and P.~O. Schmidt.
\newblock Optical atomic clocks.
\newblock {\em Rev. Mod. Phys.}, 87:637--701, Jun 2015.

\bibitem{qrevolution}
Jonathan~P. Dowling and Gerard~J. Milburn.
\newblock Quantum technology: the second quantum revolution.
\newblock {\em Philos. Trans. Royal Soc. A}, 361(1809):1655--1674, 2003.

\bibitem{pauli1}
Wolfgang Pauli.
\newblock {Handbuch der Physik}.
\newblock {\em Springer, Berlin}, 24:83—272, 1933.

\bibitem{pauli2}
Wolfgang Pauli.
\newblock {Encyclopedia of Physics}.
\newblock {\em Springer, Berlin}, 1:60, 1958.

\bibitem{PauliGeneralPrinciples}
Wolfgang Pauli.
\newblock {\em General principles of quantum mechanics}.
\newblock Springer Science \& Business Media, 2012.

\bibitem{holevo2011probabilistic}
A.S. Holevo.
\newblock {\em Probabilistic and Statistical Aspects of Quantum Theory}.
\newblock Publications of the Scuola Normale Superiore. Scuola Normale
  Superiore, 2011.

\bibitem{RaLiRe15}
Sandra Rankovi{\'{c}}, Yeong-Cherng Liang, and Renato Renner.
\newblock {Quantum clocks and their synchronisation - the Alternate Ticks
  Game}.
\newblock {\em arXiv:1506.01373}, 2015.

\bibitem{FuchsPeres96}
Christopher~A. Fuchs and Asher Peres.
\newblock Quantum-state disturbance versus information gain: Uncertainty
  relations for quantum information.
\newblock {\em Phys. Rev. A}, 53:2038--2045, Apr 1996.

\bibitem{WSO16}
Mischa~P. Woods, Ralph Silva, and Jonathan Oppenheim.
\newblock {Autonomous Quantum Machines and Finite-Sized Clocks}.
\newblock {\em Annales Henri Poincar{\'e}}, Oct 2018.

\bibitem{Pauletal2017}
Paul Erker, Mark~T. Mitchison, Ralph Silva, Mischa~P. Woods, Nicolas Brunner,
  and Marcus Huber.
\newblock Autonomous quantum clocks: Does thermodynamics limit our ability to
  measure time?
\newblock {\em Phys. Rev. X}, 7:031022, Aug 2017.

\bibitem{OptimalStopwatch}
Vladimir Bu\ifmmode~\check{z}\else \v{z}\fi{}ek, Radoslav Derka, and Serge
  Massar.
\newblock Optimal quantum clocks.
\newblock {\em Phys. Rev. Lett.}, 82:2207--2210, Mar 1999.

\bibitem{SaleckerWigner58}
Helmut Salecker and Eugene~P. Wigner.
\newblock Quantum limitations of the measurement of space-time distances.
\newblock {\em Phys. Rev.}, 109:571--577, Jan 1958.

\bibitem{TiQMVol1}
Gonzalo Muga, Rafael~Sala Mayato, and \'{I}nigo Egusquiza, editors.
\newblock {\em {Time in Quantum Mechanics Vol 1}}.
\newblock Lecture Notes in Physics. Springer Berlin Heidelberg, 2007.

\bibitem{TiQMVol2}
Gonzalo Muga, Andreas Ruschhaupt, and Adolfo~del Campo, editors.
\newblock {\em Time in Quantum Mechanics Vol 2}.
\newblock Lecture Notes in Physics. Springer Berlin Heidelberg, 2010.

\bibitem{ATGRandomWalk}
Sandra Stupar, Christian Klumpp, Nicolas Gisin, and Renato Renner.
\newblock {Performance of Stochastic Clocks in the Alternate Ticks Game}.
\newblock 2018.
\newblock ArXiv:1806.08812.

\bibitem{Pashby14}
Thomas Pashby.
\newblock {\em Time and the foundations of quantum mechanics}.
\newblock PhD thesis, University of Pittsburgh, 2014.

\bibitem{Peres80}
Asher Peres.
\newblock Measurement of time by quantum clocks.
\newblock {\em American Journal of Physics}, 48(7):552--557, 1980.

\bibitem{Gross2012}
David Gross, Vincent Nesme, Holger Vogts, and Reinhard.~F. Werner.
\newblock Index theory of one dimensional quantum walks and cellular automata.
\newblock {\em Communications in Mathematical Physics}, 310(2):419--454, Mar
  2012.

\bibitem{BUSCH1994357}
Paul Busch, Marian Grabowski, and Pekka~J. Lahti.
\newblock Time observables in quantum theory.
\newblock {\em Physics Letters A}, 191(5):357 -- 361, 1994.

\bibitem{mischa}
Mischa~P. Woods.
\newblock Autonomous {T}icking {C}locks from {A}xiomatic {P}rinciples.
\newblock {\em {Quantum}}, 5:381, January 2021.

\bibitem{weaver2001mathematical}
N.~Weaver.
\newblock {\em Mathematical Quantization}.
\newblock Studies in Advanced Mathematics. CRC Press, 2001.

\bibitem{Erker}
Paul Erker.
\newblock {The Quantum Hourglass}.
\newblock 2014.
\newblock {ETH Z\"urich}.

\bibitem{Lindblad}
G{\"{o}}ran Lindblad.
\newblock On the generators of quantum dynamical semigroups.
\newblock {\em Commun. Math. Phys.}, 48(2):119--130, Jun 1976.

\bibitem{PreskillLectureNotes}
John Preskill.
\newblock {Chapter 3. Foundations of Quantum Theory II: Measurement and
  Evolution}.
\newblock July 2015.
\newblock {Lecture notes available}.

\bibitem{2007decoherence}
Maximilian~A. Schlosshauer.
\newblock {\em Decoherence: and the Quantum-To-Classical Transition}.
\newblock The Frontiers Collection. Springer Berlin Heidelberg, 2007.

\bibitem{UnversalDeco}
Igor Pikovski, Magdalena Zych, Fabio Costa, and {\v C}aslav Brukner.
\newblock {Universal decoherence due to gravitational time dilation}.
\newblock {\em Nature Physics}, 11:668, 2015.

\bibitem{QuantumDarwinism}
Wojciech Zurek.
\newblock {Quantum Darwinism}.
\newblock {\em Nature Physics}, 5, 2009.

\bibitem{2004.07857}
Yuxiang Yang and Renato Renner.
\newblock Ultimate limit on time signal generation, Apr 2020.
\newblock arXiv:2004.07857.

\bibitem{Busch2008}
Paul Busch.
\newblock {\em The Time--Energy Uncertainty Relation}, pages 73--105.
\newblock Springer Berlin Heidelberg, Berlin, Heidelberg, 2008.

\bibitem{PhysRevA.66.052107}
Yakir Aharonov, Serge Massar, and Sandu Popescu.
\newblock Measuring energy, estimating hamiltonians, and the time-energy
  uncertainty relation.
\newblock {\em Phys. Rev. A}, 66:052107, Nov 2002.

\bibitem{GiLlMa11}
Vittorio Giovannetti, Seth Lloyd, and Lorenzo Maccone.
\newblock Advances in quantum metrology.
\newblock {\em Nature Photonics}, 5(4):222--229, Apr 2011.

\bibitem{QSL_review}
Sebastian Deffner and Steve Campbell.
\newblock Quantum speed limits: from heisenberg’s uncertainty principle to
  optimal quantum control.
\newblock {\em J. Phys. A}, 50(45):453001, 2017.

\bibitem{PhysRevD.9.3292}
Jacob~D. Bekenstein.
\newblock Generalized second law of thermodynamics in black-hole physics.
\newblock {\em Phys. Rev. D}, 9:3292--3300, Jun 1974.

\bibitem{Bekenstein90}
Jacob~D. Bekenstein and Marcelo Schiffer.
\newblock Quantum limitations on the storage and transmission of information.
\newblock {\em Int. J. Mod. Phys. C}, 01(04):355--422, 1990.

\bibitem{Oelker2019}
E.~Oelker, R.~B. Hutson, C.~J. Kennedy, L.~Sonderhouse, T.~Bothwell, A.~Goban,
  D.~Kedar, C.~Sanner, J.~M. Robinson, G.~E. Marti, D.~G. Matei, T.~Legero,
  M.~Giunta, R.~Holzwarth, F.~Riehle, U.~Sterr, and J.~Ye.
\newblock Demonstration of 4.8 {$\times{}10^{-17}$} stability at 1 s for two
  independent optical clocks.
\newblock {\em Nature Photonics}, 13(10):714--719, 2019.

\bibitem{1904.02178}
Shishir Khandelwal, Maximilian~P.E. Lock, and Mischa~P. Woods.
\newblock Universal quantum modifications to general relativistic time dilation
  in delocalised clocks.
\newblock {\em {Quantum}}, 4:309, August 2020.

\bibitem{Anastopoulos_2018}
Charis Anastopoulos and Bei~Lok Hu.
\newblock Equivalence principle for quantum systems: dephasing and phase shift
  of free-falling particles.
\newblock {\em Classical and Quantum Gravity}, 35(3):035011, Jan 2018.

\bibitem{zych2011quantum}
Magdalena Zych, Fabio Costa, Igor Pikovski, and {\v{C}}aslav Brukner.
\newblock Quantum interferometric visibility as a witness of general
  relativistic proper time.
\newblock {\em Nat. Commun.}, 2:505, 2011.

\bibitem{PhysRevResearch.3.033051}
Costantino Budroni, Giuseppe Vitagliano, and Mischa~P. Woods.
\newblock Ticking-clock performance enhanced by nonclassical temporal
  correlations.
\newblock {\em Phys. Rev. Research}, 3:033051, Jul 2021.

\bibitem{axiomaticDef}
Andrzej Kossakowski.
\newblock {On quantum statistical mechanics of non-Hamiltonian systems}.
\newblock {\em Rep. Math. Phys.}, 3:247—274, 1972.

\bibitem{Gilles}
Gilles P{\"{u}}tz et~al.
\newblock {\em In preparation}, 2019.

\bibitem{probability}
Athanasios Papoulis and S.~Unnikrishna Pillai.
\newblock {\em Probability, random variables, and stochastic processes}.
\newblock Boston: McGraw-Hill, 2002.

\bibitem{rudin1976principles}
Walter Rudin.
\newblock {\em Principles of Mathematical Analysis}.
\newblock International series in pure and applied mathematics. McGraw-Hill,
  1976.

\bibitem{kolmo}
Andrey Kolmogoroff.
\newblock {\"U}ber die analytischen methoden in der
  wahrscheinlichkeitsrechnung.
\newblock {\em Mathematische Annalen}, 104(1):415--458, Dec 1931.

\bibitem{Bhatia}
Rajendra Bhatia.
\newblock {\em Matrix analysis}.
\newblock Graduate Texts in Mathematics, 169. Springer-Verlag, New York, 1997.

\end{thebibliography}
\bibliographystyle{unsrt}


\newpage
\onecolumngrid
\begin{center}
	{\Huge{Appendices and Table of Contents}}
\end{center}

\appendix
\addappheadtotoc
\tableofcontents

\section{Modelling of Clocks} \label{sec_ModellingClocksProofs}

\subsection{Proof of Lemma~\ref{lem_clockgenerators}} \label{sec_GeneratorLemmaProof}

\lemclockgenerators*

\begin{proof}
  Consider an operator-sum representation of the map $\cM^{\delta}_{C \to C T}$, i.e., 
  \begin{align}\label{eq:C to CT gen in proof}
    \cM^{\delta}_{C \to C T}(\rho_C) = \sum_{j=0}^m M_j \rho_C M_j^{\dagger} \otimes \proj{0}_T + \sum_{j=1}^n N_j \rho_C N_j^{\dagger} \otimes \proj{1}_T ,
  \end{align}
  where $\{M_j\}_{j=0}^m$ and $\{N_j\}_{j=1}^n$ are families of operators on $C$. We will assume without loss of generality that the labelling of the operators $M_j$ is such that $\| M_0 \| \geq \| M_j \|$ for any $j \neq 0$. Define furthermore Hermitian operators $H$ and $K$ such that
  \begin{align}\label{eq:k - i H}
   K - i H  = \frac{1}{\delta} (M_0 - \id)
  \end{align}
  and 
  \begin{align}\label{eq:L_j J_j}
    \begin{split}
    L_j & = \frac{1}{\sqrt{\delta}} M_j \\
    J_j & = \frac{1}{\sqrt{\delta}} N_j  .
    \end{split}
  \end{align}
  Furthermore, we can assume without loss of generality that $n=m$. The operators $M_j$ and $N_j$ can be chosen such that $K$, $H$, $L_j$ and $J_j$ are independent of $\delta$ w.l.o.g., which one can prove as follows. 
   From Eq. \eqref{eq:C to CT gen in proof} it follows that
  \be \label{eq:C to C gen in proof}
   \cM^{\delta}_{C \to C}(\rho_C) = \sum_{j=0}^m M_j \rho_C M_j^{\dagger} + \sum_{j=1}^n N_j \rho_C N_j^{\dagger}.
  \ee 
  Given the defining properties stated in Section \ref{Quantum Clocks}, Lindblad proved \cite{Lindblad} that the map $\cM^{\delta}_{C \to C}(\rho_C)$ can be written as a Lindbladian with semi-group parameter $\delta$.\footnote{For the axiomatic definitions of a dynamical semi-group see Definition 2 in \cite{axiomaticDef}. In \cite{Lindblad}, Lindblad states these conditions on the adjoint map, which is equivalent to the conditions stated here since the clock is finite dimensional.} By Taylor expanding the generic expression for $\cM^{\delta}_{C \to C}(\rho_C)$ when expressed in Lindblad form to 1st order in $\delta$, and equating 0$^\text{th}$ and 1$^\text{st}$ order terms with the R.H.S. of \eqref{eq:C to C gen in proof}, the dependency on $\delta$ in Eqs. \eqref{eq:k - i H}, \eqref{eq:L_j J_j} follows.
  
  Eq. \eqref{eq:C to CT gen in proof} then reads
  \begin{align*}
    \cM^{\delta}_{C \to C T}(\rho_C) 
    = (\id_C + \delta K- \delta i H) \rho_C (\id_C + \delta K+ \delta i H) \otimes \proj{0}_T + \sum_{j=1}^m \delta L_j \rho_C L_j^{\dagger} \otimes \proj{0}_T +  \sum_{j=1}^m \delta J_j \rho_C J_j^{\dagger}  \otimes \proj{1}_T  .
  \end{align*}
  We then have to first order in $\delta$
  \begin{align*}
    \cM^{\delta}_{C \to C T}(\rho_C)
    = \rho_C \otimes \proj{0}_T + \delta \left( \{K, \rho \} - i [H,\rho] \right) \otimes \proj{0} + \delta  \sum_{j=1}^m  \Bigl( L_j \rho_C L_j^{\dagger} \otimes \proj{0}_T +  J_j \rho_C J_j^{\dagger}  \otimes \proj{1}_T \Bigr) + O(\delta^2) .
  \end{align*}
  The requirement that $\cM^{\delta}_{C \to C T}$ be trace-preserving furthermore implies that
  \begin{align*}
     (\id_C + \delta K + \delta i H)  (\id_C + \delta K- \delta i H) 
     + \delta  \sum_{j=1}^m \Bigl(L_j^{\dagger} L_j  + J_j^{\dagger} J_j  \Bigr) = \id_C  .
  \end{align*}
  Hence, again to first order in $\delta$, we have
  \begin{align*}
    2 \delta K = - \delta  \sum_{j=1}^m \Bigl( L_j^{\dagger} L_j - J_j^{\dagger}J_j  \Bigr)+ O(\delta^2)  .
  \end{align*}
      Inserting this into the above yields the claim. \Mspace
      
  Finally, to prove the converse part of the Lemma, we first note that for any Hermitian operator $H$ and families of orthogonal operators $\{L_j\}_{j=1}^m$, $\{J_j\}_{j=1}^m$ from $\mathcal{H}_C$ to $\mathcal{H}_C$, the reduced map $\cM^{\delta}_{C \to C}(\cdot)$ in Eq. \eqref{eq:MC to CT gens} takes on the following form
    \begin{align}\label{eq:MC to CT gens 2}
  \cM^{\delta}_{C \to C}(\cdot)= \tr_T\circ \cM^{\delta}_{C \to C T}(\cdot)= \me^{\delta\mathcal{L}} (\cdot) + O(\delta^2),
  \end{align} 
with the Lindbladian 
\be 
 \mathcal{L}(\cdot)=  - \Bigl(  i [H, (\cdot)]  + \sum_{j=1}^m \frac{1}{2} \{L^{\dagger}_j L_j + J^{\dagger}_j J_j, (\cdot) \} - L_j  (\cdot) L_j^{\dagger}\Bigr)  +  \sum_{j=1}^m J_j (\cdot) J_j^{\dagger}.
  \ee
We thus have
 \begin{align} \label{eq_quantumclockcondition 2}
 \lim_{\Delta \to 0}  \lim_{\delta \to 0} \bigl(\tr_T \circ \cM^{\delta}_{C \to C T}\bigr)^{\lfloor \frac{\Delta}{\delta} \rfloor}   =\lim_{\Delta \to 0}  \lim_{\delta \to 0} \bigl(  \me^{\delta\mathcal{L}} + O(\delta^2)  \bigr)^{\lfloor \frac{\Delta}{\delta} \rfloor} = \lim_{\Delta \to 0}  \lim_{\delta \to 0} \bigl( \me^{\delta {\lfloor \frac{\Delta}{\delta} \rfloor}  \mathcal{L}}  +   O({\lfloor {\Delta}/{\delta} \rfloor}  \delta^2)  \bigr) = \lim_{\Delta\to 0} \me^{\Delta \mathcal{L}}= \cI_{C},
\end{align}
thus proving that the map $\cM^{\delta}_{C \to CT}(\cdot)$ satisfies Eq. \eqref{eq_quantumclockcondition} in Def. \ref{def_quantumclock} and thus is a clock.

  \end{proof}

\subsection{Lindbladian semigroup for clock and register qubit}\label{appsec:Lindbladian}

While Eq. \eqref{eq:MC to CT gens} of the main text does not define a dynamical semigroup $CT \to CT$, one can do so quite simply via the map $e^{\delta \; \mathcal{L}_{CT}} \left( \cdot \right)$, where the Lindbladian on the clockwork and register is
\begin{align}
\mathcal{L}_{CT} \left( \cdot \right) =& -i[\tilde{H},(\cdot)] + \sum_{j=1}^m \sum_{a\in{0,1}} \tilde{L}_{ja} (\cdot) \tilde{L}^\dagger_{ja} - \frac{1}{2} \{ \tilde{L}^\dagger_{ja} \tilde{L}_{ja}, (\cdot) \}\nonumber\\
&+ \sum_{j=1}^m \tilde{J}_j (\cdot) \tilde{J}^\dagger_j - \frac{1}{2} \{ \tilde{J}^\dagger_j \tilde{J}_j, (\cdot) \},
\end{align}
where the extended operators are $\tilde{H} = H \otimes \id_T$, $\tilde{L}_{j0} = L_j \otimes \proj{0}_T$, $\tilde{L}_{j1} = L_j \otimes \proj{1}_T$, and $\tilde{J}_j = J_j \otimes \ket{1}\!\bra{0}_T$. The clock map is then defined via
\begin{align}
\cM^\delta_{C \to C T} \left( \rho_C \right) = e^{\delta \, \mathcal{L}_{CT}} \left( \rho_C \otimes \proj{0}_T \right).
\end{align}

\subsection{Proof of Corollary~\ref{coll:classical clocks}} \label{sec_ClassicalClocksCorollaryProof}

We prove the corollary using density matrix notation:
\begin{corollary*}
	Let $(\rho^0_C, \{\cM_{C \to C T}\})$ be a classical clock with basis $\{\ket{c}_j\}_{j=1}^d$ and suppose that the tick register has basis $\{\ket{0}, \ket{1}\}$. Then there exist $d\times d$-matrices $\cN$ and $\cT$ such that
	\begin{align}
	&\cM^{\delta}_{C \to C T}(\rho_C)
	=  \rho_C \otimes \proj{0} +  \delta\,\sum_{m, n} \bra{c_n} \rho_C \ket{c_n}  \proj{c_m}_C \otimes \bigl( \cN_{m n} \proj{0}_T + \cT_{m n} \proj{1}_T \bigr) + O(\delta^2). 
	\end{align}
	with
	\begin{align} 
	\cN_{m n} & \begin{cases} \leq 0 & \text{for $m=n$} \\ \geq 0 & \text{for $m \neq n$} \end{cases} \\
	\cT_{m n} & \geq 0
	\end{align}
	for any $m, n$, and
	\begin{align} \label{appeq_sumticknotick}
	\sum_{m=1}^d \cN_{m n} + \sum_{m=1}^d \cT_{m n} = 0
	\end{align}
	for any $m$.	 
\end{corollary*}

  \begin{proof}
	Let $L_j$ and $J_j$ be the operators defined by Lemma~\ref{lem_clockgenerators} and set
	\begin{align}
	\cN_{m n} & = - \delta_{m, n} \bra{c_m} \sum_j (L_j^{\dagger} L_j + J_j^{\dagger} J_j) \ket{c_m} + \sum_j  |\bra{c_m} L_j \ket{c_n}|^2 \\
	\cT_{m n} &= \sum_j  |\bra{c_m} J_j \ket{c_n}|^2  .
	\end{align}
	It is then straightforward to verify the claimed expression for the map. Furthermore, for any $n \in \{1, \ldots, d\}$, 
	\begin{align}
	\sum_{m=1}^d (\cN_{m n} +  \cT_{m n}) 
	= -  \bra{c_n} \sum_j (L_j^{\dagger} L_j + J_j^{\dagger} J_j) \ket{c_n} + \sum_{j, m} \bra{c_n} L_j^\dagger \proj{c_m} L_j \ket{c_n} + \bra{c_n} J_j^\dagger \proj{c_m} J_j \ket{c_n}  = 0,
	\end{align}
	which proves~\eqref{appeq_sumticknotick}. The non-negativity conditions for $\cN_{n m}$ (for $m \neq n$) and $\cT_{m n}$ (for any $m, n$) hold trivially.  Together with~\eqref{appeq_sumticknotick} they also imply that $\cN_{m n} \leq 0$ for $m = n$. 
\end{proof}

\section{Delay functions, precision, and i.i.d sequences}\label{Delay functions, accuracy, and i.i.d sequences}

The appendix is structured as follows. In \ref{app:delayfunction}, we define a \textit{delay function}, that characterises the probability distribution, w.r.t. time, of when an event occurs. This will eventually be applied to the ticks of the clock. We define the moments of the delay function, and introduce the precision $R$. We discuss the special case of a convolution of delay functions and the scaling of the precision in this case (Remark \ref{remark:delaysequence}). Finally, in \ref{app:Rlemmas}, we prove a number of important lemmas concerning the precision, including one for sequences (i.e. convolutions) of delay functions (Lemma \ref{lemma:delaysequence}, and another  for mixtures (Lemma \ref{lemma:delaymixture}).

\subsection{Delay functions: Definition, and behaviour of the precision}\label{app:delayfunction}

\begin{definition}\label{def:delayfunction}

By a \textbf{delay function}, we refer to a non-negative integrable function of time $t \geq 0$, $\tau \; : \; \mathbb{R}^+_0 \longrightarrow \mathbb{R}^+_0$, that is normalised or sub-normalised,
\begin{align}\label{eq:integrable}
	\int_0^\infty \tau(t) dt = Q&\leq 1.
\end{align}
	
\end{definition}

\begin{definition}\label{def:delayfunctionmoments}
	
We define the \textbf{mean} \textbf{(first moment)}, \textbf{second moment}, and \textbf{variance} of a delay function with respect to the normalized version of the delay function. That is, given a delay function $\tau(t)$ that integrates to $Q$ (see. Eq. \ref{eq:integrable}), the moments are calculated from $\tau(t)/Q$ which is a normalised probability distribution,
\begin{subequations}\label{eq:delayfunctionmoments}
\begin{align}
	Q &= \braket{t^0} = \int_0^\infty \tau(t) dt, \label{eqdef:delayfunctionprobability}\\
	\mu &= \frac{\braket{t^1}}{Q} = \int_0^\infty t \; \frac{\tau(t)}{Q} dt, \label{eqdef:delayfunctionmean}\\
	\chi &= \frac{\braket{t^2}}{Q} = \int_0^\infty t^2 \frac{\tau(t)}{Q} dt, \label{eqdef:delayfunctionchi}
\end{align}
\end{subequations}
while the variance, denoted by $\sigma$, is defined in the usual manner,
\begin{align}\label{eqdef:delayfunctionvariance}
	\sigma &= \sqrt{ \chi - \mu^2 }.
\end{align}

\end{definition}

\begin{remark}

Note that the first and second moments may diverge.

\end{remark}

\begin{definition}\label{def:accuracydelayfunction}

The precision of a delay function $\tau(t)$ is defined by
\begin{align}\label{eqdef:delayfunctionR}
	R \left[ \tau \right] &= \frac{\mu^2}{\sigma^2},
\end{align}
if the first moment $\mu$ of the delay function is finite, and $R=0$ if it diverges.
\end{definition}

For simplicity of expression, we will often omit the functional notation $R[\cdot]$, referring to the precision by simply $R$, augmented with a subscript or superscript when necessary.

\begin{remark}

This definition is discussed in the main text (Section \ref{Accuracy of Quantum Clocks}), in the context of sequences of independent events all described by the same delay function. There we see that the ratio above is the average number of events before the uncertainty in the occurrence of the next event equals the average interval between events. For a discussion of the same, see Appendix \ref{sec:justifyR}.

\end{remark}

\begin{remark}

There are delay functions of arguably high precision, whose precision is not reflected well by $R$ (\cite{Gilles}). Consider, for example, the following pathological\footnote{Pathological in the sense that such a function cannot be generated by the dynamics of finite-dimensional systems of bounded energy.} delay functions parametrised by $\epsilon>0$,
\begin{align}
	\tau(t) &= \left( 1 - \epsilon \right) \delta (t-1) + \epsilon \delta (t - 1/\epsilon^2),
\end{align}
where $\delta(x)$ is the Dirac-Delta function. This delay function corresponds to a large probability of an event occurring at precisely $t=1$, and a small probability $\epsilon$ of the event occurring much later, at $t=1/\epsilon^2$. For small $\epsilon$, the precision $R$ is of the order of $\epsilon$, and goes to zero in the limit $\epsilon \rightarrow 0$, even though the limiting delay function, $\delta(t-1)$, is very precise. Other notions of precision, such as the operational number of alternate ticks ``$N$" \cite{RaLiRe15} may be better suited to characterize these types of delay functions. The precise relationship between $R$ and $N$ is dealt with in detail in \cite{Gilles} (from which this delay function is sourced).
\end{remark}

\begin{remark}[The precision and $\gamma$-value]
	
The precision can alternatively be expressed as
\begin{align}\label{eq:accuracygamma}
	R &= \frac{1}{\frac{\chi}{\mu^2} - 1} = \frac{1}{\gamma - 1}, \quad \text{where} \quad \gamma = \frac{\chi}{\mu^2}.
\end{align}

$\gamma$ has been introduced because it is usually easier to deal with than $R$, being the ratio between two moments of the delay function. Note that the relationship is bijective, and inverted (the smaller $\gamma$ is, the higher the value of $R$). The range of the precision is $R \in [0,\infty)$, corresponding to $\gamma \in (1,\infty)$.
	
\end{remark}

\begin{remark}[The precision is invariant under re-normalisation]

As the precision $R$ is defined via the first and second moments of the \emph{normalised} version of the delay function, changing the normalisation of $\tau(t)$ by multiplying it by a positive constant does not affect $R$.

\end{remark}

\subsubsection{Combining delay functions in sequence - convolution}\label{app:delaysequence}

Consider an arbitrary sequence of delay functions $\{\tau_i\}_{i=1}^M$. Then the following convolution of a subsequence
\begin{subequations}\label{eq:delaysequence}
\begin{align}
	\tau^{(m)}(t) &= \int_0^t dt_{m-1} \iint... \int_0^{t_3} dt_{2} \int_0^{t_2} dt_{1} \tau_1(t_1) \tau_2(t_2-t_1) \tau_3(t_3-t_2) ... \tau_m(t-t_{m-1}) \\
	&= \left( \tau_1 \conv \tau_2 \conv ... \conv \tau_m \right) (t),
\end{align}
\end{subequations}
where $\conv$ denotes convolution, is also a delay function. To see this, note that the integrand above is always non-negative, and thus $\tau^{(m)}$ is also non-negative. Furthermore, one may calculate, by direct integration, that the moments of $\tau^{(m)}$ are given by
\begin{subequations}\label{eq:sequencemoments}
\begin{align}
	Q^{(m)} = \braket{t^0} &= \prod_{i=1}^m Q_i, & \text{where} \quad Q_i &= \int_0^\infty \tau_i(t) dt, \\
	\mu^{(m)} = \frac{\braket{t^1}}{Q^{(m)}} &= \sum_{i=0}^m \mu_i, & \text{where} \quad \mu_i &= \frac{1}{Q_i} \int_0^\infty t \cdot \tau_i(t) dt, \\
	\text{and} \quad \chi^{(m)} = \frac{\braket{t^2}}{Q^{(m)}} &= \left( \sum_{i=0}^m \chi_i + \sum_{\substack{i,j=0 \\i \neq j}}^m \mu_i \mu_j \right), & \text{where} \quad \chi_i &= \frac{1}{Q_i} \int_0^\infty t^2 \cdot \tau_i(t) dt.
\end{align}
\end{subequations}

Since each $Q_i$ is within $[0,1]$, it follows that $Q^{(m)}$ is as well, and thus $\tau^{(m)}$ is a delay function (see Def. \ref{def:delayfunction}). Note that $\{\mu^{(m)},\chi^{(m)}\}$ refer to the moments of the normalized version of the delay function (see Def. \ref{def:delayfunctionmoments}), and are finite if and only if every one of the corresponding moments of the individual delay functions do not diverge.

\begin{remark}\label{remark:delaysequence}

A convolution of delay functions describes the case of a sequence of independent distributed events, and in the particular case of clocks, corresponds to reset clocks, those that go to a fixed state after ticking. In this case, each tick of the clock has an identical delay function w.r.t. the previous tick, (all of the $\tau_i$ are the same and equal to $\tau_1$) and one can use Eq. \ref{eq:sequencemoments} to calculate the moments, and thus the precision of the $m^{th}$ tick,
\begin{align}
	Q^{(m)} &= Q_1^m \\
	\mu^{(m)} &= m \mu_1 \\
	\chi^{(m)} &= m \chi + m(m-1) \mu^2 \\
	\sigma^{(m)} &= \sqrt{m} \sigma_1,
\end{align}
from which we find that the precision $R^{(m)} = m R_1$.
\end{remark}

\subsection{Lemmas on the precision of sequences, mixtures, and scaled delay functions}\label{app:Rlemmas}

\begin{lemma}[The precision of a convolution of delay functions is limited by the sum of the precisions]\label{lemma:delaysequence}

If a delay function $\tau^{(m)}$ is formed out of the convolution of a sequence of delay functions as in Eq. \ref{eq:delaysequence}, then its precision $R\left[\tau^{(m)}\right]$ (Def. \ref{def:accuracydelayfunction}) is upper bounded by the sum of the precisions of the individual delay functions that form the sequence, i.e.
\begin{align}
	R \left[ \tau^{(m)} \right] \leq \sum_{i=1}^m R_i,
\end{align}
where $R_i$ is the precision of the $i^{th}$ delay function $\tau_i$ in the sequence. Furthermore, this optimal precision is achieved if and only if the individual delay functions satisfy
\begin{align}
	\frac{\mu_i}{R_i} &= \frac{\mu_j}{R_j} \quad \forall {i,j} \in \{1,2,...,m\},
\end{align}
or equivalently, that they satisfy
\begin{align}
	\frac{\sigma_i^2}{\mu_i} &= \frac{\sigma_j^2}{\mu_j} \quad \forall i,j \in \{1,2,...,m\},
\end{align}
where $\mu_i$ and $\sigma_i$ are the mean and variance of the $i^{th}$ delay function, Eq. \ref{eq:sequencemoments}.

\end{lemma}

\begin{proof}

We prove the lemma by induction. Consider that we express the delay function of the convolution as
\begin{align}
	\tau^{(m)} &= \tau^{(m-1)} \conv \tau_m, \\
	\text{where} \quad \tau^{(m-1)} &= \tau_1 \conv \tau_2 \conv ... \tau_{m-1}.
\end{align}

Calculating the moments of $\tau^{(m)}$ w.r.t. the above subdivision, via Eq. \ref{eq:sequencemoments}, we get
\begin{align}
	\mu^{(m)} &= \mu^{(m-1)} + \mu_m, \\
	\chi^{(m)} &= \chi^{(m-1)} + \chi_m + 2 \mu^{(m-1)} \mu_m \\
	&= \left( \frac{1}{R^{(m-1)}} + 1 \right) [\mu^{(m-1)}]^2 + \left( \frac{1}{R_m} + 1 \right) \mu_m^2 + 2 \mu^{(m-1)} \mu_m,
\end{align}
where we have used Eq. \ref{eq:accuracygamma} to re-express the second moments $\chi$ w.r.t. the precisions $R$.

Calculating the precision for $\tau^{(m)}$, using Eq.  \ref{eq:accuracygamma} again,
\begin{align}
	R^{(m)} &= \frac{1}{\frac{\chi}{\mu^2} - 1 } = \left( 1 + \frac{\mu^{(m-1)}}{\mu_m} \right)^2 \left( \frac{1}{R_m} + \frac{1}{R^{(m-1)}} \left( \frac{\mu^{(m-1)}}{\mu_m} \right)^2 \right)^{-1}.
\end{align}

To find the optimal precision $R^{(m)}$ given the sub-precisions $R^{(m-1)},R_m$, we optimize the above expression w.r.t. the ratio of means $\mu^{(m-1)}/\mu_m$, which lies in $(0,\infty)$. At the limit points of the ratio, when $\mu^{(m-1)}/\mu_m \rightarrow 0$, one recovers $R^{(m)} \rightarrow R_m$, whereas when $\mu^{(m-1)}/\mu_m \rightarrow \infty$, one recovers $R^{(m)} \rightarrow R^{(m-1)}$.

In between, there is a single extremal (maximum) value, found by differentiating w.r.t. $\mu^{(m-1)}/\mu_m$, corresponding to
\begin{align}
	\quad R^{(m)} &= R^{(m-1)} + R_m,
\end{align}
when the means satisfy
\begin{align}
	\frac{\mu^{(m-1)}}{R^{(m-1)}} &= \frac{\mu_m}{R_m} = \frac{ \mu^{(m-1)} + \mu_m }{ R^{(m-1)} + R_m} = \frac{\mu^{(m)}}{R^{(m)}}.
\end{align}

Thus in general $R^{(m)} \leq R^{(m-1)} + R_m$.

One continues by induction, maximizing the precision of $R^{(m-1)}$ by splitting $\tau^{(m-1)}$ into the convolution of $\tau^{(m-2)}$ and $\tau_{m-1}$. Analogously to the above, one obtains that
\begin{align}
	R^{(m)} \leq R^{(m-2)} + R_{m-1} + R_m,
\end{align}
with equality if and only if
\begin{align}
	\frac{\mu^{(m-2)}}{R^{(m-2)}} &= \frac{\mu_{m-1}}{R_{m-1}} = \frac{\mu_m}{R_m}.
\end{align}

Proceeding in the same manner, one recovers the lemma.

\end{proof}

\begin{definition}\label{def:partialnorm}

We define the \textbf{partial norm} $P_t[\tau]$of a delay function $\tau$ (Def. \ref{def:delayfunction}) to be the finite integral
\begin{align}
	P_t[\tau] &= \int_0^t \tau(t^\prime) dt^\prime.
\end{align}

Thus $P_t[\tau] \leq Q \leq 1$ for all $t$.
\end{definition}

\begin{lemma}\label{lemma:delaynorm}

Given a delay function $\tau^{(m)}$ that is a convolution of a sequence of delay functions (Eq. \ref{eq:delaysequence}, the partial norm (Def. \ref{def:partialnorm}) of $\tau^{(m)}$ is upper bounded by the products of the corresponding partial norms of the individual delay functions in the sequence, i.e.
\begin{align}
	P_t\left[ \tau^{(m)} \right] \leq \prod_{i=0}^m P_t \left[ \tau_i \right].
\end{align}

\end{lemma}

\begin{proof}

Proof by induction. For a sequence of a single delay function, the statement of the lemma holds trivially (and is an equality). Next consider that the statement is proven for a sequence of $m$ arbitrary delay functions. Appending a single delay function $\tau_{m+1}$, we have that the convolution and its partial norm are given by
\begin{align}
	\tau^{(m+1)} (t) &= \left( \tau^{(m)} \conv \tau_{m+1} \right) (t), \\
	P_t \left[ \tau^{(m+1)} \right] &= \int_0^t \left( \int_0^{t^\prime} \tau_{m+1} (t^\prime - t^{\prime \prime} ) \tau^{(m)} (t^{\prime\prime}) dt^{\prime\prime} \right) dt^\prime.
\end{align}
In the integral above, the argument of $\tau^{(m)}$ runs within the interval $[0,t^\prime]$, which is contained in $[0,t]$, as $t^\prime$ itself runs within $[0,t]$. Furthermore, the argument of $\tau_{m+1}$, which is $(t^\prime - t^{\prime\prime})$, is also constrained to be within the interval $[0,t]$. Thus we can upper bound the above integral by
\begin{align}
	P_t \left[ \tau^{(m+1)} \right] &\leq \int_0^t \int_0^t \tau_{m+1}(x) \tau^{(m)}(y) dy dx \\
	\therefore \; P_t \left[ \tau^{(m+1)} \right] &\leq P_t \left[ \tau_{m+1} \right] P_t \left[ \tau^{(m)} \right] \\
	&\leq P_t \left[ \tau_{m+1} \right] \prod_{i=0}^m P_t \left[ \tau_i \right] = \prod_{i=0}^{m+1} P_t \left[ \tau_i \right].
\end{align}

\end{proof}

\subsubsection{Combining delay functions in mixtures}

Another manner in which one can combine delay functions is by mixing them in a convex combination.

\begin{lemma}[The precision of a mixture is bounded by the best precision from among its components] \label{lemma:delaymixture}

\bigskip

Let a delay function be given by a sum of delay functions,
\begin{align}
	\tau(t) &= \sum_{i=1}^m \tau_i(t),
\end{align}
where each $\tau_i$ is a sub-normalized delay function (with non-zero zeroth moment), and the sum is also either normalised or sub-normalised. Then the precision $R$ of the mixture is upper bounded by
\begin{align}
	R \left[ \tau \right] \leq \max_i R_i,
\end{align}
where $R_i$ is the precision of the $i^{th}$ delay function $\tau_i$. Furthermore, this inequality is only saturated in the case that every delay function $\tau_i$ has the same mean and precision, i.e.
\begin{align}
	\mu_i = \mu_j \quad \text{and} \quad R_i = R_j \quad \forall i,j \in \{1,2,...,m\}.
\end{align}

\end{lemma}

\begin{proof}

We prove the statement by induction. Split the set of delay functions that comprise the mixture into two subsets, and label the partial sums as $\tau^{(1)}(t) = \sum_{i=1}^k \tau_i(t)$ and $\tau^{(2)}(t) = \sum_{i=k+1}^m \tau_i(t)$, thus
\begin{align}
	\tau(t) = \tau^{(1)}(t) + \tau^{(2)}(t).
\end{align}

We label the zeroth and (normalised) first moments of the two delay functions as $\{Q^{(1)},\mu^{(1)}\}$ and $\{Q^{(2)},\mu^{(2)}\}$ respectively (Eq. \ref{eq:delayfunctionmoments}), and take their respective $\gamma$-values (ref. Eq. \ref{eq:accuracygamma}) to be $\{\gamma^{(1)},\gamma^{(2)}\}$. Note that both $Q^{(1)},Q^{(2)} > 0$. Thus for the composite delay function, calculating the moments explicitly from Eq. \ref{eq:delayfunctionmoments},
\begin{align}
	Q &= \braket{t^0} = Q^{(1)} + Q^{(2)}, \\
	\mu &= \frac{\braket{t^1}}{Q} = \frac{ Q^{(1)} \mu^{(1)} + Q^{(2)} \mu^{(2)} }{ Q^{(1)} + Q^{(2)} }, \\
	\chi &= \frac{\braket{t^2}}{Q} = \frac{ Q^{(1)} \chi^{(1)} + Q^{(2)} \chi^{(2)} }{ Q^{(1)} + Q^{(2)} } = \frac{ Q^{(1)} [\mu^{(1)}]^2 \gamma^{(1)} + Q^{(2)} [\mu^{(2)}]^2 \gamma^{(2)} }{ Q^{(1)} + Q^{(2)} }
\end{align}

Calculating the $\gamma$-value of the composite delay function,
\begin{align}
	\gamma &= (Q^{(1)} + Q^{(2)})\frac{ Q^{(1)} [\mu^{(1)}]^2 \gamma^{(1)} + Q^{(2)} [\mu^{(2)}]^2 \gamma^{(2)}}{\left( Q^{(1)} \mu^{(1)} + Q^{(2)} \mu^{(2)} \right)^2}.
\end{align}

Denote $q^{(1)} = Q^{(1)}/(Q^{(1)} + Q^{(2)})$ and $q^{(2)} = Q^{(2)}/(Q^{(1)} + Q^{(2)})$, so that $q^{(1)},q^{(2)} >0$ and $q^{(1)} + q^{(2)} = 1$,
\begin{align}\label{eq:convexity}
	\gamma &= \frac{q^{(1)} [\mu^{(1)}]^2 \gamma^{(1)} + q^{(2)} [\mu^{(2)}]^2 \gamma^{(2)}}{\left( q^{(1)} \mu^{(1)} + q^{(2)} \mu^{(2)} \right)^2} \geq \frac{q^{(1)} [\mu^{(1)}]^2 \gamma^{(1)} + q^{(2)} [\mu^{(2)}]^2 \gamma^{(2)}}{q^{(1)} [\mu^{(1)}]^2 + q^{(2)} [\mu^{(2)}]^2}
\end{align}
by the convexity of the square. Next, denote $p^{(1)} = q^{(1)} [\mu^{(1)}]^2/(q^{(1)} [\mu^{(1)}]^2 + q^{(2)}[\mu^{(2)}]^2)$ and $p^{(2)} = q^{(2)} [\mu^{(2)}]^2/(q^{(1)}[\mu^{(1)}]^2 + q^{(2)}[\mu^{(2)}]^2)$, so that $p^{(1)},p^{(2)} >0$ and $p^{(1)} + p^{(2)} = 1$,
\begin{align}
	\gamma &\geq p^{(1)} \gamma^{(1)} + p^{(2)} \gamma^{(2)} \geq \min \{ \gamma^{(1)}, \gamma^{(2)} \}
\end{align}

As $\gamma$ is inversely related to $R$ (Eq. \ref{eq:accuracygamma}), it follows that
\begin{align}
	R \leq \max \{R^{(1)},R^{(2)}\}.
\end{align}

To saturate the upper bound, note that the inequality in Eq. \ref{eq:convexity} is only an equality if
\begin{align}
	\left( q^{(1)} \mu^{(1)} + q^{(2)} \mu^{(2)} \right)^2 = q^{(1)} [\mu^{(1)}]^2 + q^{(2)} [\mu^{(2)}]^2,
\end{align}
which by the strict convexity of the square function, is only satisfied when $\mu^{(1)} = \mu^{(2)}$. In this case, one has that
\begin{align}
	\gamma &= q^{(1)} \gamma^{(1)} + q^{(2)} \gamma^{(2)},
\end{align}
which is equal to the minimum from among $\{\gamma^{(1)},\gamma^{(2)}\}$ if and only if $\gamma^{(1)}=\gamma^{(2)}$. Thus to saturate the upper bound, both the mean and the precision of both delay functions must be equal, in which case, the mixture has the same mean and precision.

Proceeding by further splitting $\tau^{(1)}$ and $\tau^{(2)}$ until one recovers the original mixture, one arrives at the statement of the lemma.
\end{proof}

\subsubsection{Scaling the time-scale of a delay function - the invariance of the precision}

\begin{lemma}\label{lemma:delayscale}

If $\tau(t)$ is a delay function (Def. \ref{def:delayfunction}), then so is
\begin{align}
	\tau^\prime(t) &= a \tau(a t), \quad a>0.
\end{align}

Furthermore, the zeroth moment of $\tau^\prime$ is the same as that of $\tau$, the (normalised) first and second moments (ref. Eq. \ref{eq:delayfunctionmoments}) scale as $\{\mu/a,\chi/a^2\}$, and the precision $R$ is left unchanged.

\end{lemma}

\begin{proof}

Since $a>0$, $\tau^\prime(t)$ is a non-negative function. Calculating the zeroth moment,
\begin{align}
	Q^\prime &= \int_0^\infty a \tau(at) dt \\
		&= \int_0^\infty \tau(s) ds, \quad \text{where $s=at$} \\
		&= Q \leq 1.
\end{align}
Thus $\tau^\prime$ is also a delay function. For arbitrary moments, via the same change of variable, one finds that
\begin{align}
	\frac{\braket{t^n}^\prime}{Q^\prime} &= \frac{1}{a^n} \frac{\braket{t^n}}{Q}.
\end{align}
Thus $\mu^\prime = \mu/a$ and $\chi^\prime = \chi/a^2$. Finally, with respect to the $\gamma$-value (ref Eq. \ref{eq:accuracygamma}),
\begin{align}
	\gamma^\prime &= \frac{\chi^\prime}{\mu^{\prime 2}} = \frac{\chi/a^2}{\left(\mu/a\right)^2} = \frac{\chi}{\mu^2} = \gamma,
\end{align}
from which it follows that the precision $R^\prime = R$.

\end{proof}

\begin{remark}

The above lemma implies that the precision is independent of how quickly the event takes place, as would be measured by its frequency/resolution. In other words, the precision of a delay function is measured w.r.t. the natural timescale of the delay function itself, rather than an external reference.

\end{remark}

\section{An optimal classical clock - The Ladder Clock}\label{app:optimalclassical}

\subsection{The Ladder Clock achieves precision $R_j=j\, d$}\label{sec:example classical}

Here we discuss a simple classical clock that saturates the upper bound for the precision $R$ of classical clocks, and as far as we know, is the only one to do so. This is the classical clock used in Fig. \ref{fig:3clocks} a). This clock, more precisely, the discrete version of the continuous clock we discuss below, was introduced in \cite{ATGRandomWalk}. There it was shown to perform with a similar linear scaling in the dimension $d$, but for the alternative definition of precision via the Alternate Ticks Game \cite{RaLiRe15}. This optimal clock may also be approached thermodynamically as in \cite{Pauletal2017}, in the limit of semi-classical dynamics, and of infinite entropy cost.

As discussed in Corollary \ref{coll:classical clocks}, a classical clock is completely specified by a pair $\mathcal{N},\mathcal{T}\in \rr^{(d\times d)}$ of stochastic generators and the initial state of the clock, $V_0\in\cP_d$. In the case of the Ladder Clock, these are
\begin{align}\label{eq:stochastic gens foe exnaple classicla clock}
\mathcal{N}=
\begin{bmatrix}
-1 & 0 & \ldots & & \, & \ldots & 0 \\
 1 &-1 & 0 & \ldots & \, &\, & \vdots \\
 0 & 1 &-1 & 0 & \ldots & \, &\, \\
 \vdots & \ddots & \ddots & \ddots & \ddots &  & \, \\
 \, & \, & \ddots & \ddots &\ddots & \ddots & \vdots \\
  \vdots & \, & \, & 0 & 1 &-1 & 0 \\
  0  & \ldots &\, & \ldots & 0 & 1 &-1
\end{bmatrix}
,\quad\quad 
[\mathcal{T}]_{i,j}=
\begin{cases}
1 &\mbox{if } i=1,j=d\\ 
0 & \mbox{otherwise}
\end{cases},
\end{align}
and the initial state is chosen to be the vector of probabilities $[V_0(0)]_i=\delta_{i,1}$. Physically, the choice of tick generator means that we can only tick from the $d^\text{th}$ site, while the choice of initial state $V_0(0)$ means that the clock starts out with all its population on the 1st site. Furthermore, $\mathcal{N}$ generates movement from the basis state $i$ to the basis state $i+1$, and thus we call this clock the \textit{Ladder Clock}.

As the tick generator $\mathcal{T}$ is rank-1, it follows that the clock is a reset clock (Def. \ref{def:reset clock}). Each tick is thus identical and independent of the others, and the delay function of the $n^{th}$ tick of the clock is just the convolution of that of the first tick with itself $n$ times (for more details, see Appendix \ref{Sec:classical clocks theorem proof}, and specifically \ref{sec:justifyR}). From Remark \ref{remark:delaysequence}, we conclude that the precision of the $n^{th}$ tick is $n R_1$, where $R_1$ is the precision of a single tick.

Furthermore, one can understand a single tick itself as a sequence of $d$ identical and independent events, that of moving from site $1$ to site $2$, $2$ to $3$, and so on, with the final event being the tick moving the state from site $d$ back to site $1$. The precision of a single tick is thus $d$ times the precision of the delay function of a single site jump, which turns out to be $R=1$. In turn, the precision of a single tick of the clock is $R_1=d$, and that of the $n^{th}$ tick is $R_n = n d$.

We proceed to prove these statements from the analytical form of the clock state and delay function of the tick.

\subsection{The time evolution of the Ladder clock}

The unnormalised state of the ladder clock corresponding to the case in which no tick is generated during the time interval $[0,t]$ is given by the equation (for a more detailed discussion on this, see Section \ref{Sec:classical clocks theorem proof})
\be 
V_0(t)=\me^{t \mathcal{N}} V_0(0)= \me^{-t}\, \begin{bmatrix}
	1 & 0 & 0 &\ldots & \, & \,  & \, & \ldots & 0 \\
	f_2(t) & 1 & 0 & 0 & \, & \, & \, &\, & \vdots \\
	f_3(t) &f_2(t) & 1 & 0 & 0 & \, & \, & \, &\, \\
	\vdots &f_3(t)&f_2(t)& 1 & 0 & 0 & \, & \, & \, \\
	\, & \, & \ddots & \ddots &\ddots & \ddots & \ddots & \,& \,\\
	\, & \, & \, &f_3(t)&f_2(t)& 1 & 0 & 0 &\vdots  \\
	\,  & \, &\, & \, &f_3(t)&f_2(t)& 1 & 0 & 0\\
	\vdots & \, &\, &\, & \, &f_3(t)&f_2(t)& 1 & 0 \\
	f_d(t)& \ldots &\, &\, &\, & \ldots &f_3(t)&f_2(t)& 1
\end{bmatrix} 
\begin{bmatrix}
	1  \\
	0  \\
	\vdots  \\
	\vdots \\
	\vdots\\
	\vdots  \\
	\vdots\\
	\vdots \\
	0
\end{bmatrix}
=\me^{-t}\,
\begin{bmatrix}
	1  \\
	f_2(t)  \\
	f_3(t)  \\
	\vdots \\
	\vdots\\
	\vdots  \\
	\vdots\\
	\vdots \\
	f_d(t)
\end{bmatrix}
,\quad f_k(t):= \frac{t^{k-1}}{(k-1)!}.
\label{eq:V of t emaple classical clock}
\ee
A direct calculation of the mean $\mu(t)$ and standard deviation $\sigma(t)$ associated with the population at time $t$ in the limit $0\leq t\ll d$, using Eq. \eqref{eq:V of t emaple classical clock} yields
\ba
\mu(t)&= \sum_{n=1}^d n \,[ V_0(t) ]_n \\
\sigma(t)&=\sqrt{\sum_{n=1}^d (n-\mu(t))^2\, [ V_0(t) ]_n  }.
\ea
From these expressions we see that
\ba
\mu(t)& \approx \sum_{n=1}^\infty n \,[ V_0(t) ]_n= t+1,\quad & \sigma(t)\approx \sqrt{\sum_{n=1}^\infty (n-\mu(t))^2\, [ V_0(t) ]_n  }=\sqrt{t}.
\ea
In other words, the mean is approximately proportional to the distance travelled by the initial state, while the standard deviation is proportional to the square root of the distance travelled by the initial state. Physically, the reason why this is only approximately true for short times, is because when population reaches the last site $n=d$, the population is ``lost" from the no-tick space due to a tick happening.

From Eqs. \eqref{eq:stochastic gens foe exnaple classicla clock}, \eqref{eq:V of t emaple classical clock}, it follows that the delay function associated with the 1st tick, $P_{T_1}(t)$  (Defined in Section \ref{Accuracy of Quantum Clocks}) is
\be \label{eq:delay tau}
P_{T_1}(t)=\|\mathcal{T} V_0(t) \|_1= \me^{-t} \frac{t^{d-1}}{(d-1)!}.
\ee
A direct calculation of the mean $\mu_1$, standard deviation $\sigma_1$, and precision $R_1$; for these classical clocks according to the delay function Eq. \eqref{eq:delay tau}, yields
\be 
\mu_1=d,\quad \sigma_1=\sqrt{d}, \quad R_1=\frac{\mu_1^2}{\sigma_1^2}=d.
\ee 
Since the Ladder Clock is a re-set clock, it follows (see Section \ref{Sec:classical clocks theorem proof} for a proof) that the precision of the $j^\text{th}$ tick is $R_j= j R_1$, and thus this completes the proof of Eq. \eqref{eq:saturates accuracy classical} in Theorem \ref{thm_classicalbound}.

\section{Classically generated Markovian sequences}\label{app:markovianpart}

In This appendix we discuss Markovian dynamics of classical systems. We begin in \ref{app:markoviandynamics} by discussing classical finite-dimensional states, the differential equations that govern the generation of events, and their associated delay functions. We demonstrate the properties of the dynamical generators that we will require later in the proof. In \ref{app:eventsequences} we discuss sequences of events generated by Markovian dynamics, and some of their properties, continuing in \ref{app:indisequences} to differentiate between general sequences of events and \emph{independent} sequences, those in which each event is independent of the dynamics of the prior event. Finally in \ref{app:resetresult}, we prove that the precision of general sequences is upper bound by that of independent sequences, a result that will translate directly to the case of clocks, in which case reset clocks will upper bound the precision of general clocks.

\subsection{Events generated by classical (stochastic) dynamics}\label{app:markoviandynamics}

\subsubsection{Classical states are population vectors}

Here we define the state space for classical Markovian dynamics, of which clocks are a special case. In the main text (Section \ref{sec:classicalspecialcase}), we argued that classical clocks, which are a special case of quantum clocks restricted to a fixed orthonormal basis, may be described by stochastic dynamics on the vector of diagonal elements of the quantum state w.r.t. the restricted basis. In light of this description, we define classical states accordingly.

\begin{definition}\label{def:canonicalbasis}

The \textbf{canonical} basis is a preferred orthonormal basis $\{\mathbf{e}_i\}_{i=0}^d$ for a real vector space of finite dimension $d$. Each $\mathbf{e}_i$ is referred to as a \textbf{canonical state}.
	
\end{definition}

\begin{definition}\label{def:popvector}
	
A \textbf{classical state}, or \textbf{population vector}, is a vector $v \in \mathbb{R}^d$ all of whose elements w.r.t. the canonical basis (Def. \ref{def:canonicalbasis}) are greater or equal to zero, and whose sum of the elements in the canonical basis is less or equal to one. Thus $v = \sum_{i=1}^d v_i \mathbf{e}_i$ is a population vector if and only if
\begin{align}
	v_i &\geq 0 \quad \forall i, \\
	\text{and} \quad \sum_i v_i &\leq 1.
\end{align}
	
\end{definition}

\begin{remark}

In this appendix, we only deal with classical states, and so will omit the prefix classical for simplicity.

\end{remark}

\begin{definition}\label{def:norm}

For this appendix, we define the \textbf{e-sum} (that stands for ``element-sum") of a real vector $v \in \mathbb{R}^d$ as the sum of its elements w.r.t. the canonical basis (Def. \ref{def:canonicalbasis}). Thus if $v = \sum_{i=1}^d v_i \mathbf{e}_i$, then
\begin{align}
	\onenorm{v} &= \sum_{i=1}^d v_i.
\end{align}

\end{definition}

\begin{remark}

Note that for population vectors, the e-sum coincides with the 1-norm, as all of the elements in said basis are non-negative. We use the simple sum because we will deal with negative vectors as well, and the e-sum has the advantage of commuting with other finite sums and finite integrals. 

\end{remark}

\subsubsection{Generators of stochastic (Markovian) dynamics}

From standard probability theory, the differential equation for Markovian dynamics is given by
\begin{align}
	\frac{d}{dt} V(t) &= \hat{M} V(t),
\end{align}
where $V(t)$ is the state at time $t$ and $\hat{M}$ is a time-independent linear operator, referred to as the Kolmogorov generator of forward dynamics. If the dynamics are also responsible for the generation of events, then the generator $\hat{M}$ is split into
\begin{align}
	\hat{M} &= \notickmatrix + \tickmatrix,
\end{align}
where $\tickmatrix$ is associated with the generation of events, and the probability density (probability per unit time) that the event occurs at a time $t$ is given by
\begin{align}\label{eq:probdensityevent}
	P(t) &= \onenorm{\tickmatrix V(t)}.
\end{align}

We proceed to discuss the precise properties of these generators that we require in the proof.

\begin{definition}\label{def:stochasticgenerators}

By a \textbf{pair of stochastic generators}, we mean a pair of linear operators $\{ \notickmatrix, \tickmatrix \}$on a finite dimensional real-valued vector space $\mathbb{R}^d$, that we refer to as the \textbf{non-event generator and event generator} respectively, such that the following conditions are satisfied:
\begin{itemize}
	\item $\notickmatrix$ is an endomorphism, i.e. it takes vectors from $\mathbb{R}^d$ to $\mathbb{R}^d$, and when expressed as a matrix w.r.t. the canonical basis,
	\begin{align}\label{eq:nontickcondition}
		\notickmatrix_{ij} & \begin{cases}
				\leq 0 & \text{for $i=j$}, \\
				\geq 0 & \text{for $i\neq j$}.
			\end{cases}
	\end{align}
	\item $\tickmatrix$ takes vectors from $\mathbb{R}^d$ to a real vector space of possibly different dimension $\mathbb{R}^{d^\prime}$. Expressed as a matrix w.r.t. the canonical bases of both (input and output) spaces, $\tickmatrix$ has $d^\prime$ rows and $d$ columns, and must be element-wise non-negative, i.e.
	\begin{align}\label{eq:tickcondition}
		\tickmatrix_{ij} &\geq 0 \quad \forall i,j.
	\end{align}
	\item The pair of generators must satisfy
	\begin{align}\label{eq:stochasticcondition}
		\sum_{i=1}^{d^\prime} \tickmatrix_{ij} + \sum_{i=1}^d \notickmatrix_{ij} &\leq 0 \quad \forall j.
	\end{align}
\end{itemize}

\end{definition}

\begin{remark}

The non-negative nature of $\tickmatrix$ follows from the requirement that the probability density of the event (Eq. \ref{eq:probdensityevent}) must be non-negative, and the state following the event must also be a population vector. The off-diagonal elements of both $\notickmatrix$ and $\tickmatrix$ represent transition probabilities between different canonical states, and are therefore also non-negative. Finally, the condition on the column sums (Eq. \ref{eq:stochasticcondition}) is necessary for the e-sum norm of the state to be non-increasing, from which the condition on the diagonal elements of $\notickmatrix$ follows.

\end{remark}

\begin{remark}

We use the more general norm non-increasing rather than norm-preserving generators in the appendix because 1) the results apply to the more general case anyway, and 2) we will require results on norm non-increasing dynamics later in the proof.

\end{remark}

\begin{corollary}\label{cor:zerotickmatrix}

If $\{\notickmatrix,\tickmatrix\}$ are a pair of stochastic generators, then so is $\{\notickmatrix,\hat{\mathbf{0}}\}$, where $\hat{\mathbf{0}}$ is the zero operator with the same input and output spaces as $\tickmatrix$. 

\end{corollary}

\begin{proof}

The zero operator clearly satisfies Eq. \ref{eq:tickcondition}. Furthermore, we verify that Eq. \ref{eq:stochasticcondition} is satisfied for the new pair,
\begin{align}
	\sum_{i=1}^d \notickmatrix_{ij} &\leq - \left( \sum_{i=1}^{d^\prime} \tickmatrix_{ij} \right) < 0 \quad \forall j,
\end{align}
as all of the elements of $\tickmatrix$ are non-negative.

\end{proof}

\begin{remark}

Given an linear endomorphism $\notickmatrix$, we consider it to be a non-event generator if it forms a pair of stochastic generators with the zero operator, which is equivalent to requiring it to for a pair of stochastic generators with at least one element wise generator $\tickmatrix$. On the other hand, every linear operator $\tickmatrix$ that is element-wise non-negative w.r.t. the canonical bases for the input and output space is an event generator, as one can always find a non-event generator with diagonal elements that are negative enough so that the two taken together form a pair of stochastic generators.

\end{remark}

\begin{corollary}\label{cor:zerogenerators}

If $\{\notickmatrix,\tickmatrix\}$ are a pair of stochastic generators (Def. \ref{def:stochasticgenerators}) and $\tickmatrix$ is an endomorphism, i.e. it takes vectors from $\mathbb{R}^d$ to $\mathbb{R}^d$, then the pair $\{\notickmatrix + \tickmatrix, \hat{\mathbf{0}} \}$ is also a pair of stochastic generators. Equivalently, $\notickmatrix + \tickmatrix$ is a non-event generator.

\end{corollary}

\begin{proof}

Eq. \ref{eq:stochasticcondition} is trivially satisfied, while the zero operator by definition satisfies Eq. \ref{eq:tickcondition}. The off-diagonal elements of $\notickmatrix$ and $\tickmatrix$ are individually non-negative, and thus the same holds for the sum, satisfying the second condition in Eq. \ref{eq:nontickcondition}. 

Finally, to verify the condition on the diagonal elements (first line of Eq. \ref{eq:nontickcondition}),
\begin{align}
	\notickmatrix_{jj} + \tickmatrix_{jj} &= \sum_{i=1}^d \left( \notickmatrix_{ij} + \tickmatrix_{ij} \right) - \sum_{i\neq j}^d \left( \notickmatrix_{ij} + \tickmatrix_{ij} \right) < 0 \quad \forall j,
\end{align}
as the first sum on the right is non-positive by Eq. \ref{eq:stochasticcondition}, and the second (subtracted) sum comprises solely non-negative numbers, by Eqs. \ref{eq:nontickcondition} and \ref{eq:tickcondition}.

\end{proof}

\begin{corollary}\label{cor:rowgenerators}

If $\{\notickmatrix,\tickmatrix\}$ are a pair of stochastic generators, then so are all of the pairs $\{\notickmatrix,\tickmatrix_j\}$, where $\tickmatrix_j$ is constructed by setting all of the rows in $\tickmatrix$ to zero, save the $j^{th}$ one.

\end{corollary}

\begin{proof}

Eq. \ref{eq:nontickcondition} is satisfied trivially as the non-event generator $\notickmatrix$ is left unchanged. Setting rows to zero in $\tickmatrix$ also maintains Eq. \ref{eq:tickcondition}. Finally, for Eq. \ref{eq:stochasticcondition}, for the $j^{th}$ column, the modified LHS of the expression is missing some non-zero elements from the original event generator $\tickmatrix$, and is therefore less or equal than the original expression, still satisfying the inequality.

\end{proof}

\begin{corollary}\label{cor:scaledgenerators}

If $\{\notickmatrix,\tickmatrix\}$ are a pair of stochastic generators, then so is $\{a\notickmatrix,a\tickmatrix\}$, where $a>0$.

\end{corollary}

\begin{proof}

By direct substitution into Def. \ref{def:stochasticgenerators}, one verifies that all of the inequalities are maintained as the scaling factor $a$ is positive.

\end{proof}

\begin{lemma}[From infinitesimal generators to finite-time transition matrices]\label{lemma:stochasticgenerators}

Given any stochastic non-event generator $\notickmatrix$, and a population vector $V$,
\begin{itemize}
	\item $e^{\notickmatrix t}$ is a transition matrix \cite{probability} for all $t\geq 0$, i.e. If $\mathcal{M} = e^{\notickmatrix t}$, then
	\begin{subequations}
		\begin{align}
			\mathcal{M}_{ij} &\geq 0 \quad \forall i,j, \\
			\sum_i \mathcal{M}_{ij} &\leq 1 \quad \forall j.
		\end{align}
	\end{subequations}
	\item $e^{\notickmatrix t} V$ is a population vector (Def. \ref{def:popvector}) for all $t\geq 0$, and its norm (e-sum) is non-increasing in time.
\end{itemize}

\end{lemma}

\begin{proof}

One of the definitions of the exponential function is via
\begin{align}
	e^{\notickmatrix t} &= \lim_{m\rightarrow\infty} \left( \mathds{1} + \frac{\notickmatrix t}{m} \right)^m
\end{align}

The matrix within the parentheses has non-negative off-diagonal entries since $\notickmatrix_{ij} > 0$ for $i\neq j$. Furthermore, for the diagonal entries, one can always pick $m$ large enough so that $\mathds{1} + t \notickmatrix_{ii}/m > 0$, and thus for large enough $m$, the above expression is a positive power of a matrix with non-negative entries, which must thus be non-negative element-wise. Thus $e^{\notickmatrix t}$ is element wise non-negative.

To prove that the column sums of $\mathcal{M}$ are less or equal to $1$, we first prove the properties of $e^{\notickmatrix t} V$. As $V$ is a population vector, all of its elements are non-negative, and thus the elements of $e^{\notickmatrix t} V$ are also non-negative. Labelling this vector as $V(t) = \sum_i V_i(t) \mathbf{e}_i$ where each $V_i(t) \geq 0$, we may calculate the rate of change of its e-sum (Def. \ref{def:norm}),
\begin{align}
	\frac{d}{dt} \onenorm{V(t)} &= \onenorm{ \frac{d}{dt} e^{\notickmatrix t} V } \\
	&= \onenorm{ \notickmatrix V(t) } \\
	&= \sum_{ij} \notickmatrix_{ij} V_j(t) \\
	&= \sum_j \left( \sum_i \notickmatrix_{ij} \right) V_j \leq 0
\end{align}
as each of the column sums of $\notickmatrix$ is non-positive (Eq. \ref{eq:nontickcondition}). Thus $V(t)$ is also a population vector for all $t \geq 0$, with non-increasing norm. Finally take $V = \mathbf{e}_j$. As the norm of $V(t)$ is less or equal to $1$ for all $t \geq 0$, one has for the transition matrix $\mathcal{M} = e^{\notickmatrix t}$ that
\begin{align}
	\onenorm{V(t)} &= \onenorm{\mathcal{M} \mathbf{e}_j} \\
	&= \sum_i \mathcal{M}_{ij} \leq 1.
\end{align}

\end{proof}

\subsection{Sequences of events}\label{app:eventsequences}

\begin{definition}\label{def:markoviansequence}

By a \textbf{Markovian sequence of events}, we mean a sequence, finite or infinite, of time-dependent population vectors $\{ V^{(n)}(t) : \mathbb{R}^+ \rightarrow \mathbb{R}^{d^{(n)}} \}$ together with a corresponding sequence of pairs of stochastic generators $\{ \notickmatrix^{(n)} : \mathbb{R}^{d^{(n)}} \rightarrow \mathbb{R}^{d^{(n)}}, \tickmatrix^{(n)} : \mathbb{R}^{d^{(n)}} \rightarrow \mathbb{R}^{d^{(n+1)}} \}$, the dynamics of which are given by
\begin{align}\label{eq:markovdynamics}
	\frac{d}{dt} V^{(n)}(t) &= \begin{cases}
		\notickmatrix^{(0)} V^{(0)}(t) & \text{for $n=0$}, \\
		\notickmatrix^{(n)} V^{(n)}(t) + \tickmatrix^{(n-1)} V^{(n-1)}(t) & \text{for $ n \geq 1 $},
	\end{cases}
\end{align}
and that satisfy the initial conditions
\begin{itemize}
	\item $V^{(0)}(0)$ is a normalised population vector,
	\item $V^{(n)}(0) = \mathbf{0}$ for all $n \geq 1$.
\end{itemize}

Furthermore, we define the \textbf{delay function of the $n^{th}$ event} as
\begin{align}\label{eq:sequencedelay}
	\tau^{(n)}(t) &= \onenorm{ \tickmatrix^{(n-1)} V^{(n-1)}(t) },
\end{align}
for $n \geq 1$. For $n=0$, the delay function is $\tau^{(0)}(t) = \delta (t)$, the Dirac-delta distribution.

\end{definition}

\begin{remark}

Note that the fact that the $V^{(n)}$ are population vectors and that the $\tau^{(n)}$ are delay functions is not immediate from the above definition, but is nevertheless implied, as the following lemmas will establish.

\end{remark}

\begin{lemma}[Events in a Markovian sequence are recursive convolutions]\label{lemma:sequenceconvolution}

For a Markovian sequence of events (Def. \ref{def:markoviansequence}), each vector in the sequence can be expressed as a convolution w.r.t. the previous one, i.e. for $n \geq 1$,
\begin{align}\label{eq:sequenceconvolution}
	V^{(n)}(t) &= \int_0^t e^{\notickmatrix^{(n)} \left( t - t^\prime \right)} \tickmatrix^{(n-1)} V^{(n-1)} \left( t^\prime \right) dt^\prime,
\end{align}
and for $n=0$,
\begin{align}
	V^{(0)}(t) &= e^{\notickmatrix^{(0)} t} V^{(0)}(0).
\end{align}
\end{lemma}

\begin{proof}

Proof by induction. Consider we define the sequence $f^{(n)}(t)$ by
\begin{align}\label{eq:prooftempsequence}
	f^{(n)}(t) &= \begin{cases}
		e^{\notickmatrix^{(0)} t} V^{(0)}(0) & \text{for $n=0$}, \\
		\int_0^t e^{\notickmatrix^{(n)} \left( t - t^\prime \right)} \tickmatrix^{(n-1)} f^{(n-1)} \left( t^\prime \right) dt^\prime & \text{for $n \geq 1$.}
	\end{cases}
\end{align}

For $n=0$, by construction, $f^{(0)}(t)$ has the same initial conditions and obeys the same differential equation as $V^{(0)}(t)$, and is thus equal to $V^{(0)}(t)$.

Proceeding, let $f^{(n-1)}(t) = V^{(n-1)}(t)$ for some $n \geq 1$. Then differentiating $f^{(n)}(t)$ from Eq. \ref{eq:prooftempsequence} using Leibniz's rule, we find that
\begin{align}
	\frac{d}{dt} f^{(n)}(t) &= \frac{d}{dt} \int_0^t e^{\notickmatrix^{(n)} \left( t - t^\prime \right)} \tickmatrix^{(n-1)} f^{(n-1)} \left( t^\prime \right) dt^\prime \\
	&= \notickmatrix^{(n)} \int_0^t e^{\notickmatrix^{(n)} \left( t - t^\prime \right)} \tickmatrix^{(n-1)} f^{(n-1)} \left( t^\prime \right) dt^\prime + \tickmatrix^{(n-1)} f^{(n-1)} \left( t \right) \\
	&= \notickmatrix^{(n)} f^{(n)}(t) + \tickmatrix^{(n-1)} V^{(n-1)}(t),
\end{align}
and thus $f^{(n)}(t)$ obeys the same differential equation as $V^{(n)}(t)$. Their initial conditions are also the same as $f^{(n)}(0) = \mathbf{0} = V^{(n)}(0)$ for $n \geq 1$. Thus $f^{(n)}(t) = V^{(n)}(t)$, and the rest follows by induction.

\end{proof}

\begin{lemma}[States and delay functions of a Markovian sequence of events]\label{lemma:markovstatedelay}

For a Markovian sequence of events (Def. \ref{def:markoviansequence}),
\begin{itemize}
	\item Every vector $V^{(n)}(t)$ in the sequence is a population vector (Def. \ref{def:popvector}).
	\item Every $\tau^{(n)}(t)$ is a delay function (Def. \ref{def:delayfunction}).
\end{itemize}

\end{lemma}

\begin{proof}

We first prove by induction that the elements of every $V^{(n)}(t)$ w.r.t. the canonical basis are non-negative. At $t=0$, this is true by definition. Furthermore, by Lemma \ref{lemma:sequenceconvolution}, $V^{(0)}(t)$ evolves via the transition operator $e^{\notickmatrix^{(0)} t}$, and is thus always non-negative, by Lemma \ref{lemma:stochasticgenerators}. Proceeding to $V^{(n)}(t)$ from Eq. \ref{eq:sequenceconvolution},
\begin{align}
	V^{(n)}(t) &= \int_0^t e^{\notickmatrix^{(n)} \left( t - t^\prime \right)} \tickmatrix^{(n-1)} V^{(n-1)} \left( t^\prime \right) dt^\prime,
\end{align}
note that if $V^{(n-1)}(t)$ is non-negative for all $t$, then so is $\tickmatrix^{(n-1)} V^{(n-1)}(t)$ as the elements of $\tickmatrix^{(n-1)}$ are all non-negative (Def. \ref{def:stochasticgenerators}), and since this is multiplied by another transition matrix (Lemma \ref{lemma:stochasticgenerators}), whose elements are all non-negative, the integrand is non-negative, and thus $V^{(n)}(t)$ is also non-negative w.r.t. the canonical basis for $t \geq 0$.

To prove that each vector is a population vector, i.e. has an e-sum (Def. \ref{def:norm}) less or equal to $1$, we prove the statement for the sum of e-sums $\sum_{n=0}^m \onenorm{V^{(n)}(t)}$, $m \geq 0$, from which the weaker statement follows.
\begin{align}
	\frac{d}{dt} \sum_{n=0}^m \onenorm{ V^{(n)}(t) } &=  \sum_{n=0}^m  \onenorm{ \frac{d}{dt} V^{(n)}(t) } \\
	&= \onenorm{ \notickmatrix^{(0)} V^{(0)}(t) } + \sum_{n=1}^m \left( \onenorm{ \notickmatrix^{(n)} V^{(n)}(t) } + \onenorm{ \tickmatrix^{(n-1)} V^{(n-1)}(t) } \right) \\
	&= \sum_{n=0}^{m-1} \left( \onenorm{ \notickmatrix^{(n)} V^{(n)}(t)} + \onenorm{ \tickmatrix^{(n)} V^{(n)}(t) } \right) + \onenorm{ \notickmatrix^{(m)} V^{(m)}(t) } \\
	&= \sum_{n=0}^{m-1} \sum_{j=1}^{d^{(n)}} \left( \sum_{i=1}^{d^{(n)}} \mat{\notickmatrix^{(n)}}_{ij} + \sum_{i=1}^{d^{(n+1)}} \mat{\tickmatrix^{(n)}}_{ij} \right) \mat{V^{(n)}(t)}_j + \sum_{j=1}^{d^{(m)}} \left( \sum_{i=1}^{d^{(m)}} \mat{\notickmatrix^{(m)}}_{ij} \right) \mat{V^{(m)}(t)}_j \label{eq:proofseqsum},
\end{align}
where we have used the definition of the e-sum, Def. \ref{def:norm}. The RHS of the above equation is non-positive, as the two column sums in parentheses above are non-positive (see Def. \ref{def:stochasticgenerators} and Corollary \ref{cor:zerotickmatrix}). Thus
\begin{align}
	\frac{d}{dt} \onenorm{ \sum_{n=0}^m V^{(n)}(t) } &\leq 0,
\end{align}
and thus the sum of the norms is non-increasing. At $t=0$ the sum is $1$ by definition (Def. \ref{def:markoviansequence}), and thus the sum is less or equal to unity for $t \geq 0$, from which it follows that each term in the sum, which is also non-negative, must also be between $0$ and $1$. Thus every $V^{(n)}$ is a population vector for all $t \geq 0$.

\bigskip

To prove that every $\tau^{(m)}$ is a delay function, note in Eq. \ref{eq:sequencedelay} that as $\tickmatrix^{(m-1)}$ and $V^{(m-1)}$ are element-wise non-negative and bounded, it follows that $\tau^{(m)}$ is non-negative and bounded as well. To prove integrability, we add up Eq. \ref{eq:sequencedelay} (with $n=m+1$) and Eq. \ref{eq:proofseqsum} to get
\begin{align}
	\tau^{(m+1)}(t) + \frac{d}{dt} \sum_{n=0}^m \onenorm{ V^{(n)}(t) } &= \sum_{n=0}^{m} \sum_{j=1}^{d^{(n)}} \left( \sum_{i=1}^{d^{(n)}} \mat{\notickmatrix^{(n)}}_{ij} + \sum_{i=1}^{d^{(n+1)}} \mat{\tickmatrix^{(n)}}_{ij} \right) \mat{V^{(n)}(t)}_j,
\end{align}
where $m \geq 0$. The RHS is non-positive from the properties of the column sums of the generators, (Eq. \ref{eq:stochasticcondition}), and we label it as the function $g(t)$. Thus $g(t) \leq 0$.

Integrating from $t=0$ to $t=T$, one obtains
\begin{align}
	\int_0^T \tau^{(m+1)}(t) dt &= - \int_0^T \frac{d}{dt} \sum_{n=0}^m \onenorm{ V^{(n)}(t) } dt + \int_0^T g(t) dt \\
	&= \sum_{n=0}^m \onenorm{ V^{(n)}(0) } - \onenorm{ \sum_{n=0}^m V^{(n)}(T) } + \int_0^T g(t) dt \\
	&= 1 - \sum_{n=0}^m \onenorm{ V^{(n)}(T) } + \int_0^T g(t) dt,
\end{align}
using the initial conditions for the sequence of states (Def. \ref{def:markoviansequence}). As $g(t) \leq 0$ for all $t \geq 0$, the RHS of the above expression is less or equal to $1$ for all $T$. However, the integrand on the LHS is non-negative, and thus the integral is non-decreasing w.r.t. $T$. It follows from the monotone convergence theorem\cite{rudin1976principles} that the limit $\lim_{T\rightarrow\infty} \int_0^T \tau^{(m+1)}(t) dt$ exists and is less or equal to $1$.

Thus every $\tau^{(m+1)}$ for $m \geq 0$, or equivalently, every $\tau^{(m)}$ for $m \geq 1$, is a delay function.

\end{proof}

\begin{corollary}\label{cor:validdelayfunction}

Given a single pair of stochastic generators $\{\notickmatrix,\tickmatrix\}$ acting on $\mathbb{R}^d$, together with an arbitrary population vector $V \in \mathbb{R}^d$, the following is a delay function,
\begin{align}
	\tau(t) &= \onenorm{ \tickmatrix e^{\notickmatrix t} V }.
\end{align}

\end{corollary}

\begin{proof}

Take an arbitrary Markovian sequence (Def. \ref{def:markoviansequence}) with the choice of $V^{(0)}(0) = V$,  $\tickmatrix^{(0)} = \tickmatrix$, and $\notickmatrix^{(0)} = \notickmatrix$. Apply Lemma \ref{lemma:markovstatedelay} to deduce that $\tau^{(1)}$ is a delay function, and substitute the solution for the time-evolved state $V^{(0)}(t)$ from Lemma \ref{lemma:sequenceconvolution} to the expression of $\tau^{(1)}(t)$ from Def. \ref{def:markoviansequence}.

\end{proof}

\begin{lemma}[Invariance of the dynamics of an event w.r.t. the following event.]\label{lemma:resetinvariance}
	
For any Markovian sequence of events, the state and delay function of the $n^{th}$ event in the sequence (as well as those of every previous event) is invariant w.r.t. a change in the $n^{th}$ event generator $\tickmatrix^{(n)}$ that leaves its column sums unchanged, i.e. if $\tickmatrix^{(n)^\prime}$ is chosen such that its matrix elements satisfy
\begin{align}
	\sum_i \mat{\tickmatrix^{(n)\prime}}_{ij} &= \sum_i \mat{\tickmatrix^{(n)}}_{ij} \quad \forall j.
\end{align}
	
\end{lemma}

\begin{proof}

From Lemma \ref{lemma:sequenceconvolution}, it is clear that the state $V^{(n)}(t)$ corresponding to the $n^{th}$ event is independent of $\tickmatrix^{(n)}$, and does not affect the states and delay functions of any of the events prior to the $n^{th}$. The independence of the $V^{(i)}(t)$ and $\tau^{(i)}(t)$, where $i<n$ follows from the fact that the initial conditions and differential equations for all of the $i<n$ are independent of $i \geq n$.

As for the delay function of the $n^{th}$ event, from Eq. \ref{eq:sequencedelay},
\begin{align}
	\tau^{(n)}(t) &= \onenorm{ \tickmatrix^{(n-1)} V^{(n-1)}(t) } \\
	&= \sum_{ij} \mat{\tickmatrix^{(n-1)}}_{ij} \mat{V^{(n-1)}(t)}_j = \sum_j \left( \sum_i \mat{\tickmatrix^{(n-1)}}_{ij} \right) \mat{V^{(n-1)}(t)}_j,
\end{align}
and is thus invariant under any operation that leaves the column sums of $\tickmatrix^{(n-1)}$ unchanged.
	
\end{proof}

\begin{remark}

Lemma \ref{lemma:resetinvariance} reflects the fact that the tick generator $\tickmatrix^{(n)}$ encodes both the amount of probability of the event being generated from each canonical state (as reflected by each of the column sums) as well as the state following the generation of an event (encoded in the individual elements of $\tickmatrix^{(n)}$), the first of which affects the dynamics of the $n^{th}$ event, while the second only affects subsequent events.

\end{remark}

\subsection{Independent sequences of events}\label{app:indisequences}

Here we discuss a special class of Markovian sequences, those in which the state of the system immediately following an event is a fixed state, invariant w.r.t. the state of the system prior to the event. Such sequences have useful properties that we use later in the proof.

\begin{definition}\label{def:independentsequence}

A Markovian sequence of events (Def. \ref{def:markoviansequence}) is called an \textbf{independent Markovian sequence of events} if every one of its event generators $\tickmatrix^{(n)}$ is a rank-1 linear operator.
	
\end{definition}

\begin{remark}[\it Properties of an independent sequence]

If $\tickmatrix$ is a rank-1 linear operator, then all of its columns are proportional to one non-zero column, and therefore, given any vector $V$ in the domain of $\tickmatrix$, the product $\tickmatrix V$ is also proportional to the same fixed column. For a Markovian sequence of events, the term in the dynamical equation $\tickmatrix^{(n-1)} V^{(n-1)}(t)$ (see Eqs. \ref{eq:markovdynamics} and \ref{eq:sequencedelay}) represents the state following the occurrence of the $n^{th}$ event, and thus an independent sequence is one in which this product is proportional to some fixed state for all $t\geq 0$.
\end{remark}

\begin{definition}\label{def:eventresetstate}

For an independent Markovian sequence of events (Def. \ref{def:markoviansequence}), we define the sequence of \textbf{event reset states} $\{V_R^{(n)}\}$ as follows: for each $n$, $V_R^{(n)}$ is the unique normalised vector that every column of $\tickmatrix^{(n)}$ is proportional to.

\end{definition}

\begin{remark}

Note that as every element of $\tickmatrix^{(n)}$ is non-negative, it follows that every event reset state is a normalised population vector (Def. \ref{def:popvector}).

\end{remark}

\begin{lemma}\label{lemma:eventresetstate}

For an independent Markovian sequence,
\begin{align}
	\tickmatrix^{(n-1)} V^{(n-1)}(t) &= \tau^{(n)}(t) V_R^{(n)}.	
\end{align}

\end{lemma}

\begin{proof}
	The proof follows from the definition of the event reset state, Def. \ref{def:eventresetstate} and the event delay function, Def. \ref{def:markoviansequence}, Eq. \ref{eq:sequencedelay}.
\end{proof}

\begin{definition}\label{def:canonicalsequence}

A \textbf{canonical independent Markovian sequence of events} is an independent sequence (Def. \ref{def:independentsequence}) in which the initial state and every reset state (Def. \ref{def:eventresetstate}) in the sequence is a canonical state (Def. \ref{def:canonicalbasis}). Alternatively, it is a sequence in which the initial state is a canonical state, and every event generator $\tickmatrix^{(n)}$ has a single non-zero row.

\end{definition}

\begin{definition}\label{def:resetsubeventdelay}

For an independent sequence of Markovian events (Def. \ref{def:independentsequence}), one defines the \textbf{sub-event delay function} as
\begin{align}
	\nu^{(n)}(t) &= \onenorm{ \tickmatrix^{(n-1)} e^{\notickmatrix^{(n-1)} t} V_R^{(n-1)} }.
\end{align}

\begin{remark}

The $n^{th}$ sub-event delay function refers to the delay function of the $n^{th}$ event \textit{from the time of occurrence of event $n-1$}, as opposed to the $n^{th}$ delay function, which is constructed w.r.t. the initial time $t=0$.

\end{remark}

\begin{remark}

That $\nu^{(n)}$ is indeed a delay function follows from Corollary \ref{cor:validdelayfunction}. Note that $\nu^{(n)}(t)$ is independent of all other quantities in the sequence other than those that determine the occurrence of the $n^{th}$ event, namely the stochastic generators $\{\notickmatrix^{(n-1)},\tickmatrix^{(n-1)}\}$ and the event reset state $V_R^{(n-1)}$.

\end{remark}

\end{definition}

\begin{lemma}[The states and delay functions of an independent Markovian sequence]\label{lemma:resetsequenceconvolution}

For an independent Markovian sequence of events (Def. \ref{def:independentsequence}), every vector and delay function in the sequence can be expressed recursively for $n \geq 1$ as
\begin{align}
	V^{(n)}(t) &= \int_0^t e^{\notickmatrix^{(n)} (t-t^\prime)} V_R^{(n)} \tau^{(n)}(t^\prime) dt^\prime \label{eq:resetsequencestate} \\
	\tau^{(n)}(t) &= \left( \tau^{(n-1)} \conv \nu^{(n)} \right) (t), \label{eq:resetsequencedelay}
\end{align}
where $\nu^{(n)}(t)$ is the $n^{th}$ sub-event delay function, Def. \ref{def:resetsubeventdelay}. Note that $\tau^{(0)}(t) = \delta(t)$ by definition (Def. \ref{def:markoviansequence}), and we equate the reset state for the zeroth event $V_R^{(0)}$ to the initial state $V^{(0)}(0)$.

It follows that the delay function of the $n^{th}$ event is the sequential convolution of the delay function of all previous events,
\begin{align}\label{eq:resetsequencedelayconvolution}
	\tau^{(n)}(t) &= \left( \nu^{(1)} \conv \nu^{(2)} \conv ... \conv \nu^{(n)} \right) (t),
\end{align}
and the state of the $n^{th}$ event is similarly given by the sequential convolution
\begin{align}\label{eq:resetsequencestateconvolution}
	V^{(n)}(t) &= \int_0^t e^{\notickmatrix^{(n)} (t-t^\prime)} V_R^{(n)} \left( \nu^{(1)} \conv \nu^{(2)} \conv ... \conv \nu^{(n)} \right) (t^\prime) dt^\prime.
\end{align}

\end{lemma}

\begin{proof}

The first statement of the Lemma, Eq. \ref{eq:resetsequencestate}, follows from Lemmas \ref{lemma:sequenceconvolution} and \ref{lemma:eventresetstate}. For the second statement (Eq. \ref{eq:resetsequencedelay}), we take the definition of the delay function (Eq. \ref{eq:sequencedelay}). For $n \geq 1$,
\begin{align}
	\tau^{(n)}(t) &= \onenorm{ \tickmatrix^{(n-1)} \int_0^t e^{\notickmatrix^{(n-1)} (t-t^\prime)} V_R^{(n-1)} \tau^{(n-1)}(t^\prime) dt^\prime } \\
	&= \int_0^t \onenorm{ \tickmatrix^{(n-1)} e^{\notickmatrix^{(n-1)} (t-t^\prime)} V_R^{(n-1)} } \tau^{(n-1)}(t^\prime) dt^\prime \\
	&= \int_0^t \nu^{(n)}(t-t^\prime) \; \tau^{(n-1)}(t^\prime) dt^\prime \\
	&= \left( \tau^{(n-1)} \conv \nu^{(n)} \right) (t),
\end{align}
where $\nu^{(n)}$ is the sub-event delay function, Def. \ref{def:resetsubeventdelay}. Proceeding by induction, one recovers Eqs. \ref{eq:resetsequencedelayconvolution} and \ref{eq:resetsequencestateconvolution}.

\end{proof}

\begin{lemma}\label{lemma:resetsequencescale}

In an independent Markovian sequence of events (Def. \ref{def:independentsequence}), if a pair of stochastic generators is scaled as
\begin{align}
	\{\notickmatrix^{(n)},\tickmatrix^{(n)}\} \rightarrow  	\{a\notickmatrix^{(n)},a\tickmatrix^{(n)}\},
\end{align}
where $a>0$, the delay function of the $m^{th}$ event in the sequence, where $m>n$, is modified to be
\begin{align}
	\tau^{(m)}(t) &= \left( \nu^{(1)} \conv \nu^{(2)} \conv ... \conv \tilde{\nu}^{(n+1)} \conv ... \conv \nu^{(m)} \right) (t), \\
	\text{where} \quad \tilde{\nu}^{(n+1)}(t) &= a \cdot \nu^{(n+1)}(at).
\end{align}

For $m \leq n$, the delay function is left unchanged.

\end{lemma}

\begin{proof}

The sequence is still independent, as scaling up the event generator $\tickmatrix^{(n)}$ by a positive constant does not change its rank. Thus Lemma \ref{lemma:resetsequenceconvolution} still applies, and by Eq. \ref{eq:resetsequencedelayconvolution}, one only has to consider the modification to the $(n+1)^{th}$ sub-event delay function, (see Def. \ref{def:resetsubeventdelay}),
\begin{align}
	\tilde{\nu}^{(n+1)} &= \onenorm{ a \cdot \tickmatrix^{(n-1)} e^{a \cdot \notickmatrix^{(n-1)} t} V_R^{(n-1)} } \\
	 &= a \onenorm{ \tickmatrix^{(n-1)} e^{\notickmatrix^{(n-1)} (at)} V_R^{(n-1)} } \\
	 &= a \cdot \nu^{(n+1)}(at).
\end{align}

\end{proof}

\begin{lemma}[The delay function of a sequence in terms of canonical independent sequences]\label{lemma:generalindependent}

Given a Markovian sequence (Def. \ref{def:markoviansequence}), the $n^{th}$ delay function and state (for $n \geq 1$ in the sequence may be expressed in terms of convolutions of sub-event delay functions (Def. \ref{def:resetsubeventdelay}) as follows:
\begin{align}
	\tau^{(n)}(t) &= \sum_{i_0=1}^{d^{(0)}} \sum_{i_1=1}^{d^{(1)}} ... \sum_{i_n=1}^{d^{(n)}} v_{i_0} \left( \nu^{(1)}_{i_0i_1} \conv \nu^{(2)}_{i_1i_2} \conv ... \conv \nu^{(1)}_{i_{n-1}i_n} \right)(t), \\
	V^{(n)}(t) &= \int_0^t e^{\notickmatrix^{(n)} (t-t^\prime)}  \sum_{i_0=1}^{d^{(0)}} \sum_{i_1=1}^{d^{(1)}} ... \sum_{i_n=1}^{d^{(n)}} \mathbf{e}^{(n)}_{i_n} v_{i_0} \left( \nu^{(1)}_{i_0i_1} \conv \nu^{(2)}_{i_1i_2} \conv ... \conv \nu^{(1)}_{i_{n-1}i_n} \right)(t^\prime) dt^\prime,
\end{align}
where $v_{i_0}$ are the coefficients of the initial state in its canonical basis,
\begin{align}\label{eq:initialstatecanonical}
	V^{(0)}(0) &= \sum_{i_0}^{d^{(0)}} v_{i_0} \mathbf{e}^{(0)}_{i_0},
\end{align}
the \textbf{canonical sub-event delay functions} $\nu_{i_{j-1}i_j}$ are defined as
\begin{align}\label{eq:canonicalsub-eventdelay}
	\nu^{j}_{i_{j-1}i_j} (t) &= \onenorm{ \tickmatrix^{(j-1)}_{i_j} e^{\notickmatrix^{(j-1)} t} \mathbf{e}^{(j-1)}_{i_{j-1}} }, \quad i_{j-1} \in \{1,2,...,d^{(j-1)}\}, \; i_j \in \{1,2,...,d^{(j)}\},
\end{align}
and in turn the \textbf{canonical event generators} $\tickmatrix^{(j-1)}_{i_j}$ are defined to be the original generator $\tickmatrix^{(j-1)}$ with all the rows except row $i_j$ set to zero. Thus
\begin{subequations}\label{eq:canonicalgenerators}
\begin{align}
	\tickmatrix^{(j-1)} &= \sum_{i_j=1}^{d^{(j)}} \tickmatrix^{(j-1)}_{i_j}, \\
	\text{where} \mat{\tickmatrix^{(j-1)}_{i_j}}_{kl} &= \delta_{i_j,k} \mat{\tickmatrix^{(j-1)}}_{kl}.
\end{align}
\end{subequations}

\end{lemma}

\begin{remark} \textbf{Explanation of Lemma \ref{lemma:generalindependent}.}
	
The difference between a canonical independent Markovian sequence (Def. \ref{def:canonicalsequence}) and a general Markovian sequence (Def. \ref{def:markoviansequence}) is as follows. In the first case, the initial state and the states following the occurrence of every event in the sequence are canonical states, which is equivalent to every event generator in the sequence having a single non-zero row. For a general sequence, this is not the case.

However, even in the case of a general sequence, one can split each event generator into the sum of matrices with single non-zero rows, each corresponding to a single canonical state (hence the term \textbf{canonical generators}). Note from Lemma \ref{lemma:sequenceconvolution} that the states and delay functions of every event are linear w.r.t. the event generators, and thus by splitting the event generators of a general sequence into canonical generators, we find that states and delay functions may be expressed as a linear combination of each of the terms in the split. Thus the general sequence turns into a tree-like graph, where each event corresponds to branching out w.r.t. all the possible canonical states after the event has occurred.

In other words, we decompose a Markovian sequence into the sum of all the possible canonical independent Markovian sequences that occur within it, and the above lemma expresses the states and delay functions of the original sequence as the same sum over the canonical sequences.

\end{remark}

\begin{proof}

Proof by induction. Consider the case of $n=1$. For the delay function $\tau^{(1)}$, from Eq. \ref{eq:sequencedelay},
\begin{align}\label{eq:firstdelayfunction}
	\tau^{(1)} &= \onenorm{ \tickmatrix^{(0)} V^{(0)}(t) }.
\end{align}

We may use the solution for $V^{(0)}(t)$ from Lemma \ref{lemma:sequenceconvolution}, and the decomposition of the initial state in the canonical basis (Eq. \ref{eq:initialstatecanonical},
\begin{align}
	V^{(0)}(t) &= e^{\notickmatrix^{(0)} t} V^{(0)}(0) = \sum_{i_0=1}^{d^{(0)}} e^{\notickmatrix^{(0)} t} v_{i_0} \mathbf{e}^{(0)}_{i_0}.
\end{align}

We also split the event generator $\tickmatrix^{(0)}$ into the sum of canonical generators (Eq. \ref{eq:canonicalgenerators}),
\begin{align}
	\tickmatrix^{(0)} &= \sum_{i_1=1}^{d^{(1)}} \tickmatrix^{(0)}_{i_1}.
\end{align}

Substituting these back into the delay function of the first event (Eq. \ref{eq:firstdelayfunction}), one obtains
\begin{align}
	\tau^{(1)} &= \sum_{i_0=1}^{d^{(0)}} \sum_{i_1=1}^{d^{(1)}} v_{i_0} \onenorm{ \tickmatrix^{(0)}_{i_1} e^{\notickmatrix^{(0)} t} \mathbf{e}^{(0)}_{i_0} } \\
	&= \sum_{i_0=1}^{d^{(0)}} \sum_{i_1=1}^{d^{(1)}} v_{i_0} \nu^{(1)}_{i_0i_1},
\end{align}
using the definition of the canonical sub-event delay function (Eq. \ref{eq:canonicalsub-eventdelay}). Thus $\tau^{(1)}$ satisfies the statement of the lemma.

For the state $V^{(1)}$, we apply Lemma \ref{lemma:sequenceconvolution},
\begin{align}
	V^{(1)}(t) &= \int_0^t e^{\notickmatrix^{(1)} \left( t - t^\prime \right)} \tickmatrix^{(0)} V^{(0)} \left( t^\prime \right) dt^\prime,
\end{align}
and follow the same procedure of decomposing the state and event generator, resulting in
\begin{align}\label{eq:tempfirststate}
		V^{(1)}(t) &= \int_0^t e^{\notickmatrix^{(1)} \left( t - t^\prime \right)} \sum_{i_0=1}^{d^{(0)}} \sum_{i_1=1}^{d^{(1)}}  v_{i_0} \left( \tickmatrix^{(0)}_{i_1} e^{\notickmatrix^{(0)} t^\prime} \mathbf{e}^{(0)}_{i_0} \right) dt^\prime.
\end{align}

The term in parentheses above is proportional to $\mathbf{e}^{(1)}_{i_1}$, as the generator $\tickmatrix^{(0)}_{i_1}$ has only a single non-zero row, which is row $i_1$, corresponding to the canonical state $\mathbf{e}^{(1)}_{i_1}$. Furthermore, the e-sum of the term has already been defined as the canonical sub-event delay function (Eq. \ref{eq:canonicalsub-eventdelay}), and thus we can express it as
\begin{align}\label{eq:associatesub-event}
	\tickmatrix^{(0)}_{i_1} e^{\notickmatrix^{(0)} t^\prime} \mathbf{e}^{(0)}_{i_0} = \mathbf{e}^{(1)}_{i_1} \nu^{(1)}_{i_0i_1}(t^\prime).
\end{align}

Substituting this back into Eq. \ref{eq:tempfirststate},
\begin{align}
	V^{(1)}(t) &= \int_0^t e^{\notickmatrix^{(1)} \left( t - t^\prime \right)} \sum_{i_0=1}^{d^{(0)}} \sum_{i_1=1}^{d^{(1)}} v_{i_0} \mathbf{e}^{(1)}_{i_1} \nu^{(1)}_{i_0i_1}(t^\prime)  dt^\prime,
\end{align}
which satisfies the statement of the lemma.

Proceeding, assume that the lemma is satisfied up to and including $n=k$. For the next delay function $\tau^{(k+1)}$, from Eq. \ref{eq:sequencedelay},
\begin{align}
	\tau^{(k+1)} &= \onenorm{ \tickmatrix^{(k)} V^{(k)}(t) } \\
	&= \onenorm{ \sum_{i_{k+1}=1}^{d^{(k+1)}} \tickmatrix^{(k)}_{i_{k+1}} \int_0^t e^{\notickmatrix^{(k)} (t-t^\prime)}  \sum_{i_0=1}^{d^{(0)}} \sum_{i_1=1}^{d^{(1)}} ... \sum_{i_k=1}^{d^{(k)}} \mathbf{e}^{(k)}_{i_k} v_{i_0} \left( \nu^{(1)}_{i_0i_1} \conv \nu^{(2)}_{i_1i_2} \conv ... \conv \nu^{(k)}_{i_{k-1}i_k} \right)(t^\prime) dt^\prime } \\
	&= \sum_{i_0=1}^{d^{(0)}} \sum_{i_1=1}^{d^{(1)}} ... \sum_{i_{k+1}=1}^{d^{(k+1)}} v_{i_0} \int_0^t \onenorm{ \tickmatrix^{(k)}_{i_{k+1}} e^{\notickmatrix^{(k)} (t-t^\prime)} \mathbf{e}^{(k)}_{i_k} } \left( \nu^{(1)}_{i_0i_1} \conv \nu^{(2)}_{i_1i_2} \conv ... \conv \nu^{(k)}_{i_{k-1}i_k} \right)(t^\prime) dt^\prime,
\end{align}
where we have applied the lemma for the state $V^{(k)}(t)$, and then simply rearranged terms (all of the states, sums are finite, and so is the integral).

Identifying the term in the e-sum as another canonical sub-event delay function (Eq. \ref{eq:canonicalsub-eventdelay}), we get
\begin{align}
	\tau^{(k+1)} &= \sum_{i_0=1}^{d^{(0)}} \sum_{i_1=1}^{d^{(1)}} ... \sum_{i_{k+1}=1}^{d^{(k+1)}} v_{i_0} \int_0^t \nu^{(k+1)}_{i_ki_{k+1}}(t-t^\prime) \left( \nu^{(1)}_{i_0i_1} \conv \nu^{(2)}_{i_1i_2} \conv ... \conv \nu^{(k)}_{i_{k-1}i_k} \right)(t^\prime) dt^\prime \\
	&= \sum_{i_0=1}^{d^{(0)}} \sum_{i_1=1}^{d^{(1)}} ... \sum_{i_{k+1}=1}^{d^{(k+1)}} v_{i_0} \left( \nu^{(1)}_{i_0i_1} \conv \nu^{(2)}_{i_1i_2} \conv ... \conv \nu^{(k)}_{i_{k-1}i_n} \conv \nu^{(k+1)}_{i_ki_{k+1}} \right)(t).
\end{align}
Thus $\tau^{(k+1)}$ also satisfies the lemma.

Finally, we express $V^{(k+1)}$ using Lemma \ref{lemma:sequenceconvolution},
\begin{align}
	V^{(k+1)}(t) &= \int_0^t e^{\notickmatrix^{(k+1)} \left( t - t^\prime \right)} \tickmatrix^{(k)} V^{(k)} \left( t^\prime \right) dt^\prime \\
	&= \int_0^t e^{\notickmatrix^{(k+1)} \left( t - t^\prime \right)} \tickmatrix^{(k)} \int_0^{t^\prime} e^{\notickmatrix^{(k)} (t^\prime-t^{\prime\prime})}  \sum_{i_0=1}^{d^{(0)}} \sum_{i_1=1}^{d^{(1)}} ... \sum_{i_k=1}^{d^{(k)}} \mathbf{e}^{(k)}_{i_k} v_{i_0} \left( \nu^{(1)}_{i_0i_1} \conv \nu^{(2)}_{i_1i_2} \conv ... \conv \nu^{(k)}_{i_{k-1}i_k} \right)(t^{\prime\prime}) dt^{\prime\prime} dt^\prime,
\end{align}
where we have applied the lemma to $V^{(k)}$. Once again, we can split the generator $\tickmatrix^{(k)}$ into canonical generators (Eq. \ref{eq:canonicalgenerators}), and rearrange the integral and sums appropriately,
\begin{align}
	V^{(k+1)}(t) &= \int_0^t e^{\notickmatrix^{(k+1)} \left( t - t^\prime \right)} \sum_{i_0=1}^{d^{(0)}} \sum_{i_1=1}^{d^{(1)}} ... \sum_{i_{k+1}=1}^{d^{(k+1)}} \int_0^{t^\prime} \tickmatrix^{(k)}_{i_{k+1}} e^{\notickmatrix^{(k)} (t^\prime-t^{\prime\prime})}   \mathbf{e}^{(k)}_{i_k} v_{i_0} \left( \nu^{(1)}_{i_0i_1} \conv \nu^{(2)}_{i_1i_2} \conv ... \conv \nu^{(k)}_{i_{k-1}i_k} \right)(t^{\prime\prime}) dt^{\prime\prime} dt^\prime.
\end{align}

Once again, we can associate one of the terms above to a canonical sub-event delay function (Eq. \ref{eq:canonicalsub-eventdelay}), as we did in Eq. \ref{eq:associatesub-event},
\begin{align}
	\tickmatrix^{(k)}_{i_{k+1}} e^{\notickmatrix^{(k)} (t^\prime-t^{\prime\prime})}   \mathbf{e}^{(k)}_{i_k} &= \mathbf{e}^{(k+1)}_{i_{k+1}} \nu^{(k+1)}_{i_ki_{k+1}} (t^\prime - t^{\prime\prime}),
\end{align}
which results in
\begin{align}
	V^{(k+1)}(t) &= \int_0^t e^{\notickmatrix^{(k+1)} \left( t - t^\prime \right)} \sum_{i_0=1}^{d^{(0)}} \sum_{i_1=1}^{d^{(1)}} ... \sum_{i_{k+1}=1}^{d^{(k+1)}} \int_0^{t^\prime} \mathbf{e}^{(k+1)}_{i_{k+1}} \nu^{(k+1)}_{i_ki_{k+1}} (t^\prime - t^{\prime\prime}) v_{i_0} \left( \nu^{(1)}_{i_0i_1} \conv \nu^{(2)}_{i_1i_2} \conv ... \conv \nu^{(1)}_{i_{k-1}i_k} \right)(t^{\prime\prime}) dt^{\prime\prime} dt^\prime \\
	&= \int_0^t e^{\notickmatrix^{(k+1)} \left( t - t^\prime \right)} \sum_{i_0=1}^{d^{(0)}} \sum_{i_1=1}^{d^{(1)}} ... \sum_{i_{k+1}=1}^{d^{(k+1)}} \mathbf{e}^{(k+1)}_{i_{k+1}} v_{i_0} \left( \nu^{(1)}_{i_0i_1} \conv \nu^{(2)}_{i_1i_2} \conv ... \conv \nu^{(k)}_{i_{k-1}i_k} \conv \nu^{(k+1)}_{i_ki_{k+1}} \right) (t^\prime) dt^\prime,
\end{align}
which also satisfies the lemma. By induction, the lemma applies for all $n \geq 1$.

\end{proof}

\subsection{Result on the precision of Markovian sequences}\label{app:resetresult}

\begin{theorem}\label{theorem:markovianaccuracy}
	
Given a Markovian sequence (Def. \ref{def:markoviansequence}), there exists a canonical independent Markovian sequence (Def. \ref{def:canonicalsequence}) such that
\begin{itemize}
	\item the vectors and generators of the canonical sequence are of the same dimension as the original sequence, i.e. $d^{(n)}$ remain the same for all $n$, and
	\item the precisions (Def. \ref{def:accuracydelayfunction}) of every event in the canonical sequence, i.e. the precisions of each of the event delay functions (Eq. \ref{eq:sequencedelay}) upper bounds those of the original sequence.
\end{itemize}
	
Furthermore, one such canonical independent Markovian sequence can be explicitly constructed from the original sequence by
\begin{itemize}
	\item changing the initial state to a well chosen canonical state,
	\item making all of the event generators $\tickmatrix^{(n)}$ into rank-1 matrices by keeping only a single non-zero row, possibly shifted to a different row, and
	\item scaling each pair of event generators $\{\notickmatrix^{(n)},\tickmatrix^{(n)}\}$ by well chosen positive constants.	
\end{itemize}

\end{theorem}

\begin{remark}
	
The major implication of this theorem is that if one's interest is in upper bounding the precision of Markovian sequences, then it suffices to optimize over the much smaller subset of canonical independent Markovian sequences, which are far more tractable.
	
Our main use for this theorem is to straightforwardly apply it to imply that ``reset" clocks, i.e. those that go to a fixed state every time they tick, are the most precise. However, the content of this theorem is more general than its application to clocks.
	
\end{remark}

\begin{proof}

We begin by taking restating the first part of Lemma \ref{lemma:generalindependent}, which states that each of the delay functions of a Markovian sequence of events may be expressed as 
\begin{align}\label{eq:prooftempdelayconvolution}
	\tau^{(n)}(t) &= \sum_{i_0=1}^{d^{(0)}} \sum_{i_1=1}^{d^{(1)}} ... \sum_{i_n=1}^{d^{(n)}} v_{i_0} \left( \nu^{(1)}_{i_0i_1} \conv \nu^{(2)}_{i_1i_2} \conv ... \conv \nu^{(1)}_{i_{n-1}i_n} \right)(t),
\end{align}
where the canonical sub-event delay functions $v_{i_{j-1}i_j}$ are defined in Eq. \ref{eq:canonicalsub-eventdelay}, with the associated canonical event generators subsequently in Eq. \ref{eq:canonicalgenerators}.

Applying Lemma \ref{lemma:delaymixture} regarding the precision of mixtures of delay functions to Eq. \ref{eq:prooftempdelayconvolution},
\begin{align}
	R \left[ \tau^{(n)} \right] \leq \max_{i_0,i_1,...,i_n} R \left[ \nu^{(1)}_{i_0i_1} \conv \nu^{(2)}_{i_1i_2} \conv ... \conv \nu^{(1)}_{i_{n-1}i_n} \right],
\end{align}
where the range of each $i_j$ above (and in what follows) is $\{1,2,...,d^{(j)}\}$ . We follow by applying Lemma \ref{lemma:delaysequence} regarding the precision of convolutions of delay functions, and find that
\begin{align}
	R \left[ \tau^{(n)} \right] \leq \max_{i_0,i_1,...,i_n} \left( R \left[ \nu^{(1)}_{i_0i_1} \right] + R \left[ \nu^{(2)}_{i_1i_2} \right] + ... + R \left[ \nu^{(n)}_{i_{n-1}i_n} \right] \right).
\end{align}

In the maximization above, each $i_j$ for $1 \leq j \leq n-1$ appears twice, in $\nu^{(j)}_{i_{j-1}i_j}$ and $\nu^{(j+1)}_{i_ji_{j+1}}$. Relaxing the maximization to allow for independent maximization of each precision, we end up with
\begin{align}
	R \left[ \tau^{(n)} \right] \leq \max_{i_0,i_1^\prime,i_1,i_2^\prime,i_2,...,i_{n-1}^\prime,i_{n-1},i_n} \left( R \left[ \nu^{(1)}_{i_0i_1^\prime} \right] + R \left[ \nu^{(2)}_{i_1i_2^\prime} \right] + ... + R \left[ \nu^{(n-1)}_{i_{n-2}i_{n-1}^\prime} \right] + R \left[ \nu^{(n)}_{i_{n-1}i_n} \right] \right).
\end{align}

In other words, if we label by $R^{max}_j$ the best precision from among the sub-event delay functions of the $j^{th}$ event,
\begin{align}\label{eq:accuracymaximization}
	R^{max}_j &= \max_{i_{j-1},i_j} R \left[ \nu^{(j)}_{i_{j-1}i_j} \right],
\end{align}
then the precision of the $n^{th}$ event in the sequence is bounded by
\begin{align}\label{eq:accuracyresetbound}
	R \left[ \tau^{(n)} \right] \leq \sum_{j=1}^n R^{max}_j.
\end{align}

To complete the proof, we construct a canonical independent Markovian sequence from the original sequence, that saturates the above inequality.

Consider that one performs the maximization in Eq. \ref{eq:accuracymaximization} for every event $j \in \{1,2,...,n\}$. Doing so for the $j^{th}$ event will return the optimal indices $i_{j-1}$ and $i_j$, that we label as $l_{j-1}$ and $m_j$ respectively. We also label the optimal canonical sub-event delay function as $\zeta^{(j)}$,
\begin{align}\label{eq:optimalsub-eventdelay}
	\zeta^{(j)} &= \nu^{(j)}_{l_{j-1}m_j}.
\end{align}

In other words, the optimal canonical sub-event delay function for the $j^{th}$ event is obtained by starting in the canonical state $\mathbf{e}^{(j-1)}_{l_{j-1}}$, and associated with the canonical event generator $\tickmatrix^{(j)}_{m_j}$ (Eq. \ref{eq:canonicalgenerators}) that results in the state $\mathbf{e}^{(j)}_{m_j}$ when the event occurs.

Our first step in the construction of the new independent sequence is to take the original Markovian sequence, and change all of the event generators $\tickmatrix^{(j)}$ into the optimal generators $\tickmatrix^{(j)}_{m_j}$, by setting all of the rows except row $m_j$ to zero. The new sequence is now a canonical independent Markovian sequence (Def. \ref{def:independentsequence}).

However, this is not enough, as we note that $m_j \neq l_j$ in general, i.e. the optimal canonical generator for the $j^{th}$ event may not correspond to the optimal initial canonical state for event $j+1$. However, by Lemma \ref{lemma:resetinvariance}, a delay function is left unchanged under operations that leave the column sums of its event generator the same. Since $\tickmatrix^{(j)}_{m_j}$ has a single non-zero row, we shift the row from the $m_j$ position to the $l_j$ position, to obtain a new canonical generator $\tilde{\tickmatrix}^{(j)}$, that is both optimal w.r.t. its own event as well as the initial canonical state for the next event.

By construction, we now have a canonical independent Markovian sequence, whose individual sub-event delay functions (Def. \ref{def:resetsubeventdelay}) are the ones with the optimal precision, i.e. the $\zeta^{(j)}$ from Eq. \ref{eq:optimalsub-eventdelay}. From Lemma \ref{lemma:resetsequenceconvolution}, the delay function of the $n^{th}$ event in the new sequence is now
\begin{align}
	\tilde{\tau}^{(n)} (t) &= \left( \zeta^{(1)} \conv \zeta^{(2)} \conv ... \conv \zeta^{(n)} \right) (t).
\end{align}

Applying Lemma \ref{lemma:delaysequence} to the above delay function, we get that the precision is bounded by
\begin{align}
	R \left( \tilde{\tau}^{(n)} \right) &\leq \sum_{j=1}^n R \left( \zeta^{(j)} \right) = \sum_{j=1}^n R^{max}_j.
\end{align}

From Lemma \ref{lemma:delaysequence}, equality holds only if the mean $\mu$ of each optimal delay function $\zeta^{(j)}$ satisfies
\begin{align}\label{eq:havetosatisfy}
	\frac{ \mu \left( \zeta^{(j)} \right) }{ R \left( \zeta^{(j)} \right)} &= \frac{ \mu \left( \zeta^{(j)} \right) }{ R \left( \zeta^{(j)} \right)} \quad \forall j,k \in \{1,2,...,n\}.
\end{align}

This is not automatically true from our construction so far. However, from Lemma \ref{lemma:resetsequencescale}, we can scale the $j^{th}$ pair of event generators by a positive constant $a_j$ in the sequence to change the sub-event delay function from $\zeta^{(j)}(t)$ to $a \zeta^{(j)}(at)$, leaving every other sub-event delay function unchanged. Furthermore, by Lemma \ref{lemma:delayscale}, such a scaling operation does not affect the precision of the sub-event delay function, but does scale its mean (the new mean is divided by $a_j$).

Thus we can pick a series of positive constants $\{a_j\}$ to scale each pair of generators in the independent sequence we have constructed so that every sub-event delay function has the same ratio of mean to precision, and thus satisfies Eq. \ref{eq:havetosatisfy}.

As a result, we are left with the canonical independent sequence satisfying the statement of the lemma, whose precision saturates Eq. \ref{eq:accuracyresetbound}, which itself is an upper bound to the precision of the original sequence.

\end{proof}

\section{Classical clocks}\label{Sec:classical clocks theorem proof}

In this part of the appendix, we review the behaviour of classical clocks, showing that reset clocks (those with a fixed state after ticking) can reach the highest precision, and that, for a classical clock of dimension $d$, the precision is upper bound by its dimension. Finally, we discuss a simple classical clock that saturates this bound.

In fact, the dynamics of classical clocks, at least those that are self-contained, is a subset of more general \textit{Markovian dynamics}, described by the Kolmogorov equation\cite{kolmo}. The ticks of a finite-dimensional clock may be understood as a sequence of events generated by Markovian dynamics on a finite-dimensional vector space. As such, we have included Appendix \ref{app:markovianpart} that introduces the mathematics of Markovian sequences that we require to prove our main results on clocks. Furthermore, we will draw on some lemmas on the behaviour of the precision of mixtures and sequences of delay functions, which are covered in Appendix \ref{app:delayfunction} (where delay functions and the precision $R$ are discussed).

Appendices \ref{app:delayfunction} and \ref{app:markoviandynamics} derive from standard theory on stochastic process and random variables, and are included here to place clocks within the larger context of the theory of stochastic processes, and for ease of reading and understanding of the proofs of our main results, which rely on multiple lemmas within this theory. For an in depth discussion on stochastic processes, see for example \cite{probability}.

In Appendix \ref{app:classicalclocks}, we introduce classical clocks and they were discussed in the main text, and relate them to Markovian sequences. Reset clocks, those that tick to a fixed state, are discussed and shown to be the same as independent Markovian sequences.

In Appendix \ref{app:resetclockaccuracy}, we first apply Theorem \ref{theorem:markovianaccuracy} to clocks to obtain our first result on clocks, that for every classical clock, there is a reset clock of at least as high precision. Furthermore, in Appendix \ref{sec:justifyR}, by focusing on the precision of reset clocks, we find that the quantity $R$ that we have denoted as the precision is in fact a good quantifier of how long the clock can run before being in error.

Our final result on classical clocks, that their precision $R$ is upper bound by their dimension, is stated and proven in Appendix \ref{app:upperbound}.

\subsection{Classical clocks}\label{app:classicalclocks}

\begin{definition}\label{def:c-clock}

As discussed in the main text (Section \ref{sec:classicalspecialcase}, a \textbf{self-contained stochastic finite-dimensional classical clock}, or for simplicity, a \textbf{classical clock}, is represented by a time-dependent vector $V(t) \in \mathbb{R}^d$, that represents the state of the clock, such that $V(0)$ is a population vector (Def. \ref{def:popvector}), together with a pair of time-independent stochastic generators $\{\notickmatrix, \tickmatrix\}$ (Def. \ref{def:stochasticgenerators}), that generate the dynamics of the clock via the relation
\begin{align}\label{eq:clockstate}
	\frac{d}{dt} V(t) &= \left( \notickmatrix + \tickmatrix \right) V(t),
\end{align}
and where the probability per unit time of a tick being observed, that we label the ``tick density" and denote by $p_{tick}(t)$, is given by
\begin{align}\label{eq:tickdensity}
	p_{tick}(t) &= \onenorm{\tickmatrix V(t)}.
\end{align}

\end{definition}

\begin{remark}

The clock is \textbf{self-contained} because its generators are time-independent, and \textbf{finite-dimensional} as the vector space of clock states is taken to be finite dimensional. Finally, the clock is \textbf{stochastic/classical}, as the states are population vectors evolving under stochastic generators, which is equivalent, in the context of quantum theory, to restricting the initial state to be diagonal in some preferred basis, and restricting the dynamical generators to Lindbladian operators that keep the states diagonal in the same basis, as discussed in Section \ref{sec:classicalspecialcase} of the main text.

In the entirety of this work, we only consider self-contained and finite-dimensional clocks. For simplicity, we continue for the remainder of this appendix by shortening the term to simply \textbf{classical clocks}, with the implicit understanding that they are also self-contained, and finite-dimensional, and that the term classical implies stochastic.

\end{remark}

\begin{remark}

In the context of classical clocks, since $\tickmatrix$ is associated to the generation of ticks, we shall refer to it as the \textbf{tick generator}, and to $\notickmatrix$ as the \textbf{non-tick generator}.

\end{remark}

\subsubsection{Tick-states and tick delay functions}

The state $V(t)$ and tick density $p_{tick}(t)$ in the above description of classical clocks (Def. \ref{def:c-clock}) do not contain any information about how many times the clock has already ticked in the past. In order to distinguish between each tick, one can split the state into a sequence $\{V^{(n)}\}$ of ``tick-states", each corresponding to a fixed number of ticks.

\begin{definition}\label{def:tickstates}

The \textbf{tick-states} of a classical clock (Def. \ref{def:c-clock} are a sequence of states $\{V^{(n)}(t)\}$, $n \in \{0,1,2,...\}$, where the initial conditions and dynamics for each state in the sequence are determined w.r.t. the pair of stochastic generators $\{\notickmatrix,\tickmatrix\}$ of the clock as
\begin{subequations}\label{eq:tickstates}
\begin{align}
	\quad V^{(n)}(0) &= \begin{cases}
			V(0) & \text{for $n=0$}, \\
			0 & \text{for $n>0$}.
		\end{cases}, \label{eq:tickstateinitial}\\
	\text{and} \quad \frac{d}{dt} V^{(n)}(t) &= \begin{cases}
			\notickmatrix V_0(t) & \text{for $n=0$}, \\
			\notickmatrix V^{(n)}(t) + \tickmatrix V^{(n-1)}(t) & \text{for $n>0$}.
		\end{cases} \label{eq:tickstatedynamics}
\end{align}
\end{subequations}

\end{definition}

\begin{remark}

The choice of initial conditions (Eq. \ref{eq:tickstateinitial}) corresponds to the fact that the clock has not ticked yet, and thus the entire state of the clock is associated with $V^{(0)}$, the state corresponding to no ticks.

For $t>0$, in every infinitesimal time interval, the tick generator moves probability from a state in the sequence to the next, while the non-tick generator keeps the probability within the same tick subspace (albeit moving it around, and possibly decreasing it). Thus each $V^{(n)}(t)$ evolves due to only two contributions: firstly, its own evolution, as $\notickmatrix V^{(n)}(t)$, and because of a tick, from $\tickmatrix V^{(n-1)}(t)$.

\end{remark}

In an analogous manner to the state, one can split the tick density (Eq. \ref{eq:tickdensity}) of the clock into the contributions of the first, second, and following ticks.

\begin{definition}\label{def:tickdelay}
	
We define the \textbf{tick delay function} of the $n^{th}$ tick to be
\begin{align}\label{eq:tickdelayfunction}
	\tau^{(n)}(t) &= \onenorm{\tickmatrix V^{(n-1)}(t)}.
\end{align}
	
\end{definition}

\begin{remark}

The expression for the $n^{th}$ tick delay function reflects the fact that the $n^{th}$ tick can only occur after $n-1$ ticks, and the corresponding state of the clock is $V^{(n-1)}(t)$.

\end{remark}

\begin{remark}

By construction, the tick-states (Def. \ref{def:tickstates} and the tick delay functions (Def. \ref{def:tickdelay}) of a classical clock form a Markovian sequence of events (Def. \ref{def:markoviansequence}). In the case of clocks however, the dynamics of every event (tick) is identical to the previous one, featuring the same pair of generators on the same vector space.

\end{remark}

\subsubsection{Reset clocks}

In our discussion on Markovian sequences of events, it was useful to discuss independent Markovian sequences (Def. \ref{def:independentsequence}), those that had fixed states following the occurrence of each event. The equivalent for clocks is a ``reset clock", which we now characterize.

\begin{definition}\label{def:resetclock}
	
A \textbf{classical reset clock}, (in this appendix simply a \textbf{reset clock}), is a classical clock (Def. \ref{def:c-clock}) for which the event of ticking always causes the clock to return to a fixed state which is also the initial state, i.e.
\begin{align}\label{eq:resetdefinition}
	\tickmatrix V \propto V(0) \quad \forall V \in \mathbb{R}^d.
\end{align}
Equivalently, a reset clock is one for which the tick generator $\tickmatrix$ is a rank-1 operator, all of whose columns are proportional to the initial state $V(0)$.

\end{definition}

\begin{remark}

By definition, the tick-states and tick delay functions of a reset clock form an independent Markovian sequence of events, Def. \ref{def:independentsequence}. Since the events in the case of clocks are ticks, and the dynamics of each tick is generated identically to all the others, the ticks of a reset clock correspond to an \textit{i.i.d.} sequence, i.e. a sequence of independent and identically distributed events \cite{probability}.

\end{remark}

\subsection{The precision of reset clocks}\label{app:resetclockaccuracy}

\subsubsection{The precision $R$ quantifies the average run-time of reset clocks}\label{sec:justifyR}

If we apply Lemma \ref{lemma:resetsequenceconvolution} to the case of reset clocks (Def. \ref{def:resetclock}), we find that the delay function of the $n^{th}$ tick of the clock is given by the $n$-fold convolution of the delay function of the first tick,
\begin{align}
	\tau^{(n)}(t) &= \left( \tau_1 \conv \tau_1 \conv... \conv \tau_1 \right) (t).
\end{align}

If we denote the zeroth moment, first moment, second moment, variance and precision (Defs. \ref{def:delayfunctionmoments}, \ref{def:accuracydelayfunction}) of a single tick of a reset clock (i.e. the delay function of a single tick) by $Q,\mu,\chi,\sigma,R$ respectively, and those of the $n^{th}$ tick by $Q^{(n)},\mu^{(n)},\chi^{(n)},\sigma^{(n)},R^{(n)}$, we find from the application of Eq. \ref{eq:sequencemoments} for the moments of a convolution of delay functions, and Eq. \ref{eq:accuracygamma} for the precision of a delay function, that
\begin{align}\label{eq:resetmomentscaling}
	\left\{ \begin{array}{c}
		Q^{(n)} \\
		\mu^{(n)} \\
		\chi^{(n)} \\
		\sigma^{(n)} \\
		R^{(n)}
	\end{array} \right\} &= 
	\left\{ \begin{array}{c}
		Q^n \\
		n \mu \\
		n \chi + n(n-1) \mu^2 \\
		\sqrt{n} \cdot \sigma \\
		n R
	\end{array} \right\}.
\end{align}

Consider that one asks the question \textit{``How many ticks can the clock produce until the uncertainty in the time of occurrence of the next tick has grown to be equal to the time interval between ticks?"}, which is a universally accepted mark of a clock's precision.

For reset clocks, the average time interval between ticks is $\mu^{(n)} - \mu^{(n-1)} = \mu$,  and is independent of which tick we are at. On the other hand, the uncertainty in the time of occurrence of the $n^{th}$ tick, $\sigma^{(n)} = \sqrt{n} \sigma$, grows with the number of ticks. If we denote the average number of ticks before the uncertainty equals the interval between ticks as $N$,
\begin{align}
	\sigma^{(n)} &= \mu^{(n)} - \mu^{(n-1)} \\
	\therefore \; \sqrt{N} \sigma &= \mu \\
	\therefore \; N &= \frac{\mu^2}{\sigma^2} = R.
\end{align}

\subsubsection{Result: The precision of all classical clocks is bound by that of reset classical clocks}

\begin{theorem}\label{theorem:resetclockaccuracy}

For every classical clock (Def. \ref{def:c-clock}), there exists a reset clock (Def. \ref{def:resetclock}) of the same dimension such that the precisions of the delay functions of every tick of the original clock are upper bounded by the corresponding precisions of the reset clock. Furthermore, such a reset clock can be obtained from the original clock by
\begin{itemize}
	\item picking a single well chosen canonical state to be the initial state,
	\item setting all but one of the rows of the tick generator $\tickmatrix$ to be zero, and shifting the single non-zero row to the location corresponding to the initial canonical state.
\end{itemize}

\end{theorem}

\begin{proof}

The theorem follows from the direct application of Theorem \ref{theorem:markovianaccuracy} to the case of classical clocks, when one identifies the tick-states (Def.  \ref{def:tickstates}) and tick delay functions (Def. \ref{def:tickdelay}) as a Markovian sequence of events (Def. \ref{def:markoviansequence}), while those of a reset clock (Def. \ref{def:resetclock}) are an independent Markovian sequence (Def. \ref{def:independentsequence}).

\end{proof}

\subsection{An upper bound on the precision of classical clocks, proof of Theorem \ref{thm_classicalbound}}\label{app:upperbound}

\begin{theorem}[Theorem \ref{thm_classicalbound} from the main text]\label{theorem:classicalclockaccuracy}

For a classical clock of dimension $d$ (Def. \ref{def:c-clock}), where $d \in \mathbb{N}^+$, the precision of its $n^{th}$ tick, i.e. the precision $R$ (Def. \ref{def:accuracydelayfunction}) of the delay function of the $n^{th}$ tick (Def. \ref{def:tickdelay}), is upper bound by
\begin{align}
	R\left[ \tau^{(n)} \right] \leq n d,
\end{align}

In particular, for reset clocks of dimension $d$, which upper bound the precisions of the ticks of arbitrary $d$-dimensional clocks (see Theorem \ref{theorem:resetclockaccuracy}), the precision of every single tick w.r.t. the previous one is upper bound by
\begin{align}
	R \left[ \tau \right] \leq d,
\end{align}
which, as discussed in Sec. \ref{sec:justifyR}, is a quantifier for the number of ticks outputted before the clock is expected to fail.

\end{theorem}

\textbf{Proof.} The rest of this section is dedicated to the proof of Theorem \ref{theorem:classicalclockaccuracy}. We prove the theorem by fixing the dimension $d$ and optimizing the precision over classical clocks of that dimension. From Theorem \ref{theorem:resetclockaccuracy}, we know that the precisions of the ticks of any clock are upper bound by that of at least one reset clock with a canonical initial (and reset) state, and thus we may restrict our optimization to the case of reset clocks of dimension $d$ whose initial and reset state is the same canonical state.

Furthermore, as discussed in Sec. \ref{sec:justifyR}, the precision of the $n^{th}$ tick of a reset clock is simply $n$ times the precision of the first tick, and the delay functions of every tick w.r.t. the time of occurrence of the previous tick are identical to each other, and equal to the delay function of the first tick.

Thus the quantity we are left to optimize is simply the precision of the first tick, and the restricted set we optimize over is that of reset clocks with a canonical initial state, and a canonical tick generator chosen so that the reset state is the initial state.

More precisely, we work in a $d$-dimensional real vector space, spanned by the canonical basis $\mathbf{e}_i$, where $i \in \{0,1,...,d-1\}$. Since the initial and reset state must be a canonical state, we label this as $\mathbf{e}_0$ without loss of generality. Thus the tick generator $\tickmatrix$ has only one non-zero row, i.e. its first row, that corresponds to $\mathbf{e}_0$. On the other hand, the non-tick generator $\notickmatrix$ is arbitrary. The state and corresponding delay function of a single tick of the clock are thus
\begin{subequations}\label{eq:singletickstatedelay}
\begin{align}
	V(t) &= e^{\notickmatrix t} \mathbf{e}_0, \\
	\tau(t) &= \onenorm{ \tickmatrix e^{\notickmatrix t} \mathbf{e}_0 }.
\end{align}
\end{subequations}

To prove that the precision of the above delay function is bounded by the dimension $d$ for all $d \in \mathbb{N}^+$, we employ the method of induction, starting with the case $d=1$. In this case, both $\notickmatrix$ and $\tickmatrix$ are real  numbers, that we label $\notickmatrix = -p$ and $\tickmatrix = g$ respectively, where $p\geq g > 0$ in order to ensure that $\{\notickmatrix,\tickmatrix\}$ are a pair of stochastic generators (Def. \ref{def:stochasticgenerators}). The delay function is thus
\begin{align}
	\tau(t) &= g e^{-pt}.
\end{align}

One can calculate all of the moments (Def. \ref{def:delayfunctionmoments}) explicitly,
\begin{align}
	Q &= \braket{t^0} = \frac{g}{p} \\
	\mu &= \frac{\braket{t^1}}{Q} = \frac{1}{p} \\
	\chi &= \frac{\braket{t^2}}{Q} = \frac{2}{p^2},
\end{align}
from which the $\gamma$-value (Eq. \ref{eq:accuracygamma}) is $\gamma = \chi/\mu^2 = 2$, and the precision is
\begin{align}
	R \left[ \tau \right] &= \frac{1}{\gamma - 1} = 1,
\end{align}
which satisfies the statement of the theorem.

We continue by assuming the theorem applies for dimensions $\{1,2,...,d-1\}$, and prove that it applies to a $d$-dimensional clock.

The rest of the proof is structured as follows. In Sec. \ref{proof:splitdimension}, we divide the $d$ dimensional vector space of the clock into two subspaces, a one-dimensional space corresponding to the initial (canonical) state, and the complementary $d-1$ dimensional space, and define sequences of states and delay functions that correspond to this division. In Sec. \ref{proof:sequenceproperties}, we prove that we can recover the dynamics of the state $V(t)$ and the delay function $\tau(t)$ of a single tick from these defined sequences. In Sec. \ref{proof:delayexplicit}, we calculate the moments of the delay function $\tau(t)$ explicitly using the defined sequences, and upper bound the precision of $\tau(t)$.

\subsubsection{Dividing the vector space of the clock into a sum of one-dimensional and $d-1$ dimensional spaces}\label{proof:splitdimension}

In order to use the result for $d-1$ dimensional clocks, we divide the $d$-dimensional space of states into the direct sum of two subspaces, firstly, the one-dimensional space corresponding to the initial state $\mathbf{e}_0$, that we label $S_0$, and its complement, that is spanned by the rest of the canonical basis, that we label $S_1$. Thus $S_0$ is one-dimensional, while $S_1$ is $d-1$ dimensional. We denote the projectors onto these spaces as $\Pi_0$ and $\Pi_1$. In matrix form w.r.t. the canonical basis, these projectors are
\begin{align}
	\Pi_0 &= 	\left( \begin{array}{c|ccc}
	1 & & \mathbf{0} &\\
	\hline & & &\\
	\mathbf{0} & & \mathbf{0} & \\
	& & &
	\end{array} \right), & 
	\Pi_1 &= 	\left( \begin{array}{c|ccc}
	0 & & \mathbf{0} &\\
	\hline & & &\\
	\mathbf{0} & & \mathds{1}_{d-1} & \\
	& & &
	\end{array} \right),
\end{align}
where $\mathds{1}_{d-1}$ represents the identity operator on a $d-1$ dimensional real vector space. Note that $\Pi_0 + \Pi_1 = \mathds{1}_d$, corresponding to $S_0 \oplus S_1 = \mathbb{R}^d$.

One may split the non-tick generator $\notickmatrix$ into the corresponding sum of four matrices, w.r.t. the subspaces $S_0$ and $S_1$,
\begin{subequations}\label{eq:spacegenerators}
\begin{align}
	\notickmatrix &= \left( \Pi_0 + \Pi_1 \right) \notickmatrix \left( \Pi_0 + \Pi_1 \right) \\
	&= \notickmatrix_{00} + \notickmatrix_{01} + \notickmatrix_{10} + \notickmatrix_{11}, \\
	\text{where} \quad \notickmatrix_{xy} &= \Pi_x \; \notickmatrix \; \Pi_y. 
\end{align}
\end{subequations}

Visually, this corresponds to the expressing $\notickmatrix$ w.r.t. the canonical basis as
\begin{align}\label{eq:generatorvisual}
\notickmatrix &\equiv \quad \left( \begin{array}{c|ccc}
\star & & \mathbf{0} &\\
\hline & & &\\
\mathbf{0} & & \mathbf{0} & \\
& & &
\end{array} \right)
\quad + \quad \left( \begin{array}{c|ccc}
0 & & \star &\\
\hline & & &\\
\mathbf{0} & & \mathbf{0} & \\
& & &
\end{array} \right)
\quad + \quad \left( \begin{array}{c|ccc}
0 & & \mathbf{0} &\\
\hline & & &\\
\star & & \mathbf{0} & \\
& & &
\end{array} \right)
\quad + \quad \left( \begin{array}{c|ccc}
0 & & \mathbf{0} &\\
\hline & & &\\
\mathbf{0} & & \star & \\
& & &
\end{array} \right),
\end{align}
where $\star$ denote the original elements, and the matrices above are respectively $\notickmatrix_{00},\notickmatrix_{01},\notickmatrix_{10}\notickmatrix_{11}$.

One may understand each of the above generators $\notickmatrix_{xy}$ as the part of the non-tick generator that is responsible for moving the state \emph{from} the subspace $S_y$ \emph{into} the subspace $S_x$. More precisely, given the state of the clock at some time is $V(t)$, the part of the state that is in the subspace $S_y$ is given by the projection $\Pi_y V(t)$. On this part of the state, the infinitesimal change generated by the non-tick generator is $\notickmatrix \Pi_y V(t)$. Finally, the part of this infinitesimal state-change that is in the space $S_x$ is the projection $\Pi_x \notickmatrix \Pi_y V(t)$.

We follow the same procedure for the tick generator $\tickmatrix$. However, $\tickmatrix$ only takes states to the reset state $\mathbf{e}_0$, i.e. for all $V \in \mathbb{R}^d$,
\begin{align}
	\tickmatrix V \propto \mathbf{e}_0 \quad \Longleftrightarrow \quad \tickmatrix V \in S_0 \quad \Longleftrightarrow \quad \Pi_1 \tickmatrix = 0,
\end{align}
and thus
\begin{align}\label{eq:spacetickers}
	\tickmatrix &= \tickmatrix_{00} + \tickmatrix_{01}, \\
	\text{where} \quad \tickmatrix_{xy} &= \Pi_x \; \tickmatrix \; \Pi_y, \\
	\text{and} \quad \tickmatrix_{10} &= \tickmatrix_{11} = 0.
\end{align}

At this point, we have split both of the generators into components that describe the movement of the state within and between the subspaces $S_0$ and $S_1$. We proceed to do the same for the state of the clock, and construct an independent Markovian sequence of states and corresponding delay functions that does this. Each state in the sequence must be distinguished by two indices, first, an $x \in \{0,1\}$ to denote that the state belongs to $S_x$, and $n \in \mathbb{N}^0 = \{0,1,2,...\}$ to mean that the state corresponds to population having moved from the initial state $\mathbf{e}_0 \in S_0$ to the space $S_1$ and back $n$ times.

\begin{definition}\label{def:pathstates}

We define the set of \textbf{path-specific clock states} $v_{n,x}(t)$, where $n \in \{0,1,2,...\}$ and $x \in \{0,1\}$, by the initial conditions
\begin{align}\label{eq:spacestateinitial}
	v_{n,x} (0) &= \begin{cases}
		\mathbf{e}_0 & \text{if $n=x=0$}, \\
		\mathbf{0} &\text{otherwise},
	\end{cases},
\end{align}
and the dynamics
\begin{align}\label{eq:spacestatedynamics}
	\frac{d}{dt} v_{n,x}(t) &= \begin{cases}
		\notickmatrix_{00} v_{0,0}(t) &\text{if $n=x=0$}, \\
		\notickmatrix_{xx} v_{n,x}(t) + \notickmatrix_{x\bar{x}} v_{n-\bar{x},\bar{x}}(t), &\text{otherwise}.
		\end{cases}
\end{align}
where $\bar{x} = x \oplus 1$ is the complement of $x$.

\end{definition}

\begin{remark}

Roughly speaking, $v_{n,x}(t)$ is the part of the state of the clock (of the first tick) corresponding to being in the subspace $S_x$ and having gone from $S_0$ to $S_1$ and back $n$ times.

\end{remark}

\begin{definition}\label{def:pathdensities}
	
We define the set of \textbf{path-specific delay functions} 
$\xi_{n,x}(t)$, where $n \in \{0,1,2,...\}$ and $x \in \{0,1\}$, barring $n=x=0$, by
\begin{align}
	\xi_{n,x}(t) &= \onenorm{\notickmatrix_{x\bar{x}} v_{n-\bar{x},\bar{x}}(t)},
\end{align}
	
\end{definition}

\begin{remark}

$\xi_{n,x}(t)$ refers to the delay function of arriving at the $\{n,x\}$ state $v_{n,x}$ at time $t$.

\end{remark}

In anticipation of dividing the delay function of a single tick of the clock using the above sequences, we define the following ``component delay functions".

\begin{definition}\label{def:componentdelay}
	
We define the set of \textbf{component delay functions} $\tau_{n,x}(t)$, where $n \in \{0,1,2,...\}$ and $x \in \{0,1\}$, barring $n=x=0$, by
\begin{align}
	\tau_{n,x}(t) &= \onenorm{ \tickmatrix_{0\bar{x}} v_{n-\bar{x},\bar{x}}(t) }.
\end{align}

\end{definition}

\begin{remark}

The term ``delay function" (Def. \ref{def:delayfunction}) for the above is appropriate. We will shortly show that the $v_{n,x}(t)$ are an independent Markovian sequence of states, and the event generator $\tickmatrix_{x\bar{x}}$, taken together with $\notickmatrix_{\bar{x}\bar{x}}$, which is the relevant non-event generator for the state $v_{n-\bar{x},\bar{x}}(t)$, form a pair of stochastic generators $\{\notickmatrix_{\bar{x}\bar{x}},\tickmatrix_{x\bar{x}}\}$ (Def. \ref{def:stochasticgenerators}). From Corollary \ref{cor:validdelayfunction}, we conclude that the $\tau_{n,x}(t)$ are delay functions.

\end{remark}

Finally, as we will eventually prove, the sequences above are found to involve repetitive convolutions of a small set of delay functions, which we proceed to define.

\begin{definition}\label{def:spaceresetstates}
	
We define the pair of \textbf{path-generating reset states} $\mathbf{u}_x$, where $x \in \{0,1\}$ in the following manner. $\mathbf{u}_0 = \mathbf{e}_0$, and $\mathbf{u}_1$ is constructed by taking the single non-zero column of $\notickmatrix_{10}$ (Eq. \ref{eq:generatorvisual}), and normalising it. 

\end{definition}

\begin{definition}\label{def:pathgeneratingdelay}

We define the pair of \textbf{path-generating delay functions} $\Theta_x$, where $x \in \{0,1\}$ by
\begin{align}
	\Theta_x &= \onenorm{ \notickmatrix_{x\bar{x}} e^{\notickmatrix_{\bar{x}\bar{x}} t} \mathbf{u}_{\bar{x}} },
\end{align}
where $\mathbf{u}_x$, are the path-generating reset states (Def. \ref{def:spaceresetstates}).

\end{definition}

\begin{definition}\label{def:tickgeneratingdelay}
	
We define the pair of \textbf{tick-generating delay functions} $\Gamma_x$, where $x \in \{0,1\}$ by
\begin{align}
	\Gamma_x &= \onenorm{ \tickmatrix_{0\bar{x}} e^{\notickmatrix_{\bar{x}\bar{x}} t} \mathbf{u}_{\bar{x}} },
\end{align}
where $\mathbf{u}_x$, are the path-generating reset states (Def. \ref{def:spaceresetstates}).
	
\end{definition}

\begin{remark}

Put very simply, $\Theta_x$ is the delay function of the event of moving from $S_{\bar{x}}$ to $S_x$, while $\Gamma_x$ is the delay function of ticking from the subspace $S_{\bar{x}}$.

\end{remark}

\begin{remark}

The above two definitions are justified in using the terminology ``delay function" as both the pair $\{\notickmatrix_{\bar{x}\bar{x}},\notickmatrix_{x\bar{x}}\}$ as well as the pair $\{\notickmatrix_{\bar{x}\bar{x}},\tickmatrix_{x\bar{x}}\}$ can be proven to be pairs of stochastic generators, and from Corollary \ref{cor:validdelayfunction}, it follows that the $\Theta_x$ and $\Gamma_x$ defined above satisfy the requirements of a delay function (Def. \ref{def:delayfunction}).

\end{remark}

\subsubsection{Proving the necessary properties of the path-specific states and delay functions.}\label{proof:sequenceproperties}

\begin{lemma}\label{lemma:pathmarkoviansequence}

The set of path-specific clock states (Def. \ref{def:pathstates}) and delay functions (Def. \ref{def:pathdensities}) form an independent Markovian sequence of events (Defs. \ref{def:markoviansequence}, \ref{def:independentsequence}) w.r.t. the ordering where the event $\{n,x\}$ is followed by $\{n+x,\bar{x}\}$, corresponding to $\{ \{0,0\}, \{0,1\}, \{1,0\}, \{1,1\}, \{2,0\}, \{2,1\}, \{3,0\}, ... \}$; and via the identification of $\{\notickmatrix_{xx},\notickmatrix_{\bar{x}x}\}$ as the pair of event and non-event stochastic generators for the event $\{n,x\}$.

Furthermore, the state $v_{n,x}(t) \in S_x$ for all $t \geq 0$, i.e.
\begin{align}
	\Pi_0 v_{n,0}(t) &= v_{n,0}(t) \quad \text{and} \quad \Pi_1 v_{n,1}(t) = v_{n,1}(t) \quad \text{and} \quad \Pi_0 v_{n,1}(t) = \Pi_1 v_{n,0}(t) = \mathbf{0}.
\end{align}

\end{lemma}

\begin{proof}

Following the definition of a Markovian sequence of events, Def. \ref{def:markoviansequence}, we note that the definition of the initial states in the sequence satisfy the definition by construction. In addition, we require the pairs of event generators for each event in the sequence, $\{\notickmatrix_{xx},\notickmatrix_{\bar{x}x}\}$ to be a pair of stochastic generators, Def. \ref{def:stochasticgenerators}.

First off, to prove that $\notickmatrix_{xx}$ is a non-event generator, note that since $\notickmatrix_{xx} = \Pi_x \notickmatrix \Pi_x$, each element of $\notickmatrix_{xx}$ is equal to either the original value of the element in $\notickmatrix$ or to zero. Thus Eq. \ref{eq:nontickcondition} is still satisfied because $\notickmatrix$ is a non-event generator. On the other hand, $\notickmatrix_{\bar{x}x} = \Pi_{\bar{x}} \notickmatrix \Pi_x$, and thus all of its diagonal elements are zero. Its off-diagonal elements are either equal to those of $\notickmatrix$ which are non-negative, or zero. Thus Eq. \ref{eq:tickcondition} is satisfied, proving that $\notickmatrix_{\bar{x}x}$ is an event generator. Finally, the sum of the generators satisfies (from Eq. \ref{eq:spacegenerators})
\begin{align}
	\notickmatrix_{xx} + \notickmatrix_{\bar{x}x} &= \Pi_{\bar{x}} \notickmatrix \Pi_x + \Pi_x \notickmatrix \Pi_x = \notickmatrix \Pi_x, \quad \text{because $\Pi_{\bar{x}} + \Pi_x = \mathds{1}$}.
\end{align}
Thus the sum of the two generators is the original non-event generator $\notickmatrix$ with some of its columns set to zero (those outside the support of $\Pi_x$). Thus adding the column sums of $\notickmatrix_{xx}$ and $\notickmatrix_{\bar{x}x}$ gives either the corresponding column sum of $\notickmatrix$, or zero. In either case, Eq. \ref{eq:stochasticcondition} is satisfied, completing the conditions that determine that $\{\notickmatrix_{xx}, \notickmatrix_{\bar{x}x}\}$ is a pair of stochastic generators.

To prove that the sequence is an independent one (Def. \ref{def:independentsequence}), one has to show that both of the non-event generators in the sequence, the $\notickmatrix_{\bar{x}x}$, are rank-1. From Eq. \ref{eq:generatorvisual}, one observes that $\notickmatrix_{01}$ has a single non-zero row, while $\notickmatrix_{10}$ has a single non-zero column. Thus Def. \ref{def:independentsequence} is satisfied.

To demonstrate that each $v_{n,x}(t) \in S_x$, we express them as (for all cases except $n=x=0$) the solution to their differential equations (Eq. \ref{eq:spacestatedynamics}),
\begin{align}
v_{n,x}(t) &= v_{n,x}(t=0) + \int_0^t \frac{d}{dt^\prime} v_{n,x}(t^\prime) dt^\prime \\
&= \int_0^t \notickmatrix_{xx} v_{n,x}(t^\prime) + \notickmatrix_{x\bar{x}} v_{n-\bar{x},\bar{x}}(t^\prime) dt^\prime \\
&= \int_0^t \Pi_x \notickmatrix \Pi_x v_{n,x}(t^\prime) + \Pi_x \notickmatrix \Pi_{\bar{x}} v_{n-\bar{x},\bar{x}}(t^\prime) dt^\prime,
\end{align}
using the definition of the $\notickmatrix_{xy}$ (Eq. \ref{eq:spacegenerators}, and since the initial states are all zero vectors except for $n=x=0$. Left-multiplying the above expression by the projects $\Pi_x$ and $\Pi_{\bar{x}}$ respectively recovers the statement of the lemma.

For the special case $n=x=0$, the initial state $v_{0,0}(0)$ is the canonical state $\mathbf{e}_0$ and thus satisfies the lemma. Furthermore, the derivative of the state is proportional only to $\notickmatrix_{00} = \Pi_0 \notickmatrix \Pi_0$. The proof thus follows in an analogous manner.

\end{proof}

\begin{lemma}[Decomposing the path-specific delay functions w.r.t. the path-generating delay functions]\label{lemma:pathdensityconv}

The path-specific delay functions $\xi_{n,x}(t)$ (Def. \ref{def:pathdensities}) are sequential convolutions of the path-generating delay functions $\Theta_x$ (Def. \ref{def:pathgeneratingdelay}),
\begin{align}\label{eq:tempdelayconvolution}
\xi_{n,x}(t) &= \left( \Theta_1 \conv \Theta_0 \conv \Theta_1 \conv \Theta_0 \conv ... \conv \Theta_x \right) (t),
\end{align}
where $\Theta_1$ appears $n+x$ times and $\Theta_0$ appears $n$ times.

\end{lemma}

\begin{proof}

As Lemma \ref{lemma:pathmarkoviansequence} has proven that $\xi_{n,x}(t)$ are the delay functions of an independent Markovian sequence, we may apply Lemma \ref{lemma:resetsequenceconvolution}, to express them as
\begin{align}
\xi_{n,x}(t) &= \left( \nu_{0,1} \conv \nu_{1,0} \conv \nu_{1,1} \conv \nu_{2,0} \conv ... \conv \nu_{n,x} \right) (t),
\end{align}
where the sub-event delay functions $\nu_{n,x}$ are defined in Def. \ref{def:resetsubeventdelay}, and for the present sequence, take on the form
\begin{align}\label{eq:tempgeneratingdelay}
	\nu_{n,x}(t) &= \onenorm{ \notickmatrix_{x\bar{x}} e^{\notickmatrix_{\bar{x}\bar{x}} t} \mathbf{w}_{n,\bar{x}} },
\end{align}
where $\mathbf{w}_{n,\bar{x}}$ is the reset state (Def. \ref{def:eventresetstate} for the event $\{n,x\}$ in the sequence.

Consider the reset state $\mathbf{w}_{n,0}$. This is defined to be (Def. \ref{def:eventresetstate}) the unique normalised state that is proportional to every column of the corresponding generator in the sequence, which in this case is $\notickmatrix_{01}$ (see Lemma \ref{lemma:pathmarkoviansequence}). From the definition of $\notickmatrix_{xy}$ (see Eq. \ref{eq:generatorvisual}), we conclude that this is simply the canonical state $\mathbf{e}_0$, and thus $\mathbf{w}_{n,0} = \mathbf{u}_0$, the path-generating reset state defined in Def. \ref{def:spaceresetstates}.

In a similar manner, the reset state $\mathbf{w}_{n,1}$ is the unique normalised state proportional to every column of $\notickmatrix_{10}$. This operator has only a single non-zero column, and the corresponding normalised state has already been defined to be the other path-generating reset state, $\mathbf{u}_1$ (Def. \ref{def:spaceresetstates}).

Thus, returning to Eq. \ref{eq:tempgeneratingdelay}, we see that the sub-event delay functions $\nu_{n,x}(t)$ are in fact, indpendent of $n$, and equal to the path-generating delay functions $\Theta_x$ defined in Def. \ref{def:pathgeneratingdelay}. Substituting these back into Eq. \ref{eq:tempdelayconvolution}, we recover the statement of the lemma.

\end{proof}

\begin{lemma}[Explicit form and properties of the path-generating delay function $\Theta_1$ (Def. \ref{def:pathgeneratingdelay}) and tick-generating delay function $\Gamma_1$ (Def. \ref{def:tickgeneratingdelay}).]\label{lemma:pathgeneratingdelay1}

\begin{align}\label{eq:theta1}
	\Theta_1(t) &= A e^{-gt}, \\
	\Gamma_1(t) &= B e^{-gt},
\end{align}
where $A \geq 0$ is the sum of the first (and only non-zero) column of $\notickmatrix_{10}$, $B \geq 0$ is the only (possibly) non-zero element of $\tickmatrix_{00}$, and $g \geq 0$ is the negation of the singular non-zero element in $\notickmatrix_{00}$. (see Eq. \ref{eq:generatorvisual}). If $g=0$, then $\Theta_1(t) = \Gamma_1(t) = 0$ for all $t$. Denoting the moments (Def. \ref{def:delayfunctionmoments}) of $\Theta_1$ by $\{Q_1,\mu_1,\chi_1\}$ and those of $\Gamma_1$ by $\{Q_3,\mu_3,\chi_3\}$, they are (in the case $g > 0$)
\begin{subequations}\label{eq:1dmoments}
	\begin{align}
		Q_1 &= \frac{A}{g} \\
		Q_3 &= \frac{B}{g} \\
		\mu_1 &= \mu_3 =  \frac{1}{g} \\
		\chi_1 &= \chi_3 = \frac{2}{g^2},
	\end{align}
\end{subequations}
and their $\gamma$-values and precisions $R$ (Def. \ref{def:accuracydelayfunction} and Eq. \ref{eq:accuracygamma}) are therefore
\begin{align}
	\gamma_1 &= \gamma_3 = 2, \\
	R \left[ \Theta_1 \right] &= R \left[ \Gamma_1 \right] = 1.
\end{align}

\end{lemma}

\begin{proof}

Via the definition of $\Theta_1$ in Def. \ref{def:pathgeneratingdelay}, and since $\mathbf{u}_0 = \mathbf{e}_0$ (Def. \ref{def:spaceresetstates},
\begin{align}
	\Theta_1(t) &= \onenorm{ \notickmatrix_{10} e^{\notickmatrix_{00} t} \mathbf{e}_0 }.
\end{align}

However, the operator $\notickmatrix_{00}$, (see Eq. \ref{eq:generatorvisual}), has only a single non-zero element, on the diagonal, and corresponding to $\mathbf{e}_0$. We label this element by , $-g$, where $g \geq 0$ (recall that the diagonal elements of $\notickmatrix$ are non-positive (Def. \ref{def:stochasticgenerators}). One thus simplifies the action of $e^{\notickmatrix_{00} t}$ on $\mathbf{e}_0$, obtaining
\begin{align}
	\Theta_1(t) &= \onenorm{ \notickmatrix_{10} e^{-g t} \mathbf{e}_0 } \\
	&= A e^{-gt},
\end{align}
where $A = \onenorm{ \notickmatrix_{10} \mathbf{e}_0 }$ is the sum of the first column of $\notickmatrix_{10}$.

If $g=0$, then from the definition of stochastic generators, Def. \ref{def:stochasticgenerators} and Corollary \ref{cor:zerotickmatrix}, it follows that the entire first column of $\notickmatrix$ is zero, and thus $A$ is also zero, leading to $\Theta_1$ being the zero function.

In a similar manner, using the definition of $\Gamma_1$ (Def. \ref{def:tickgeneratingdelay}),
\begin{align}
	\Gamma_1 &= \onenorm{ \tickmatrix_{00} e^{\notickmatrix_{00} t} \mathbf{e}_0 } = B e^{-gt},
\end{align}
where $B$ is the single (possibly) non-zero element in $\tickmatrix_{00}$. If $g=0$, then it follows from the fact that $\{\notickmatrix_{00},\tickmatrix_{00}\}$ is also a pair of stochastic generators (Def. \ref{def:stochasticgenerators}), that $B=0$ as well, and therefore $\Gamma_1$ is the zero function in this case.

The rest of the lemma follows from the direct application of the definition of the moments (Def. \ref{def:delayfunctionmoments}), the precision $R$ (Def. \ref{def:accuracydelayfunction}), and the $\gamma$-value (Eq. \ref{eq:accuracygamma}), in the case that $g>0$.

\end{proof}

\begin{corollary}\label{cor:boundedpartialnorm}

The partial norm (Def. \ref{def:partialnorm}) of the path-generation delay function $\Theta_1$ is strictly smaller than $1$ for all $t \geq 0$.

\end{corollary}

\begin{proof}

From the explicit form of $\Theta_1$ (Lemma \ref{lemma:pathgeneratingdelay1}), we can calculate its partial norm (Def. \ref{def:partialnorm}) explicitly,
\begin{align}
	P_t \left[ \Theta_1 \right] &= \int_0^t \Theta_1(t) dt = \begin{cases}
		0 & \text{if $g=0$}, \\
		\frac{A}{g} \left( 1 - e^{-gt} \right) & \text{if $g>0$.}
	\end{cases}
\end{align}

The case $g=0$ satisfies the corollary trivially. For the case $g>0$, note that
\begin{align}
	P_t \left[ \Theta_1 \right] &= \frac{A}{g} \left( 1 - e^{-gt} \right) = Q_1 \left( 1 - e^{-gt} \right) < Q_1 \quad \forall t,
\end{align}
where $Q_1$ is the zeroth moment of $\Theta_1$ (Lemma \ref{lemma:pathgeneratingdelay1}), and is itself upper bounded by $1$ (see Def. \ref{def:delayfunction}). Thus the partial norm is strictly less than $1$ for all $t$.

\end{proof}

\begin{lemma}\label{lemma:pathgeneratindelay2}

The delay functions $\Theta_0$ (Def. \ref{def:pathgeneratingdelay}) and $\Gamma_0$ (Def. \ref{def:tickgeneratingdelay}) can be generated by $d-1$ dimensional clocks.

\end{lemma}

\begin{proof}

We prove the statement for $\Theta_0$, the proof for $\Gamma_0$ is analogous. From Def. \ref{def:pathgeneratingdelay},
\begin{align}\label{eq:pathdelay0}
\Theta_0(t) &= \onenorm{ \notickmatrix_{01} e^{\notickmatrix_{11} t} \mathbf{u}_1 }.
\end{align}

While the expression above appears to involve the entire vector space $\mathbb{R}^d$, in fact, one can generate the same delay function with only $d-1$ dimensional objects, as we proceed to show.

First off, take the event generator $\notickmatrix_{01}$. From Eq. \ref{eq:generatorvisual}, we observe that only the first row is non-zero, and of this row, the first element is zero. Consider the modified event generator $\notickmatrix_{01}^\prime$, formed by swapping the first row with any other. By Lemma \ref{lemma:resetinvariance}, this leaves the delay function $\Theta_0(t)$ unchanged. The entire first row and column of the modified generator are zero. Next, we note that $\notickmatrix_{11}$ by construction (see Eq. \ref{eq:generatorvisual}) already has a zero first row and column. Thus both $\notickmatrix_{11}$ and $\notickmatrix_{01}^\prime$ act trivially on the canonical state $\mathbf{e}_0$.

Finally, $\mathbf{u}_1$ (Def. \ref{def:spaceresetstates}) is constructed from the first column of $\notickmatrix_{10}$, that has a zero element at the top (see Eq. \ref{eq:generatorvisual}), and thus $\mathbf{u}_1$ has no component from the canonical state $\mathbf{e}_0$. Thus we may simply remove this space entirely from $\mathbf{u}_1$, and correspondingly from the operators $\notickmatrix_{11}$ and $\notickmatrix_{01}^\prime$, and still generate the same delay function $\Theta_0$.

\end{proof}

\begin{corollary}\label{cor:previousbound}

Assuming that the precision of a single tick of a $d-1$ dimensional clock is upper bound by $R \leq d-1$, as is done during this proof, one can lower bound the second moments of $\Theta_0$ and $\Gamma_0$ w.r.t. their first moments,
\begin{subequations}
\begin{align}
	\chi_0 &\geq \mu_0^2 \left( 1 + \frac{1}{d-1} \right) \\
	\chi_2 &\geq \mu_2^2 \left( 1 + \frac{1}{d-1} \right),
\end{align}
\end{subequations}
if the first moments $\mu_0$ and $\mu_2$ do not diverge.
\end{corollary}

\begin{proof}

From Lemma \ref{lemma:pathgeneratindelay2}, we know that both $\Theta_0$ and $\Gamma_0$ can be generated by $d-1$ dimensional clocks, but we have assumed the precision $R$ (Def. \ref{def:accuracydelayfunction}) of these clocks to be upper bound by $R \leq d-1$. The corollary then follows from the definition of the $\gamma$-value (Eq. \ref{eq:accuracygamma}) and its relationship with the precision $R$.

\end{proof}

\begin{lemma}[Recovering the clock state and delay function from the path-specific states and component delay functions]\label{lemma:recoverstatedelay}

The state of the clock $V(t)$ corresponding to the first tick (Eq. \ref{eq:singletickstatedelay}) is the series sum of the path-specific states (Def. \ref{def:pathstates}),
\begin{align}
	V(t) &= \sum_{n=0}^\infty \sum_{x\in\{0,1\}} v_{n,x}(t).
\end{align}

Furthermore, the delay function $\tau(t)$ of a single tick of the clock (Eq. \ref{eq:singletickstatedelay}) is the series sum of the component delay functions (Def. \ref{def:componentdelay}),
\begin{align}
	\tau(t) &= \tau_{0,1}(t) + \sum_{n=1}^\infty \sum_{x\in\{0,1\}} \tau_{n,x}(t).
\end{align}
	
\end{lemma}

\begin{proof}

Consider the following sequence of states, $\{v^{(M)}_{n,x}(t)\}$, defined for $x \in \{0,1\}$ and $n \in \{0,1,...,M\}$, where $M \in \mathbb{N}^+$, and whose initial states and dynamics are identical to those of the path-specific states $v_{n,x}(t)$ (see Def. \ref{def:pathstates}) except for the last one $v^{(M)}_{M,1}(t)$, that evolves as
\begin{align}\label{eq:endstatedynamics}
	\frac{d}{dt} v^{(M)}_{M,1} &= \notickmatrix_{10} v^{(M)}_{M,0} + \notickmatrix v^{(M)}_{M,1}.
\end{align}
where the last term above differentiates it from the dynamics of the original sequence, by including the entire non-event generator $\notickmatrix$ rather than only $\notickmatrix_{11}$. Intuitively this corresponds to interrupting the sequence at $n=M,x=1$ by stopping the flow of of population to further states. It is straightforward to verify that this too is an independent Markovian sequence of events (Def. \ref{def:independentsequence}).

Since the initial states and dynamical equations of all but the last state in the sequence $v^{(M)}_{n,x}$ are identical to those of the sequence $v_{n,x}$, and because the dynamics of every state in a Markovian sequence only depend on the previous states in the sequence, it follows that $v^{(M)}_{n,x}(t) = v_{n,x}(t)$ for all $x$ and $M$ except for the final state $v^{(M)}_{M,1}(t)$. Furthermore, this also implies that the corresponding delay functions of the sequence of states $v^{(M)}_{n,x}$ are the path-specific delay functions $\xi_{n,x}$ (Def. \ref{def:pathdensities}), upto and including $n=M,x=1$.

Consider the vector $V^{(M)}$, defined as the sum
\begin{align}\label{eq:tempstatesum}
	V^{(M)}(t) &= \sum_{n=0}^{M} \sum_{x\in\{0,1\}} v^{(M)}_{n,x}(t),
\end{align}
that we proceed to show is equal to $V(t)$, the state of the clock for a single tick (Eq. \ref{eq:singletickstatedelay}). For $t=0$, by construction $V^{(M)}(0) = V(0)$. For $t \geq 0$, the dynamics of the sum is given by (Eqs. \ref{eq:spacestatedynamics} and \ref{eq:endstatedynamics})
\begin{align}
	\frac{d}{dt} V^{(M)}(t) &= \sum_{n=0}^{M} \sum_{x\in\{0,1\}} \frac{d}{dt} v^{(M)}_{n,x}(t) \\
	&= \sum_{n=0}^{M-1} \sum_{x\in\{0,1\}} \left( \notickmatrix_{xx} + \notickmatrix_{\bar{x}x} \right) v^{(M)}_{n,x}(t) + \left( \notickmatrix_{00} + \notickmatrix_{10} \right) v^{(M)}_{M,0}(t) + \notickmatrix v^{(M)}_{M,1}(t) \\
	&= \sum_{n=0}^{M-1} \sum_{x\in\{0,1\}} \left( \notickmatrix_{xx} + \notickmatrix_{\bar{x}x} \right) v_{n,x}(t) + \left( \notickmatrix_{00} + \notickmatrix_{10} \right) v_{M,0}(t) + \notickmatrix v^{(M)}_{M,1}(t). \label{eq:tempref1}
\end{align}

We may use the fact that each $v_{n,x}(t) \in S_x$ for all $t \geq 0$ (Lemma \ref{lemma:pathmarkoviansequence}), together with the definition of the generators $\notickmatrix_{xy}$ (Eq. \ref{eq:spacegenerators}) to conclude that for all $n$ and $x$,
\begin{align}
	\left( \notickmatrix_{\bar{x}\bar{x}} + \notickmatrix_{x\bar{x}} \right) v_{n,x}(t) = \mathbf{0}.
\end{align}
Adding these trivial zeros to Eq. \ref{eq:tempref1}, and using the decomposition of $\notickmatrix$, (Eq. \ref{eq:spacegenerators}),
\begin{align}
	\frac{d}{dt} V^{(M)}(t) &= \sum_{n=0}^M \sum_{x\in\{0,1\}} \notickmatrix v^{(M)}_{n,x}(t) \\
	&= \notickmatrix V^{(M)}(t). \\
	\therefore \; V^{(M)}(t) &= e^{\notickmatrix t} V^{(M)}(0) = e^{\notickmatrix t} V(0) = V(t). \label{eq:tempstateequal}
\end{align}

Finally, we investigate the limit $M \rightarrow \infty$ of $V^{(M)}(t)$,
\begin{align}
	\lim_{M\rightarrow\infty} V^{(M)}(t) &= \lim_{M\rightarrow\infty} \left( \sum_{n=0}^{M-1} \sum_{x\in\{0,1\}} v_{n,x}(t) + v_{M,0}(t) + v^{(M)}_{M,1}(t) \right) \\
	&= \sum_{n=0}^\infty \sum_{x\in\{0,1\}} v_{n,x}(t) + \lim_{M\rightarrow\infty} v^{(M)}_{M,1}(t). \label{eq:tempstateMinfty}
\end{align}

To prove that the state $v^{(M)}_{M,1}(t)$ goes to zero, we first apply Lemma \ref{lemma:resetsequenceconvolution} to express it as
\begin{align}
	v^{(M)}_{M,1}(t) &= \int_0^t e^{\notickmatrix (t-t^\prime)} \mathbf{w}_{M,1} \; \xi_{M,1}(t^\prime) (t^\prime) dt^\prime,
\end{align}
where $\mathbf{w}_{M,1}$ is the event reset state (Def. \ref{def:eventresetstate}) for the $n=M,x=1$ event, that has already been proven to be equal to $\mathbf{u}_1$ (see discussion after Eq. \ref{eq:tempgeneratingdelay}). Calculating the norm (Def. \ref{def:norm}) of the above,
\begin{align}
	\onenorm{v_{M,1}(t)} &= \int_0^t \onenorm{ e^{\notickmatrix_{11} (t-t^\prime)} \mathbf{u}_1 } \xi_{M,1}(t^\prime) dt^\prime.
\end{align}
However the term in the e-sum above is a population vector for all $t-t^\prime \geq 0$ (see Lemma \ref{lemma:stochasticgenerators}), and thus the norm is upper bounded by $1$, leading to
\begin{align}
	\onenorm{v_{M,1}(t)} &\leq \int_0^t \xi_{M,1}(t^\prime) dt^\prime = P_t \left[ \xi_{M,1} \right].
\end{align}
where the partial norm $P_t[\cdot]$ is defined in Def. \ref{def:partialnorm}. As $\xi_{n,x}$ is itself a convolution of delay functions (Lemma \ref{lemma:pathdensityconv}), we can apply Lemma \ref{lemma:delaynorm}, to obtain
\begin{align}
	\onenorm{v_{M,1}(t)} &\leq \left( P_t \left[ \Theta_1 \right] \right)^{M+1} \left( P_t \left[ \Theta_0 \right] \right)^M \\
	&\leq \left( P_t \left[ \Theta_1 \right] \right)^{M+1},
\end{align}
using the fact that the partial norm is upper bounded by 1.

However, from Corollary \ref{cor:boundedpartialnorm}, the partial norm of $\Theta_1$ is srictly less than $1$ for all $t \geq 0$, and thus
\begin{align}\label{eq:templimitzero}
	\lim_{M\rightarrow\infty} \onenorm{ v_{M,1}(t) } = 0,
\end{align}
which in turn implies that every element of $v_{M,1}(t)$ must individually go to zero. Combining this with Eqs. \ref{eq:tempstateequal} and \ref{eq:tempstateMinfty}, we recover
\begin{align}
	V(t) &= \sum_{n=0}^\infty \sum_{x\in\{0,1\}} v_{n,x}(t).
\end{align}

The proof of the second statement of the lemma, concerning the delay function, proceeds in an analogous manner. Consider the following sequence of delay functions, $\{\tau^{(M)}(t)\}$, defined w.r.t. the sequence of states $v^{(M)}_{n,x}(t)$ as
\begin{align}
	\tau^{(M)}_{n,x}(t) &= \onenorm{ \tickmatrix_{0\bar{x}} v^{(M)}_{n-\bar{x},\bar{x}}(t) },
\end{align}
for all $n,x$ from $n=0,x=1$ up to and including $n=M,x=1$. Finally, the end delay function of the sequence is defined as
\begin{align}
	\tau^{(M)}_{M+1,0}(t) &= \onenorm{ \tickmatrix v^{(M)}_{M,1}(t) }.
\end{align}

Note immediately that since $v^{(M)}_{n,x}(t) = v_{n,x}(t)$ for all the states except $v^{(M)}_{M,1}(t)$, it follows that $\tau^{(M)}_{n,x}(t) = \tau_{n,x}(t)$ (Def. \ref{def:pathdensities}) except for the final delay function $\tau^{(M)}_{M+1,0}(t)$.

Consider the sum
\begin{align}
	f^{(M)}(t) &= \tau^{(M)}_{0,1}(t) + \sum_{n=1}^{M} \sum_{x\in\{0,1\}} \tau^{(M)}_{n,x}(t) + \tau^{(M)}_{M+1,0}(t),
\end{align}
which we proceed to show is equal to the delay function $\tau(t)$ of a single tick of the clock, Eq. \ref{eq:singletickstatedelay}.
\begin{align}
	f^{(M)}(t) &= \onenorm{ \tickmatrix_{00} v^{(M)}_{0,0}(t) } + \sum_{n=1}^{M-1} \sum_{x\in\{0,1\}} \onenorm{ \tickmatrix_{0\bar{x}} v^{(M)}_{n-\bar{x},\bar{x}}(t) } + \onenorm{ \tickmatrix_{00} v^{(M)}_{M,0}(t) } + \onenorm{ \tickmatrix v^{(M)}_{M,1}(t) } \\
	&= \onenorm{ \tickmatrix_{00} v_{0,0}(t) } + \sum_{n=1}^{M} \sum_{x\in\{0,1\}} \onenorm{ \tickmatrix_{0\bar{x}} v_{n-\bar{x},\bar{x}}(t) } + \onenorm{ \tickmatrix_{00} v_{M,0}(t) }  + \onenorm{ \tickmatrix v^{(M)}_{M,1}(t) }
\end{align}

Using the definition of $\tickmatrix_{xy} = \Pi_x \tickmatrix \Pi_y$ (Eq. \ref{eq:spacetickers}), and also the fact that each $v_{n,x} \in S_x$ (Lemma \ref{lemma:pathmarkoviansequence}), we can trivially add $\tickmatrix_{0 \bar{x}} v^{(M)}_{n,x} = \mathbf{0}$ to any expression, and doing so to the above equation results in
\begin{align}
	f^{(M)}(t) &= \onenorm{ \sum_{n=0}^{M-1} \sum_{x\in\{0,1\}} \left( \tickmatrix_{0x} + \tickmatrix_{0\bar{x}} \right) v^{(M)}_{n,x}(t) } + \onenorm{ \left( \tickmatrix_{0x} + \tickmatrix_{0\bar{x}} \right) v^{(M)}_{M,0}(t) } + \onenorm{ \tickmatrix v^{(M)}_{M,1}(t) } \\
	&= \onenorm{ \tickmatrix V^{(M)}(t) } = \onenorm{ \tickmatrix V(t) } = \tau(t). \label{eq:temppartialsumdelay}
\end{align}

Thus $f^{(M)}(t)$ is the delay function of a single tick of the clock (Eq. \ref{eq:singletickstatedelay}). To complete the proof, we investigate the limit as $M \rightarrow \infty$,
\begin{align}
	\lim_{M\rightarrow\infty} f^{(M)} &= \lim_{M\rightarrow\infty} \left( \tau^{(M)}_{0,1}(t) + \sum_{n=1}^{M} \sum_{x\in\{0,1\}} \tau^{(M)}_{n,x}(t) + \tau^{(M)}_{M+1,0}(t) \right) \\
	&= \lim_{M\rightarrow\infty} \left( \tau_{0,1}(t) + \sum_{n=1}^{M} \sum_{x\in\{0,1\}} \tau_{n,x}(t) \right) + \lim_{M\rightarrow\infty} \tau^{(M)}_{M+1,0}(t) \\
	&= \sum_{n=0}^\infty \sum_{x\in\{0,1\}} \tau_{n,x}(t) + \lim_{M\rightarrow\infty}  \onenorm{ \tickmatrix v^{(M)}_{M,1}(t) }.
\end{align}

But the remaining term on the right goes to zero, as we have already proved that $v^{(M)}_{M,1}(t)$ goes to zero (see Eq. \ref{eq:templimitzero}), and $\tickmatrix$ is a bounded linear operator. From the above expression and Eq. \ref{eq:temppartialsumdelay}, we recover
\begin{align}
	\tau(t) &= \sum_{n=0}^\infty \sum_{x\in\{0,1\}} \tau_{n,x}(t),
\end{align}
which completes the proof of the lemma.

\end{proof}

\subsubsection{The moments of the delay function of a single tick of the clock}\label{proof:delayexplicit}

We proceed by calculating the moments of the delay function of a single tick of a clock by using the path-specific delay functions and component delay functions.

Consider the component delay function (Def. \ref{def:componentdelay}),
\begin{align}
	\tau_{n,x}(t) &= \onenorm{ \tickmatrix_{0\bar{x}} v_{n-\bar{x},\bar{x}}(t) } \\
	&= \onenorm{ \int_0^t \tickmatrix_{0\bar{x}} e^{\notickmatrix_{\bar{x}\bar{x}} (t-t^\prime)} \mathbf{u}_{\bar{x}} \xi_{n-\bar{x},\bar{x}}(t^\prime) dt^\prime },
\end{align}
where we have applied Lemma \ref{lemma:resetsequenceconvolution} to express the path-specific state $v_{n-\bar{x},\bar{x}}(t)$ w.r.t. it's corresponding path-specific delay function, as they form an independent Markovian sequence (Lemma \ref{lemma:pathmarkoviansequence}). Continuing, we simplify the above to
\begin{align}
	\tau_{n,x}(t) &= \int_0^t \onenorm{ \tickmatrix_{0\bar{x}} e^{\notickmatrix_{\bar{x}\bar{x}} (t-t^\prime)} \mathbf{u}_{\bar{x}} } \xi_{n-\bar{x},\bar{x}}(t^\prime) dt^\prime \\
	&= \int_0^t \Gamma_x(t-t^\prime)  \xi_{n-\bar{x},\bar{x}}(t^\prime) dt^\prime,
\end{align}
using the definition of the tick-generating delay functions $\Gamma_x$ (Def. \ref{def:tickgeneratingdelay}). Finally,
\begin{align}
	\tau_{n,x}(t) &= \left( \Gamma_x \conv \xi_{n-\bar{x},\bar{x}} \right) (t) \\
	&= \left( \Gamma_x \conv \Theta_1 \conv \Theta_0 \conv \Theta_1 \conv \Theta_0 \conv ... \conv \Theta_{\bar{x}} \right) (t),
\end{align}
using Lemma \ref{lemma:pathdensityconv}. In the above convolution, $\Theta_1$ appears $n$ times, and $\Theta_0$ appears $n-\bar{x}$ times.

We can now calculate the moments of the component delay function $\tau_{n,x}$ using Eq. \ref{eq:sequencemoments} for a convolution of a sequence of delay functions. The moments (Def. \ref{def:delayfunctionmoments}) of $\Theta_1$ and $\Gamma_1$ have already been labelled and calculated in Lemma \ref{lemma:pathgeneratingdelay1}. For now, we label them (in order of zeroth, first and second moments) as $\{Q_1,\mu_1,\chi_1\}$ for $\Theta_1$ and $\{Q_3,\mu_1,\chi_1\}$ for $\Gamma_1$ (the first and second moments of these two delay functions were proven equal in Lemma \ref{lemma:pathgeneratingdelay1}). The moments of $\Theta_0$ (Def. \ref{def:pathgeneratingdelay}) are denoted as $\{Q_0,\mu_0,\chi_0\}$ and the moments of $\Gamma_0$ (Def. \ref{def:tickgeneratingdelay}) by $\{Q_2,\mu_2,\chi_2\}$.

Finally, denoting the zeroth, first and second moments of the component delay function $\tau_{n,x}(t)$ by $Q_{n,x}$, $\mu_{n,x}$ and $\chi_{n,x}$ respectively, from Eq. \ref{eq:sequencemoments},
\begin{align}
	Q_{n,0} &= Q_1^n Q_0^{n-1} Q_2 \\
	\mu_{n,0} &= n \mu_1 + (n-1) \mu_0 + \mu_2 \\
	\chi_{n,0} &= n \chi_1 + (n-1) \chi_0 + \chi_2 + n(n-1) \mu_1^2 + (n-1)(n-2) \mu_0^2 + 2 n(n-1) \mu_1 \mu_0 + 2 n \mu_1 \mu_2 + 2 (n-1) \mu_0 \mu_2 \\
	Q_{n,1} &= Q_1^n Q_0^n Q_3 \\
	\mu_{n,1} &= (n+1) \mu_1 + n \mu_0 \\
	\chi_{n,1} &= (n+1) \chi_1 + n \chi_0 + n(n+1) \mu_1^2 + n(n-1) \chi_0^2 + 2 n(n+1) \mu_1 \mu_0. \\
\end{align}

We are now in a position to calculate the precision of a single tick of the clock. First off, we split the delay function of a single tick of the clock into two parts, w.r.t. the some over component delay functions,
\begin{align}
	\tau(t) &= \tau_{0,1}(t) + \sum_{n=0}^\infty \sum_{x\in\{0,1\}} \tau_{n,x}(t) \\
	&= \sum_{n=1}^\infty \tau_{n,0}(t) + \sum_{n=0}^\infty \tau_{n,1}(t) \\
	&= \tau^{(1)}(t) + \tau^{(0)}(t).
\end{align}

We may understand the above division in the following manner. The component delay function $\tau_{n,x}$ is generated by $\tickmatrix_{0\bar{x}}$ (see Def. \ref{def:componentdelay}), and thus represents a tick being generated from the $S_{\bar{x}}$ subspace of the clock vector space. Thus the two $\tau^{(x)}(t)$ above each correspond to the delay function of a single tick generated from $S_{x}$.

From Lemma \ref{lemma:delaymixture}, the precision of the sum of two delay functions is upper bound by the maximum precision from among the two, and we proceed to bound each of the precisions individually.

For $\tau^{(0)}$, (that corresponds to the tick arising from the $S_0$ subspace), we find the zeroth, first and second moments of the entire tick density to be (on application of Def. \ref{def:delayfunctionmoments}),
\begin{align}
	Q^{(0)} &= \sum_{n=0}^\infty Q_{n,1} = \frac{Q_3}{1 - Q_1 Q_0} \\
	\mu^{(0)} &= \frac{\sum_{n=0}^\infty Q_{n,1} \mu_{n,1}}{Q^{(0)}} = \frac{Q_0 Q_1 \mu_0 + \mu_1}{1 - Q_0 Q_1} \\
	\chi^{(0)} &= \frac{\sum_{n=0}^\infty Q_{n,1} \chi_{n,1}}{Q^{(0)}} = \frac{(Q_0 Q_1 \mu_0 + \mu_1)^2 + Q_0 Q_1 \chi_0 (1 - Q_0 Q_1)}{(1 - Q_0 Q_1)^2},
\end{align}
where we have replaced $\chi_1 = 2 \mu_1^2$ (Eq. \ref{eq:1dmoments}) to simplify the expressions. Proceeding, we find that the precision is given by (Eq. \ref{eq:accuracygamma})
\begin{align}
	R \left[ \tau^{(0)} \right] &= \left( \frac{\chi^{(0)}}{\left( \mu^{(0)} \right)^2} - 1 \right)^{-1} = \frac{\left( Q_0 Q_1 \mu_0 + \mu_1 \right)^2}{ \left( Q_0 Q_1 \mu_0 + \mu_1 \right)^2 + Q_0 Q_1 \chi_0 \left( 1 - Q_0 Q_1 \right) }.
\end{align}

As $Q_0 Q_1$ is the product of the zeroth moments of delay functions and thus $\leq 1$, we see that the precision decreases monotonically w.r.t. $\chi_0$. As such, we may upper bound it by replacing $\chi_0$ by its minimum value using Corollary \ref{cor:previousbound}. Furthermore, noting that the moments $Q_0$ and $Q_1$ only appear as the product, and that $\mu_0$ and $\mu_1$ only affect the precision via their relative ratio, we can simplify the expression to
\begin{align}
	R \left[ \tau^{(0)} \right] &\leq \frac{(d-1)(1+pr)^2}{d + dpr(2+r) - (1+pr)^2},
\end{align} 
where $p = Q_0 Q_1 \in (0,1]$ and $r=\mu_0/\mu_1 \in (0,\infty)$. One can analytically optimize the above expression w.r.t. $p$ and $r$, and we find that the maximal value of the precision is
\begin{align}
	R \left[ \tau^{(0)} \right] &\leq 1,
\end{align}
which is only the precision of a one-dimensional clock, and not very interesting.

On the other hand, consider the delay function $\tau^{(1)}(t)$, that corresponds to the tick being generated from $S_1$. We find the moments of the delay function to be
\begin{align}
Q^{(1)} &= \sum_{n=0}^\infty Q_{n,0} = \frac{Q_1 Q_2}{1 - Q_0 Q_1} \\
\mu^{(1)} &= \frac{\sum_{n=0}^\infty Q_{n,0} \mu_{n,0}}{Q^{(1)}} = \frac{Q_0 Q_1 \mu_0 + \mu_1}{1 - Q_0 Q_1} + \mu_2 \\
\chi^{(1)} &= \frac{\sum_{n=0}^\infty Q_{n,0} \chi_{n,0}}{Q^{(1)}} = \frac{2 \left( Q_0 Q_1 \mu_0 + \mu_1 \right) \left( \mu_1 + Q_0 Q_1 \left( \mu_0 - \mu_2 \right) + \mu_2 \right) + \left( 1 - Q_0 Q_1 \right) \left( Q_0 Q_1 \chi_0 + \left( 1 - Q_0 Q_1 \right) \chi_2 \right) }{(1 - Q_0 Q_1)^2}.
\end{align}

Here too, the second moment is seen to increase monotonically w.r.t both $\chi_0$ and $\chi_2$. As the precision itself is inversely related to the second moment, we may substitute the minimum $\chi_0$ and $\chi_2$ by their minimum values using Corollary \ref{cor:previousbound} to upper bound the precision. Furthermore, we replace $Q_0 Q_1 = p \in (0,1]$, $\mu_0/\mu_1 = r \in (0,\infty)$ and $\mu_2/\mu_1 = s \in (0,\infty)$ to find that the bound on the precision takes on the form
\begin{align}
	R \left[ \tau^{(1)} \right] &= \left( \frac{\chi^{(1)}}{\left( \mu^{(1)} \right)^2} - 1 \right)^{-1} \leq \frac{ (d-1) (1 + p(r-s) + s)^2}{ d + dpr(2+r) - (1-pr)^2 + (1-p)^2 s^2}.
\end{align}

One may analytically optimize the above expression w.r.t. $\{p,r,s\}$, and doing so returns the maximal value
\begin{align}
	R \left[ \tau^{(1)} \right] &\leq d.
\end{align}

Thus the precision of $\tau^{(1)}$, and thus $\tau$, the delay function of a single tick of a $d$-dimensional reset clock, is upper bound by $d$. By induction the result follows for all $d$. From Sec. \ref{sec:justifyR}, we can then conclude that the precision of the $n^{th}$ tick of the reset clock (and thus any clock, from Theorem \ref{theorem:resetclockaccuracy}) is upper bound by $n d$, completing the statement of the theorem.

\section{Quantum Clocks | Proof of Theorem \ref{thm:quantum lower bound}}\label{sec:Quantum Clocks appendix}





  


The purpose of this section is to prove Theorem \ref{thm:quantum lower bound} in the main text. For the sake of generality, the version of the theorem which we will derive, Theorem \ref{Thm:R lower quant bound}, is more general than Theorem \ref{thm:quantum lower bound}, which has not been stated in the main text since it requires some more technical definitions. After summarising the setup in the example section \ref{sec:example}, we will start by introducing the \wso~Clock and stating the necessary theorems about its dynamical properties from \cite{WSO16}. Then we will summarise the exact form of the potential to which Theorem \ref{Thm:R lower quant bound} applies. Then we will proceed by proving the necessary lemmas, and mathematical statements  before proving Theorem \ref{Thm:R lower quant bound} in Section \ref{Final Theorem}.
\subsection{Setup}\label{sec:setup}
In this section, we will briefly recall the important results from the example of a quantum clock in Section \ref{sec:example}, since they will form the starting point of the proof.

The probability of not getting a tick in time interval $[0,t]$, is $\tr[\rho(t)]$, where $\ket{\bar{\psi}_t}$ is given by Eq. \eqref{eq:psi bar t ex}, namely
\be 
\rho(t):=\ketbra{\bar{\psi}_t}{\bar{\psi}_t} =\me^{-\mi t H }\rho_0  \me^{\mi t H^\dag}, 
\ee
where $\rho_0=\ketbra{\psi_0}{\psi_0}$ is the initial state, and 
\be \label{H total clock plus V}
H=\hat H_\cl-\mi \hat V_\cl,
\ee 
describes the evolution of the system with $\hat H_\cl$ being the Hamiltonian and $\hat V_\cl$ the potential. 
The probability of getting the 1st tick in the infinitesimal time step $\delta>0$ is then 
\be\label{eq:def P(t)}
P(t)= 2\, \tr[\hat V_C \rho(t)] = \frac{d}{dt}\left(1- \tr[\rho(t)]\right)=-\frac{d}{dt}\tr[\rho(t)], \quad t\geq 0,
\ee  
where the last equality follows by taking the derivative inside the trace and noting that 
\be 
\frac{d}{dt} \rho(t)= -\mi  \hat H \rho(t) +\mi \rho(t) \hat H^\dag.
\ee 
The mean and standard deviation of the tick distribution is
\be 
\mu:= \int_0^{\infty} t P(t) dt, \quad  \underline{\sigma}:= \sqrt{\int_0^{\infty} (t-\mu)^2 P(t) dt}.
\ee
Since this example of a quantum clock is a reset clock, we are only interested in the precision of the 1st tick, ${R_1}$ since the precision of later ticks is determined solely from it; $R_j= j R_1$ (see Section \ref{Accuracy of Quantum Clocks}, or \app~\ref{app:delaysequence}, and Remark \ref{remark:delaysequence} for a detailed argument). Namely, we are interested in bounding the quantity
\be 
{R_1}:= \frac{\mu^2}{\underline{\sigma}^2}.
\ee 

\subsection{Overview of the \wso~Clock}\label{Overview of the quasi-ideal clock}
In this section, we will recall some of the definitions and the core theorem from \cite{WSO16} which we will need in this \app. This section will set some of the terminology and definitions needed for the proof of Theorem \ref{Thm:R lower quant bound}. 
\subsubsection{The generator of dynamics and potential function $\bar V_0$}\label{The generator of dynamics and potential function bar V0}

We start by introducing a generator of dynamics for the \wso~Clock.
\be \label{eq:qusi ideal clock gen}
\hat H=\hat V_d +\hat H_\cl, 
\ee 
The clock's free Hamiltonian, $\hat H_\cl$ is a truncated Harmonic Oscillator Hamiltonian. Namely, $\hat H_\cl=\sum_{n=0}^{d-1}\omega n \ketbra{n}{n}$. The free evolution of any initial clock state under this Hamiltonian has a period of $T_0=2\pi/\omega$, specifically, $\me^{-\mi T_0 \hat H_\cl}\rho_\cl\me^{\mi T_0\hat H_\cl} = \rho_\cl$ for all $\rho_\cl$. The clock interaction term $\hat V_d$, takes the form,
\be 
\hat{V}_d = \frac{d}{T_0} \sum_{k=0}^{d-1} {V}_d(k) \ketbra{\theta_k}{\theta_k},
\ee 
where the basis $\{\ket{\theta_k}\}_{k=0}^{d-1}$ is the Fourier transform of the energy eigenbasis $\{\ket{n}\}_{n=0}^{d-1}$, specifically,
\begin{align}\label{finitetimestates_main}
\ket{\theta_k} &= \frac{1}{\sqrt{d}} \sum_{n=0}^{d-1} e^{-i2\pi n k/d} \ket{E_n}.
\end{align}
 In other words, we choose $\{\ket{\theta_k}\}$ to be the time basis. It will also be useful later to have the range of $k$ extended to $\zz$.\footnote{Note that $k$ will belong to a set of only $d$ consecutive integers so that $\{\ket{\theta_k}\}$ form a complete orthonormal basis without repetition.} Extending the range of $k$ in Eq.  \eqref{finitetimestates_main} it follows that $\ket{\theta_k}=\ket{\theta_{k \textup{ mod. } d}}$ for $k\in\zz$. It is this basis which we identify with the time basis, i.e. we set $\ket{\theta_k}=\ket{t_k}$, $k\in\nno$.
The function ${V}_d:\rr\mapsto\rr\cup\hh^{-}$ (where $\hh^{-}:= \{ a_0+\mi b_0 : a_0\in\rr, b_0 < 0 \}$ denotes the lower-half complex plane) is defined by 
\be 
V_d(x)= \frac{2\pi}{d} V_0\left( \frac{2\pi}{d}x \right),
\ee 
where $V_0$ is any infinitely differentiable periodic function of period $2\pi$. We will only be interested in a specialised case, and from here on will choose $V_d$ to map $ \rr \mapsto -\mi\, \rr_{\geq 0}$. As one might imagine, it will be more convenient to work with a real function. We will also want to make explicit the normalisation. We will thus write the potential function in terms of an explicitly positive, normalised function by defining $\bar V_d:\rr\rightarrow \rrp$ and $\bar V_0:\rr\rightarrow \rrp$,
\ba
-\mi \delta \,\bar{V}_d(x)&:=V_d(x),\\
-\mi \delta \,\bar{V}_0(x)&:=V_0(x),\quad\quad \int_0^{2\pi} dx \bar{V}_0(x)=1, \quad\quad \bar V_0(x)>0\quad \forall x\in\rr,\label{eq:pod eqs def}
\ea 
where we will specialise to the case $\delta\geq 1$. The last constraint ensures that $\bar V_0$ has full support. Secondly, it will be useful to let $\bar V_0$ have a unique global maximum in the interval $x\in[0,2\pi]$ at $x=x_0$. Let the parameters $\tilde\epsilon_V, x_{vl}, x_{xr}$ be such that
\be\label{eq:tilde ep V def}
1-\tilde\epsilon_V=\int_{x_{vl}}^{x_{vr}} dx\, \bar{V}_0(x+x_0)
\ee
for some $-\pi\leq x_{vl}< 0< x_{vr}\leq \pi$, where, for simplicity, we set $x_{vl}=- x_{vr}$. Furthermore, we will find that it is convenient to set $x_0=\pi+x_{vr}+\pi\gamma$. These are all the properties of $\bar V_0$ which we will need for the lemmas in Sections \ref{Approximating R1}, \ref{Upper and lower bounds for trace rho}, \ref{Calculating Delta(0) and Delta(1) and an explicit lower bound on R1}, \ref{R1 to leading order}. In Section \ref{Showing that the limit requirements}, we will find explicit parametrisations of $x_{vr}$ and $\gamma$ in terms of $d$.

\subsubsection{Explicit form on the potential $\bar V_0$}\label{Explicit form on the potential bar V0}
In Section \ref{Showing that the limit requirements} we will need to introduce an explicit function for the potential $\bar V_0$ satisfying the criteria from Section \ref{The generator of dynamics and potential function bar V0}. The function will be written in terms of parameters $\delta \geq 1$, $n\geq 1$, $N\in\nnp$. Specific values for these parameters will be chosen in Section \ref{Showing that the limit requirements} in order to satisfy specific conditions. Here we summarise the final form of the potential $\bar V_0$ to which Theorem \ref{Thm:R lower quant bound} applies.

For $d\geq 3$, let
\be \label{eq:def of rapped pot 0}
\bar{V}_0(x):=\frac{1}{\delta d^2}+n A_0 \sum_{p=-\infty}^{\infty} V_B\left (n(x-x_0-2\pi p) \right),
\ee
where
\be 
A_0=\frac{1-2\pi/(\delta d^2)}{\int_{-\infty}^\infty dx V_B(x)},
\ee 
is a normalization parameter chosen such that
\be \label{eq:norr V0 requirement 0}
\int_0^{2\pi} \bar{V}_0(x)\,dx =1,
\ee
and
\be 
V_B(x):=\textup{sinc}^{2 N}(x):=\left( \frac{\sin(\pi x)}{\pi x}\right)^{2 N},
\ee 
where $V_B(0)=1$ and $N\in \nnp$. $N$ is given by
\be 
N=\left\lceil\frac{3-4\epsilon_5-\epsilon_9}{2(\epsilon_7-\epsilon_8-\epsilon_5)}\right\rceil,
\ee
with $\lceil\cdot\rceil$ the ceiling function, and where $\epsilon_7=\eta/4,$ $\epsilon_5=\epsilon_8=\eta/16,$ and $\epsilon_9=\eta/2,$ and the role of $\eta\in(0,1)$ becomes apparent in Theorem \ref{eq:lower quantum}. The parametre $\delta$ is given by
\be 
\delta=d^{\epsilon_5}.
\ee 
Let $C_0: \nnp\mapsto\rrp$ be any function which is only a function of $N$, (i.e. independent of $n$, $d$, and $k$), such that
\be \label{eq:main eq lemma dev bound 0}
\max_{x\in[0,2\pi]} \bigg{|} \frac{d^k}{dx^k}  \bar{V}_0(x)\bigg{|} \leq n^{k+1} C_0^{k+1},\quad \forall\, k\in\nno, \forall n\geq 1,
\ee
holds. We will prove later by explicit construction that such functions exist.
We now define $n$ to be
\be\label{eq:n in terms of d def 0} 
n= \frac{\ln(\pi \alpha_0\sigma^2)}{2\pi C_0\alpha_0\kappa} \frac{d^{1-\epsilon_5}}{\delta \sigma},
\ee
where $\kappa=0.792$ and $\alpha_0\in(0,1]$ depends on the mean energy of the initial  \wso~Clock state and is defined in Section \ref{The wso clock states}. Finally, the location of the peak of the potential, $x_0$, is given by $x_0=\pi+x_{vr}+\pi\gamma$, where
\ba \label{eq:x vr limit contraint 0}
x_{vr}&=d^{\epsilon_7} \frac{\sigma}{\pi\,d},\\
\gamma&= \frac{m-2}{d},
\ea 
where 
\be\label{Eq: bar m def 1 0}
m=
\begin{cases}
	2\lfloor \bar m \rfloor &\mbox{ if } d=2,4,6,\ldots,\\ \vspace{-0.3cm}\\
	2\lfloor \bar m \rfloor+1 &\mbox{ if } d=3,5,7,\ldots,
\end{cases}
\quad\quad 
\bar m =
\begin{cases}
	\frac{d^{\eta/2}}{2}\sigma +1  &\mbox{ if } d=2,4,6,\ldots,\\ \vspace{-0.3cm}\\
	\frac{d^{\eta/2}}{2}\sigma +\frac{1}{2}   &\mbox{ if } d=3,5,7,\ldots.
\end{cases}
\ee
With the above definitions, one can write the function $\bar V_0(x)$ solely in terms of parameters $d$, $\eta$, $\sigma$, and $\alpha_0$.

\subsubsection{The \wso~Clock states}\label{The wso clock states}

Recall that for the quasi-ideal clock, the initial state is pure $\rho_\cl=\ketbra{\Psi_\textup{nor}(k_0)}{\Psi_\textup{nor}(k_0)},$ where 
\begin{align}\label{gaussianclock}
\ket{\Psi_\textup{nor}(k_0)} &= \sum_{\mathclap{\substack{k\in \mathcal{S}_d(k_0)}}}\psi_\textup{nor}(k_0;k) \ket{\theta_k},\\
\psi_\textup{nor}(k_0;x) &= A e^{-\frac{\pi}{\sigma^2}(x-k_0)^2} e^{i2\pi n_0(x-k_0)/d}, \quad x\in\rr.\label{eq:psi qudi ideal def}
\end{align}
with $\sigma \in (0,d)$, $n_0 \in (0,d-1)$, $k_0\in\rr$, $A \in \rr^+$, and $\mathcal{S}_d(k_0)$ is the set of $d$ integers closest to $k_0$, defined as
\begin{align}\label{eq: mathcal S def}
\mathcal{S}_d(k_0) = \left\{ k \; : \; k\in \mathbb{Z} \text{   and  }  -\frac{d}{2} \leq  k_0-k < \frac{d}{2} \right\}.
\end{align}
$A$ is defined so that the state is normalised, namely
\be\label{eq:A normalised}
A=A(\sigma;k_0)=\frac{1}{\sqrt{\sum_{k\in \mathcal{S}_d(k_0)} \me^{-\frac{2\pi}{\sigma^2}(k-k_0)^2}}},
\ee
which is of order $A=\bo\left( \left(\frac{2}{\sigma^2}\right)^{1/4} \right)$ as $d\rightarrow \infty$ for all $k_0\in\rr$.
$\omega n_0$ is approximately the mean energy of the initial clock state. We use $n_0$ to define  $\alpha_0\in(0,1]$. It is defined in Def. 1 in \cite{WSO16} as 
\begin{align}\label{eq:alpha_0 def}
\alpha_0&=\left(\frac{2}{d-1}\right) \min\{n_0,(d-1)-n_0\}\\
&=1-\left|1-n_0\,\left(\frac{2}{d-1}\right)\right|\in(0,1].
\end{align}
Physically, it is a measure of the distance of the mean energy from the edge of the energy spectrum, and has its maximum value $\alpha =1$ when the the mean energy is in the midpoint of the spectrum, namely when $n_0=(d-1)/2$. We pick $\alpha_0$ to be a fixed constant thought this \doc. This means that $n_0$ takes on the value $n_0=\tilde n_0(d-1)$ with $\tilde n_0\in(0,1)$ fixed constant.

In the proof of Theorem \ref{thm:quantum lower bound}, we will need the core theorem in \cite{WSO16} (Theorem IX.1, page 35). For brevity, we will state it in a slightly reduced which is nevertheless adequate for our purposes. Intuitively, the theorem states that the evolution of \wso~Clock states under the generator $\hat H$ (Eq. \eqref{eq:qusi ideal clock gen}) mimics the evolution of the Idealised clock to a good approximation. In order to state it, we need to recall some definitions from \cite{WSO16}. Namely Definitions 8 and 9 on page 21 \cite{WSO16}:

Let $b$ be any real number satisfying
\be\label{eq:b def eq 0}
b\geq\; \sup_{k\in\nn^+}\left(2\max_{x\in[0,2\pi]} \left|  V_0^{(k-1)}(x) \right|\,\right)^{1/k},
\ee
where $ V_0^{(p)}(x)$ is the $p^\textup{th}$ derivative with respect to $x$ of $ V_0(x)$ and $V_0^{(0)}:=V_0$.
We can use $b$ to define $\mathcal{N}\in\nn^0$ as follows
\be \label{eq:mathcal N def}
\mathcal{N}=
	\left\lfloor \frac{\pi\alpha_0^2}{2\left(\bar \upsilon+\frac{d}{\sigma^2}\right)^2} \left(\frac{d}{\sigma}\right)^2 \right\rfloor,
\ee
where $\kappa=0.792$ and

\be \label{eq:upsilon def}
\bar\upsilon= \frac{\pi\alpha_0\kappa}{\ln\left(\pi\alpha_0\sigma^2  \right)}b.
\ee

The theorem states that for all $t\geq 0$, $k_0\in\rr$, 
\ba
\me^{-\mi t \left(\hat V_d+\hat H_\cl\right)} \ket{\Psi_\textup{nor}(k_0)} &= \ket{\bar\Psi_\textup{nor}(k_0+td/T_0, t d/T_0)}+ \ket{\varepsilon_\nu} \\
&= \sum_{k\in\mathcal{S}_d (k_0+td/T_0)} \me^{-\mi \int_{k-t d/T_0}^k dy  V_d(y)} \psi_\textup{nor}(k_0+td/T_0;\,k)\ket{\theta_k} + \ket{\varepsilon_\nu}, \label{eq:non free evollution of Quasi ideal}
\ea 
where $\ket{\varepsilon_\nu}=\ket{\varepsilon_\nu}(t,d)$. The important question, is how small can the error term be, namely how does
\be\label{eq:def varepsiln nu}
\varepsilon_\nu(t,d):=\| \ket{\varepsilon_\nu}(t,d)\|_2
\ee
scale with the properties of the \wso~Clock and generator. The theorem puts bounds on this scaling. Namely if 
\ba 
\begin{split}\label{eq:mathcal N contraints}
	\bar \upsilon &\geq 0, \quad \quad	\mathcal{N} \geq 8,
\end{split}
\ea
are both satisfied,  the theorem tells us that in the limits $d\rightarrow \infty$, $(0,d)\ni \sigma\rightarrow \infty$,
\be \label{eq:epsilon t d}
 \varepsilon_\nu(t,d)=\,          |t| \frac{d}{T_0}\!\left(\bo\left( \frac{\sigma^3}{\bar\upsilon \sigma^2/d+1}\right)^{1/2}\!\!+\bo\left(\frac{d^2}{\sigma^2}\right)\right) \exp\left(-\frac{\pi}{4}\frac{\alpha_0^2}{\left(\frac{d}{\sigma^2}+\bar\upsilon\right)^2} \left(\frac{d}{\sigma}\right)^2 \right)+\bo\left(|t|\frac{d^2}{\sigma^2}+1\right)\me^{-\frac{\pi}{4}\frac{d^2}{\sigma^2}}+\bo\left( \me^{-\frac{\pi}{2}\sigma^2} \right).
\ee 

\subsubsection{An expression for ${R_1}$ for the \wso~Clock}\label{sec:expressing R in terms of WSO}
In this section, we will show how to express $\tr[\rho(t)]$ from Section \ref{sec:setup} in terms of the \wso~Clock from \cite{WSO16} discussed in Sections \ref{The generator of dynamics and potential function bar V0}, \ref{The wso clock states}.
Recalling the initial state of the \wso~Clock, namely $\rho_0=\ketbra{\psi_0}{\psi_0}=\ketbra{\psi_\textup{nor}(k_0)}{\psi_\textup{nor}(k_0)}$. From Eq. \eqref{eq:non free evollution of Quasi ideal} if follows
\ba\label{eq:tr rho t from WSO paper}
\tr[\rho(t)]=&\sum_{k\in\mathcal{S}_d (k_0+td/T_0)} \big{|} \bra{\theta_k}\me^{-\mi t \hat H}\ket{\psi_\textup{nor}(k_0)} \big{|}^2\\
=& \sum_{k\in\mathcal{S}_d (k_0+td/T_0)} \left|  \me^{-\mi \int_{k-t d/T_0}^k dy V_d(y)} \psi_\textup{nor}(k_0;k-t d/T_0) +\braket{\theta_k|\varepsilon_\nu}(t,d) \right|^2,
\ea
where we have used that $\psi_\textup{nor}(k_0+t d/T_0;k)=\psi_\textup{nor}(k_0;k-t d/T_0)$ which follows directly from Eq. \eqref{eq:psi qudi ideal def}. Due to the normalisation of $\ket{\psi_\textup{nor}(k_0)}$,  we have
\be \label{eq:normalisation}
\sum_{k\in\mathcal{S}_d (k_0+td/T_0)}\left| \psi_\textup{nor}(k_0;k-t d/T_0)\right|^2=1.
\ee

\subsection{Approximating ${R_1}$}\label{Approximating R1}
We will now proceed to lower bound ${R_1}$ in terms of dimension and mean clock energy for the \wso~Clock. Specifically in this section, up to additive error terms, the aim will be to express $R_1$ in terms of the 1st and 2nd moments of $\tr[\rho(t)]$ with integration range limited to one period of the free clock Hamiltonian, namely to the interval $[0,T_0]$.

\subsubsection{A more useful expression for ${R_1}$} 
\begin{lemma}\label{lem:1st up bound on R}
	Let $\Delta(0)$ and $\Delta(1)$ satisfy
\ba
\Delta^2(0)&\leq \left(\int_0^{T_0} dt \,\tr [\rho(t)]\right)^2,\quad\quad
\Delta(1)\geq\int_0^{T_0}dt\,t\,\tr [\rho(t)]\geq 0.
\ea 
then 
\ba\label{eq:R lo bound 2} 
{R_1}&\geq \frac{\Delta^2(0)+\epsilon_2}{2\Delta(1)-\Delta^2(0)+\epsilon_1},
\ea
where the epsilon terms $\epsilon_2$, $\epsilon_1$ are defined by
\be \label{eq:epsilon 2}
\epsilon_2:= 2\mu_0\mu_\epsilon+\mu_\epsilon^2, 
\ee
with
\ba \label{eqdef:mu0 and mu epsilon}
\mu_0&:=\int_0^{T_0} dt\, \tr[\rho(t)],\\ 
\mu_\epsilon&:=-\lim_{t\rightarrow \infty} t\,\tr[\rho(t)] 
+\int_{T_0}^\infty dt\, \tr[\rho(t)],
\ea
and 
\be \label{eq:epsilon 1}
\epsilon_1:=-2\mu_0\mu_\epsilon-\mu_\epsilon^2 +2\int_{T_0}^\infty dt \, t\, \tr[\rho(t)] - \lim_{t\rightarrow \infty} \left( 2\mu t+(t-\mu)^2 \right)\tr[\rho(t)].
\ee
\end{lemma}
\begin{proof}
By integration by parts,
\ba
\mu&= \int_0^{\infty} t P(t) dt=-\left(t\, \tr[\rho(t)]\right)\Big{|}_{0}^{\infty} + \int_0^{\infty} dt\, \tr[\rho(t)] \\
&= \mu_0+\mu_\epsilon.\label{mu simple}
\ea
Similarly,
\ba
\underline{\sigma}^2 &= \int_0^\infty (t-\mu)^2 P(t)dt = -\left( (t-\mu)^2\, \tr[\rho(t)]\right)\Big{|}_{0}^{\infty} + 2\int_0^{\infty} dt\, (t-\mu)\tr[\rho(t)] \\
&= -\lim_{t\rightarrow \infty}(t-\mu)^2 \tr [\rho(t)]+\mu^2+2\int_0^{\infty} dt\,t\,\tr[\rho(t)] -2\mu\big(\mu+\lim_{t\rightarrow \infty} t\,\tr [\rho(t)]\big)\\
&=-\mu^2+ 2\int_0^{\infty}dt\,t\,\tr [\rho(t)]-\lim_{t\rightarrow \infty} \big(2\mu t+(t-\mu)^2\big)\tr[\rho(t)]\\
&=-\mu_0^2+2\int_0^{T_0} dt \,t\,\tr[\rho(t)] +\epsilon_1\leq -\mu_0^2+2\Delta(1) +\epsilon_1\\
& \leq -\Delta^2(0)+2\Delta(1) +\epsilon_1 \label{sigmaSqrt}
\ea
From Eqs. \eqref{sigmaSqrt}, \eqref{mu simple}, we find
\ba
{R_1}&= \frac{\mu^2}{\underline\sigma^2}\geq \frac{\mu_0^2+2\mu_0\mu_\epsilon+\mu_\epsilon^2}{- \Delta^2(0)+2\Delta(1) +\epsilon_1}\\
&\geq \frac{\Delta^2(0)+2\mu_0\mu_\epsilon+\mu_\epsilon^2}{ -\Delta^2(0)+2\Delta(1) +\epsilon_1}.\label{eq:R lo bound inter}
\ea 
\end{proof}
We will now bound $\epsilon_1$ and $\epsilon_2$ appearing in Lemma \ref{lem:1st up bound on R}. These terms are negligible in the large $d$ limit and originate from the tails of the integrals in the definition of $\mu$ and $\underline{\sigma}$.
\begin{lemma}\label{lem: ep1 and ep2 bounds} $\epsilon_1$ and $\epsilon_2$, defined in Eqs. \eqref{eq:epsilon 1} and \eqref{eq:epsilon 2} respectively, are bounded by
	\ba
	|\epsilon_1|&\leq \left( a^2\left(\me^{-\delta}+\varepsilon_\nu(T_0,d)\right)^2  +2(2a T_0+a^2)\right)\left(\me^{-\delta}+\varepsilon_\nu(T_0,d)\right)^2,\\
	0\leq\epsilon_2&\leq \left(2T_0 a+a^2 \left(\me^{-\delta}+\varepsilon_\nu(T_0,d)\right)^2\right) \left(\me^{-\delta}+\varepsilon_\nu(T_0,d)\right)^2 ,
	\ea 
	where $a$ is any parameter satisfying
	\be \label{eqdef: a}
	a\geq \frac{T_0}{4\pi\,\delta} \left(\min_{x\in[0,2\pi]} \bar{V}_0(x) \right)^{-1}.
	\ee 
	and $\varepsilon_\nu(T_0,d)$ is defined by Eq. \eqref{eq:def varepsiln nu}. Note that since we have chosen $\bar{V}_0$ to have full support, $a$ is finite.
\end{lemma}
\begin{proof}
	We will start by bounding $\tr[\rho(t)]$. We find
	\ba 
	\tr[\rho(t)]&=  \tr[\me^{-\mi t\hat H}\rho_0\me^{\mi t\hat H^\dag}] \leq \| \me^{-\mi t\hat H}\me^{\mi t\hat H^\dag}\|_\infty\leq \| \me^{-\mi t\hat H}\|_\infty\|\me^{\mi t\hat H^\dag}\|_\infty,
	\ea 
	where $\|\cdot\|_\infty$ is the Schatten $p=\infty$ norm also called the operator norm. Recalling $\hat H=\hat H_C-\mi \hat V_C$, and applying the Golden-Thompson inequality, which holds for the operator norm \cite{Bhatia} it we find
	\ba 
	\tr[\rho(t)]&\leq \| \me^{-\mi t\hat H_C}\|_\infty\|\me^{\mi t\hat H_C}\|_\infty \|\me^{-t \hat V_C}\|_\infty^2 = \|\me^{-t \hat V_C}\|_\infty^2 =\exp\left[-\delta\, t \frac{4\pi}{T_0}\left( \min_{k\in \{0,1,2,\ldots,d-1\}} \bar{V}_0\left( \frac{2\pi}{d}k \right)\right)\right],
	\ea
	where we have used $ \hat V_C=\mi \hat V_d:=\delta \frac{d}{T_0}\sum_{k=0}^{d-1}\bar V_d(k)\ketbra{\theta_k}{\theta_k}$, where $\frac{d}{T_0} \bar V_d(k)= \frac{2\pi}{T_0} \bar{V}_0\left( \frac{2\pi}{d}k \right)$. Now, in this manuscript, $\bar{V}_0$ is positive (Eq. \eqref{eq:pod eqs def}). As such, for all $d\in\nnp$ 
	\be 
	\min_{k\in \{0,1,2,\ldots,d-1\}} \bar{V}_0\left( \frac{2\pi}{d}k \right)>0
	\ee 
	and thus
	\be \label{eq:lim powers rho}
	\lim_{t\rightarrow \infty} t^n \tr[\rho(t)]=0
	\ee 
	for all $n\geq 0$ and $d\in\nnp$.
	Our next task will be to bound the two integrals $\int_{T_0}^\infty dt\, \tr[\rho(t)]$ and $\int_{T_0}^\infty dt\,t\, \tr[\rho(t)]$. We will start by finding a $t$ independent parameter $a>0$ such that $\tr[\rho(t)]\leq a P(t)$ for all $t\geq 0$, where $P(t)=-\frac{d}{dt}\tr[\rho(t) ]$ was defined in Eq. \eqref{eq:def P(t)}.
	\ba 
	\frac{\tr[\rho(t)]}{P(t)}&=	
	\frac{\tr[\rho(t)]}{\mi\tr[\hat H\rho(t)]-\mi\tr[\hat H^\dagger\rho(t)]}=\frac{\tr[\rho(t)]}{\mi\tr[(\hat H_C-\mi\hat V_C)\rho(t)]-\mi\tr[(\hat H_C+\mi\hat V_C)\rho(t)]}=\frac{1}{2\tr [\hat V_C \bar{\rho}(t)]}\\
	&\leq \frac{1}{2}\left( \inf_{\ket{\psi}\in\mathcal{S}_p} \tr [\hat V_C \ketbra{\psi}{\psi}] \right)^{-1},\quad \forall\, t\geq 0
	\ea 
	where $\bar{\rho}(t):=\rho(t)/\tr[\rho(t)]$ is a normalised rank one density matrix for all $t\geq 0$ and $\mathcal{S}_p$ is the set of normalised pure quantum states (rank one density matrices). Crucially, note that $P(t)>0$ since $\tr [\hat V_C \bar{\rho}(t)]>0$ because $\hat V_C$ is positive-definite since $\bar{V}_0$ has full support. Hence taking the trace in the orthonormal basis $\{\ket{\theta_k}\}_{k=0}^{d-1}$,
	\ba 
	\frac{\tr[\rho(t)]}{P(t)}&\leq \frac{1}{2}\left( \min_{k\in\{0,1,2,\ldots,d-1\}} \bra{\theta_k}\hat V_C \ket{\theta_k} \right)^{-1}= \frac{T_0}{4\pi\,\delta} \left(\min_{k\in\{0,1,2,\ldots,d-1\}} \bar{V}_0\left(2\pi k/d\right) \right)^{-1}\\
	&\leq \frac{T_0}{4\pi\,\delta} \left(\min_{x\in[0,2\pi]} \bar{V}_0(x) \right)^{-1} \leq a ,\quad \forall\, t\geq 0. \label{eq:bound  tr/P}
	\ea 
	Using Eq. \eqref{eq:bound  tr/P} can now bound the first integral: 
	\be \label{eq:int rho}
	\int_{T_0}^\infty dt\, \tr [\rho(t)]\leq a \int_{T_0}^\infty dt P(t)=-a \Big[ \tr [\rho(t)]\Big]_{T_0}^\infty=a\left( \tr [\rho(T_0)]-\lim_{t\rightarrow \infty}\tr [\rho(t)] \right)=a\,\tr [\rho(T_0)].
	\ee 
	For the second integral, we will additionally have to integrate by parts and recall $T_0>0$\,:
	\ba \label{eq:int 1 rho}
	\int_{T_0}^\infty dt\,t\, \tr [\rho(t)]&\leq -a\int_{T_0}^\infty dt\,t \frac{d}{dt}\tr [\rho(t)]=-a \left( \Big[t\,\tr[\rho(t)]\Big]_{T_0}^\infty-\int_{T_0}^\infty dt\, \tr [\rho(t)] \right)=a T_0 \,\tr[\rho(T_0)]+ a\int_{T_0}^\infty dt\, \tr [\rho(t)]\\
	&\leq(a T_0+a^2)\tr[\rho(T_0)],
	\ea 
	where to achieve the last line, we have used Eq. \eqref{eq:int rho}. In order for Eqs. \eqref{eq:int rho} and Eq. \eqref{eq:int 1 rho} to be useful, we need to bound $\tr [\rho(T_0)]$. This can be achieved using Eq. \eqref{eq:tr rho t from WSO paper}. One finds	
	\be \label{def:ep1}
	\tr[\rho(T_0)]= |\me^{-\delta} + \braket{\theta_k|\varepsilon_\nu}(T_0,d) |^2\leq \left(\me ^{-\delta} +\varepsilon_\nu(T_0,d)\right)^2.
	\ee
	Finally, recalling that $-\frac{d}{dt}\tr[\rho(t) ]>0$, we have
	\be \label{eq:int mu}
	\mu_0=\int_0^{T_0} dt \tr [\rho (t)]\leq T_0 \max_{t\in[0,T_0]}  \tr [\rho (t)]=T_0 \tr [\rho(0)]=T_0.
	\ee 
	Thus using Eqs. \eqref{eq:lim powers rho}, \eqref{eq:int rho}, \eqref{eq:int 1 rho}, \eqref{def:ep1} and \eqref{eq:int mu}, it follows that $\epsilon_1$ defined in Eq. \eqref{eq:epsilon 1} 
	is bounded by
	\be
	|\epsilon_1|\leq 2T_0\mu_\epsilon+\mu_\epsilon^2 +2(aT_0+a^2) \left(\me ^{-\delta} +\varepsilon_\nu(T_0,d)\right)^2.
	\ee
	Similarly, we find that $\mu_\epsilon$ defined in Eq. \eqref{eqdef:mu0 and mu epsilon} is bounded by
	\be 
	\mu_\epsilon\leq a \left(\me ^{-\delta} +\varepsilon_\nu(T_0,d)\right)^2.
	\ee 
	Thus using the definition of $\epsilon_2$ in Eq. \eqref{eq:epsilon 2}, we complete the proof.
\end{proof}

\subsection{Upper and lower bounds for $\tr[\rho(t)]$}\label{Upper and lower bounds for trace rho}
\begin{lemma}\label{lem:1st up lowe on trace rho t}
	\ba \label{eq:trace rho bound}
	\left(\sum_{k\in\mathcal{S}_d(k_0+td/T_0)} \Delta_k\right) -\epsilon_0\leq \tr[\rho(t)] \leq  \left(\sum_{k\in\mathcal{S}_d (k_0+td/T_0)} \Delta_k\right) +\epsilon_0,
	\ea
	where
	\ba
	\Delta_k:=   \me^{-2\delta \int_{k-t d/T_0}^k dy {\bar V_d}(y)}\left| \psi_\textup{nor}(k_0;k-t d/T_0)   \right|^2, \quad \epsilon_0=\epsilon_0(t,d):= \varepsilon_\nu(t,d)\big(\varepsilon_\nu(t,d) +2\big)d,
	\ea
	and $\varepsilon_\nu(t,d)$ is defined by Eq. \eqref{eq:def varepsiln nu}. 

\end{lemma}
\begin{proof}
	Here we will use the core theorem in \cite{WSO16} discussed in Section \ref{sec:expressing R in terms of WSO}.
	Using Eq. \eqref{eq:tr rho t from WSO paper}, Eq. \eqref{eq:trace rho bound} follows from noting $|a+\epsilon|^2= |a|^2+|\epsilon|^2+2\,\mathfrak{Re}(a\,\epsilon)$, $\mathfrak{Re}(a)\leq |a|$ for $a,\epsilon\in\cc$, and $|\braket{\theta_k|\varepsilon_\nu}(t,d)|\leq \varepsilon_\nu(t,d)$ for all $k$.
\end{proof}
We will now calculate time dependent upper and lower bounds for $\sum_{k\in\mathcal{S}_d (k_0+td/T_0)} \Delta_k$, where $\Delta_k$ is defined in Lemma \ref{lem:1st up lowe on trace rho t}.
\begin{lemma}\label{lem: t bouns for DeltaC}
	Let $k_0=0$. There exists $\Delta_L\geq 0$, $\Delta_R\geq 0$, $\Delta_C\geq 0$ such that 
	\ba 
	\sum_{k\in\mathcal{S}_d (k_0+td/T_0)} \Delta_k= \Delta_L+\Delta_C+\Delta_R,
	\ea
	where
$\Delta_L$, $\Delta_R$ satisfy the bounds
	\ba 
	\Delta_L,\Delta_R &\leq A^2 \frac{\me^{-\frac{2\pi}{\sigma^2}(\gamma d/2-\bar k(t))^2}}{1-\me^{-4\pi|\gamma d/2- \bar k(t)|/\sigma^2}},
	\ea 
	where $	\bar{k}(t):=\lfloor -d/2 +td/T_0+1\rfloor +d/2 -td/T_0\in[0,1]$ and
	
		\ba\label{eq:gamma in term of m} 
	\gamma &=\gamma(m):= \frac{m-2}{d}\in(0,1), \quad \quad m=
	\begin{cases}
		4,6,8\ldots,d+2 \mbox{ if } d=2,4,6,\ldots\\
		3,5,7,\ldots,d+2 \mbox{ if } d=3,5,7,\ldots
	\end{cases}
	\ea
 \\
	Furthermore, there exists $\Delta_C^{(-)}\geq 0,$ $\Delta_C^{(+)}\geq 0,$ such that $\Delta_C$ satisfies the bounds
	\ba \label{eq:Delta C def simplified}
\Delta_C^{(-)}-\epsilon_\textup{C}\leq	\Delta_C \leq \Delta_C^{(+)}
	\ea
	where
	\be\label{eq: ep C def} 
	 \epsilon_\textup{C}=\epsilon_\textup{C}(t):=
	2  A^2 \frac{\me^{-\frac{2\pi}{\sigma^2}(\gamma d/2-\bar k(t))^2}}{1-\me^{-4\pi|\gamma d/2- \bar k(t)|/\sigma^2}},
	\ee 
	and if $x_0$ (defined in Section \ref{The generator of dynamics and potential function bar V0}) satisfies $x_{vr} +\pi\gamma \leq x_0 \leq 2\pi-x_{vr}-\pi\gamma$, then
	\ba
	\begin{cases} 
		\me^{-2\delta \tilde \epsilon_V} \leq \Delta_C ^{(-)}(t)\leq \Delta_C ^{(+)}(t) \leq 1 \quad&\textup{if}\quad 0 \leq t\frac{2\pi}{T_0} \leq x_0-x_{vr}-\pi\gamma,\\
	\quad\quad\;	0  \leq \Delta_C ^{(-)}(t)\leq \Delta_C ^{(+)}(t)  \leq 	\me^{-2\delta(1-\tilde\epsilon_V)}  \quad&\textup{if}\quad x_{vr}+\pi\gamma+x_0 \leq t\frac{2\pi}{T_0}.
	\end{cases}
	\ea	
	When the value of $t$ is such that the above bounds do not hold, we can also use the bounds
	\ba 
	0 \leq \Delta_C (t) &\leq 1, \quad \forall\, t\geq 0,\,x_0\in\rr.
	\ea 
\end{lemma}
\begin{proof}
	\ba 
	\sum_{k\in\mathcal{S}_d (k_0+td/T_0)} \Delta_k &=  	\sum_{k\in\mathcal{S}_d (k_0+td/T_0)} \me^{-2\delta \int_{k-t d/T_0}^k dy {\bar V_d}(y)}\left| \psi_\textup{nor}(k_0;k-t d/T_0)   \right|^2 \\
	&=  	\sum_{k=\min \{\mathcal{S}_d(k_0+td/T_0)\}-td/T_0}^{\max \{\mathcal{S}_d(k_0+td/T_0)\}-td/T_0}  \me^{-2\delta \int_{k}^{k+td/T_0} dy {\bar V_d}(y)}\left| \psi_\textup{nor}(k_0;k)   \right|^2\\
	&=  	\sum_{k=\lfloor -d/2+k_0+1+td/T_0\rfloor-td/T_0}^{\lfloor d/2+k_0+td/T_0\rfloor-td/T_0}  \me^{-2\delta \int_{k}^{k+td/T_0} dy {\bar V_d}(y)}\left| \psi_\textup{nor}(k_0;k)   \right|^2\\
		&=  	\sum_{k= k_0 -d/2+\bar k(t)}^{ k_0 +d/2+\bar k(t)-1}  \me^{-2\delta \int_{k}^{k+td/T_0} dy {\bar V_d}(y)}\left| \psi_\textup{nor}(k_0;k)   \right|^2,
	\ea 
	where recall $\bar{k}(t)=\lfloor -d/2 +k_0+td/T_0+1\rfloor +d/2 -k_0-td/T_0\in[0,1]$. This follows from noting that $\bar{k}(t)$ is a solution to both equations $\lfloor -d/2+k_0+1+td/T_0\rfloor-td/T_0=k_0 -d/2+\bar k(t)$ and $\lfloor d/2+k_0+td/T_0\rfloor-td/T_0= k_0 +d/2+\bar k(t)-1$, since we can use the identity $1=\lfloor-x+y+1\rfloor - \lfloor x+y\rfloor +2x$, for $2x\in\zz$, $y\in \rr$ and set $x=d/2$, $y=k_0+td/T_0$. For simplicity, we will now take into account that $k_0=0$. We will now break the sum up into three contributions where the first two will correspond to the ``Gaussian tails" to the ``left" ($\Delta_L$) and ``right" ($\Delta_R$) of $k=k_0=0$, and a ``central term" ($\Delta_C$) corresponding to the region $k\approx k_0=0$. Namely,
	\ba 
		\sum_{k\in\mathcal{S}_d (k_0+td/T_0)} \Delta_k= \Delta_L+\Delta_C+\Delta_R.
	\ea 
	We introduce $\gamma\in(0,1)$ and start with bounding $\Delta_L$:
	\ba\label{eq:bound for Delta L} 
	\Delta_L&:= \sum_{k= -d/2+\bar k(t)}^{\bar k(t)-\gamma d/2}  \me^{-2\delta \int_{k}^{k+td/T_0} dy {\bar V_d}(y)}\left| \psi_\textup{nor}(k_0;k)   \right|^2\leq \sum_{k= -d/2+\bar k(t)}^{\bar k(t)-\gamma d/2}  \left| \psi_\textup{nor}(k_0;k)   \right|^2= \sum_{y=\gamma d/2-\bar k(t)}^{d/2-\bar k(t)}A^2 \me^{-\frac{2\pi}{\sigma^2}(-y)^2}\\
	&\leq \sum_{y=\gamma d/2-\bar k(t)}^{\infty}A^2 \me^{-\frac{2\pi}{\sigma^2}(-y)^2}\leq A^2 \frac{\me^{-\frac{2\pi}{\sigma^2}(\bar k(t)-\gamma d/2)^2}}{1-\me^{-4\pi| \bar k(t)-\gamma d/2|/\sigma^2}}.\label{eq:last line}
	\ea  
	where the summations are defined as,
	\be 
	\sum_{y=a}^{b} f(y)=f(a)+f(a+1)+\ldots +f(b),
	\ee
	where $a,b\in\rr$, $b-a\in\zz$ and the sum is defined to be zero if $b-a\leq 0$. We will use this convention throughout. We will constrain $\gamma$ appropriately  later in Eq. \eqref{eq:gamma in term of m in proof}. In the last line of Eq. \eqref{eq:last line} we have used Lemma J.0.1, page 59 from \cite{WSO16}. Similarly, for $\Delta_R$,
	\ba\label{eq:bound for Delta R} 
		\Delta_R&:= \sum_{k=\bar k(t) +\gamma d/2}^{\bar k(t)+ d/2}  \me^{-2\delta \int_{k}^{k+td/T_0} dy {\bar V_d}(y)}\left| \psi_\textup{nor}(k_0;k)   \right|^2\leq \sum_{k=\bar k(t) +\gamma d/2}^{\bar k(t)+ d/2}  \left| \psi_\textup{nor}(k_0;k)   \right|^2\leq \sum_{k=\bar k(t) +\gamma d/2}^{\infty}  \left| \psi_\textup{nor}(k_0;k)   \right|^2 \\
		&\leq \sum_{k=\bar k(t) -\gamma d/2}^{\infty}  \left| \psi_\textup{nor}(k_0;k)   \right|^2 \leq A^2 \frac{\me^{-\frac{2\pi}{\sigma^2}(\bar k(t)-\gamma d/2)^2}}{1-\me^{-4\pi| \bar k(t)-\gamma d/2|/\sigma^2}}= A^2 \frac{\me^{-\frac{2\pi}{\sigma^2}(\gamma d/2-\bar k(t))^2}}{1-\me^{-4\pi|\gamma d/2- \bar k(t)|/\sigma^2}}.
	\ea
	For $\Delta_C$, we have
	\ba \label{eq:Delta C def}
	\Delta_C&:= \sum_{k=\bar k(t) -\gamma d/2+1}^{\bar k(t)+ \gamma d/2-1}  \me^{-2\delta \int_{k}^{k+td/T_0} dy {\bar V_d}(y)}\left| \psi_\textup{nor}(k_0;k)   \right|^2\\
	&
	\begin{cases}
	 \leq \bigg(\max_{y\in \mathcal{I}_\gamma} \Big{\{} \me^{-2\delta \int_{y}^{y+td/T_0} dx {\bar V_d}(x)} \Big{\}}\bigg)\sum_{k=\bar k(t) -\gamma d/2+1}^{\bar k(t)+ \gamma d/2-1} \left| \psi_\textup{nor}(k_0;k)   \right|^2,\\ \vspace{-0.3cm}\\
	 \geq \bigg(\min_{y\in \mathcal{I}_\gamma} \Big{\{} \me^{-2\delta \int_{y}^{y+td/T_0} dx {\bar V_d}(x)} \Big{\}}\bigg)\sum_{k=\bar k(t) -\gamma d/2+1}^{\bar k(t)+ \gamma d/2-1} \left| \psi_\textup{nor}(k_0;k)   \right|^2,\label{eq: Delta C upper bottom bound split}
	\end{cases}
	\ea
	where $\mathcal{I}_\gamma:=\{ \bar k(t)-\gamma d/2+1, \bar k(t)-\gamma d/2+2,\ldots,\bar k(t)+\gamma d/2-1 \}.$ In order for the summation in the definitions of $\Delta_L, \Delta_C, \Delta_R$ to be well defined, we need the constraints, $-d\gamma
	/2+\bar k(t)-(-d/2+\bar k(t))=n_1$, $d\gamma
	/2+\bar k(t)-1-(-d\gamma/2+\bar k(t)+1)=n_2$, $d\gamma
	/2+\bar k(t)-(d/2+\bar k(t))=n_3$, for $n_1,n_2,n_3 \in \zz$. A solution for $\gamma$ is
	\be\label{eq:gamma in term of m in proof} 
	\gamma=\gamma(m)= \frac{m-2}{d}\in(0,1), \quad \text{where } m=
	\begin{cases}
		4,6,8,\ldots,d+2 \mbox{ if } d=2,4,6,\ldots\\
		3,5,7,\ldots,d+2 \mbox{ if } d=3,5,7,\ldots
	\end{cases}
	\ee
	Before proceeding to bound the maximization and minimization in Eq. \eqref{eq: Delta C upper bottom bound split}, we will bound the common factor term which is approximately one. We find
	\ba \label{eq:gam gam sum equality}
	\sum_{k=\bar k(t) -\gamma d/2+1}^{\bar k(t)+ \gamma d/2-1} \left| \psi_\textup{nor}(k_0;k)   \right|^2&=\sum_{k\in\mathcal{S}_d(k_0+t d/T_0)} \left| \psi_\textup{nor}(k_0;k-t d/T_0)   \right|^2-\sum_{k=\bar k(t) -d/2}^{\bar k(t)-\gamma d/2} \left| \psi_\textup{nor}(k_0;k)   \right|^2- \sum_{k=\bar k(t) +\gamma d/2}^{\bar k(t)+ d/2} \left| \psi_\textup{nor}(k_0;k)   \right|^2\\
	&=1-\sum_{k=\bar k(t) -d/2}^{\bar k(t)-\gamma d/2} \left| \psi_\textup{nor}(k_0;k)   \right|^2- \sum_{k=\bar k(t) +\gamma d/2}^{\bar k(t)+ d/2} \left| \psi_\textup{nor}(k_0;k)   \right|^2.
	\ea 
	where in the last line we have used Eq. \eqref{eq:normalisation}. 
The remaining two terms in Eq. \eqref{eq:gam gam sum equality} have been bounded in Eqs. \eqref{eq:bound for Delta L}, \eqref{eq:bound for Delta R}. Noting that $0\leq  \me^{-2\delta \int_{y}^{y+td/T_0} dx {\bar V_d}(x)} \leq 1$, we can thus simplify Eq. \eqref{eq: Delta C upper bottom bound split}, to find 
	\ba \label{eq:Delta C def simplified in  proof}
\Delta_C
\begin{cases}
	\leq \Delta_C^{(+)} \\
	 \vspace{-0.3cm}\\
	\geq \Delta_C^{(-)}-\epsilon_\textup{C},
\end{cases}
\ea
where 
$\epsilon_\textup{C}$ is defined in the statement of the Lemma and
\be 
\Delta_C^{(+)}:= \max_{y\in \mathcal{I}_\gamma} \Big{\{} \me^{-2\delta \int_{y}^{y+td/T_0} dx {\bar V_d}(x)} \Big{\}},\quad \Delta_C^{(-)}:= \min_{y\in \mathcal{I}_\gamma} \Big{\{} \me^{-2\delta \int_{y}^{y+td/T_0} dx {\bar V_d}(x)} \Big{\}}.
\ee \\

	Our next aim will be to find bounds on $t$ which determine whether $\Delta_C$ is approximately $1$ or $\me^{-2\delta}$. First note,
	\be \label{eq:re writing int to t cond}
	\int_{y}^{y+td/T_0} dx {\bar V_d}(x)=\frac{2\pi}{d} \int_{y}^{y+td/T_0} dx \bar{V}_0\left(\frac{2\pi}{d}x\right)=  \int_{2\pi y/d-x_0}^{2\pi y/d+t2\pi/T_0-x_0} dx \bar{V}_0\left(x+x_0\right).
	\ee 
	\begin{figure}
		\includegraphics[scale=0.25]{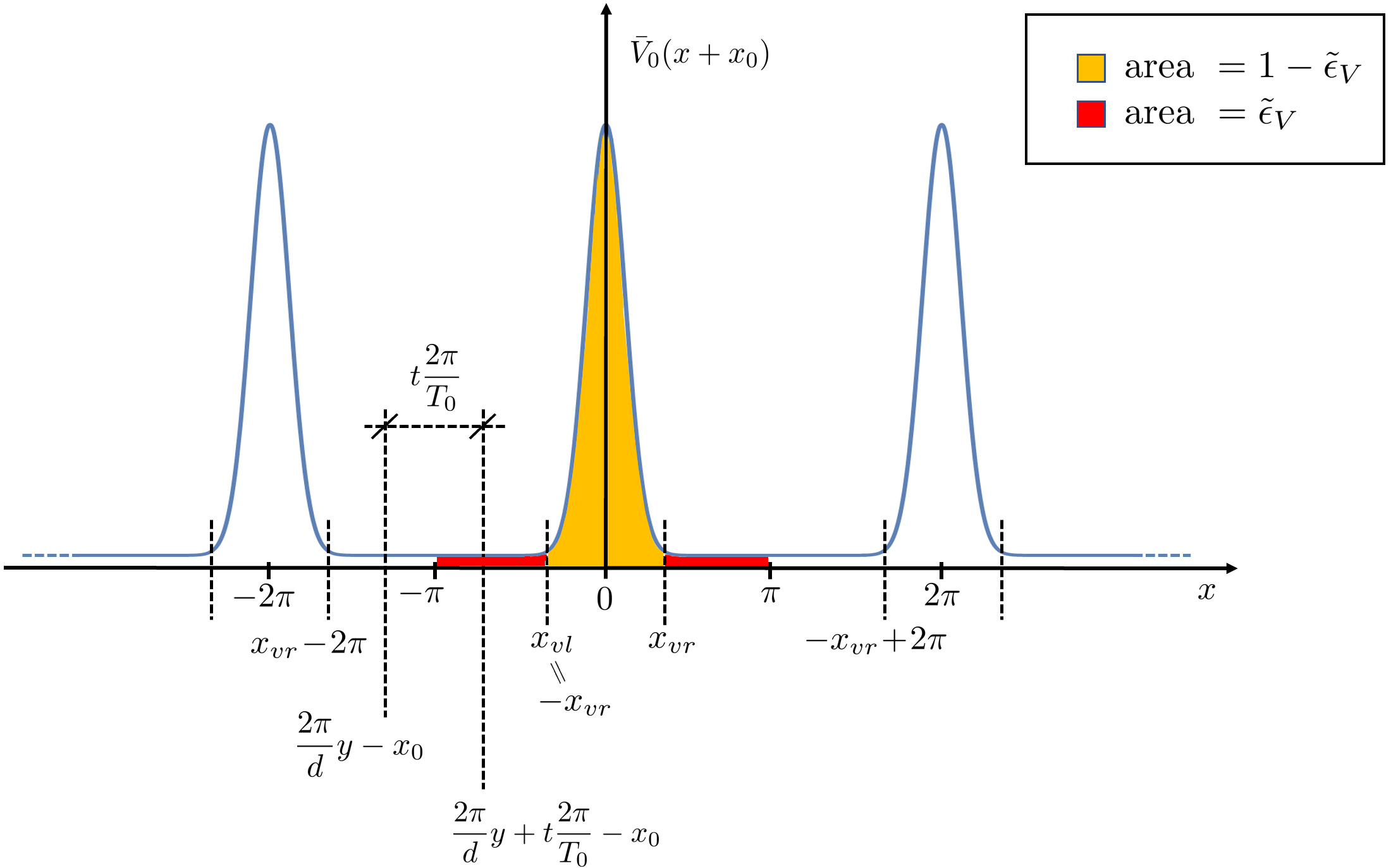}
		\centering
		\caption{Depiction of the periodic potential function $\bar{V}_0(x+x_0)$ with the parameters $x_{vl}$, $x_{vr}$ and $\tilde{\epsilon}_V$ introduced in Eq. \eqref{eq:tilde ep V def}. Observe that the area under the potential $\bar{V}_0(x+x_0)$ (i.e. its integral) between $x=x_{vr}-2\pi$ and $x=x_{vl}$ is $\tilde\epsilon_V$ due to symmetry. The quantities $\frac{2\pi}{d} y-x_0$ and $\frac{2\pi}{d} y+t\frac{2\pi}{T_0}-x_0$ are the lower and upper integration limits in Eq. \eqref{eq:re writing int to t cond}.\label{fig:pot pic proof} }
	\end{figure}
	We can now use Eq. \eqref{eq:re writing int to t cond} to find conditions for the time "before a tick is measured", namely conditions for which
	\be\label{eq:before mtick condition} 
	\int_{y}^{y+td/T_0} dx {\bar V_d}(x)\leq 1-(1-\tilde{\epsilon}_V)=\tilde{\epsilon}_V.
	\ee 
	From Fig. \ref{fig:pot pic proof} and Eq. \eqref{eq:re writing int to t cond} we observe that Eq. \eqref{eq:before mtick condition} is satisfied for all $y\in \mathcal{I}_\gamma$ if 
	\ba
	x_{vr}-2\pi&\leq \frac{2\pi}{d}y-x_0\\
	\frac{2\pi}{d}y+t\frac{2\pi}{T_0} -x_0&\leq x_{vl}=-x_{vr},
	\ea
	for all $y\in \mathcal{I}_\gamma$. Or equivalently, if
	\ba
	\frac{2\pi}{d} \min_{y\in\mathcal{I}_\gamma}\{y\} &\geq x_{vr}-2\pi+x_0\\
	\frac{2\pi}{d}\max_{y\in\mathcal{I}_\gamma}\{y\}&\leq x_0-\frac{2\pi}{T_0}t-x_{vr},
	\ea
	from which it follows,
	\ba 
	x_0 &\leq 2\pi -x_{vr}-\pi\gamma+\frac{2\pi}{d}(\bar k(t)+1)\label{eq:1LowBoundOnt pre}\\
	0\leq t\frac{2\pi}{T_0} &\leq x_0-x_{vr}-\pi\gamma-\frac{2\pi}{d}(\bar k(t)-1).\label{eq:2LowBoundOnt pre}
	\ea 
	Thus recalling that $\bar k(t)\in[0,1]$, sufficient conditions on $x_0$ and $t$ for Eq. \eqref{eq:before mtick condition} to be satisfied are
	\ba 
	x_0 &\leq 2\pi -x_{vr}-\pi\gamma\label{eq:1LowBoundOnt}\\
	0\leq t\frac{2\pi}{T_0} &\leq x_0-x_{vr}-\pi\gamma.\label{eq:2LowBoundOnt}
	\ea 
Similarly, we can work out conditions for the time ``after a tick has occurred", i.e. 
	\be \label{eq:after tick has happened condition}
	\int_{y}^{y+td/T_0} dx {\bar V_d}(x)\geq 1-\tilde{\epsilon}_V, 
	\ee
	for all $y\in \mathcal{I}_\gamma$ if 
	\ba 
	\frac{2\pi}{d}y-x_0 &\leq -x_{vr}\\
	x_{vr} &\leq  \frac{2\pi}{d}y+t\frac{2\pi}{T_0}-x_0 
	\ea 
	for all $y\in \mathcal{I}_\gamma$. Or equivalently, if
	\ba 
	x_{vr}+\max_{y\in\mathcal{I}_\gamma} \{y\}\frac{2\pi}{d} &\leq x_0 
	\\
		x_{vr}+\frac{2\pi}{d}\max_{y\in \mathcal{I}_\gamma}\{-y\}+x_0 &\leq t\frac{2\pi}{T_0}, 
	\ea
	from which it follows
	\ba 
	x_{vr} +\frac{2\pi}{d}(\bar k(t)-1)+\pi\gamma &\leq x_0 
	\label{eq:x0 contraint pre}\\
	x_{vr}-\frac{2\pi}{d}(\bar k(t)+1)+\pi\gamma+x_0 &\leq t\frac{2\pi}{T_0} 
	\label{eq:afterclick pre}
	\ea 
		Thus recalling that $\bar k(t)\in[0,1]$, sufficient conditions on $x_0$ and $t$ for Eq. \eqref{eq:after tick has happened condition} to be satisfied are
\ba 
x_{vr} +\pi\gamma &\leq x_0, 
\label{eq:x0 contraint}\\
x_{vr}+\pi\gamma+x_0 &\leq t\frac{2\pi}{T_0}. 
 \label{eq:afterclick}
\ea 
	Thus from Eqs. \eqref{eq:Delta C def}, \eqref{eq:1LowBoundOnt}, \eqref{eq:2LowBoundOnt},\eqref{eq:x0 contraint}, and recalling $\bar{V}_0\geq 0$, we conclude that if $x_0$ satisfies 
	\be \label{eq:up low bound on x0}
	x_{vr} +\pi\gamma \leq x_0 \leq 2\pi-x_{vr}-\pi\gamma,
	\ee
	then 
	\ba 
	\me^{-2\delta \tilde \epsilon_V} \leq \Delta_C ^{(-)}(t)\leq \Delta_C ^{(+)}(t) &\leq 1 \quad\textup{if}\quad 0 \leq t\frac{2\pi}{T_0} \leq x_0-x_{vr}-\pi\gamma,\\
	0  \leq \Delta_C ^{(-)}(t)\leq \Delta_C ^{(+)}(t)  &\leq 	\me^{-2\delta(1-\tilde\epsilon_V)}  \quad\textup{if}\quad x_{vr}+\pi\gamma+x_0 \leq t\frac{2\pi}{T_0}.
	\ea	


\end{proof}

\subsection{Calculating $\Delta(0)$ and $\Delta(1)$ and an explicit lower bound on ${R_1}$}\label{Calculating Delta(0) and Delta(1) and an explicit lower bound on R1}
Define 
\be\label{eq:def delt t}
 t_1:=(x_0-x_{vr}-\pi\gamma)\frac{T_0}{2\pi},\quad t_2:=(x_0+x_{vr}+\pi\gamma)\frac{T_0}{2\pi},\quad \Delta t:= t_2-t_1=(x_{vr}+\pi\gamma)\frac{T_0}{\pi}.
\ee
Furthermore, we will want $t_1$ to be in the centre of the range $[0,T_0]$. As such, we set $t_1=T_0/2$ which taking into account the definition of $t_1$ in Eq. \eqref{eq:def delt t} implies
\be \label{eq:x0 def}
x_0=\pi+x_{vr}+\pi\gamma.
\ee
Physically, Eq. \eqref{eq:x0 def} means that the potential is peaked near the mid point $\pi$ and thus that the continuous measurements will occur approximately at a time $T_0/2$.
Note that Eq. \eqref{eq:x0 def} is consistent with Eq. \eqref{eq:up low bound on x0} as long as $x_{vr}+\pi\gamma\leq \pi/2$. As we will see later, both $\gamma$ and $x_{vr}$ will be parametrized such that $\gamma$, $x_{vr}$ $\rightarrow 0$ as $d\rightarrow \infty$.\footnote{See Eqs. \eqref{eq:up bound gamma even},\eqref{eq:bound for gamma d over sigma odd},\eqref{eq:x vr limit contraint} for explicit parametrization.} As such, this constraint will always be satisfied for sufficiently large $d$.\footnote{It is expected that one could actually weaken the constraint imposed on $x_{vr}$ and $\gamma$ if one choose a slightly different parametrization of $t_1$, say $t_1=T_0/2-(x_{vr}+\pi\gamma)T_0/\pi$. But for simplicity, we will not do this here.}

\begin{lemma}\label{lem:R lowe pre leading order}
$R_1$ satisfies the bound 
	\be 
{R_1}	\geq\frac{-\left( \left(2T_0\epsilon_0(T_0)+\int_0^{T_0}dt |\epsilon_\textup{C}(t)|\right) \left(t_1 \me^{-2\delta\tilde{\epsilon}_V}+\Delta t\right)  +\Delta t^2\right)/T_0^2+\me^{-4\delta \tilde{\epsilon}_V}/4}{\epsilon_4+\epsilon_3+\frac{3}{4}\me^{-\delta (1-\tilde{\epsilon}_V)}+(1-\me^{-4\delta \tilde{\epsilon}_V})/4+\frac{2}{T_0^2} \left( t_2-t_1\me^{-2\delta\tilde \epsilon_V}\right) \Delta t +\left(\frac{\Delta t}{T_0}\right)^2},
	\ee 
where we have defined
\ba \label{eq:ep 3 4 def}
\epsilon_3&:= \frac{2}{T_0^2}\int_0^{T_0} dt\,t (\Delta_L(t)+\Delta_R(t)),\\
\epsilon_4&:= \epsilon_0(T_0)  + \frac{1}{T_0^2}    \left(2T_0\epsilon_0(T_0)+\int_0^{T_0}dt \,\epsilon_\textup{C}(t)\right) \left(t_1 \me^{-2\delta\tilde{\epsilon}_V}+\Delta t\right)  +\epsilon_1/T_0^2.
\ea 
\end{lemma}
\begin{proof}
Using Lemma \ref{lem:1st up lowe on trace rho t}, followed by Lemma \ref{lem: t bouns for DeltaC},	
\ba
\int_0^{T_0} dt \,\tr [\rho(t)]&\geq \int_0^{T_0} dt \left(\Delta_L(t)+\Delta_C(t)+\Delta_R(t)-\epsilon_0\right)\\
&\geq -T_0\epsilon_0(T_0)-\int_0^{T_0}dt \epsilon_\textup{C}(t) + \int_0^{t_1}dt\, \Delta_C^{(-)}(t)+\int_{t_1}^{t_2}dt\, \Delta_C(t)+\int_{t_2}^{T_0}dt\,\Delta_C(t)\\
&\geq -T_0\epsilon_0(T_0)-\int_0^{T_0}dt \epsilon_\textup{C}(t)+ t_1 \me^{-2\delta\tilde{\epsilon}_V}+\int_{t_1}^{t_2}dt\, \Delta_C(t).
\ea 
Thus,
\ba 
\left( \int_0^{T_0} dt \,\tr [\rho(t)]\right)^2 \geq& \left( -T_0\epsilon_0(T_0)-\int_0^{T_0}dt \epsilon_\textup{C}(t)+ t_1 \me^{-2\delta\tilde{\epsilon}_V}+\int_{t_1}^{t_2}dt\, \Delta_C(t)\right)^2\\
\geq& \left( -T_0\epsilon_0(T_0)-\int_0^{T_0}dt \epsilon_\textup{C}(t)\right)^2 -\left(2T_0\epsilon_0(T_0)+\int_0^{T_0}dt \epsilon_\textup{C}\right) \left(t_1 \me^{-2\delta\tilde{\epsilon}_V}+\int_{t_1}^{t_2}dt\, \Delta_C(t)\right)\\
&+  \left(t_1 \me^{-2\delta\tilde{\epsilon}_V}+\int_{t_1}^{t_2}dt\, \Delta_C(t)\right)^2\\
\geq&  -\left(2T_0\epsilon_0(T_0)+\int_0^{T_0}dt \epsilon_\textup{C}(t)\right) \left(t_1 \me^{-2\delta\tilde{\epsilon}_V}+\Delta t\right)+  t_1^2\me^{-4\delta\tilde \epsilon_V} +2\me^{-2\delta\tilde \epsilon_V}\,t_1 \int_{t_1}^{t_2}dt\, \Delta_C(t) \\&+\left(\int_{t_1}^{t_2}dt\, \Delta_C(t)\right)^2\\
\geq&  -\left(2T_0\epsilon_0(T_0)+\int_0^{T_0}dt \epsilon_\textup{C}(t)\right) \left(t_1 \me^{-2\delta\tilde{\epsilon}_V}+\Delta t\right)+  t_1^2 \me^{-4\delta\tilde \epsilon_V} +2\me^{-2\delta\tilde \epsilon_V}\,t_1 \int_{t_1}^{t_2}dt\, \Delta_C(t) -\Delta t^2\\
 =&\Delta^2(0),
\ea
where we have used Definition \eqref{eq:def delt t} and Lemma \ref{lem: t bouns for DeltaC}.
Similarly,
\ba 
\int_0^{T_0} dt \,t\,\tr [\rho(t)]& \leq \int_0^{T_0} dt \left(\Delta_L(t)+\Delta_C(t)+\Delta_R(t)+\epsilon_0\right)\\
&\leq \frac{T_0^2}{2}\left(\epsilon_0(T_0) +\epsilon_3\right)+ \int_0^{t_1}dt\, t\,\Delta_C(t)+\int_{t_1}^{t_2}dt\, t\,\Delta_C(t)+\int_{t_2}^{T_0}dt\,t\,\Delta_C(t)\\
&\leq \frac{T_0^2}{2}\left(\epsilon_0(T_0) +\epsilon_3\right)+ \frac{t_1^2}{2}+\int_{t_1}^{t_2}dt\, t\,\Delta_C(t)+\me^{-\delta (1-\tilde{\epsilon}_V)}\frac{T_0^2-t_2^2}{2}\\
&=\Delta(1),
\ea 
where $\epsilon_3$ is defined in Eq. \eqref{eq:ep 3 4 def}. Setting $t_1=T_0/2$ and simplifying Eq. \eqref{eq:R lo bound 2}, we achieve 
\ba 
{R_1}&\geq \frac{\Delta^2(0)+\epsilon_2}{2\Delta(1)-\Delta^2(0)+\epsilon_1}\\
&=\frac{\left( -\left(2T_0\epsilon_0(T_0)+\int_0^{T_0}dt \epsilon_\textup{C}(t)\right) \left(t_1 \me^{-2\delta\tilde{\epsilon}_V}+\Delta t\right)+  T_0^2\me^{-4\delta \tilde{\epsilon}_V}/4-\Delta t^2\right)/T_0^2 +\left(2\me^{-2\delta\tilde \epsilon_V}\,t_1 \int_{t_1}^{t_2}dt\, \Delta_C(t) +\epsilon_2\right)/T_0^2}{\epsilon_4+\epsilon_3+\frac{3}{4}\me^{-\delta (1-\tilde{\epsilon}_V)}+(1-\me^{-4\delta \tilde{\epsilon}_V})/4+\frac{2}{T_0^2}\int_{t_1}^{t_2}dt\left( t-t_1\me^{-2\delta\tilde \epsilon_V}\right) \Delta_C(t) +\left(\frac{\Delta t}{T_0}\right)^2}\\
&\geq\frac{-\left( \left(2T_0\epsilon_0(T_0)+\int_0^{T_0}dt \epsilon_\textup{C}(t)\right) \left(t_1 \me^{-2\delta\tilde{\epsilon}_V}+\Delta t\right)  +\Delta t^2\right)/T_0^2+\me^{-4\delta \tilde{\epsilon}_V}/4}{\epsilon_4+\epsilon_3+\frac{3}{4}\me^{-\delta (1-\tilde{\epsilon}_V)}+(1-\me^{-4\delta \tilde{\epsilon}_V})/4+\frac{2}{T_0^2} \left( t_2-t_1\me^{-2\delta\tilde \epsilon_V}\right)\int_{t_1}^{t_2}dt \Delta_C(t) +\left(\frac{\Delta t}{T_0}\right)^2}\\
&\geq\frac{-\left( \left(2T_0\epsilon_0(T_0)+\int_0^{T_0}dt |\epsilon_\textup{C}(t)|\right) \left(t_1 \me^{-2\delta\tilde{\epsilon}_V}+\Delta t\right)  +\Delta t^2\right)/T_0^2+\me^{-4\delta \tilde{\epsilon}_V}/4}{\epsilon_4+\epsilon_3+\frac{3}{4}\me^{-\delta (1-\tilde{\epsilon}_V)}+(1-\me^{-4\delta \tilde{\epsilon}_V})/4+\frac{2}{T_0^2} \left( t_2-t_1\me^{-2\delta\tilde \epsilon_V}\right) \Delta t +\left(\frac{\Delta t}{T_0}\right)^2},
\ea
where $\epsilon_4$ is defined in Eq. \eqref{eq:ep 3 4 def}.
\end{proof}

\subsection{${R_1}$ to leading order}\label{R1 to leading order}

\begin{lemma}\label{R to leading order}
Assume $\tilde{\epsilon}_V\rightarrow 0$, $\epsilon(T_0,d)\rightarrow 0$, $\delta\rightarrow \infty$, $a\rightarrow\infty$, and $\delta\,\tilde{\epsilon}_V\rightarrow 0$, in the limit $d\rightarrow \infty$, then using Big-O notation, $\bo$, to leading order in $\me^{-\delta},$ $\Delta t$, and $\delta\,\tilde{\epsilon}_V$; ${R_1}$ is lower bounded by
\be \label{eq:R with order terms}
{R_1} \geq\frac{1/4+\bo\left(\epsilon(T_0,d)\,d\right)+\bo(\Delta t^2)}{\bo\left( A^2 \left(\frac{\sigma}{d^{\eta/2}}+1\right)\me^{-\frac{\pi}{2}d^\eta}  \right)+\bo\left(a^2d\,\epsilon(T_0,d)\right)+\bo(a^2\,\me^{-\delta})+\bo(\delta\tilde{\epsilon}_V)+3\left(\frac{\Delta t}{T_0}\right)^2},
\ee 
for all fixed constants $\eta>0$ and $\sigma, d$ satisfying $4/\sigma< d^{\eta/2}\leq d/\sigma$.
\end{lemma}
\begin{proof}
	We start with a technical definition which allows us to upper and lower bound $\gamma$  (defined in terms of the integer $m$ in Eq. \eqref{eq:gamma in term of m}) in terms of $\sigma$ and a power of $d^{\eta/2}$, which will be crucial for the rest of the proof.\\
	
	For $\eta>0$ with $4/\sigma< d^{\eta/2}\leq d/\sigma$, parametrize $m$ by
	\be\label{Eq: bar m def 1}
	m=
	\begin{cases}
		2\lfloor \bar m \rfloor &\mbox{ if } d=2,4,6,\ldots,\\ \vspace{-0.3cm}\\
		2\lfloor \bar m \rfloor+1 &\mbox{ if } d=3,5,7,\ldots,
	\end{cases}
	\ee
	where
	\be\label{Eq: bar m def 2} 
	\bar m :=
	\begin{cases}
		\frac{d^{\eta/2}}{2}\sigma +1 \in\big(3,d/2+1\big] &\mbox{ if } d=2,4,6,\ldots,\\ \vspace{-0.3cm}\\
		\frac{d^{\eta/2}}{2}\sigma +\frac{1}{2} \in\big(5/2,(d+1)/2\big]  &\mbox{ if } d=3,5,7,\ldots.
	\end{cases}
	\ee 
	Note the consistency of the domains of $\bar m$ and $m$: the domain of $\bar m$ (which follows from the constraint $4/\sigma< d^{\eta/2}\leq d/\sigma$) in Eq. \eqref{Eq: bar m def 2} implies (via Eq. \eqref{Eq: bar m def 1}) a range for $m$ of $m\in\{6,8,10,\ldots,d+2\}$ for $d=2,4,6,\ldots,$ and $m\in\{5,7,9,\ldots,d+2\}$ for $d=3,5,7\ldots,$ which is within the domain of $m$ defined in Eq. \eqref{eq:gamma in term of m}.\\
	
	From the bound $x-1\leq \lfloor x \rfloor \leq x$, $x\in\rr$, it follows $2\bar m-2\leq m\leq 2\bar m$ and thus from Eqs. \eqref{eq:gamma in term of m}, \eqref{Eq: bar m def 1}, \eqref{Eq: bar m def 2}  it follows
	\be \label{eq:up bound gamma even}
	\gamma=\frac{m-2}{d} \leq \frac{2\bar m-2}{d}=d^{\eta/2}\frac{\sigma}{d}\quad \textup{for } d=2,4,6,\ldots,
	\ee
	and 
	\be \label{eq:bound for gamma d over sigma even}
	\gamma =\frac{m-2}{d}\geq \frac{2\bar m-4}{d}= d^{\eta/2}\frac{\sigma}{d}-\frac{2}{d} \quad \textup{for } d=2,4,6,\ldots
	\ee 
	Similarly, we find
	\ba\label{eq:bound for gamma d over sigma odd}
	d^{\eta/2} \frac{\sigma}{d} -\frac{1}{d} \leq \gamma \leq d^{\eta/2}\frac{\sigma}{d} \quad \text{for } d=1,3,5,
	\ea 
	We will now bound $\epsilon_3$ using Eqs. \eqref{eq:bound for gamma d over sigma even}, \eqref{eq:bound for gamma d over sigma odd},
	\ba 
	\epsilon_3 &=  \frac{2}{T_0^2}\int_0^{T_0} dt\,t (\Delta_L(t)+\Delta_R(t))\leq \left(\frac{2 A}{T_0}\right)^2\int_0^{T_0} dt   \frac{t\,\me^{-\frac{2\pi}{\sigma^2}(\gamma d/2-\bar k(t))^2}}{1-\me^{-4\pi|\gamma d/2- \bar k(t)|/\sigma^2}}\leq \left(\frac{2 A}{T_0}\right)^2\int_0^{T_0} dt  \frac{t\,\me^{-\frac{2\pi}{\sigma^2}(\gamma d/2-1)^2}}{1-\me^{-4\pi(\gamma d/2-1)/\sigma^2}}\\
	&= {2 A^2}  \frac{\me^{-\frac{2\pi}{\sigma^2}(\gamma d/2-1)^2}}{1-\me^{-4\pi(\gamma d/2-1)/\sigma^2}}	
	\leq 2 A^2 \frac{\me^{-\frac{\pi}{2}(d^{\eta/2}-4/\sigma)^2}}{1-\me^{-2\pi (d^{\eta/2}-4/\sigma)/\sigma}},\label{eq:line last line in ep 3 bound}
	\ea 
	where the denominator of line \eqref{eq:line last line in ep 3 bound} is finite since $d^{\eta/2}-4/\sigma>0$ and we have used the observation that $1/(1-\me^{-x})$ is monotonic decreasing in $x$ for $x>0$.
	In the $d\rightarrow\infty$ limit, there are a number of possibilities, depending on whether $d^{\eta/2}/\sigma$ tends to infinity, zero or a positive constant. By calculating these cases separately. 
	\begin{itemize}
	\item [1)] If $\lim_{d\rightarrow \infty} d^{\eta/2}/\sigma=\infty$ we find
		\be 
		\epsilon_3= \bo\left( A^2 \frac{\sigma}{d^{\eta/2}}\me^{-\frac{\pi}{2}d^\eta}  \right).
		\ee 
	\item [2)] 	If $\lim_{d\rightarrow \infty} d^{\eta/2}/\sigma$ converges to a constant,  we find
\be 
\epsilon_3= \bo\left( A^2 \me^{-\frac{\pi}{2}d^\eta}  \right).
\ee 
	\end{itemize}

	Thus recalling that Big-O notation captures the worse case scenario (i.e. it is an upper bound in the limit), combining the conclusions from 1) and 2) above we find for arbitrary $d^{\eta/2}/\sigma$,
	\be 
	\epsilon_3= \bo\left( A^2 \left(\frac{\sigma}{d^{\eta/2}}+1\right)\me^{-\frac{\pi}{2}d^\eta}  \right).
	\ee 
	Similarly, we find that $\epsilon_\textup{C}$ defined in Eq. \eqref{eq: ep C def} is of order
		\be 
	\epsilon_\textup{C}= \bo\left( A^2 \left(\frac{\sigma}{d^{\eta/2}}+1\right)\me^{-\frac{\pi}{2}d^\eta}  \right).
	\ee 
	
	Recalling the definition of $\epsilon_0$, $\epsilon_1$, $\epsilon_2$ and $\epsilon_4$, we find
	\ba 
	&\epsilon_0 = \bo(d\,\epsilon(T_0,d)),&\quad
	&\epsilon_1 = \bo(a^2\,\me^{-\delta})+\bo(a^2\,\epsilon(T_0,d)),&\quad
	&\epsilon_4 = \bo(d\,\epsilon(T_0,d))+\bo(\me^{-\delta})+\bo(\epsilon_\textup{C}),&
	\ea
	Thus concluding the proof.
\end{proof}

\begin{corollary}\label{corrolary w lims}
	Recall that $\Delta t/T_0=x_{vr}/\pi+\gamma$. Let the conditions in Lemma \ref{R to leading order} be satisfied. In addition, if $x_{vr}$ is of higher order than $\gamma$, $(x_{vr}=\lo(\gamma))$; and the order terms in Eq. \eqref{eq:R with order terms} are of higher order than $\gamma^{2}$, $\big(\text{i.e.\,\,\,} A^2 (\sigma/d^{\eta/2}+1)e^{-\frac{\pi}{2}d^{\eta}}=\lo(\gamma)$, $a^2 d \,\epsilon(T_0,d)=\lo(\gamma)$, $a^2 \me^{-\delta}=\lo(\gamma)$, $\delta \tilde \epsilon_V=\lo(\gamma)\big)$ it follows from Eqs. \eqref{eq:up bound gamma even}, \eqref{eq:bound for gamma d over sigma odd}, that
	\be 
	{R_1}\geq \frac{1}{12} \gamma^{-2}+\lo(\gamma^{-2})\geq \frac{1}{12} \frac{d^{2-\eta}}{\sigma^2}+\lo\left(\frac{d^{2-\eta}}{\sigma^2}\right),
	\ee 
	where $\lo$ is little-o notation.
\end{corollary}
\begin{proof}
	Follows directly from Lemma \ref{R to leading order} by keeping leading order terms only.
\end{proof}
\subsection{Showing that the limit requirements of Corollary \ref{corrolary w lims} can be met}\label{Showing that the limit requirements}
The difficulty in proving the conditions given in Lemma \ref{R to leading order} can be satisfied, and that $x_{vr}$ can be of higher order than $\gamma$, consists in finding a potential $\bar{V}_0$ which can be parametrized in terms of $d$ adequately. In this section, we will demonstrate via a particular $\bar{V}_0$ that this can indeed be achieved. Recall that the potential $\bar V_0$ which we will introduce here, has been summarised in Section \ref{Explicit form on the potential bar V0} along with the particular parametrizations which we will pick in this section, as need be.

\subsubsection{Generic definition of $\bar{V}_0$}
We start by defining $\bar{V}_0:\rr\rightarrow\rr$ appropriately. Let\footnote{One could also exchange the term $1/(\delta d^2)$ with $1/(\delta d^K)$ for some constant $K>0$ and the main result, Theorem \ref{Thm:R lower quant bound} would still go through.}
\be \label{eq:def of rapped pot}
\bar{V}_0(x):=\frac{1}{\delta d^2}+n A_0 \sum_{p=-\infty}^{\infty} V_B\left (n(x-x_0-2\pi p) \right),
\ee
where $n\geq 1$ is a free parameter, $A_0$ a normalization constant such that
\be \label{eq:norr V0 requirement}
\int_0^{2\pi} \bar{V}_0(x)\,dx =1,
\ee
and
\be 
V_B(x):=\textup{sinc}^{2 N}(x):=\left( \frac{\sin(\pi x)}{\pi x}\right)^{2 N},
\ee 
where $V_B(0)=1$ (defined by continuity in $x$) and $N\in \nnp$ is another free parameter. In the following if statements about $\bar{V}_0$ or $V_B$ are made without mentioning either $n$ or $N$, it is to be understood that these properties hold for all $n\geq 1$ and $N\in \nnp$.\\
We will now show some useful properties of $\bar{V}_0$. Later we will parameterize $n$ in terms of $d$, while $N$ will be a constant dependent on the small ($d$ independent) parameter $\eta$ introduced in Lemma \ref{R to leading order}. Such parametrizations will eventually lead to the explicitly form of the potential summarised in Section \ref{Explicit form on the potential bar V0}.

\subsubsection{Finding a value for $a$}
It will now become apparent why we included the term $1/(\delta d^2)$ in the definition of the potential $\bar V_0$.
\begin{corollary}\label{cor:a parametrisation for a}
	The constant $a$ defined in Eq. \eqref{eqdef: a} can be set equal to the following 
	\be\label{eq:a set to}
	a= \frac{T_0}{4\pi} d^2.
	\ee 
\end{corollary}
\begin{proof}
	Since all the terms in the definition of $\bar V_0$ are non-negative, it follows $\min_{x\in[0,2\pi]} \bar V_0(x)\geq 1/(\delta d^2)$. Hence
	\be \label{eq:theorem q clock}
	\frac{T_0}{4\pi\delta}\frac{1}{\min_{x\in[0,2\pi]}\bar V_0(x)}\leq \frac{T_0}{4\pi}d^2.
	\ee
	Eq. \eqref{eq:a set to} is thus a direct consequence of Eq. \eqref{eqdef: a}.
\end{proof}

\subsubsection{Properties of $\bar{V}_0$}
First note that $\bar{V}_0$ is periodic with period $2\pi$, and is infinitely differentiable. The first of these properties is clear from Eq. \eqref{eq:def of rapped pot} while the second will be demonstrated in Lemma \ref{lem: devs grow like powers}. These two properties are required by definition in \cite{WSO16}. We will now show some additional properties, specific to this particular choice of potential function $\bar{V}_0$. 
\begin{lemma}[Normalization]\label{lem:Normalization A0}
$A_0$ is $n$ independent. In particular, it satisfies $A_0<\underline{A}_0$, where $\underline{A}_0$ is only a function of $N$.
\end{lemma}
\begin{proof}
	Since for all $x\in[0,2\pi]$,
	\be 
	\sum_{p=-\infty}^\infty\big{|}V_B(n(x-x_0-2\pi p))\big{|}\leq 1+\sum_{p\in\zz\backslash\{0\}} \frac{1}{ p^{2N}}<\infty,
	\ee 
	it follows from the Weierstrass M-test (see Theorem 7.10 in \cite{rudin1976principles}), that the sum in Eq. \eqref{eq:def of rapped pot} converges uniformly for all $x\in[0,2\pi]$. We will use this in the following to exchange summation and integration limits over a finite interval. 
Using Eq. \eqref{eq:norr V0 requirement}, we have
\ba 
1&=\int_0^{2\pi} dx \bar{V}_0(x)=\int_0^{2\pi}dx \bar{V}_0(x+x_0)=A_0 n  \int_0^{2\pi}dx \sum_{p=-\infty}^\infty V_B(n(x-2\pi p))+\int_0^{2\pi}\frac{dx}{\delta d^2}\\
&=A_0 n \sum_{p=-\infty}^\infty \int_0^{2\pi}dx V_B(n(x-2\pi p))+\frac{2\pi}{\delta d^2}= A_0 n \sum_{p=-\infty}^\infty \int_{-2\pi p}^{2\pi(1-p)}dx V_B(nx)+\frac{2\pi}{\delta d^2}= A_0 n \int_{-\infty}^\infty dx V_B(n x)+\frac{2\pi}{\delta d^2}\\
&= A_0  \int_{-\infty}^\infty dx V_B(x)+\frac{2\pi}{\delta d^2}.
\ea 
Hence,
\be 
A_0=\frac{1-2\pi/(\delta d^2)}{\int_{-\infty}^\infty dx V_B(x)}< \frac{1}{\int_{-\infty}^\infty dx V_B(x)}=:\underline{A}_0 ,
\ee 
which is $n$ independent, since $V_B$ is, and well defined since $V_B\in L^1$ with a non-zero integral.
\end{proof}
\begin{lemma}[Technical Lemma needed for Lemma \ref{lem: devs grow like powers}]\label{tech lemm dev interchange} 
\be 
\frac{d^k}{dx^k} \sum_{p=-\infty}^\infty  V_B(x n -2\pi n p)= 	 \sum_{p=-\infty}^\infty \frac{d^k}{dx^k} V_B(x n -2\pi n p),\label{lem stament pre power}
\ee
	for all $k\in\nnp$ and $x\in[0,2 \pi]$.

\end{lemma}
\begin{proof}
	For $p=\pm 2,\pm 3,\pm 4,\ldots$, $x\in[0,2\pi]$ and for all $k\in\nnp$,
\ba \label{eq:M p k}
\bigg{|}\frac{d^k}{dx^k} V_B&(x n -2\pi n p)\bigg{|}\\
&\leq \Bigg{|} \frac{n^k}{(\pi n ( x-2\pi\ p))^{2N}} \sum_{q=0}^k \binom{k}{q} \frac{ (-{2N})(-{2N}-1)\ldots (-{2N}-q+1) }{n^q (x -2 \pi p)^q} \bigg{[}\frac{d^{k-q}}{dz^{k-q}} \sin^{2N}(\pi z)\bigg{]}_{z=n(x-2\pi p)}	\Bigg{|}\\ 
& \leq  \frac{n^k}{(\pi^2 n |p|)^{2N}} \sum_{q=0}^k \binom{k}{q} \frac{ \big{|}(-{2N})(-{2N}-1)\ldots (-{2N}-q+1) \big{|} }{n^q(2 \pi)^q} D_{k-q}=:M_p^k,
\ea
where we have defined $D_k:=\sup_{z\in\rr} \big{|}{d^k}/{dz^k} \,\sin^{2N}(\pi z)	\big{|}<\infty$ for $k\in\nno$. Hence, since \be \label{eq:uniform converg M test}
\sum_{p\in\zz,\, |p|\geq 2} M_p^k<\infty,
\ee
from the Weierstrass M-test (see Theorem 7.10 in \cite{rudin1976principles}), uniform convergence of $\sum_{p=-\infty}^\infty \, {d^k}/{dx^k} V_B(x n -2\pi n p)$ follows for all $k\in\nnp$.

If $f_n:\rr\rightarrow\rr$ is a sequence of differentiable functions with domain $[a,b]$, $a,b\in\rr$ with derivatives $f_n'$, and $\sum_{n=0}^\infty f_n'$ converges uniformly on $[a,b]$, then if we define $f:=\sum_{n=0}^\infty f_n$ and $|f|<\infty$, it follows that $f'=\sum_{n=0}^\infty f_n'$, where $f'$ is the derivative of $f$. This is a well known Theorem and can be found in e.g. Theorem 7.17 \cite{rudin1976principles}. Thus from Eq. \eqref{eq:uniform converg M test} for $k=1$, Eq. \eqref{lem stament pre power} follows immediate for $k=1$. Proceeding inductively, we prove Eq. \eqref{lem stament pre power} for all $k\in\nnp$.
\end{proof}

\begin{lemma}[Technical Lemma needed for Lemma \ref{lem: devs grow like powers}]\label{lem:finite set of discont}
	Let $f$ denote the \textit{Triangle function}, namely
	\be \label{eq:tringle func tech lem}
	f(x):=
	\begin{cases}
		0 &\mbox{ if } |x|\geq 1\\
		1-|x| &\mbox{ otherwise},
	\end{cases}
	\ee 
	and let $\star$ denote Convolution. Then,
		\be \label{eq:lech lemm Conve}
	f^{\,\mathlarger{\star} N}:=\underbrace{	f\star  f \star \ldots \star	f}_{N \text{-times convolution product}}\in L^1,
	\ee 
	is a continuous function which is infinitely differentiable on the intervals $(-N,-N+1)\cup(-N+1,-N+2)\cup\ldots\cup(N-1,N)$ and zero on the intervals $(-\infty,-N)\cup(N,+\infty)$ for all $N\in\nnp$.
\end{lemma}
\begin{proof}
	The proof is by induction. Let us start by defining the continuous functions $g_n:\rr\rightarrow\rr$ for $n\in \nnp$,
	\be \label{eq:gn def}
	g_n(x):= P_{n,m}(x)\quad \mbox{if } m\leq x \leq m+1 \text{ for } m\in\zz, 
	\ee
	where
	\be \label{eq:gn def 2}
	P_{n,m}(x)=0 \quad\mbox{if  } m\in(-\infty, -n-1)]\cup[n,+\infty),
	\ee 
	and $\{P_{n,m}(x)\}$ are a set of real polynomials in $x$. Continuity of $g$ implies 
	\be\label{eq:continuity Pnm} 
	P_{n,m}(m+1)=P_{n,m+1}(m+1)\quad \text{for all } m\in\zz.
	\ee
	We denote the set of such functions $g_n$ by $\mathcal{S}_n$. Note that $f$ in Eq. \eqref{eq:tringle func tech lem} belongs to $\mathcal{S}_1$. We will start by showing that $(g_1\star g_n)(y)\in\mathcal{S}_{n+1}$ for all $n\in\nnp$, $g_n\in\mathcal{S}_n$ and $g_1\in\mathcal{S}_1$. For this we will calculate $(g_1\star g_n)(y)$ for $y\in[q,q+1]$; for $q\in\zz$ and define $\delta_y:=y-q$, $0\leq \delta_y\leq 1$. We find
	\ba \label{eq:g1 star gn 1sr eq}
	(g_1\star g_n)(y)=& \int_{-\infty}^\infty dx g_1(x) g_n(x-y)=\sum_{m=-1,0} \int_m^{m+1}dx P_{1,m}(x) g_n(x-y)\\
	=& \sum_{m=-1,0}\left( \int_m^{m+\delta_y}dx P_{1,m}(x) g_n(x-q-\delta_y)+\int_{m+\delta_y}^{m+1}dx P_{1,m}(x) g_n(x-q-\delta_y) \right)\\
	=& \sum_{m=-1,0}\left( \int_m^{m-q+y}dx P_{1,m}(x) P_{n,m-q-1}(x-y)+\int_{m-q+y}^{m+1}dx P_{1,m}(x) P_{n,m-q}(x-y) \right)\label{line: g1 gn}\\
	=:&P'(n,q;y),
	\ea 
	where $P'(n,q;y)$ is a polynomial in $y$ with coefficient depending on $n$ and $q$. This follows by noting: 1) the integral of $P_{1,m}(x) P_{n,m-q-1}(x-y)$ and $P_{1,m}(x) P_{n,m-q}(x-y)$ w.r.t. $x$ are polynomials in both $x$ and $y$ (this is trivial to see by a formal power-law expansion). 2) The $x$ variable is then evaluated at a linear function in $y$ which leaves us with a polynomial in $y$. 3) these polynomials are then summed over the coefficients in $m$, which again is another polynomial in $y$. $(g_1\star g_n)(y)$ is continuous in $y\in\rr$, since both $g_1$ and $g_n$ are continuous and the convolution of two continuous functions is continuous.

	All that is left to show, to prove that $(g_1\star g_n)(y)\in\mathcal{S}_{n+1}$, is to show that $P'(n,q;y)=P_{n+1,q}(y)$ for some $P_{n+1,q}(y)$ obeying Eqs. \eqref{eq:gn def}, \eqref{eq:gn def 2}. That $P'(n,q;y)$ obeys Eq. \eqref{eq:gn def} follows from the relationship between $y$ and $q$, namely $y\in[q,q+1]$. Eq. \eqref{eq:gn def 2} is true for $P'(n,q;y)$ if
	\be \label{eq:to be verified for P'}
	P'(n,q;y)=0 \quad \text{if } q\in(-\infty, -n-2] \cup [n+1,\infty).
	\ee
	We now verify that Eq. \eqref{eq:to be verified for P'} is satisfied. $q\leq -n-2$ implies $m-q-1\geq m+n+1\geq n$ for $m=-1,0$. Thus it follows using Eq. \eqref{eq:gn def 2} that the term $P_{n,m-q-1}$ in line \eqref{line: g1 gn} is zero for $q\leq -n-2$. Furthermore $q\leq -n-2$ implies $m-q\geq m+n+2\geq n+1$ for $m=-1,0$. Thus it follows using Eq. \eqref{eq:gn def 2} that the term $P_{n,m-q}$ in line \eqref{line: g1 gn} is zero for $q\leq -n-2$. Hence $P'(n,q;y)=0$ for $q\leq -n-2$. Similarly we can verify that $P'(n,q;y)=0$ for $q\geq n+1$. This concludes the proof that $(g_1\star g_n)(y)\in\mathcal{S}_{n+1}$ for all $n\in\nnp$, $g_n\in\mathcal{S}_n$ and $g_1\in\mathcal{S}_1$.
	
	Now work inductively to conclude that $g_1^{\,\mathlarger{\star} N}\in\mathcal{S}_{N}$. Thus recalling that $f\in\mathcal{S}_1$ and noting that all functions in $\mathcal{S}_N$ satisfy the conditions on $f^{\,\mathlarger{\star} N}$ of the Lemma, we conclude the proof.

\end{proof}

\begin{lemma}\label{lem: devs grow like powers}
	There exists $C_0=C_0(N)>0$ which is only a function of $N$, i.e. independent of $n$, $d$, and $k$, such that
\be \label{eq:main eq lemma dev bound}
\max_{x\in[0,2\pi]} \bigg{|} \frac{d^k}{dx^k}  \bar{V}_0(x)\bigg{|} \leq n^{k+1} C_0^{k+1},\quad \forall\, k\in\nno,\,\, \forall\, n \geq 1.
\ee
\end{lemma}
\begin{proof} The proof will consist in writing the $k$th derivative in Fourier space and using properties of the Fourier Transform to calculate and upper bound the resultant expression.\\
	We will start by simplifying the expression for the derivative.

	We start by noting that the Fourier Transform of $V_B\in L^1$ can be computed via the Convolution theorem, 
	\be \label{F trans of VB}
	\mathcal{F}(V_B) (y)=	\mathcal{F}((\textup{sinc}^2)^N) (y)= \underbrace{[	\mathcal{F}(\textup{sinc}^2)\star  	\mathcal{F}(\textup{sinc}^2)\star \ldots \star	\mathcal{F}(\textup{sinc}^2)]}_{N \text{-times convolution product}}(y),
	\ee
	which is well-defined, since the Convolution Theorem maps two $L^1$ functions to an $L^1$ function, and $\textup{sinc}^2\in L^1$. 
	By direct calculation, we have that
	\be 
	\mathcal{F}(\textup{sinc}^2)(y)=
	\begin{cases}
	0 &\mbox{ if } |y|\geq 1\\
	1-|y| &\mbox{ otherwise},
	\end{cases}
	\ee
	thus since the Convolution of two finite support functions, has finite support, we conclude by induction from Eq. \eqref{F trans of VB} that $\mathcal{F}(V_B) (y)$ has finite support. We denote the finite interval containing the support of $\mathcal{F}(V_B) (y)$ by $[y_\textup{min}, y_\textup{max}]$. Furthermore, note that $\mathcal{F}(\textup{sinc}^2)$ has a discontinuous derivative at three points on its domain $(0, y_\textup{min},y_\textup{max})$. When $	\mathcal{F}(\textup{sinc}^2)$ is convoluted with itself, the resultant function might also have a finite set of points in its support interval at which it is not differentiable (see Lemma \ref{lem:finite set of discont})\footnote{In fact, it is not differentiable at these points, but we will not need to prove this for our purposes.\label{refnote}}. Similarly due to Eq. \eqref{F trans of VB}, $\mathcal{F}(V_B) (y)$ may have a finite set of points contained in $[y_\textup{min}, y_\textup{max}]$ at which the function is not differentiable\footref{refnote}. From Lemma \ref{lem:finite set of discont}, we conclude that a set containing all such points is $\{y_l=-N-1+l\}_{l=1}^{2N+1}$, where $y_1=y_\textup{min}$, $y_{2N+1}=y_\textup{max}$. Finally, the last property of $\mathcal{F}(V_B)(y)$ which we will need; is that since $\mathcal{F}(\textup{sinc}^2)(y)$ has finite right and left 1st and 2nd derivatives in the interval $y\in[y_\textup{min},y_\textup{max}]$, it follows that $\mathcal{F}(V_B)(y)$ also has finite left and right 1st and 2nd derivatives in the interval $y\in[y_\textup{min},y_\textup{max}]$ (see Lemma \ref{lem:finite set of discont}).\\
	Using Lemma \ref{tech lemm dev interchange} and the change of variable $z=nx -2\pi p$, we have for $k\in\nno$
	\ba\label{eq:max over x kth up inter}
	\max_{x\in[0,2\pi]} \bigg{|} \frac{d^k}{dx^k} \bar{V}_0(x)\bigg{|}&=\max_{x\in[0,2\pi]} \bigg{|} \frac{d^k}{dx^k} \bar{V}_0(x+x_0)\bigg{|}&\\
	&=
	nA_0	\max_{x\in[0,2\pi]} \bigg{|} \frac{d^k}{dx^k}\sum_{p=-\infty}^\infty  V_B(x n -2\pi n p) \bigg{|}\\
		&=
	nA_0	\max_{x\in[0,2\pi]} \bigg{|} \sum_{p=-\infty}^\infty \frac{d^k}{dx^k} V_B(x n -2\pi n p) \bigg{|}\\
	&= n^{k+1} 	A_0	\max_{x\in[0,2\pi]} \bigg{|}	\sum_{p=-\infty}^\infty \bigg{[} \frac{d^k}{dz^k} V_B(z) \bigg{]}_{z=nx -2\pi n p}\bigg{|}\\
	&=  n^{k+1} 	A_0	\max_{x\in[0,2\pi]} \bigg{|} 	\sum_{p=-\infty}^\infty \bigg{[} \mathcal{F}^{-1}\bigg( (2 \pi \mi y)^k \mathcal{F}(V_B) \bigg)(z) \bigg{]}_{z=nx -2\pi n p}\bigg{|}\label{line:kth dev intermediate}\\
	&=  n^{k+1} 	A_0	\max_{x\in[0,2\pi]} \bigg{|}	\sum_{p=-\infty}^\infty  \int_{-\infty}^\infty dy  (2 \pi \mi y)^k \mathcal{F}(V_B)(y) \,\me^{-2\pi \mi (nx -2\pi n p)y}\bigg{|} \\
	& \leq  2  n^{k+1} 	A_0	\max_{x\in[0,2\pi]} 	\sum_{p=0}^\infty \bigg{|} \int_{y_\textup{min}}^{y_\textup{max}} dy  (2 \pi \mi y)^k \cos(4\pi^2 n p y)\mathcal{F}(V_B)(y)\, \me^{-2\pi \mi nxy}\bigg{|} .
	\ea  
	Note that the inverse Fourier Transform is well defined on the domain $L^2$. Thus  line \eqref{line:kth dev intermediate} is well justified since $(2 \pi \mi y)^k \mathcal{F}(V_B)\in L^2$ $\forall\,k\in\nno$, because $\mathcal{F}(V_B)$ has finite support.\\
	For the term $p=0$ we have
	\ba \label{eq:p is zero term}
	\bigg{|}\int_{y_\textup{min}}^{y_\textup{max}} dy  (2 \pi \mi y)^k \mathcal{F}(V_B)(y) \me^{-2\pi \mi nxy}	\bigg{|}&\leq (y_\textup{max}-y_\textup{min}) |2\pi y_\textup{max}|^k \max_{x\in [y_\textup{max},y_\textup{min}]} \big{|}\mathcal{F}(V_B)(x) \big{|}\\
	&\leq y_\textup{M} (2\pi  y_\textup{M})^k \mathcal{F}_\textup{max}^{(0)}(V_B),
	\ea 
	where we have introduced the notation $\max\{|y_\textup{min}|,|y_\textup{max}|,1 \}=N=:y_\textup{M}$,  $d^k/dy^k \mathcal{F}(V_B)(y)=:\mathcal{F}^{(k)}(V_B)(y)$, and $\sup_{y \in(y_1,y_2)\cup (y_2,y_3)\cup\ldots\cup(y_{2N},y_{2N+1})} |\mathcal{F}^{(k)}(V_B)(y)|=: \mathcal{F}^{(k)}_\textup{max}(V_B)$. We will now bound the other terms $p\neq 0$. In the following we will denote, $\lim_{\varepsilon\rightarrow 0^+} \mathcal{F}^{(k)}(V_B)(y_0\pm \varepsilon)=: \mathcal{F}^{(k)}_{\pm}(V_B)(y_0)$, and use primes to denote 1st derivatives w.r.t. $y$. We will proceed by integration by parts twice, after splitting the integral up into sections in which $\mathcal{F}(V_B)(y)$ has continuous derivatives. We find
	\ba 
	\bigg{|}\int_{y_\textup{min}}^{y_\textup{max}} &dy  (2 \pi \mi y)^k\cos(4\pi^2 n py) \mathcal{F}(V_B)(y) \me^{-2\pi \mi nxy}	\bigg{|}=	\bigg{|}\sum_{r=1,\ldots,2N}\int_{y_r}^{y_{r+1}} dy  (2 \pi \mi y)^k\cos(4\pi^2 n py) \mathcal{F}(V_B)(y)\, \me^{-2\pi \mi nxy}	\bigg{|}\\
	=&\bigg{|}\underbrace{\sum_{r=1,\ldots,2N}\bigg{[}  (2 \pi \mi y)^k\frac{\sin(4\pi^2 n py)}{4\pi^2 n p} \mathcal{F}(V_B) (y)\,\me^{-2\pi \mi nxy}\bigg{]}_{y_r}^{y_{r+1}}}_{=\,0 \text{ Since $\mathcal{F}(V_B)(y)$ is continuos}}\\ &-\sum_{r=1,\ldots,2N}\int_{y_r}^{y_{r+1}} dy \frac{\sin(4\pi^2 n py)}{4\pi^2 n p} (2 \pi \mi y)^{k-1} \me^{-2\pi \mi nxy}\,2\pi \mi\bigg{(}(k-2\pi \mi y nx)  \mathcal{F}(V_B)(y)+  \mathcal{F}(V_B)(y)+  y\mathcal{F}^{(1)}(V_B)(y)    \bigg{)}	\bigg{|}.\\
	=&	\bigg{|}\lim_{\varepsilon\rightarrow 0^+}\underbrace{\sum_{r=1,\ldots,2N} \bigg{[}\frac{-\cos(4\pi^2 n py)}{16\pi^4 n^2 p^2} (2 \pi \mi y)^{k-1} \me^{-2\pi \mi nxy}\,2\pi \mi\bigg{(}(k-2\pi \mi y nx)  \mathcal{F}(V_B)(y)+  \mathcal{F}(V_B)(y)+  y\mathcal{F}^{(1)}(V_B)(y)     \bigg{)}	\bigg{]}_{y_r+\varepsilon}^{y_{r+1}-\varepsilon}}_{\text{May not be zero since $\mathcal{F}^{(1)}(V_B)(y)$ may be discontinuos at points $y_1,y_2\ldots,y_{2N+1}$}}\\
	&-\sum_{r=1,\ldots,2N} \int_{y_r}^{y_{r+1}}dy \frac{-\cos(4\pi^2 n py)}{16\pi^4 n^2 p^2} \bigg{(} (2 \pi \mi y)^{k-1} \me^{-2\pi \mi nxy}\,2\pi \mi\big{(}(k-2\pi \mi y nx)  \mathcal{F}(V_B)(y)+  \mathcal{F}(V_B)(y)+  y\mathcal{F}^{(1)}(V_B)(y)\big{)} \bigg{)}'\,\bigg{|}\\
	\leq 
&\sum_{r=1,\ldots,2N+1}\bigg{|}	\frac{-\cos(4\pi^2 n py_r)}{16\pi^4 n^2 p^2} (2 \pi \mi y_r)^{k-1} \me^{-2\pi \mi nxy_r}\,2\pi \mi  y_r\Big{(}\mathcal{F}^{(1)}_{-}(V_B)(y_r)-\mathcal{F}^{(1)}_{+}(V_B)(y_r)     \Big{)}\bigg{|}\\
&+ \sum_{r=1,\ldots,2N} \frac{y_{r+1}-y_{r}}{16\pi^4 n^2 p^2} \sup_{y \in(y_{r},y_{r+1}) } \bigg{|}\bigg{(} (2 \pi \mi y)^{k-1} \me^{-2\pi \mi nxy}\,2\pi \mi\big{(}(k-2\pi \mi y nx)  \mathcal{F}(V_B)(y)+  \mathcal{F}(V_B)(y)+  y\mathcal{F}^{(1)}(V_B)(y)\big{)} \bigg{)}'\,\bigg{|}\\
\leq&  \frac{1}{16\pi^4 n^2 p^2} \Bigg( (2\pi y_\textup{M})^{k-1} y_\textup{M} 4\pi (2N+1) \mathcal{F}^{(1)}_\textup{max}(V_B)   \label{inter step 106}\\
&\quad\quad\quad\quad\quad +(y_\textup{max}-y_\textup{min}) 2\pi (2\pi y_\textup{M} )^{k-1} \bigg{[} \big{(}(k-1)(2\pi y_\textup{M})^{-1} +n x\big{)} \big{(}(k+2\pi y_\textup{M} n x + 1) \mathcal{F}_\textup{max}(V_B) +y_\textup{M}\mathcal{F}^{(1)}_\textup{max}(V_B)\big{)}  2\pi \nonumber\\
&\quad\quad\quad\quad\quad+2\pi n x \mathcal{F}_\textup{max}(V_B) + (k+2\pi n x +2) \mathcal{F}_\textup{max}^{(1)}(V_B)+y_\textup{M} \mathcal{F}_\textup{max}^{(2)}(V_B)	\Bigg{]} \Bigg{)}.\nonumber
	\ea 
	Note the dependency on $n$ of the expression on the R.H.S. of the inequality line \eqref{inter step 106}. It is of the form $(c_{00}+n c_{01}+n^2 c_{02})/n^2$, where $c_{00}\geq 0$, $c_{01}\geq 0$, $c_{02}\geq 0$ are $n$ independent. Thus since $n\in[1,\infty)$, the R.H.S. of line \eqref{inter step 106} is upper bounded uniformly in $n$ by setting $n=1$. Similarly, the coefficients $c_{01}\geq 0$, $c_{02}\geq 0$ are upper bounded in $x\in[0,2 \pi]$, by setting $x=2 \pi$. We thus have for $|p|\in\nnp$,
	\ba\label{eq:penultimate inter for multi dev bound}
	\bigg{|}\int_{y_\textup{min}}^{y_\textup{max}} &dy  (2 \pi \mi y)^k\cos(4\pi^2 n py) \mathcal{F}(V_B)(y) \me^{-2\pi \mi nxy}	\bigg{|}\\
	\leq&  \frac{(2\pi N)^{k-1}}{16\pi^4  p^2} \Bigg(  N 4\pi (2N+1) \mathcal{F}^{(1)}_\textup{max}(V_B)   
	\\
	&\quad\quad\quad\quad\quad +N 2\pi  \bigg{[} \big{(}(k+1)(2\pi N)^{-1} + 2\pi \big{)} \big{(}(k+(2\pi)^2 N  + 1) \mathcal{F}_\textup{max}(V_B) +N\mathcal{F}^{(1)}_\textup{max}(V_B)\big{)}  2\pi \nonumber\\
	&\quad\quad\quad\quad\quad+(2\pi)^2 \mathcal{F}_\textup{max}(V_B) + (k+(2\pi)^2 +2) \mathcal{F}_\textup{max}^{(1)}(V_B)+N \mathcal{F}_\textup{max}^{(2)}(V_B)	\Bigg{]} \Bigg{)},
	\ea
	where we have used $y_\textup{M}=N$. The R.H.S. of the inequality Eq. \eqref{eq:penultimate inter for multi dev bound} can be written as
	\be 
	\frac{(2\pi N)^{k}}{16\pi^4  p^2}\left( c_{03}+c_{04}k +c_{05}k^2  \right),
	\ee 
	with the coefficients $c_{03}\geq 0$, $c_{04}\geq 0$, $c_{05}\geq 0$ are $k$, $x$ and $p$ independent. Thus from Eqs. \eqref{eq:max over x kth up inter}, \eqref{eq:p is zero term}, \eqref{eq:penultimate inter for multi dev bound},
	\ba
		\max_{x\in[0,2\pi]} \bigg{|} \frac{d^k}{dx^k} \bar{V}_0(x)\bigg{|}&=\max_{x\in[0,2\pi]} \bigg{|} \frac{d^k}{dx^k} \bar{V}_0(x+x_0)\bigg{|}\\
		&\leq  2  n^{k+1} 	A_0 \bigg{(} N (2\pi  N)^k \mathcal{F}_\textup{max}^{(0)}(V_B) + \sum_{p\in \zz \backslash \{0\}} 	\frac{(2\pi N)^{k}}{16\pi^4  p^2}\bigg( c_{03}+c_{04}k +c_{05}k^2  \bigg)\\
		&=  2  n^{k+1} 	A_0 (2\pi N)^{k} \bigg{(} N \mathcal{F}_\textup{max}^{(0)}(V_B) +	\frac{1}{16\pi^4}\frac{2 \pi^2}{6}\bigg( c_{03}+c_{04} \me^{\ln k} +c_{05}\me^{2\ln k}  \bigg)\bigg)\\
		&\leq  2  n^{k+1} 	A_0 (2\pi N)^{k} \bigg{(} N \mathcal{F}_\textup{max}^{(0)}(V_B) +	\frac{1}{48\pi^2}\bigg( c_{03}+c_{04} \me^{ k} +c_{05}\me^{2 k}  \bigg)\bigg),\label{eq:up bound mult deov intermed} 
	\ea
	for $k\in\nno$.
	Thus recalling from Lemma \ref{lem:Normalization A0} that $A_0$ is upper bounded by $\underline{A}_0$ which is only a function of $N$, from Eq. \eqref{eq:up bound mult deov intermed}, it follows that there exists a coefficient $C_0$ which only depends on $N$ such that Eq. \eqref{eq:main eq lemma dev bound} holds.
\end{proof} 
\subsubsection{Determining a parametrization of $n$ in terms of $d$ such that $\epsilon(T_0,d)\rightarrow0$ quicker than any polynomial in $d$.}

 Let us start by introducing the constraint
\ba
\frac{d}{\bar{\upsilon} \sigma} &=d^{\epsilon_5},\quad \text{for some fixed constant } 0<{\epsilon_5}. \label{eq:epsilon 5 def}
\ea
It will also be useful to introduce the variable $\epsilon_6$ which is a function of $d$ and $\sigma$ via
\ba\sigma &= d^{\epsilon_6},
\ea 
where $\epsilon_6$ uniformly bounded to the interval\footnote{By ``uniformly bounded'', it is meant that $\epsilon_5<\lim_{d\rightarrow \infty}\epsilon_6<1$} $\epsilon_5<\epsilon_6<1.$
We can re-write Eq. \eqref{eq:mathcal N def} as
\be
\mathcal N= \Bigg\lfloor \alpha_0^2\frac{\pi }{2}\,\frac{d^{2-2\epsilon_6}}{d^{2-2\epsilon_5-2\epsilon_6}+d^{2-4\epsilon_6}+d^{2-\epsilon_5-3\epsilon_6}}\Bigg\rfloor= \Bigg\lfloor \alpha_0^2\frac{\pi }{2}\,\frac{1}{d^{-2\epsilon_5}+d^{-2\epsilon_6}+d^{-\epsilon_5-\epsilon_6}}\Bigg\rfloor.
\ee 
Thus recalling that $\alpha_0\in(0,1]$ is a fixed constant, it follows that $\mathcal{N}=\infty$ and $\bar\upsilon\geq 0$ in $d \rightarrow \infty$ limit. Hence Eqs.  \eqref{eq:mathcal N contraints} are satisfied for sufficiently large $d$  and Eq. \eqref{eq:epsilon t d} holds for the parameterisations considered in this proof in the limit $d\rightarrow\infty$.
We can now re-write eq. \eqref{eq:epsilon t d} as
\be\label{eq:epsilon t d re write}
\epsilon(t,d)=
|t| \frac{d}{T_0}\!\left(\bo\left( \frac{\sigma^3}{\sigma d^{-\epsilon_5}+1}\right)^{1/2}\!\!+\bo\left(\frac{d^2}{\sigma^2}+b\right)\right) \exp\left(-\frac{\pi}{4}\frac{\alpha_0^2}{\left(1+d^{\epsilon_5}/\sigma\right)^2} d^{2\epsilon_5} \right)+\bo\left(|t|\frac{d^2}{\sigma^2}+1\right)\me^{-\frac{\pi}{4}\frac{d^2}{\sigma^2}}+\bo\left( \me^{-\frac{\pi}{2}\sigma^2} \right). 
\ee
thus since (as we will soon show), $b$ grows at most polynomially in $d$, $\epsilon(t,d)$ decays quicker than any polynomial in $d$ in the limit $d\rightarrow\infty$. We will now workout the implications of Eq. \eqref{eq:epsilon 5 def}. From Eq. \eqref{eq:b def eq 0} and the relation $-\mi \delta \,\bar{V}_0(x)=V_0(x)$ stated in Eq. \eqref{eq:pod eqs def}, we have that $b$ is any non-negative number satisfying
\be\label{eq:b def eq}
b\geq\; \sup_{k\in\nnp}\left(2\max_{x\in[0,2\pi]} \left|  \delta \bar{V}_0^{(k-1)}(x) \right|\,\right)^{1/k}.
\ee 
Thus from Lemma \ref{lem: devs grow like powers}, it follows 
\be 
\sup_{k\in\nnp}\left(2\max_{x\in[0,2\pi]} \left| \delta \bar{V}_0^{(k-1)}(x) \right|\,\right)^{1/k}\leq \sup_{k\in\nnp}\left(2\delta n^k C_0^k \right)^{1/k}=2 \delta n C_0,
\ee
hence we will set $b= 2\delta n C_0$. Using definition the of $\bar{\upsilon}$ (Eq. \eqref{eq:upsilon def}), we have
\be 
d^{\epsilon_5}= \frac{d}{\bar{\upsilon} \sigma}= \frac{d}{\sigma} \frac{\ln(\pi \alpha_0\sigma^2)}{\pi \alpha_0\kappa}\frac{1}{b}=\frac{d}{\sigma} \frac{\ln(\pi \alpha_0\sigma^2)}{2\pi C_0\alpha_0\kappa}\frac{1}{\delta n},
\ee
with $\kappa=0.792$ from which we find the constraint on $n$,
\be\label{eq:n in terms of d def} 
n= \frac{\ln(\pi \alpha_0\sigma^2)}{2\pi C_0\alpha_0\kappa} \frac{d^{1-\epsilon_5}}{\delta \sigma}.
\ee 

\subsubsection{Determining the constant $N$ and parametrization of $\delta$ in terms of $d$ such that conditions of Corollary \ref{corrolary w lims} are satisfied}
From Eq. \eqref{eq:tilde ep V def} and using elementary properties of the potential $\bar{V}_0$ in Eq. \eqref{eq:def of rapped pot} (namely, integrates to unity over one period and is symmetric w.r.t. $x_0$, i.e. that $\bar V_0(x+x_0)=\bar V_0(-x+x_0)$ we find
\ba\label{eq:tilde ep V v2}
\tilde\epsilon_V&=1-\int_{-x_{vr}}^{x_{vr}} dx\, \bar{V}_0(x+x_0)\\
&= 1- \bigg{(} \int_{-\pi}^{\pi} dx\, \bar{V}_0(x+x_0)-\int_{-\pi }^{-x_{vr}} dx\, \bar{V}_0(x+x_0)-\int_{x_{vr}}^{\pi} dx\, \bar{V}_0(x+x_0) \bigg{)}\\
&= 1- \bigg{(} 1-\int_{x_{vr}}^{\pi} dx\, \bar{V}_0(-x+x_0)-\int_{x_{vr}}^{\pi} dx\, \bar{V}_0(x+x_0) \bigg{)}\\
&= 2\int_{x_{vr}}^{\pi} dx\, \bar{V}_0(x+x_0)\\
&\leq 2(\pi-x_{vr}) \max_{x\in[x_{vr},\pi]}\big\{ \bar{V}_0(x+x_0) \big\}\\
&\leq 2 A_0 n (\pi-x_{vr}) \sum_{p=-\infty}^\infty \max_{x\in[x_{vr},\pi]}\big\{| V_B(n x-2\pi p n)| \big\}+2\frac{(\pi-x_{vr})}{\delta d^2}\\
&\leq  2 \underline{A}_0 n (\pi-x_{vr}) \sum_{p=-\infty}^\infty\max_{x\in[x_{vr},\pi]}\big\{|\pi (n x-2\pi p n)|^{-2N} \big\}+2\frac{(\pi-x_{vr})}{\delta d^2}\\
&\leq  2 \underline{A}_0 n (\pi-x_{vr})\bigg{(} (n\pi x_{vr})^{-2 N}+ (2\pi^2 n)^{-2N}\Big{(}\sum_{p\in\zz\backslash \{0\}} p^{-2N} \Big) \bigg{)}+2\frac{(\pi-x_{vr})}{\delta d^2}\\
&\leq  2 \underline{A}_0  \pi\bigg{(} \frac{n}{(n \pi x_{vr})^{2 N}}+ \frac{2n\zeta(2N)}{(2\pi^2 n)^{2N}} \bigg{)}+\frac{2\pi}{\delta d^2}\\
&\leq  2 \underline{A}_0  \pi(1+\zeta(2N))\frac{n}{(n \pi x_{vr})^{2 N}}+\frac{2\pi}{\delta d^2}.
\ea
Now recall Corollary \ref{corrolary w lims}. One of the conditions needed for this Corollary to be satisfied is $x_{vr}=o\left(\gamma\right)$. Recalling the bound 
\be 
d^{\eta/2} \frac{\sigma}{d}-\frac{2}{d} \leq \gamma \leq d^{\eta/2} \frac{\sigma}{d},
\ee 
which was derived in Eqs \eqref{eq:up bound gamma even}, \eqref{eq:bound for gamma d over sigma odd}), we will find that the following parametrization of $x_{vr}$ and $\delta$ allows us to achieve this scaling and the other conditions introduced in Corollary \ref{corrolary w lims}.
Let
\ba \label{eq:x vr limit contraint}
&\pi x_{vr}=d^{\epsilon_7} \frac{\sigma}{d},\quad\quad \quad 
 \delta=d^{\epsilon_8},&
\ea 
for some fixed constants $0<\epsilon_7<\eta/2$, $\,0<\epsilon_8$. 
We will now show that indeed Eq. \eqref{eq:tilde ep V v2} permits the existence of constants $N$, $\epsilon_7$, $\epsilon_8$ such that Eqs. \eqref{eq:x vr limit contraint} are satisfied.

Substituting expressions for $\pi x_{vr}$ and $n$ from Eq. \eqref{eq:x vr limit contraint}, and Eq. \eqref{eq:n in terms of d def} into \eqref{eq:tilde ep V v2}, we find,
\ba 
\delta \tilde{\epsilon}_V &\leq  2 \underline{A}_0  \pi(1+\zeta(2N))\frac{\delta n}{(n \pi x_{vr})^{2 N}}+\frac{2\pi}{d^2}\\
&=  2 \underline{A}_0  \pi(1+\zeta(2N))\delta \frac{\ln(\pi \alpha_0 \sigma^2)}{2\pi C_0\alpha_0\kappa}\frac{d^{1-\epsilon_5}}{\delta \sigma} \left(\frac{2\pi C_0\alpha_0\kappa}{\ln(\pi \sigma^2)}\frac{\delta \sigma}{d^{1-\epsilon_5}}\frac{1}{d^{\epsilon_7}}\frac{d}{\sigma}  \right)^{2N}+\frac{2\pi}{d^2}\\
&=  2 \underline{A}_0  \pi(1+\zeta(2N))\left( \frac{2\pi C_0\alpha_0\kappa}{\ln(\pi\alpha_0\sigma^2)}\right)^{2N-1} \frac{d^{1-\epsilon_5+(\epsilon_8+\epsilon_5-\epsilon_7)2N}}{\sigma}+\frac{2\pi}{d^2}.
\ea 
Now define a constant $\epsilon_9$ such that $0<\epsilon_9< \eta$, and impose the constraints $\epsilon_7>\epsilon_5+\epsilon_8$ and
\be\label{eq:inter proof bound ep 9} 
\frac{d^{1-\epsilon_5+(\epsilon_8+\epsilon_5-\epsilon_7)2N}}{\sigma}\leq d^{\epsilon_9} \left(\frac{\sigma}{d}\right)^2.
\ee 
Recalling $\sigma=d^{\epsilon_6}$, it follows from Eq. \eqref{eq:inter proof bound ep 9}, 
\ba 
-3\epsilon_6-\epsilon_9+2&\leq -1+\epsilon_5+\underbrace{(\epsilon_7-\epsilon_8-\epsilon_5)}_{>0} 2N\\
\implies & N\geq \frac{-3\epsilon_6-\epsilon_9-\epsilon_5+3}{2(\epsilon_7-\epsilon_8-\epsilon_5)}.
\ea
Therefore, since $-\epsilon_5>-\epsilon_6$, we set the constant $N$ to be
\be 
 N:=\left\lceil\frac{3-4\epsilon_5-\epsilon_9}{2(\epsilon_7-\epsilon_8-\epsilon_5)}\right\rceil> \frac{3-\epsilon_5-\epsilon_9-3\epsilon_6}{2(\epsilon_7-\epsilon_8-\epsilon_5)},
\ee
with $3-4\epsilon_5-\epsilon_9>0$ and $\lceil\cdot\rceil$ denoting the \textit{Ceiling function}. Thus from Eqs. \eqref{eq:tilde ep V v2}, \eqref{eq:inter proof bound ep 9}, it follows
\be 
\delta\tilde{\epsilon}_V\leq 2 \underline{A}_0  \pi(1+\zeta(2N))\left( \frac{2\pi C_0\alpha_0\kappa}{\ln(\pi\alpha_0\sigma^2)}\right)^{2N-1} d^{\epsilon_9} \left(\frac{\sigma}{d}\right)^2+\frac{2\pi}{d^2}. 
\ee
Thus all the conditions for Corollary \ref{corrolary w lims} are satisfied as long as all the constraints on the epsilons we have introduced can simultaneously be satisfied. We will show this in the proof of Theorem \ref{Thm:R lower quant bound}. 

\subsection{Final Theorem}\label{Final Theorem}
\begin{theorem}\label{Thm:R lower quant bound}
	Consider the setup described in Section \ref{sec:setup} for \wso~clock and potential functions $\bar V_0$ introduced in Section \ref{Overview of the quasi-ideal clock} with fixed constant $\alpha_0\in(0,1]$, and $k_0=0$. For all fixed constants $0<\eta\leq1$, and $\sigma$ satisfying 
		\be\label{eq:baby constraint 2 thrm}
	d^{\eta/2}\leq \sigma<d,
	\ee
	 the clock precision ${R_1}$ is lower bounded by
	\be \label{eq:main thrm eq}
	{R_1}\geq \frac{d^{2-\eta}}{\sigma^2}+\lo\left(\frac{d^{2-\eta}}{\sigma^2}\right),
	\ee 
	in the large $d$ limit.
\end{theorem}

\begin{proof}
	The bound
	\be \label{eq:main thrm eq almost}
	{R_1}\geq \frac{1}{12}\frac{d^{2-\eta}}{\sigma^2}+\lo\left(\frac{d^{2-\eta}}{\sigma^2}\right),
	\ee 
	follows directly from Corollary \ref{corrolary w lims} and the results from Section \ref{Showing that the limit requirements}, so long as the constraints on the epsilon terms can all be simultaneously satisfied. We will check this here. The constraints introduced are:
	\ba \label{eq:epsilons contrains 1}
	&0<\epsilon_5<\epsilon_6<1,&\quad &0<\epsilon_7<\frac{\eta}{2},& \quad & 0<\epsilon_8,& \quad &\epsilon_5+\epsilon_8<\epsilon_7,& \quad & 0<\epsilon_9<\eta,& 0<3-4\epsilon_5-\epsilon_9,&
	\ea 
	and 
	\be  \label{eq:epsilons contrains 2}
	\frac{4}{\sigma}< d^{\eta/2}\leq \frac{d}{\sigma}.
	\ee 
	Let $\epsilon_7=\eta/4,$ $\epsilon_5=\epsilon_8=\eta/16,$ and $\epsilon_9=\eta/2.$ As long as
	\be \label{eq:baby constraint}
	\frac{\eta}{16}< \epsilon_6<1-\frac{\eta}{6},\quad\quad 0<\eta<4
	\ee 
	are satisfied, all constraints are met in Eq. \eqref{eq:epsilons contrains 1}. Similarly, Eq. \eqref{eq:epsilons contrains 2} is satisfied if
	\ba
	\frac{\ln 4}{\ln d}-\frac{\eta}{2}<\epsilon_6\leq 1-\frac{\eta}{2}.
	\ea
	Hence Eqs. \eqref{eq:epsilons contrains 1}, \eqref{eq:baby constraint} are satisfied for sufficiently large $d$ if 
	\be \label{eq:baby constraint 2}
	\eta< \epsilon_6\leq 1-\frac{\eta}{2},
	\ee 
	for $\eta>0$. Now make the substitution $\eta=\eta'/2$. Eq. \eqref{eq:main thrm eq almost} becomes
	\be\label{eq:R  wioth constant factor}
	{R_1}\geq \frac{d^{3\eta'/4}}{12}\left(\frac{d^{2-\eta'}}{\sigma^2}+\lo\left(\frac{d^{2-\eta'}}{\sigma^2}\right)\right)\geq \frac{d^{2-\eta'}}{\sigma^2}+\lo\left(\frac{d^{2-\eta'}}{\sigma^2}\right),
	\ee
	with constraint 
	\be \label{eq:baby constraint 3}
	d^{\frac{\eta'}{2}}< \sigma\leq d^{1-\frac{\eta'}{4}},\quad \eta'>0,
	\ee
	where we have used Eq. \eqref{eq:baby constraint 2} and recalled $\sigma=d^{\epsilon_6}$. Finally, note that w.l.o.g., we can replace the upper bound on $\sigma$ in Eq. \eqref{eq:baby constraint 3} by $< d$. This is because of two observations: 1) by choosing $\eta'$ arbitrarily close to zero, $\sigma$ is upper bounded by a number arbitrarily close to $d$. 2) The lower bound on ${R_1}$ in Eq. \eqref{eq:R  wioth constant factor} is monotonically decreasing w.r.t. $\eta'$. To finalise the proof, we simply re-name $\eta'$ by $\eta$.
\end{proof}




\end{document}